\documentclass{llncs}

\pdfoutput=1

\usepackage{xspace}
\usepackage{xcolor}
\usepackage{gensymb}
\usepackage[english]{babel}
\usepackage{graphicx}
\usepackage{amssymb}
\usepackage{enumerate}
\usepackage{amsmath}
\usepackage{fullpage}

\newcommand{\constant}{10}

\newtheorem{claimx}{Claim}

\renewenvironment{proof}
{{\bf Proof:}}{\hspace*{\fill}$\Box$\par\vspace{2mm}}

\begin{document}
\title{Planar Rectilinear Drawings of Outerplanar Graphs\\ in Linear Time\thanks{Partially supported by MIUR Project ``AHeAD'' under PRIN 20174LF3T8, by H2020-MSCA-RISE Proj.\ ``CONNECT'' n$^\circ$ 734922, and by Roma Tre University Azione 4 Project ``GeoView''.}}
\author{Fabrizio Frati}
\institute{Roma Tre University, Italy -- \email{frati@dia.uniroma3.it}\\
}
\maketitle

\begin{abstract}
We show how to test in linear time whether an outerplanar graph admits a planar rectilinear drawing, both if the graph has a prescribed plane embedding that the drawing has to respect and if it does not. Our algorithm returns a planar rectilinear drawing if \mbox{the graph admits one}. 
\end{abstract}

\section{Introduction}\label{le:intro}

The problem of constructing planar orthogonal graph drawings with a minimum number of bends has been studied for decades. In 1987, Tamassia~\cite{t-eggmnb-87} proved that, for a planar graph with a prescribed plane embedding, a planar orthogonal drawing that has the minimum number of bends and that respects the prescribed plane embedding can be constructed in polynomial time, thereby answering a question of Storer~\cite{st-ncmegpg-80} and establishing a result that lies at the very foundations of the graph drawing research area. The running time of Tamassia's algorithm is $O(n^2 \log n)$ for a graph with $n$ vertices. This bound has been improved to $O(n^{7/4}\sqrt{\log n})$ by Garg and Tamassia~\cite{DBLP:conf/gd/GargT96a} and then to $O(n^{3/2})$ by Cornelsen and Karrenbauer~\cite{ck-abm-12}. However, achieving a linear running time is still an elusive goal. 

Bend minimization in the variable embedding setting is a much harder problem; indeed, Garg and Tamassia~\cite{gt-ccurpt-01} proved that testing whether a graph admits a planar orthogonal drawing with zero bends is NP-hard. However, there are some natural restrictions on the input that make the problem tractable. A successful story is the one about degree-$3$ planar graphs. In~1998, Di Battista et al.~\cite{dlv-sood-98} proved that a planar orthogonal drawing with the minimum number of bends for an $n$-vertex planar graph with maximum degree $3$ can be constructed in $O(n^5\log n)$ time. After some improvements~\cite{cy-bmod-17,dlp-bmodqt-18}, a recent breakthrough result by Didimo et al.~\cite{dlop-ood-20} has shown that $O(n)$ time is indeed sufficient. Di Battista et al.~\cite{dlv-sood-98} also presented an $O(n^4)$-time algorithm for minimizing the number of bends in a planar orthogonal drawing of an $n$-vertex $2$-connected series-parallel graph. This result was first extended to not necessarily $2$-connected series-parallel graphs by Bl\"asius et al.~\cite{blr-odie-16} and then improved to an $O(n^3 \log n)$ running time by Di Giacomo et al.~\cite{glm-sr-19}.

Evidence has shown that the bend-minimization problem is not much easier if one is only interested in the construction of planar orthogonal drawings with zero bends; these are also called \emph{planar rectilinear drawings} (see Figures~\ref{fig:introduction}(a) and~\ref{fig:introduction}(b) for two such drawings). Namely, the cited NP-hardness proof of Garg and Tamassia~\cite{gt-ccurpt-01} is designed for planar rectilinear drawings. Further, almost every efficient algorithm for testing the existence of planar rectilinear drawings has been eventually subsumed by an algorithm in the more general bend-minimization scenario. This was indeed the case for degree-$3$ series-parallel graphs with a variable embedding (see~\cite{ren-nb-06} and~\cite{nz-odspg-08}), for degree-$3$ planar graphs with a fixed embedding (see~\cite{rnn-odpwb-03} and~\cite{rn-bmod-02}), and for degree-$3$ planar graphs with a variable embedding (see~\cite{hr-nbocd-19} and~\cite{dlop-ood-20}). A notable exception is that of planar graphs with a fixed embedding, for which the fastest known algorithm for the bend-minimization problem runs in $O(n^{3/2})$ time (this is the already cited result by Cornelsen and Karrenbauer~\cite{ck-abm-12}), while the fastest known algorithm for the rectilinear-planarity testing problem runs in $O(n \log^3 n)$ time; the latter algorithm can be obtained by applying recent results of Borradaile et al.~\cite{bkmnw-msms-17} in the Tamassia's flow network formulation of the bend-minimization problem~\cite{t-eggmnb-87}. 

\begin{figure}[tb]\tabcolsep=4pt
	\centering
	\begin{tabular}{c c c}
		\includegraphics[scale=0.7]{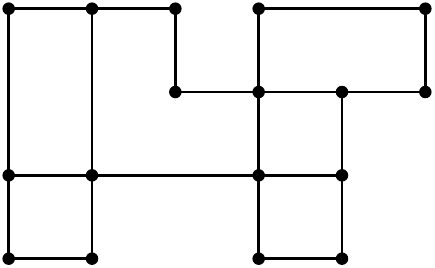} \hspace{3mm} &
		\includegraphics[scale=0.7]{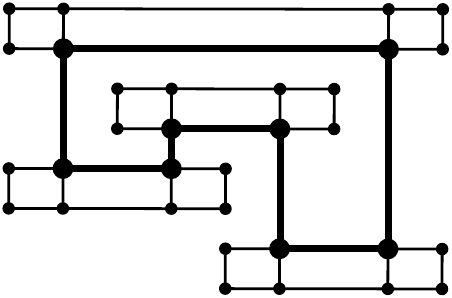} \hspace{3mm} &
		\includegraphics[scale=0.7]{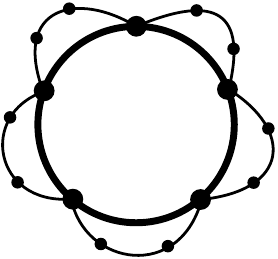} \\
		(a) \hspace{3mm} & (b) \hspace{3mm} & (c)\\
	\end{tabular}
	\caption{(a) An outerplanar rectilinear drawing. (b) A planar rectilinear drawing of an outerplanar graph $G$. The graph $G$ has no outerplanar rectilinear drawing. Namely, if all the edges incident to the vertices of the $8$-cycle $\mathcal C$, shown by thick lines and large disks, would lie outside the polygon $\mathcal P$ representing $\mathcal C$, then the sum of the angles in the interior of $\mathcal P$ would be $8 \cdot 90\degree=720\degree$, while it has to be $(8-2)\cdot 180\degree=1080\degree$. (c) An outerplanar graph $G$ that has no $3$-cycle and no planar rectilinear drawing. Namely, let $\mathcal C$ be the $5$-cycle of $G$. In any planar rectilinear drawing of $G$, the sum of the angles in the interior of the polygon $\mathcal P$ representing $\mathcal C$ would be $(5+2k)\cdot 90\degree$, where $k$ is the number of $4$-cycles of $G$ inside $\mathcal P$. However, the sum of the angles in the interior of $\mathcal P$ has to be $(5-2)\cdot 180\degree=6 \cdot  90\degree$ and $5+2k\neq 6$ for every integer $k$.}
	\label{fig:introduction}
\end{figure}

In this paper, we show that the existence of a planar rectilinear drawing can be tested in linear time for an outerplanar graph; further, if a planar rectilinear drawing exists, then our algorithm constructs one such a drawing in linear time. Our result does not assume any restriction on the maximum degree of the outerplanar graph; further, it covers both the fixed embedding scenario, where the previously best known running time for a testing algorithm was $O(n \log^3 n)$~\cite{bkmnw-msms-17,t-eggmnb-87}, and the variable embedding scenario, where the previously best known running time for a testing algorithm was $O(n^3 \log n)$~\cite{glm-sr-19}. We also present a linear-time algorithm that tests whether an outerplanar graph admits an outerplanar rectilinear drawing and constructs one such a drawing if it exists. 

Given how common it is to study outerplanar graphs for a problem which is too difficult or too computationally expensive on general planar graphs, it is quite surprising that a systematic study of planar orthogonal drawings and of planar rectilinear drawings of outerplanar graphs has not been done before. The only result we are aware of that is tailored for outerplanar graphs is the one by Nomura et al.~\cite{ntu-odog-05}, which states that an outerplanar graph with maximum degree $3$ admits a planar rectilinear drawing if and only if it does not contain any $3$-cycle. This characterization is not true for outerplanar graphs with vertices of degree $4$; see, e.g., Figure~\ref{fig:introduction}(c). 

We outline our algorithm that tests for the existence of a planar rectilinear drawing of an outerplanar graph with a variable embedding. 

The first, natural, idea is to reduce the problem to the case in which the input graph is $2$-connected. This reduction builds on (an involved version of) a technique introduced by Didimo et al.~\cite{dlop-ood-20} that, roughly speaking, allows us to perform postorder traversals of the block-cut-vertex tree of the graph in total linear time so that each edge is traversed in both directions; during these traversals, information is computed that allows us to decide whether solutions for the subproblems associated to the blocks of the graph can be combined into a solution for the entire graph. This reduction to the $2$-connected case comes at the expense of having to solve a harder problem, in which some vertices of the graph have restrictions on their incident angles in the sought planar rectilinear drawing.


An analogous technique allows us to reduce the problem to the case in which the input graph is $2$-connected and has a prescribed edge that is required to be incident to the outer face in the sought planar rectilinear drawing. The role that in the previous reduction is played by the block-cut-vertex tree is here undertaken by the ``extended dual tree'' of the outerplanar graph. Each edge of this tree is dual to an edge of the outerplanar graph; the latter edge splits the outerplanar graph into two smaller outerplanar graphs. These smaller outerplanar graphs define the sub-instances whose solutions might be combined into a solution for the entire graph; whether this combination is possible is decided based on information that is computed during the traversals of the extended dual tree.

The core of our algorithm consists of an efficient solution for the problem of testing whether a $2$-connected outerplanar graph admits a planar rectilinear drawing in which a prescribed edge is required to be incident to the outer face. The starting point for our solution is a characterization of the positive instances of the problem in terms of the existence of a sequence of numerical values satisfying some conditions; these values represent certain geometric angles in the sought planar rectilinear drawing. Some of these numerical values can be chosen optimally, based on recursive solutions to smaller subproblems; further, a constant number of these numerical values have to be chosen in all possible ways; finally, we reduce the problem of finding the remaining numerical values to the one of testing for the existence of a set of integers, each of which is required to be in a certain interval, so that a linear equation on these integers is satisfied. We characterize the solutions to the latter problem so that not only it can be solved efficiently, but a solution can be modified in constant time if the interval associated to each integer changes slightly; this change corresponds to a different edge chosen to be incident to the outer face.

\begin{table}[!tb]\footnotesize
	\centering
	\linespread{1.2}
	\begin{tabular}{|c|c|c|c|c|c|c|c|c|}
		\cline{2-9}
		\multicolumn{1}{c|}{} & \multicolumn{4}{c|}{Fixed Embedding}  & \multicolumn{4}{c|}{Variable Embedding} \\
		\cline{2-9}
		\multicolumn{1}{c|}{} & \emph{$\Delta=3$} & Ref. & \emph{$\Delta=4$} & Ref. & \emph{$\Delta=3$} & Ref. & \emph{$\Delta=4$} & Ref. \\
		\hline
		{\em Planar } & $O(n)$ & \cite{rnn-odpwb-03} & $O(n \log^3 n)$ & \cite{bkmnw-msms-17} & $O(n)$ & \cite{hr-nbocd-19} & NP-hard & \cite{gt-ccurpt-01}\\
		\hline
		{\em Series-Parallel } & $O(n)$ & \cite{rnn-odpwb-03} & $O(n \log^3 n)$ & \cite{bkmnw-msms-17} & $O(n)$ & \cite{ren-nb-06} & $O(n^3 \log n)$ & \cite{glm-sr-19} \\
		\hline
		{\em Outerplanar } & $O(n)$ & \cite{rnn-odpwb-03}  & $O(n)$ & This paper & $O(n)$ & \cite{ntu-odog-05} & $O(n)$ & This paper\\
		\hline
	\end{tabular}
	\vspace{2mm}
	\caption{\small Complexity of testing for the existence of a planar rectilinear drawing. The letter $\Delta$ denotes the maximum degree of the graph. Each reference shows the first paper to achieve the corresponding bound. The bound for series-parallel graphs with maximum degree $4$ and with fixed embedding is actually $O(n)$ if $2$-connectivity is assumed~\cite{dklo-rpt-20}.}
	\label{ta:straight-line}
\end{table}

Table~\ref{ta:straight-line} shows the best known time bounds for deciding whether a graph admits a planar rectilinear~drawing. Simultaneously with our paper, another paper on the rectilinear-planarity testing problem appeared. Namely, Didimo et al.\ presented in~\cite{dklo-rpt-20} an $O(n)$-time algorithm which tests whether an $n$-vertex $2$-connected series-parallel graph with fixed embedding admits a planar rectilinear drawing and, in the positive case, constructs one such drawing; the work of Didimo et al.\ is based on techniques which are different from those we use in this paper.

Throughout the paper, we assume every considered graph to be connected. This is not a loss of generality, as a graph admits a planar rectilinear drawing if and only if every connected component of it does. Further, we assume that the maximum degree of any vertex is $4$, as if it is larger than that, then the graph has no planar rectilinear drawing; the existence of such high-degree vertices can be clearly tested in linear time.  

The rest of the paper is organized as follows. In Section~\ref{se:preliminaries} we introduce some preliminaries. In Section~\ref{se:fixed} we present an algorithm that tests for the existence of a planar rectilinear drawing of an outerplanar graph with a prescribed plane embedding; we also show how this algorithm allows us to test for the existence of an outerplanar rectilinear drawing of an outerplanar graph. In Section~\ref{se:variable} we present an algorithm that decides whether an outerplanar graph with a variable embedding admits a planar rectilinear drawing. Finally, in Section~\ref{se:conclusions}, we conclude and present some open problems. 


\section{Preliminaries} \label{se:preliminaries}

A {\em cut-vertex} of a connected graph $G$ is a vertex whose removal disconnects $G$. A graph is {\em $2$-connected} if it has no cut-vertex. A {\em block} of a graph $G$ is a maximal (in terms of vertices and edges) $2$-connected subgraph of $G$; a block is {\em trivial} if it is a single edge, it is {\em non-trivial} otherwise. The {\em block-cut-vertex tree} $T$ of a connected graph $G$~\cite{h-gt-69,ht-aeagm-73} is the tree that has a {\em B-node} for each block of $G$ and a {\em C-node} for each cut-vertex of $G$; a B-node $b$ and a C-node $c$ are adjacent in $T$ if the cut-vertex $c$ lies in the block corresponding to $b$ (we often identify a C-node of $T$ and the corresponding cut-vertex of $G$).

A {\em drawing} of a graph maps each vertex to a point in the plane and each edge to a curve between its endpoints. A drawing is {\em planar} if no two edges intersect, except at common endpoints, and it is \emph{rectilinear} if each edge is either a horizontal or a vertical segment. A planar drawing divides the plane into topologically connected regions, called {\em faces}; the only unbounded face is the {\em outer face}, while all the other faces are {\em internal}. A graph is {\em planar} if it admits a planar drawing. 


Two planar drawings $\Gamma_1$ and $\Gamma_2$ of a connected planar graph $G$ are {\em equivalent} if: (i) for each vertex $w$ of $G$, the clockwise order of the edges incident to $w$ is the same in $\Gamma_1$ and $\Gamma_2$; and (ii) the clockwise order of the edges incident to the outer face is the same in $\Gamma_1$ and $\Gamma_2$. A {\em plane embedding} is an equivalence class of planar drawings. Two drawings that correspond to the same plane embedding have faces delimited by the same walks; for this reason, we often speak about \emph{faces of a plane embedding}. We denote by $f^*_{\mathcal E}$ the outer face of a plane embedding $\mathcal E$. The \emph{reflection} of a plane embedding $\mathcal E$ is the plane embedding that is obtained from $\mathcal E$ by (i) inverting the clockwise order of the edges incident to each vertex, and (ii) inverting the clockwise order of the edges incident to the outer face. Given a plane embedding $\mathcal E$ of a planar graph $G$ and given a subgraph $G'$ of $G$, the \emph{restriction} of $\mathcal E$ to $G'$ is the plane embedding $\mathcal E'$ defined as follows. Consider any planar drawing $\Gamma$ within the equivalence class $\mathcal E$. Remove from $\Gamma$ the vertices and edges not in $G'$, obtaining a planar drawing $\Gamma'$ of $G'$. Then $\mathcal E'$ is the equivalence class of $\Gamma'$. We also say that a face $f'$ of $\mathcal E'$ \emph{corresponds} to a face $f$ of $\mathcal E$ if the region $f$ in $\Gamma$ is a subset of the region $f'$ in $\Gamma'$.

Two planar rectilinear drawings of a $2$-connected planar graph $G$ are {\em equivalent} if they correspond to the same plane embedding $\mathcal E$ of $G$ and if, for every face $f$ of $\mathcal E$ and for every vertex $w$ incident to $f$, the angle at $w$ in the interior of $f$ is the same in both drawings. Then a {\em rectilinear representation} of $G$ is an equivalence class of planar rectilinear drawings of $G$. Thus, a rectilinear representation of a planar graph $G$ is a pair $(\mathcal E,\phi)$, where $\mathcal E$ is a plane embedding of $G$ and $\phi$ is a function assigning an {\em angle} $\phi(w,f)\in \{90\degree,180\degree,270\degree\}$ to every pair $(w,f)$ such that $w$ is a vertex incident to $f$ in $\mathcal E$ (see Fig.~\ref{fig:outerplanarrepresentation}(a)). 

For connected planar graphs that are not $2$-connected, the notions of equivalence and of rectilinear representation are similar, however since a vertex $w$ might have several occurrences $w^1,\dots,w^x$ on the boundary of a face $f$, the function $\phi$ assigns an angle to every pair $(w^k,f)$, for $k\in \{1,\dots,x\}$; further, the value $360\degree$ is admissible for an angle $\phi(w^k,f)$.  

Tamassia~\cite{t-eggmnb-87} proved that\footnote{We present here a specialized version of Tamassia's characterization which applies to rectilinear drawings, while his result deals, more in general, with orthogonal drawings that might possibly have bends.} a pair $(\mathcal E,\phi)$ defines a rectilinear representation of a $2$-connected planar graph $G$ if and only if the following two conditions are satisfied: 

\begin{figure}[tb]\tabcolsep=4pt
	\centering
	\begin{tabular}{c c}
		\includegraphics[scale=0.7]{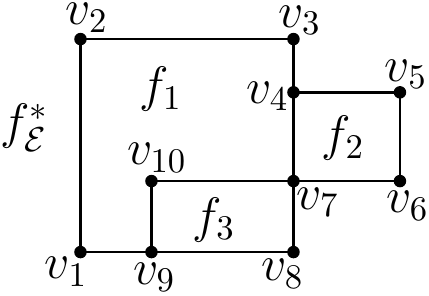} \hspace{3mm} &
		\includegraphics[scale=0.7]{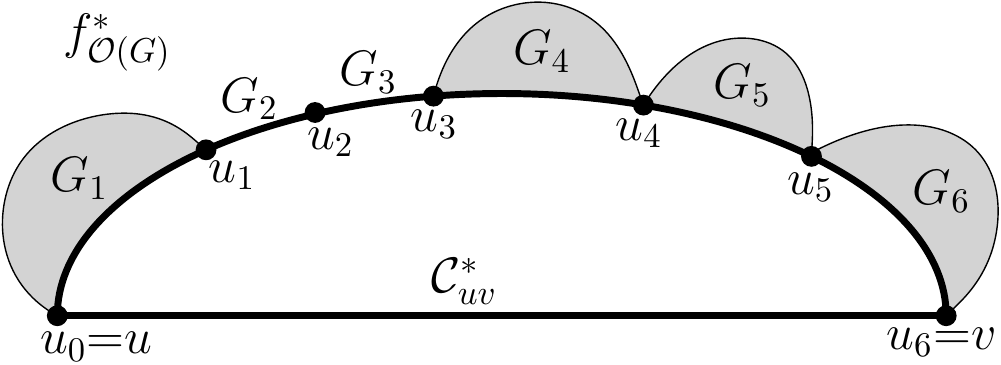} \\
		(a) \hspace{3mm} & (b)\\
	\end{tabular}
	\caption{(a) A planar rectilinear drawing. The corresponding rectilinear representation $(\mathcal E,\phi)$ has: (i) $\phi(w,f^*_{\cal E})=90\degree$ for $w\in \{v_4,v_7\}$, $\phi(v_9,f^*_{\cal E})=180\degree$, $\phi(w,f^*_{\cal E})=270\degree$ for $w\in \{v_1,v_2,v_3,v_5,v_6,v_8\}$; (ii) $\phi(w,f_1)=90\degree$ for $w\in \{v_1,v_2,v_3,v_7,v_9\}$, $\phi(v_4,f_1)=180\degree$, $\phi(v_{10},f_1)=270\degree$; (iii) $\phi(w,f_2)=90\degree$ for $w\in \{v_4,v_5,v_6,v_7\}$; and (iv) $\phi(w,f_3)=90\degree$ for $w\in \{v_7,v_8,v_9,v_{10}\}$. (b) A rooted $2$-connected outerplanar graph $G$. The cycle ${\mathcal C}^*_{uv}$ is represented by a thick line. The interior of the $uv$-subgraphs of $G$ is colored gray; $G_2$ and $G_3$ are trivial, while $G_1$, $G_4$, $G_5$, and $G_6$ are not.}
	\label{fig:outerplanarrepresentation}
\end{figure}

\begin{enumerate}[(1)]
	\item for every vertex $w$ in $G$, we have that $\sum_f \phi(w,f)=360\degree$, where the sum is over all the faces $f$ of $\mathcal E$ incident to $w$; and
	\item for every face $f$ of $\mathcal E$, we have that $\sum_w (2-\phi(w,f)/90\degree)=+4$ or $\sum_w (2-\phi(w,f)/90\degree)=-4$, depending on whether $f$ is an internal face or the outer face of $\mathcal E$, respectively, where the sum is over all the vertices $w$ incident to $f$.
\end{enumerate}

For a simply-connected planar graph $G$, the above characterization is as follows:

\begin{enumerate}[(1)]
	\item for every vertex $w$ in $G$, we have that $\sum_{f,k} \phi(w^k,f)=360\degree$, where the sum is over all the faces $f$ incident to $w$ and all the occurrences $w^k$ of $w$ on the boundary of $f$; and
	\item for every face $f$ of $\mathcal E$, we have that $\sum_{w,k} (2-\phi(w^k,f)/90\degree)=+4$ or $\sum_{w,k} (2-\phi(w^k,f)/90\degree)=-4$, depending on whether $f$ is an internal face or the outer face of $\mathcal E$, respectively, where the sums are over all the occurrences $w^k$ on the boundary of $f$ of each vertex $w$ incident to $f$.
\end{enumerate}

Tamassia~\cite{t-eggmnb-87} also proved that, given a rectilinear representation $(\mathcal E, \phi)$ of an $n$-vertex  planar graph $G$, a planar rectilinear drawing of $G$ within the equivalence class $(\mathcal E,\phi)$ can be constructed in $O(n)$ time. This allows us to often shift our attention from planar rectilinear drawings to rectilinear representations. 

Let $G$ be a graph and $G'$ be a subgraph of $G$. Given a rectilinear representation $(\mathcal E,\phi)$ of $G$, the \emph{restriction} of $(\mathcal E,\phi)$ to $G'$ is the rectilinear representation $(\mathcal E',\phi')$ of $G'$ defined as follows. Consider any planar rectilinear drawing $\Gamma$ within the equivalence class $(\mathcal E,\phi)$. Remove from $\Gamma$ the vertices and edges not in $G'$, obtaining a planar rectilinear drawing $\Gamma'$ of $G'$. Then $(\mathcal E',\phi')$ is the equivalence class of $\Gamma'$. More in general, let $(\mathcal E,\phi)$ be a pair, where $\mathcal E$ is a plane embedding of $G$ and $\phi$ is a function that assigns an angle $\phi(w^k,f) \in \{90\degree,180\degree,270\degree,360\degree\}$ to every occurrence $w^k$ of a vertex $w$ of $G$ along the boundary of a face $f$ of $\mathcal E$; note that $(\mathcal E,\phi)$ is not necessarily a rectilinear representation of $G$. Then the \emph{restriction} $(\mathcal E',\phi')$ of $(\mathcal E,\phi)$ to $G'$ is defined as follows (see Figure~\ref{fig:restriction}). First, $\mathcal E'$ is the restriction of $\mathcal E$ to $G'$. Second, consider any occurrence $w^*$ of a vertex $w$ of $G'$ along the boundary of a face $f'$ of $\mathcal E'$ and let $wu$ and $wv$ be the edges defining $w^*$, where $w^*$ is to the left of the path $(u,w,v)$. Let $e_1=wu,e_2,\dots,e_h=wv,\dots,e_k,e_{k+1}=e_1$ be the clockwise order of the edges incident to $w$ in $\mathcal E$ and let $w^{1},\dots,w^{k}$ be the occurrences of $w$ along the boundaries of faces of $\mathcal E$, where $w^i$ is defined by the edges $e_i$ and $e_{i+1}$. Finally, let $f_1,\dots,f_k$ be the faces of $\mathcal E$ the occurrences $w^{1},\dots,w^{k}$ appear along, respectively. Then we have $\phi'(w^*,f')=\sum_{i=1}^{h-1}(w^i,f_i)$.

\begin{figure}[tb]\tabcolsep=4pt
	\centering
	\begin{tabular}{c c}
		\includegraphics[scale=0.7]{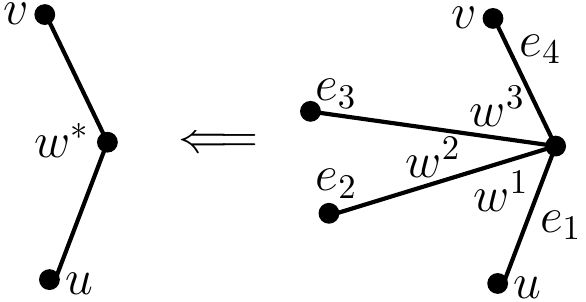}
	\end{tabular}
	\caption{Illustration for the definition of the restriction $(\mathcal E', \phi')$ of a pair $(\mathcal E,\phi)$ to a subgraph $G'$ of $G$.}
	\label{fig:restriction}
\end{figure}

Given a rectilinear representation $(\mathcal E, \phi)$ of a planar graph $G$ and given a vertex $u$ incident to $f^*_{\mathcal E}$, we denote by $\phi^{\mathrm{int}}(u)$ the sum of the \emph{internal angles} incident to $u$, that is, $\phi^{\mathrm{int}}(u)=\sum_f \phi(u,f)$ over all the internal faces $f$ of $\mathcal E$ incident to $u$. Note that $\phi^{\mathrm{int}}(u)=360\degree -\phi(u,f^*_{\mathcal E})$. 

An {\em outerplanar drawing} is a planar drawing such that all the vertices are incident to the outer face. 
An {\em outerplane embedding} is a plane embedding such that all the vertices are incident to the outer face.
A graph is {\em outerplanar} if it admits an outerplanar drawing. It will be handy to work with {\em rooted} outerplanar graphs (see Fig.~\ref{fig:outerplanarrepresentation}(b)): Given a $2$-connected outerplanar graph $G$ with outerplane embedding $\mathcal O$, select any edge $uv$ incident to $f^*_{\mathcal O}$ as the {\em root} of $G$, where $u$ immediately precedes $v$ in counterclockwise order along the boundary of $f^*_{\mathcal O}$. Denote by ${\mathcal C}^*_{uv}$ the cycle delimiting the internal face of $\mathcal O$ incident to $uv$. The blocks of the graph obtained from $G$ by removing the edge $uv$ are the {\em $uv$-subgraphs} of $G$. The root of each $uv$-subgraph of $G$ is its unique edge in ${\mathcal C}^*_{uv}$; note that each $uv$-subgraph can be either trivial or non-trivial. 

The \emph{extended dual tree} $\mathcal T$ of an outerplane embedding $\mathcal O$ of a $2$-connected outerplanar graph is defined as follows. The {\em dual graph} $\mathcal D$ of $\mathcal O$ has a vertex for each face of $\mathcal O$ and has an edge between two vertices if the corresponding faces have a common edge on their boundaries. Then $\mathcal T$ is obtained from $\mathcal D$ by removing the vertex of $\mathcal D$ corresponding to $f^*_{\mathcal O}$ and by letting its incident edges end in $n$ new degree-$1$ nodes (see Fig.~\ref{fig:outerplanardual}). 	

\begin{figure}[tb]
	\centering
	\includegraphics[scale=0.7]{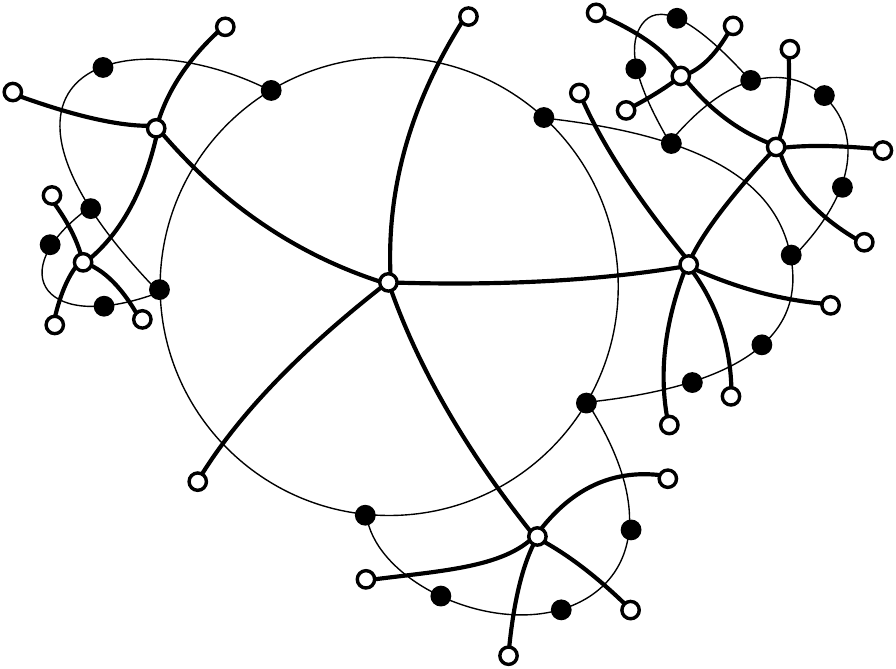}
	\caption{The extended dual tree $\mathcal T$ of an outerplane embedding $\mathcal O$ of a $2$-connected outerplanar graph; the nodes and the edges of $\mathcal T$ are represented by white disks and thick lines, respectively.}
	\label{fig:outerplanardual}
\end{figure}

\subsection{Join of Rectilinear Representations} \label{sse:join}

When dealing with an outerplanar graph $G$ containing cut-vertices, we often split $G$ along a cut-vertex~$c$, construct rectilinear representations of the resulting subgraphs, all of which contain $c$, and then combine such rectilinear representations into a rectilinear representation of $G$. We now introduce a tool to do this combination, called \emph{join}; this applies to graphs that are not necessarily outerplanar.

Let $G$ be a connected graph, let $c$ be a cut-vertex of $G$, and let $T$ be the block cut-vertex tree of $G$. Let $T_1,\dots,T_k$ be the trees obtained by removing $c$ from $T$ and, for $i=1,\dots,k$, let $G_i$ be the subgraph of $G$ composed of the blocks of $G$ corresponding to B-nodes in $T_i$; note that any two distinct graphs $G_i$ and $G_j$ share $c$ and no other vertex or edge. Finally, let $(\mathcal E_1,\phi_1),\dots,(\mathcal E_k,\phi_k)$ be rectilinear representations of $G_1,\dots,G_k$, respectively. 
A \emph{join} of $(\mathcal E_1,\phi_1),\dots,(\mathcal E_k,\phi_k)$ is a pair $(\mathcal E,\phi)$ satisfying the following properties. 

\begin{enumerate}[(a)]
	\item $\mathcal E$ is a plane embedding of $G$.
	\item For every occurrence $c^x$ of $c$ along the boundary of a face $f$ of $\mathcal E$, we have $\phi(c^x,f)\in \{90\degree,180\degree,270\degree\}$.
	\item $\sum_{f,x} \phi(c^x,f)=360\degree$, where the sum is over all the faces $f$ of $\mathcal E$ incident to $c$ and all the occurrences $c^x$ of $c$ on the boundary of $f$.
	\item For $i=1,\dots,k$, the restriction of $(\mathcal E,\phi)$ to $G_i$ is $(\mathcal E_i,\phi_i)$.
\end{enumerate}

We have the following.

\begin{lemma} \label{le:preliminaries-subgraphs-composition}
	A join $(\mathcal E,\phi)$ of $(\mathcal E_1,\phi_1),\dots,(\mathcal E_k,\phi_k)$ is a rectilinear representation of $G$.  
\end{lemma}

\begin{proof}
	By Property~(a) of $(\mathcal E,\phi)$, we have that $\mathcal E$ is a plane embedding, hence we can talk about its faces. By Property~(d) of $(\mathcal E,\phi)$, we have that, for $i=1,\dots,k$, the restriction of $(\mathcal E,\phi)$ to $G_i$ is $(\mathcal E_i,\phi_i)$; for a face $f$ of $\mathcal E$ and for any $i\in \{1,\dots,k\}$, let $f_i$ be the face of $\mathcal E_i$ corresponding to $f$. 
	
	Consider any vertex $w$ of $G$ and consider any occurrence $w^x$ of $w$ along the boundary of a face $f$ of $\mathcal E$. If $w=c$, then we have $\phi(w^x,f)\in \{90\degree,180\degree,270\degree\}$ by Property~(b) of $(\mathcal E,\phi)$. If $w\neq c$, let $G_i$ be the graph among $G_1,\dots,G_k$ the vertex $w$ belongs to; then $\phi(w^x,f)\in \{90\degree,180\degree,270\degree,360\degree\}$, given that $\phi(w^x,f)=\phi_i (w^x,f_i)$, by Property~(d) of $(\mathcal E,\phi)$, and given that $(\mathcal E_i,\phi_i)$ is a rectilinear representation of $G_i$. 
	
	
	We now prove that $(\mathcal E,\phi)$ satisfies condition (1) of Tamassia's characterization~\cite{t-eggmnb-87}. Consider any vertex $w$ of $G$. If $w=c$, then, by Property~(c) of $(\mathcal E,\phi)$, we have $\sum_{f,x} \phi(c^x,f)=360\degree$, as requested, where the sum is over all the faces $f$ of $\mathcal E$ incident to $c$ and all the occurrences $c^x$ of $c$ on the boundary of $f$. If $w\neq c$, let $G_i$ be the graph among $G_1,\dots,G_k$ the vertex $w$ belongs to; then $\sum_{f,x} \phi(w^x,f)=360\degree$, where the sum is over all the faces $f$ incident to $w$ and all the occurrences $w^x$ of $w$ on the boundary of $f$, given that $\phi(w^x,f)=\phi_i (w^x,f_i)$, by Property~(d) of $(\mathcal E,\phi)$, and given that $(\mathcal E_i,\phi_i)$ is a rectilinear representation of $G_i$.

	We next prove that $(\mathcal E,\phi)$ satisfies condition (2) of Tamassia's characterization~\cite{t-eggmnb-87}. For every face $f$ of $\mathcal E$, we need to prove the following condition. Let $S:=\sum_{w,x} (2-\phi(w^x,f)/90\degree)$, where the sum is over all the occurrences $w^x$ on the boundary of $f$ of each vertex $w$ incident to $f$. Then $S=+4$ or $S=-4$, depending on whether $f$ is an internal face or the outer face of $\mathcal E$, respectively. 
	
	\begin{itemize}
		\item Suppose first that $f$ is only incident to edges of $G_i$, for some $i\in \{1,\dots,k\}$. By Property~(d) of $(\mathcal E,\phi)$, for each occurrence $w^x$ of a vertex $w$ along the boundary of $f$, we have $\phi(w^x,f)=\phi_i(w^x,f_i)$, and hence $S=\sum_{w,x} (2-\phi_i(w^x,f_i)/90\degree)$. Since $(\mathcal E_i,\phi_i)$ is a rectilinear representation of $G_i$, we have that $\sum_{w,x} (2-\phi_i(w^x,f_i)/90\degree)=+4$ or $\sum_{w,x} (2-\phi_i(w^x,f_i)/90\degree)=-4$, depending on whether $f_i$ is an internal face or the outer face of $\mathcal E_i$, respectively.
		\item Suppose next that $f$ is incident to edges of $h$ graphs among $G_1,\dots,G_k$, for some $h\geq 2$; this implies that $c$ is incident to $f$. Denote by $G_{\lambda(1)},G_{\lambda(2)},\dots,G_{\lambda(h)}$ the graphs containing edges incident to $f$. Then the faces of $\mathcal E_{\lambda(1)},\mathcal E_{\lambda(2)},\dots,\mathcal E_{\lambda(h)}$ corresponding to $f$ are $f_{\lambda(1)},f_{\lambda(2)},\dots,f_{\lambda(h)}$, respectively. Note that $c$ has one occurrence along the boundary of $f_{\lambda(i)}$, for each $i=1,\dots,h$; namely, it has at most one occurrence since it is not a cut-vertex of $G_{\lambda(i)}$ and it has at least one occurrence as otherwise $G_{\lambda(i)}$ would not contain any edges incident to $f$. For $i=1,\dots,h$, let $S_{\lambda(i)}=\sum_{w,x} (2-\phi_{\lambda(i)}(w^x,f_{\lambda(i)})/90\degree)$, where the sum is over all the occurrences $w^x$ on the boundary of $f_{\lambda(i)}$ of each vertex $w$ of $G_{\lambda(i)}$ incident to $f_{\lambda(i)}$.

		\begin{figure}[tb]
			\centering
			\includegraphics[scale=0.7]{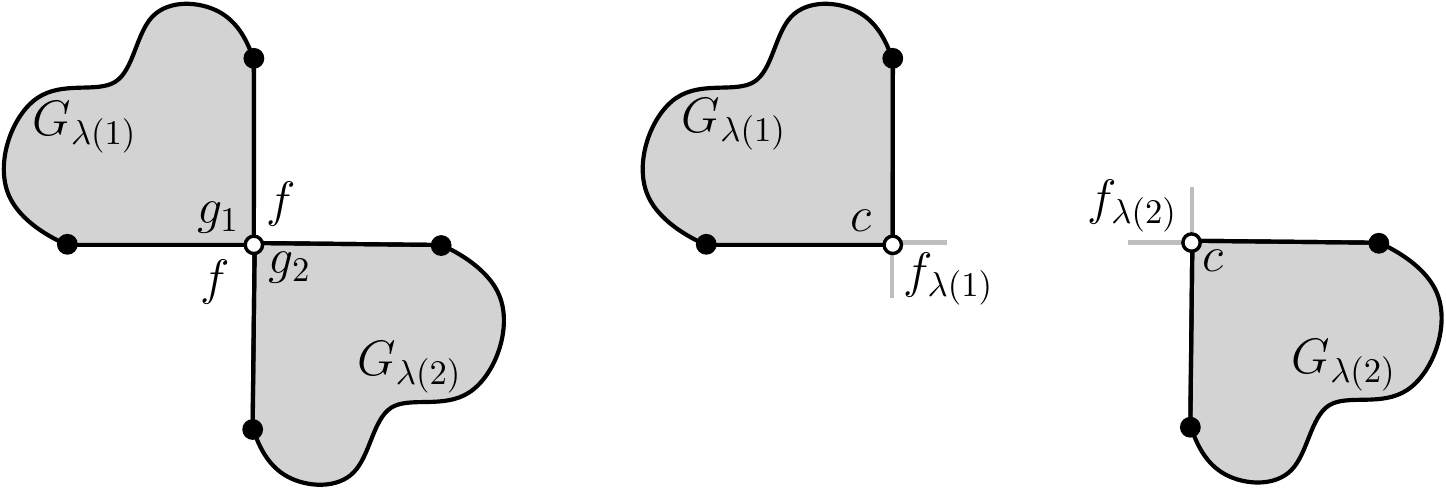}
			\caption{Illustration for the proof of the equality $\sum_{i=1}^h \phi_{\lambda(i)}(c,f_{\lambda(i)})=(h-1)\cdot 360\degree + \sum_{x}\phi(c^x,f)$. In this example $h=2$. The figure shows $\mathcal E$ (left), $\mathcal E_{\lambda(1)}$ (center), and $\mathcal E_{\lambda(2)}$ (right). Note that, for $i=1,2$, the face of $\mathcal E_{\lambda(i)}$ corresponding to $f$ is $f_{\lambda(i)}$. On the other hand, for every face $g\neq f$ of $\mathcal E$ incident to $c$, there exists an index $j\in \{1,\dots,h\}$ such that the face of $\mathcal E_{\lambda(j)}$ corresponding to $g$ is not $f_{\lambda(j)}$; in this example, the face of $\mathcal E_{\lambda(1)}$ corresponding to $g_1$ is not $f_{\lambda(1)}$ and the face of $\mathcal E_{\lambda(2)}$ corresponding to $g_2$ is not $f_{\lambda(2)}$.}
			\label{fig:join}
		\end{figure}
		
		We now prove the following auxiliary equality (refer to Figure~\ref{fig:join}): 
		\begin{eqnarray*}
			\sum_{i=1}^h \phi_{\lambda(i)}(c,f_{\lambda(i)})=(h-1)\cdot 360\degree + \sum_{x}\phi(c^x,f),
		\end{eqnarray*}
		where the latter sum is over all the occurrences $c^x$ of $c$ along the boundary of $f$. For $i=1,\dots,h$, by Property~(d) of $(\mathcal E,\phi)$, the value $\phi_{\lambda(i)}(c,f_{\lambda(i)})$ is equal to $\sum_{g,x}\phi(c^x,g)$, where the sum is over all the faces $g$ of $\mathcal E$ whose corresponding face in $\mathcal E_{\lambda(i)}$ is $f_{\lambda(i)}$. Hence
		\begin{eqnarray*}
			\sum_{i=1}^h \phi_{\lambda(i)}(c,f_{\lambda(i)})=\sum_{i=1}^h\sum_{g,x}\phi(c^x,g).
		\end{eqnarray*}
		Since, for $i=1,\dots,h$, we have that $f_{\lambda(i)}$ is the face of $\mathcal E_{\lambda(i)}$ corresponding to $f$, it follows that every occurrence $c^x$ of $c$ along the boundary of $f$ contributes  $h \cdot \phi(c^x,f)$ to the sum $\sum_{i=1}^h\sum_{g,x}\phi(c^x,g)$. Further, every occurrence $c^x$ of $c$ along the boundary of a face $g\neq f$ of $\mathcal E$ contributes  $(h-1) \cdot \phi(c^x,g)$ to the sum $\sum_{i=1}^h\sum_{g,x}\phi(c^x,g)$; indeed, there exists one index $j\in \{1,\dots,h\}$ such that the face of $\mathcal E_{\lambda(j)}$ corresponding to $g$ is not $f_{\lambda(j)}$, while for $i=1,\dots,h$ with $i\neq j$, the face of $\mathcal E_{\lambda(i)}$ corresponding to $g$ is $f_{\lambda(i)}$. Hence, 
		\begin{eqnarray*}
			\sum_{i=1}^h \phi_{\lambda(i)}(c,f_{\lambda(i)})=h\cdot \sum_{x}\phi(c^x,f)+(h-1)\cdot \sum_{g,x}\phi(c^x,g),
		\end{eqnarray*}
		where the first sum in the right term is over all the occurrences $c^x$ of $c$ along the boundary of $f$ and  the second sum in the right term is over all the faces $g\neq f$ of $\mathcal E$ incident to $c$ and over all the occurrences $c^x$ of $c$ along the boundary of $g$. Rearranging the right term we get
		\begin{eqnarray*}
			\sum_{i=1}^h \phi_{\lambda(i)}(c,f_{\lambda(i)})=(h-1)\cdot \sum_{g,x}\phi(c^x,g)+\sum_{x}\phi(c^x,f),
		\end{eqnarray*}
		where the first sum in the right term is over all the faces $g$ of $\mathcal E$ incident to $c$, including $f$, and over all the occurrences $c^x$ of $c$ along the boundary of $g$. By Property~(c) of $(\mathcal E,\phi)$, we have that $\sum_{g,x}\phi(c^x,g)=360\degree$, which concludes the proof of the equality $\sum_{i=1}^h \phi_{\lambda(i)}(c,f_{\lambda(i)})=(h-1)\cdot 360\degree + \sum_{x}\phi(c^x,f)$.	
		
		We can now proceed to the proof that $S=-4$ or $S=+4$ if $f$ is the outer face of $\mathcal E$ or an internal face of $\mathcal E$, respectively. If $f$ is the outer face of $\mathcal E$, then $f_{\lambda(i)}$ is the outer face of $\mathcal E_{\lambda(i)}$, for every $i=1,\dots,h$. Conversely, if $f$ is an internal face of $\mathcal E$, then there exists an index $j\in \{1,2,\dots,h\}$ such that $f_{\lambda(j)}$ is an internal face of $\mathcal E_{\lambda(j)}$, while $f_{\lambda(i)}$ is the outer face of $\mathcal E_{\lambda(i)}$, for every $i=1,\dots,h$ with $i\neq j$. In the following we distinguish these two cases.
		
		\begin{itemize}
			\item Suppose first that $f$ is the outer face of $\mathcal E$. Hence, we have $S_{\lambda(1)}=\dots=S_{\lambda(h)}=-4$.  Separating the contributions to $S$ given by the occurrences along $f$ of $c$ from the ones given by the occurrences along $f$ of  vertices different from $c$, we get 
			\begin{eqnarray*}
				S=\sum_{w,x} (2-\phi(w^x,f)/90\degree) + \sum_x(2-\phi(c^x,f)/90\degree),
			\end{eqnarray*} where the first sum is over all the occurrences $w^x$ on the boundary of $f$ of each vertex $w\neq c$ of $G$ incident to $f$, and the second sum is over all the occurrences $c^x$ on the boundary of $f$ of $c$. Since $c$ has $h$ occurrences along the boundary of $f$, the previous equality can be rewritten as 
			\begin{eqnarray*}
				S=\sum_{w,x} (2-\phi(w^x,f)/90\degree) - \sum_x(\phi(c^x,f)/90\degree) + 2h.
			\end{eqnarray*}
			Note that each vertex $w\neq c$ belongs to a unique graph $G_{\lambda(i)}$; then, by Property~(d) of $(\mathcal E,\phi)$, we have $\phi(w^x,f)=\phi_{\lambda(i)}(w^x,f_{\lambda(i)})$. Hence, we have
			\begin{eqnarray*}
				S=\sum_{i=1}^h\sum_{w,x} (2-\phi_{\lambda(i)}(w^x,f_{\lambda(i)})/90\degree) - \sum_x(\phi(c^x,f)/90\degree) + 2h,
			\end{eqnarray*}
			where the second sum is over all the occurrences $w^x$ on the boundary of $f_{\lambda(i)}$ of each vertex $w\neq c$ of $G_{\lambda(i)}$ incident to $f_{\lambda(i)}$.	Adding and subtracting $\sum_{i=1}^h (2-\phi_{\lambda(i)}(c,f_{\lambda(i)})/90\degree)$ from the right side of the previous equality, we~get
			\begin{eqnarray*}
				S=\sum_{i=1}^h\sum_{w,x} (2-\phi_{\lambda(i)}(w^x,f_{\lambda(i)})/90\degree) - \sum_x(\phi(c^x,f)/90\degree)+ 2h-\sum_{i=1}^h (2-\phi_{\lambda(i)}(c,f_{\lambda(i)})/90\degree),
			\end{eqnarray*}
			
			where the second sum is now over all the occurrences $w^x$ on the boundary of $f_{\lambda(i)}$ of each vertex $w$ of $G_{\lambda(i)}$ incident to $f_{\lambda(i)}$, including $w=c$. We can now exploit $S_{\lambda(i)}=\sum_{w,x} (2-\phi_{\lambda(i)}(w^x,f_{\lambda(i)})/90\degree)=-4$ and $\sum_{i=1}^h \phi_{\lambda(i)}(c,f_{\lambda(i)})=(h-1)\cdot 360\degree + \sum_{x}\phi(c^x,f)$ to get 
			\begin{eqnarray*}
				S&=&\sum_{i=1}^h S_{\lambda(i)} - \sum_x(\phi(c^x,f)/90\degree)+\sum_{i=1}^h \phi_{\lambda(i)}(c,f_{\lambda(i)})/90\degree=\\
				&=& -4h - \sum_x(\phi(c^x,f)/90\degree)+4(h-1)+\sum_{x}(\phi(c^x,f)/90\degree)=-4.
			\end{eqnarray*}
			\item Suppose next that $f$ is an internal face of $\mathcal E$. Then there exists an index $j\in \{1,2,\dots,h\}$ such that $S_{\lambda(j)}=+4$, while $S_{\lambda(i)}=-4$ for every $i\in \{1,2,\dots,h\}$ with $i\neq j$. Exactly as in the previous case, we can derive 
			\begin{eqnarray*}
				S&=&\sum_{i=1}^h S_{\lambda(i)} - \sum_x(\phi(c^x,f)/90\degree)+\sum_{i=1}^h \phi_{\lambda(i)}(c,f_{\lambda(i)})/90\degree.
			\end{eqnarray*}
			The sum $\sum_{i=1}^h S_{\lambda(i)}$ is now equal to $-4(h-1)+4=-4h+8$. Hence
			\begin{eqnarray*}
				S&=&-4h+8 - \sum_x(\phi(c^x,f)/90\degree)+4(h-1)+\sum_{x}(\phi(c^x,f)/90\degree)=+4.
			\end{eqnarray*}	
		\end{itemize}
	\end{itemize}

	This completes the proof that $(\mathcal E,\phi)$ satisfies condition (2) of Tamassia's characterization~\cite{t-eggmnb-87}. It follows that $(\mathcal E,\phi)$ is a rectilinear representation of $G$.
\end{proof}

\section{Fixed Embedding} \label{se:fixed}

In this section we show how to test in $O(n)$ time whether an $n$-vertex outerplanar graph $G$ with a fixed plane embedding $\mathcal E$ has a planar rectilinear drawing. In Section~\ref{se:fixed-biconnected} we assume that $G$ is $2$-connected; then, in Section~\ref{se:fixed-cutvertices}, we show how to extend our algorithm to handle the presence of cut-vertices.


\subsection{$2$-Connected Outerplanar Graphs} \label{se:fixed-biconnected}

%


Let $G$ be an $n$-vertex $2$-connected outerplanar graph with a prescribed plane embedding $\mathcal E$. We want to design an $O(n)$-time algorithm to test whether $G$ admits a rectilinear representation $(\mathcal E,\phi)$, for some function $\phi$. We will actually solve a more general problem, which is defined in the following and which will be used in Section~\ref{se:fixed-cutvertices} in order to deal with simply-connected planar graphs.  

Suppose that, for each face $f$ of $\cal E$ and each vertex $w$ incident to $f$, a value $\ell(w,f)\in \{90\degree,180\degree,270\degree\}$ is given. Our goal is to test whether an {\em $\ell$-constrained} representation $(\cal E,\phi)$ of $G$ exists, that is, a rectilinear representation $(\cal E,\phi)$ such that, for each face $f$ of $\cal E$ and each vertex $w$ incident to $f$, we have $\phi(w,f)\geq \ell(w,f)$. Note that the original question about the existence of a rectilinear representation $(\cal E,\phi)$ of $G$ coincides with the special case of this problem in which $\ell(w,f)= 90\degree$, for each face $f$ of $\cal E$ and each vertex $w$ incident to $f$. Intuitively speaking, the lower bound on the angle $\phi(w,f)$ given by $\ell(w,f)$ ensures that $\phi(w,f)$ is large enough to accommodate distinct blocks incident to $w$ inside $f$; this will be explained in detail in Section~\ref{se:fixed-cutvertices}. 

We start by proving that, if $G$ admits a planar rectilinear drawing with plane embedding $\mathcal E$, then an edge incident to $f^*_{\mathcal E}$ is also incident to $f^*_{\mathcal O}$.

\begin{lemma} \label{le:no-outer-edge}
	Suppose that no edge incident to $f^*_{\mathcal E}$ is also incident to $f^*_{\mathcal O}$. Then $G$ admits no planar rectilinear drawing with plane embedding $\mathcal E$.
\end{lemma} 

\begin{proof}
	Refer to Fig.~\ref{fig:edge-outer-face}. Let $\mathcal C=(v_0,v_1,\dots,v_{k-1})$ denote the cycle delimiting $f^*_{\mathcal E}$. By assumption, the edge $v_iv_{i+1}$  is an internal edge of $\mathcal O$, for every $i=0,\dots,k-1$, where the indices are modulo $k$. Thus, the removal of the vertices $v_i$ and $v_{i+1}$ splits $G$ into two outerplanar graphs; let $H_i$ be the one containing $v_0,v_1,\dots,v_{i-1},v_{i+2},,\dots,v_{k-1}$ and let $K_i$ be the other one. Denote by $\mathcal H_i$ (by $\mathcal K_i$) the subgraph of $G$ induced by $v_i$, $v_{i+1}$, and the vertex set of $H_i$ (of $K_i$). Since $\mathcal C$ delimits $f^*_{\mathcal E}$, we have that $\mathcal K_i$ lies inside $\mathcal C$ (except for the edge $v_iv_{i+1}$, which is shared by $\mathcal K_i$ and $\mathcal C$) in the plane embedding $\mathcal E$.  
	
	\begin{figure}[htb]
		\centering
		\includegraphics[]{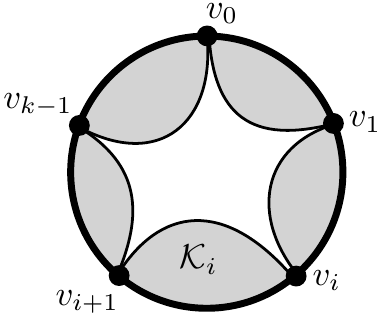}
		\caption{Illustration for the proof of Lemma~\ref{le:no-outer-edge}.}
		\label{fig:edge-outer-face}
	\end{figure}
	
	In any planar rectilinear drawing of $\mathcal K_i$ in which the edge $v_i v_{i+1}$ is incident to the outer face, the internal angles at $v_i$ and $v_{i+1}$ are at least $90\degree$ each. Hence, in any planar rectilinear drawing of $G$ with plane embedding $\mathcal E$, the sum of the internal angles at $v_0,v_1,\dots,v_{k-1}$ is at least $k \cdot 180 \degree$, while the sum of the internal angles of the polygon representing $\mathcal C$ should be $(k-2) \cdot 180 \degree$. This proves that no planar rectilinear drawing of $G$ with plane embedding $\mathcal E$ exists.
\end{proof}

Because of Lemma~\ref{le:no-outer-edge}, we can assume that an edge $uv$ incident to $f^*_{\mathcal E}$ is also incident to $f^*_{\mathcal O}$ and root $G$ at $uv$. Denote by $f^{uv}_{\mathcal E}$ the internal face of $\mathcal E$ incident to $uv$.

Our algorithm will test whether $G$ admits an $\ell$-constrained representation $(\mathcal E,\phi)$ such that the values $\phi^{\mathrm{int}}(u)$ and $\phi^{\mathrm{int}}(v)$ are prescribed. For example, the algorithm will test whether an $\ell$-constrained representation $(\mathcal E,\phi)$ exists such that $\phi^{\mathrm{int}}(u)=90\degree$ and $\phi^{\mathrm{int}}(v)=270\degree$. Since each of $\phi^{\mathrm{int}}(u)$ and $\phi^{\mathrm{int}}(v)$ can only be a value in $\{90\degree,180\degree,270\degree\}$, the algorithm will test for the existence of nine $\ell$-constrained representations of $G$. Formally, we define the following.



\begin{definition} \label{def:alpha-beta}
	An $(\mathcal E,\ell,\mu,\nu)$-representation of $G$ is an $\ell$-constrained representation $(\mathcal E,\phi)$ such that $\phi^{\mathrm{int}}(u)=\mu$ and $\phi^{\mathrm{int}}(v)=\nu$.
\end{definition}

It turns out that dealing with $(\mathcal E,\ell,\mu,\nu)$-representations rather than general $\ell$-constrained representations is easier, as in order to test whether an $(\mathcal E,\ell,\mu,\nu)$-representation of $G$ exists, we can exploit the information about the existence of an $(\mathcal E_i,\ell_i,\mu_i,\nu_i)$-representation of each $uv$-subgraph $G_i$ of $G$, for each value of $\mu_i$ and $\nu_i$, where $\mathcal E_i$ and $\ell_i$ are the restrictions of $\mathcal E$ and $\ell$ to $G_i$, respectively. This can be done due to the following main structural lemma; refer to Fig.~\ref{fig:structural-fixed-embedding}. 

Let $u=u_0,u_1,\dots,u_k=v$ be the clockwise order of the vertices of ${\mathcal C}^*_{uv}$ in $\mathcal O$. Further, for $i=1,\dots,k$, let $G_i$ be the $uv$-subgraph of $G$ with root $u_{i-1}u_i$, let $\mathcal E_i$ be the restriction of $\mathcal E$ to $G_i$, and let $\ell_i$ be the restriction of $\ell$ to $G_i$. The latter is defined as follows: Consider a vertex $w$ that belongs to $G_i$; if the faces $f_1,\dots,f_h$ of $\mathcal E$ are incident to $w$ and become a unique face $f$ in $\mathcal E_i$ because some edges incident to $w$ do not belong to $G_i$, then $\ell_i(w,f)=\sum_{j=1}^h \ell(w,f_j)$. 

\begin{figure}[tb]\tabcolsep=4pt
	\centering
	\includegraphics[width=.5\textwidth]{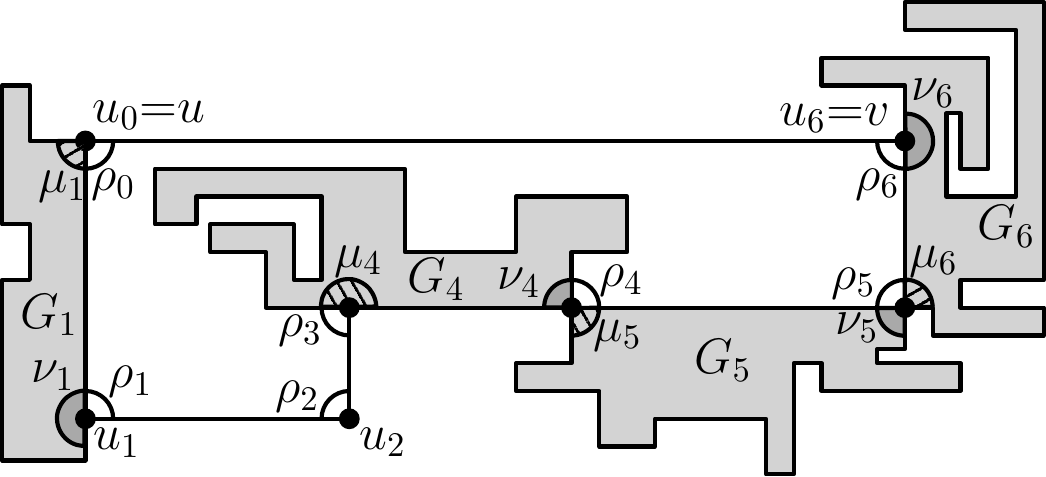}
	\caption{A planar rectilinear drawing of the outerplanar graph $G$ in Fig.~\ref{fig:outerplanarrepresentation}(b). The corresponding rectilinear representation is an $(\mathcal E,\ell,180\degree,270\degree)$-representation, where $\mathcal E$ is a plane (and not outerplane) embedding of~$G$.}
	\label{fig:structural-fixed-embedding}
\end{figure}

\begin{lemma} \label{le:structural-fixed-embedding}
	For any $\mu,\nu\in \{90\degree,180\degree,270\degree\}$, we have that $G$ admits an $(\mathcal E,\ell,\mu,\nu)$-representation if and only if values $\rho_0$, $\rho_1$, $\mu_1$, $\nu_1$, $\rho_2$, $\mu_2$, $\nu_2$, $\dots$, $\rho_k$, $\mu_k$, $\nu_k\in \{0\degree,90\degree,180\degree,270\degree\}$ exist with the following properties:
	
	\begin{enumerate}
		\item[Pr1:] for $i=0,\dots,k$, we have $\rho_i \geq \ell(u_i,f^{uv}_{\mathcal E})$;
		\item[Pr2:] for $i=1,\dots,k$, if $G_i$ is trivial then $\mu_i=\nu_i=0\degree$, otherwise $\mu_i,\nu_i \in \{90\degree,180\degree\}$ and $G_i$ admits an $(\mathcal E_i,\ell_i,\mu_i,\nu_i)$-representation;
		\item[Pr3:] for $i=1,\dots,k-1$, we have that $\nu_i+\rho_i+\mu_{i+1}\leq 360\degree-\ell(u_i,f^*_{\cal E})$; further, $\rho_0+\mu_1\leq 360\degree-\ell(u_0,f^*_{\cal E})$ and $\rho_k+\nu_k\leq 360\degree-\ell(u_k,f^*_{\cal E})$;
		\item[Pr4:] $\rho_0+\mu_1=\mu$ and $\rho_k+\nu_k=\nu$; and
		\item[Pr5:] for $i=1,\dots,k$, if $G_i$ is trivial or if it lies outside ${\mathcal C}^*_{uv}$ in $\cal E$, then let $\sigma_i=0\degree$, otherwise let $\sigma_i=\mu_i+\nu_i$; then $\sum_{i=0}^k \rho_i + \sum_{i=1}^k \sigma_i = (k-1) \cdot 180\degree$. 
	\end{enumerate}
\end{lemma}

Before proving the lemma, we describe its statement. For $i=0,\dots,k$, the value $\rho_i$ represents the angle incident to $u_i$ inside $f^{uv}_{\mathcal E}$. Hence, Property~Pr1 ensures that such an angle is at least $\ell(u_i,f^{uv}_{\mathcal E})$, as required. For $i=1,\dots,k$, the value $\mu_i$ represents the sum of the angles incident to $u_{i-1}$ inside the internal faces of $\mathcal E_i$; similarly, the value $\nu_i$ represents the sum of the angles incident to $u_{i}$ inside the internal faces of $\mathcal E_i$. Hence, Property~Pr2 ensures that each non-trivial $uv$-subgraph $G_i$ of $G$ admits an $(\mathcal E_i,\ell_i,\mu_i,\nu_i)$-representation. Property~Pr3 ensures that the sum of the angles around each vertex $u_i$ does not exceed $360\degree$, taking into account that the angle incident to $u_i$ in $f^*_{\mathcal E}$ is at least $\ell(u_i,f^*_{\cal E})$. Property~Pr4 ensures that the sum of the internal angles incident to $u_0$ is $\mu$ and the sum of the internal angles incident to $u_k$ is $\nu$. Finally, Property~Pr5 ensures that the sum of the angles internal to the polygon representing $\mathcal C^*_{uv}$ is equal to $(k-1) \cdot 180\degree$. 

\medskip
\begin{proof}
	$(\Longrightarrow)$	We first prove the necessity. Suppose that $G$ admits an $(\mathcal E,\ell,\mu,\nu)$-representation $(\mathcal E,\phi)$. For $i=0,\dots,k$, we set $\rho_i=\phi(u_i,f^{uv}_{\mathcal E})$. Further, if $G_i$ is trivial, then we set $\mu_i=\nu_i=0\degree$, otherwise, we set $\mu_i= \sum_f \phi(u_{i-1},f)$, where the sum is over all the internal faces $f$ of $\mathcal E_i$ incident to $u_{i-1}$, and we set $\nu_i= \sum_f \phi(u_{i},f)$, where the sum is over all the internal faces $f$ of $\mathcal E_i$ incident to $u_{i}$. For $i=0,\dots,k$, we have $\rho_i\in \{90\degree,180\degree,270\degree\}$, given that $\phi(u_i,f^{uv}_{\mathcal E})\geq 90\degree$ and $\phi(u_i,f^{*}_{\mathcal E})\geq 90\degree$. For the same reason, we have $\mu_i\in \{0\degree,90\degree,180\degree\}$ and $\nu_i\in \{0\degree,90\degree,180\degree\}$, for $i=1,\dots,k$. We now show that the values $\rho_0,\rho_1,\mu_1,\nu_1,\rho_2,\mu_2,\nu_2,\dots,\rho_k,\mu_k,\nu_k$ satisfy Properties Pr1--Pr5. 
	
	\begin{enumerate}[(1)]
		\item Since $(\mathcal E,\phi)$ is an $(\mathcal E,\ell,\mu,\nu)$-representation, for each $i=0,\dots,k$ we have $\phi(u_i,f^{uv}_{\mathcal E})\geq \ell(u_i,f^{uv}_{\mathcal E})$ and hence $\rho_i \geq \ell(u_i,f^{uv}_{\mathcal E})$. Property Pr1 follows. 
		\item If $G_i$ is trivial, then $\mu_i=\nu_i=0\degree$, by construction. Otherwise, let $(\mathcal E_i,\phi_i)$ be the restriction of $(\mathcal E,\phi)$ to $G_i$. Since $(\mathcal E,\phi)$ is an $\ell$-constrained representation of $G$, it follows that $(\mathcal E_i,\phi_i)$ is an $\ell_i$-constrained representation of $G_i$. By the definition of $\mu_i$ and $\nu_i$, we have $\phi^{\mathrm{int}}_i(u_{i-1})=\mu_i$ and $\phi^{\mathrm{int}}_i(u_{i})=\nu_i$, hence $(\mathcal E_i,\phi_i)$ is an $(\mathcal E_i,\ell_i,\mu_i,\nu_i)$-representation of $G_i$. 
		\item For $i=1,\dots,k-1$, by the definition of $\rho_i$, $\nu_i$, and $\mu_{i+1}$, we have $\nu_i+\rho_i+\mu_{i+1}+\phi(u_{i},f^*_{\mathcal E})=360\degree$, hence Property Pr3 follows, given that $\phi(u_{i},f^*_{\mathcal E})\geq \ell(u_i,f^*_{\cal E})$.
		\item By the definition of $\rho_0$ and $\mu_1$, we have $\rho_0+\mu_1=\phi(u_0,f^{uv}_{\mathcal E})+ \sum_f \phi(u_0,f)=\phi^{\mathrm{int}}(u_0)$, where the sum is over all the internal faces $f$ of $\mathcal E_1$ incident to $u_0$; since $(\mathcal E,\phi)$ is an $(\mathcal E,\ell,\mu,\nu)$-representation, it follows that $\phi^{\mathrm{int}}(u_0)=\mu$ and hence $\rho_0+\mu_1=\mu$. Analogously, $\rho_k+\nu_k=\nu$, and Property Pr4 follows. 
		\item Finally, by construction, $\sum_{i=0}^k \rho_i + \sum_{i=1}^k \sigma_i$ is equal to the sum of the angles $\phi(u_i,f)$ such that $f$ is a face of $\mathcal E$ incident to $u_i$ inside ${\mathcal C}^*_{uv}$. Since ${\mathcal C}^*_{uv}$ has $k+1$ vertices, the sum of its internal angles in $(\mathcal E,\phi)$ is $(k-1) \cdot 180\degree$, and Property Pr5 follows.
	\end{enumerate}

	$(\Longleftarrow)$	We now prove the sufficiency. Assume that values $\rho_0,\rho_1,\mu_1,\nu_1,\dots,\rho_k,\mu_k,\nu_k$ in $\{0\degree,90\degree,180\degree,270\degree\}$ exist satisfying Properties Pr1--Pr5. For $i=1,\dots,k$, if $G_i$ is non-trivial, let $(\mathcal E_i,\phi_i)$ be an $(\mathcal E_i,\ell_i,\mu_i,\nu_i)$-representation of $G_i$; this exists by Property~Pr2. 
	
	We define a function $\phi$ such that $(\mathcal E,\phi)$ is an $(\mathcal E,\ell,\mu,\nu)$-representation of $G$. Consider any face $f$ of $\mathcal E$ and any vertex $w$ incident to $f$. For each $i=1,\dots,k$ such that $G_i$ is non-trivial, let $f_i$ be the face of $\mathcal E_i$ corresponding to $f$. We distinguish four cases.
	
	\begin{itemize}
		\item Case~1: $w\notin \{u_0,u_1,\dots,u_k\}$; then $w$ belongs to a unique $uv$-subgraph $G_i$ of $G$. Let $\phi(w,f)=\phi_i(w,f_i)$. 
		\item Case~2: for some $i\in \{1,\dots,k\}$, we have that $f_i$ is an internal face of $\mathcal E_i$. Let $\phi(w,f)=\phi_i(w,f_i)$. 
		\item Case~3: $w=u_i$, for some $i\in \{0,1,\dots,k\}$, and $f=f^{uv}_{\mathcal E}$. Let $\phi(w,f)=\rho_i$. 
		\item Case~4: $w=u_i$, for some $i\in \{0,1,\dots,k\}$, and $f=f^*_{\mathcal E}$. Let $\phi(w,f)=360\degree - \nu_{i}-\rho_i-\mu_{i+1}$ (if $1\leq i\leq k-1$) or $\phi(w,f)=360\degree - \rho_0-\mu_1$ (if $i=0$) or $\phi(w,f)=360\degree - \nu_k-\rho_k$ (if $i=k$). 
	\end{itemize}
	
	First, we have $\phi(w,f)\in \{90\degree,180\degree,270\degree\}$, for every face $f$ of $\mathcal E$ and every vertex $w$ incident to $f$. This is because $\phi_i(w,f_i)\in \{90\degree,180\degree,270\degree\}$ if $\phi(w,f)$ was set according to Cases~1 or~2, it is because $\rho_i \in \{90\degree,180\degree,270\degree\}$ (by Property~Pr1) if $\phi(w,f)$ was set according to Case~3, and it comes from Properties~Pr1 and~Pr3 and from $\ell(w,f^*_{\mathcal E})\geq 90\degree$ if $\phi(w,f)$ was set according to Case~4. 
	
	Second, we have $\phi(w,f)\geq \ell(w,f)$, for every face $f$ of $\mathcal E$ and every vertex $w$ incident to $f$. This comes from $\phi_i(w,f_i)\geq \ell_i(w,f_i)$ if $w\notin \{u_0,u_1,\dots,u_k\}$ or if $f_i$ is a face internal to $\mathcal E_i$, by Cases~1 and~2 and since $(\mathcal E_i,\phi_i)$ is an $(\mathcal E_i,\ell_i,\mu_i,\nu_i)$-representation of $G_i$. Further, $\phi(u_i,f^{uv}_{\mathcal E})\geq \ell(u_i,f^{uv}_{\mathcal E})$ comes from Property~Pr1, by Case~3. Finally, $\phi(u_i,f^*_{\mathcal E})\geq \ell(u_i,f^*_{\mathcal E})$ comes from Property~Pr3, by Case~4. 	
	
	Third, we have $\phi^{\mathrm{int}}(u)=\mu$. Namely, by Case~3 we have $\phi(u,f^{uv}_{\mathcal E})=\rho_0$. Further, by Cases~1 and~2, we have $\sum_{f}\phi(u,f)=\sum_{f_1}\phi_1(u,f_1)=\phi_1^{\mathrm{int}}(u)=\mu_1$, where $f$ ranges over all the faces of $\mathcal E$ whose corresponding faces $f_1$ in $\mathcal E_1$ are internal to $\mathcal E_1$ and incident to $u$; the equality $\phi_1^{\mathrm{int}}(u)=\mu_1$ descends from the fact that $(\mathcal E_1,\phi_1)$ is a $(\mathcal E_1,\ell_1,\mu_1,\nu_1)$-representation, by Property~Pr2. Finally, by~Property~Pr4 we have $\rho_0+\mu_1=\mu$. It can be analogously proved that $\phi^{\mathrm{int}}(v)=\nu$.
	
	In order to prove that $(\mathcal E,\phi)$ is a rectilinear representation, we prove that it satisfies conditions~(1) and (2) of Tamassia's characterization (see~\cite{t-eggmnb-87} and Section~\ref{se:preliminaries}).
	
	We first prove that $(\mathcal E,\phi)$ satisfies condition~(1) of Tamassia's characterization~\cite{t-eggmnb-87}. Consider any vertex $w\notin \{u_0,u_1,\dots,u_k\}$ and let $G_i$ be the $uv$-subgraph of $G$ the vertex $w$ belongs to. We have $\sum_f \phi(w,f)=360\degree$, where the sum is over all the faces $f$ of $\mathcal E$ incident to $w$, given that $\phi(w,f)=\phi_i(w,f_i)$, according to Cases~1 and~2, and given that $\sum_{f_i} \phi_i(w,f_i)=360\degree$. Second, consider any vertex $u_i$, with $i\in \{0,1,\dots,k\}$. Any face $f$ incident to $u_i$ in $\mathcal E$ is of one of the following four types: (a) the face $f_i$ of $\mathcal E_i$ corresponding to $f$ is internal to $\mathcal E_i$; (b) the face $f_{i+1}$ of $\mathcal E_{i+1}$ corresponding to $f$ is internal to $\mathcal E_{i+1}$; (c) $f=f^{uv}_{\mathcal E}$; or (d) $f=f^*_{\mathcal E}$. By Case~2 and by Property~Pr2, the sum of the angles incident to $u_i$ inside faces of type~(a) and~(b) is $\nu_i+\mu_{i+1}$ (the first or the last of such two terms has to be neglected if $i=0$ or $i=k$, respectively). By Case~3, the angle incident to $u_i$ inside $f^{uv}_{\mathcal E}$ is $\rho_i$. Finally, by Case~4, the angle incident to $u_i$ inside $f^*_{\mathcal E}$ is $\phi(u_i,f^*_{\mathcal E})=360\degree - \nu_{i}-\rho_i-\mu_{i+1}$ (if $1\leq i\leq k-1$) or $\phi(u_i,f^*_{\mathcal E})=360\degree - \rho_0-\mu_1$ (if $i=0$) or $\phi(u_i,f^*_{\mathcal E})=360\degree - \nu_k-\rho_k$ (if $i=k$). Hence, we have $\sum_f \phi(u_i,f)=360\degree$, where the sum is over all the faces $f$ of $\mathcal E$ incident to $u_i$. It follows that $(\mathcal E,\phi)$ satisfies condition~(1) of Tamassia's characterization~\cite{t-eggmnb-87}.
	
	We finally prove that $(\mathcal E,\phi)$ satisfies condition~(2) of Tamassia's characterization~\cite{t-eggmnb-87}. Consider any face $f$ of $\mathcal E$ whose corresponding face $f_i$ of $\mathcal E_i$ is internal to $\mathcal E_i$, for some $i\in \{1,\dots,k\}$. Then $\sum_w (2-\phi(w,f)/90\degree)=+4$, where the sum is over all the vertices $w$ incident to $f$, given that $\phi(w,f)=\phi_i(w,f_i)$, according to Cases~1 and~2, and given that $\sum_w (2-\phi_i(w,f_i)/90\degree)=+4$. 
	
	It remains to deal with $f^{uv}_{\mathcal E}$ and $f^*_{\mathcal E}$. For every non-trivial $uv$-subgraph $G_i$ of $G$, we have that $\sum_w (2-\phi_i(w,f^{*}_{\mathcal E_i})/90\degree)=-4$, where the sum is over all the vertices $w$ of $G_i$ incident to $f^*_{\mathcal E_i}$. Note that, if $w\neq u_{i-1},u_i$, then $w$ is incident to $f^{uv}_{\mathcal E}$ or $f^*_{\mathcal E}$, depending on whether $G_i$ lies inside or outside $\mathcal C^*_{uv}$, respectively.
	
	We show that $\sum_w (2-\phi(w,f^{uv}_{\mathcal E})/90\degree)=+4$, where the sum is over all the vertices $w$ of $G$ incident to $f^{uv}_{\mathcal E}$. Consider any non-trivial $uv$-subgraph $G_i$ of $G$ that lies inside $\mathcal C^*_{uv}$. By Case~1, for every vertex $w\neq u_{i-1},u_i$ of $G_i$ incident to $f^{uv}_{\mathcal E}$, we have $\phi(w,f^{uv}_{\mathcal E})=\phi_i(w,f^{*}_{\mathcal E_i})$; further, by Property~Pr2 we have $\phi_i(u_{i-1},f^{*}_{\mathcal E_i})=360\degree-\mu_i$ and $\phi_i(u_{i},f^{*}_{\mathcal E_i})=360\degree-\nu_i$. It follows that $\sum_w (2-\phi(w,f^{uv}_{\mathcal E})/90\degree)$, where the sum is over all the vertices $w\neq u_{i-1},u_i$ of $G_i$ incident to $f^{uv}_{\mathcal E}$, is equal to $\sum_w (2-\phi_i(w,f^*_{\mathcal E_i})/90\degree)$, where the sum is over all the vertices $w$ of $G_i$ incident to $f^*_{\mathcal E_i}$, minus $(2 - (360\degree-\mu_i)/90\degree) + (2 - (360\degree-\nu_i)/90\degree)$, hence it is equal to $-4 + (2-\mu_i/90\degree) + (2-\nu_i/90\degree)=-\sigma_i / 90\degree$. Further, by Case~3, we have $\phi(u_i,f^{uv}_{\mathcal E})=\rho_i$, for $i=0,\dots,k$. Hence, $\sum_w (2-\phi(w,f^{uv}_{\mathcal E})/90\degree)$, where the sum is over all the vertices $w$ of $G$ incident to $f^{uv}_{\mathcal E}$, is equal to $2\cdot (k+1) - (\sum_{i=0}^k \rho_i +\sum_{i=1}^k \sigma_i)/90\degree=2\cdot (k+1)-2\cdot (k-1)=+4$, where we used Property~Pr5.
	
	We next show that $\sum_w (2-\phi(w,f^*_{\mathcal E})/90\degree)=-4$, where the sum is over all the vertices of $G$ incident to $f^*_{\mathcal E}$. Consider any non-trivial $uv$-subgraph $G_i$ of $G$ that lies outside $\mathcal C^*_{uv}$. By Case~1, for every vertex $w\neq u_{i-1},u_i$ of $G_i$ incident to $f^*_{\mathcal E}$, we have $\phi(w,f^*_{\mathcal E})=\phi_i(w,f^{*}_{\mathcal E_i})$; further, by Property~Pr2 we have $\phi_i(u_{i-1},f^{*}_{\mathcal E_i})=360\degree-\mu_i$ and $\phi_i(u_{i},f^{*}_{\mathcal E_i})=360\degree-\nu_i$. It follows that $\sum_w (2-\phi(w,f^*_{\mathcal E})/90\degree)$, where the sum is over all the vertices $w\neq u_{i-1},u_i$ of $G_i$ incident to $f^*_{\mathcal E_i}$, is equal to $\sum_w (2-\phi_i(w,f^*_{\mathcal E_i})/90\degree)$, where the sum is over all the vertices $w$ of $G_i$ incident to $f^*_{\mathcal E_i}$, minus $(2 - (360\degree-\mu_i)/90\degree) + (2 - (360\degree-\nu_i)/90\degree)$, hence it is equal to $-4 + (2-\mu_i/90\degree) + (2-\nu_i/90\degree)=-(\mu_i+\nu_i) / 90\degree$. Further, by Case~4, for $i=0,\dots,k$ we have $\phi(u_i,f^*_{\mathcal E})=360\degree-\rho_i-\nu_i-\mu_{i+1}$ (the third or fourth term has to be neglected if $i=0$ or $i=k$, respectively). Hence, $\sum_w (2-\phi(w,f^*_{\mathcal E})/90\degree)$, where the sum is over all the vertices $w$ incident to $f^*_{\mathcal E}$, is equal to $2\cdot (k+1) - 4 \cdot (k+1) + (\sum_{i=0}^k \rho_i + \sum_{i=1}^k \mu_i + \sum_{i=1}^k \nu_i)/90\degree - (\sum \mu_i + \sum \nu_i)/90\degree$, where the last two sums range over the indices $i$ such that $G_i$ lies outside $\mathcal C^*_{uv}$ in $\mathcal E$. Hence, $\sum_w (2-\phi(w,f^*_{\mathcal E})/90\degree)=2\cdot (k+1) - 4 \cdot (k+1) + (\sum_{i=0}^k \rho_i + \sum_{i=1}^k \sigma_i)/90\degree= 2\cdot (k+1) - 4 \cdot (k+1) + 2 \cdot (k-1) = -4$, where we used Property~Pr5.
\end{proof}

The main ingredient of our algorithm that tests whether $G$ admits an $(\mathcal E,\ell,\mu,\nu)$-representation is an $O(k)$-time solution to the following problem. Assume that, for each $i=1,\dots,k$, the pairs $(\mu_i,\nu_i)$ such that the $uv$-subgraph $G_i$ of $G$ admits an $(\mathcal E_i,\ell_i,\mu_i,\nu_i)$-representation are known. Do values $\rho_0,\rho_1,\mu_1,\nu_1,\rho_2$, $\mu_2,\nu_2,\dots,\rho_k,\mu_k,\nu_k$ exist such that Properties~Pr1--Pr5 of Lemma~\ref{le:structural-fixed-embedding} are satisfied? 

Observe that, if the number of non-trivial $uv$-subgraphs of $G$ is $t$, there might be $4^t$ distinct assignments of pairs $(\mu_i,\nu_i)$ to such subgraphs with $\mu_i,\nu_i \in \{90\degree,180\degree\}$. Hence, if $t$ is ``large'' we can not just test, for each of such assignments, whether values $\rho_0,\rho_1,\dots,\rho_k$ exist satisfying Properties Pr1--Pr5 of Lemma~\ref{le:structural-fixed-embedding}. We overcome this problem by establishing some criteria for choosing an ``optimal'' pair $(\mu_i,\nu_i)$ to assign to almost every non-trivial $uv$-subgraph $G_i$ of $G$. This is done in the following two lemmata. 


%

\begin{lemma}\label{le:smaller-is-better}
	Suppose that a sequence $\mathcal S$ of values $\rho_0,\rho_1,\mu_1,\nu_1,\dots,\rho_k,\mu_k,\nu_k$ exists satisfying Properties~Pr1--Pr5 of Lemma~\ref{le:structural-fixed-embedding}. Further, suppose that $G_i$ is a non-trivial $uv$-subgraph of $G$, for some $i\in \{2,\dots,k-1\}$. 
	
	If $G_i$ admits an $(\mathcal E_i,\ell_i,\mu'_i,\nu'_i)$-representation, where $\mu'_i,\nu'_i\in \{90\degree,180\degree\}$, $\mu'_i\leq \mu_i$, and $\nu'_i\leq \nu_i$, then there exist values $\rho'_{i-1}$ and $\rho'_i$ such that the sequence $\mathcal S'$ obtained from $\mathcal S$ by replacing the values $\rho_{i-1}$, $\rho_i$, $\mu_i$, and $\nu_i$ with $\rho'_{i-1}$, $\rho'_i$, $\mu'_i$, and $\nu'_i$, respectively, satisfies Properties~Pr1--Pr5 of Lemma~\ref{le:structural-fixed-embedding}. 
\end{lemma}

\begin{proof}
	We distinguish two cases.
	
	If $G_i$ lies outside $\mathcal C^*_{uv}$ in $\cal E$, then we let $\rho'_{i-1}=\rho_{i-1}$ and $\rho'_{i}=\rho_{i}$. Then Properties~Pr1--Pr5 are satisfied by $\mathcal S'$ since they are satisfied by $\mathcal S$. In particular, $\nu_{i-1}+\rho'_{i-1}+\mu'_i \leq \nu_{i-1}+\rho_{i-1}+\mu_i$ and $\nu'_{i}+\rho'_{i}+\mu_{i+1} \leq \nu_{i}+\rho_{i}+\mu_{i+1}$, hence Property~Pr3 is satisfied by $\mathcal S'$ since it is satisfied by $\mathcal S$. Further, Property~Pr4 is satisfied by $\mathcal S'$ since the values $\rho_0$, $\mu_1$, $\rho_k$, and $\nu_k$ are the same in $\mathcal S'$ as in $\mathcal S$, given that $2\leq i\leq k-1$. Finally, since $G_i$ lies outside $\mathcal C^*_{uv}$, the sum $\sum_{i=1}^k \sigma_i$ from Property~Pr5 has the same value in $\mathcal S$ as in $\mathcal S'$.
	
	If $G_i$ lies inside $\mathcal C^*_{uv}$ in $\cal E$, then we let $\rho'_{i-1}=\rho_{i-1}+\mu_i-\mu'_i$ and $\rho'_{i}=\rho_{i}+\nu_i-\nu'_i$. Then Properties~Pr1--Pr5 are satisfied by $\mathcal S'$ since they are satisfied by $\mathcal S$. In particular, by Property Pr1 of $\mathcal S$ and by construction, we have $\rho'_{i-1} \geq \rho_{i-1} \geq \ell(u_{i-1},f^{uv}_{\mathcal E})$ and $\rho'_i \geq \rho_i \geq \ell(u_i,f^{uv}_{\mathcal E})$, hence Property~Pr1 follows. Further, $\nu_{i-1}+\rho'_{i-1}+\mu'_i = \nu_{i-1}+\rho_{i-1}+\mu_i$ and $\nu'_{i}+\rho'_{i}+\mu_{i+1} = \nu_{i}+\rho_{i}+\mu_{i+1}$, hence Property~Pr3 is satisfied by $\mathcal S'$ since it is satisfied by $\mathcal S$. Property~Pr4 is satisfied by $\mathcal S'$ as in the case in which $G_i$ lies outside $\mathcal C^*_{uv}$ in $\cal E$. Finally, the sum $\sum_{i=1}^k \sigma_i$ is smaller in $\mathcal S'$ than in $\mathcal S$ by $\mu_i+\nu_i-\mu'_i-\nu'_i$, however the sum $\sum_{i=0}^k \rho_i$ is larger in $\mathcal S'$ than in $\mathcal S$ by the same amount, hence Property~Pr5 is satisfied by $\mathcal S'$.
\end{proof}

%
%
%


\begin{lemma}\label{le:same-size-is-the-same}
	Suppose that a sequence $\mathcal S$ of values $\rho_0,\rho_1,\mu_1,\nu_1,\dots,\rho_k,\mu_k,\nu_k$ exists satisfying Properties~Pr1--Pr5 of Lemma~\ref{le:structural-fixed-embedding}. Further, suppose that $G_i$ is a non-trivial $uv$-subgraph of $G$, for some $i\in \{2,\dots,k-1\}$, that $G_{i-1}$ and $G_{i+1}$ are both trivial $uv$-subgraphs of $G$, and that $\ell(u_{i-1},f^{uv}_{\mathcal E})=\ell(u_{i-1},f^{*}_{\mathcal E})=\ell(u_{i},f^{uv}_{\mathcal E})=\ell(u_{i},f^*_{\mathcal E})=90\degree$.
	
	If $G_i$ admits an $(\mathcal E_i,\ell_i,\mu'_i,\nu'_i)$-representation, where $\mu'_i,\nu'_i\in \{90\degree,180\degree\}$ and $\mu'_i+\nu'_i\leq \mu_i+\nu_i$, then there exist values $\rho'_{i-1}$ and $\rho'_i$ such that the sequence $\mathcal S'$ obtained from $\mathcal S$ by replacing the values $\rho_{i-1}$, $\rho_i$, $\mu_i$, and $\nu_i$ with $\rho'_{i-1}$, $\rho'_i$, $\mu'_i$, and $\nu'_i$, respectively, satisfies Properties~Pr1--Pr5 of Lemma~\ref{le:structural-fixed-embedding}. 
\end{lemma}

\begin{proof}
	If $\mu'_i\leq \mu_i$ and $\nu'_i\leq \nu_i$, then Lemma~\ref{le:smaller-is-better} implies the statement. We can hence assume that $\mu'_i>\mu_i$ and $\nu'_i<\nu_i$, or that $\mu'_i<\mu_i$ and $\nu'_i>\nu_i$. Assume the former, as the proof for the latter is symmetric. Recall that $\mu_i,\nu_i\in \{90\degree,180\degree\}$, since $\mathcal S$ satisfies Property~Pr2, and that $\mu'_i,\nu'_i\in \{90\degree,180\degree\}$, by assumption. Hence, we have $\mu_i=90\degree$, $\nu_i=180\degree$, $\mu'_i=180\degree$, and $\nu'_i=90\degree$. Further, since $G_{i-1}$ and $G_{i+1}$ are both trivial, we have $\nu_{i-1}=\mu_{i+1}=0\degree$. 
	
	We let $\rho'_{i-1}=\rho_{i}$ and $\rho'_{i}=\rho_{i-1}$. Then Properties~Pr1--Pr5 are satisfied by $\mathcal S'$ since they are satisfied by $\mathcal S$. In particular, $\rho'_{i-1}\geq \ell(u_{i-1},f^{uv}_{\mathcal E})=90\degree$, given that $\rho_{i}\geq \ell(u_i,f^{uv}_{\mathcal E})=90\degree$; similarly, $\rho'_{i}\geq \ell(u_{i},f^{uv}_{\mathcal E})=90\degree$, hence Property~Pr1 is satisfied by $\mathcal S'$ since it is satisfied by $\mathcal S$. Further, $\nu_{i-1}+\rho'_{i-1}+\mu'_i=\mu_{i+1}+\rho_{i}+\nu_i\leq 360\degree-\ell(u_{i},f^*_{\mathcal E})=360\degree-\ell(u_{i-1},f^*_{\mathcal E})$, where the inequality comes from the fact that $\mathcal S$ satisfies Property~Pr3; similarly, $\nu'_{i}+\rho'_{i}+\mu_{i+1} \leq 360\degree-\ell(u_i,f^*_{\mathcal E})$, hence Property~Pr3 is satisfied by $\mathcal S'$ since it is satisfied by $\mathcal S$. Further, Property~Pr4 is satisfied by $\mathcal S'$ since the values $\rho_0$, $\mu_1$, $\rho_k$, and $\nu_k$ are the same in $\mathcal S'$ as in $\mathcal S$, given that $2\leq i\leq k-1$. Finally, the sum $\sum_{i=1}^k \sigma_i$ from Property~Pr5 has the same value in $\mathcal S$ as in $\mathcal S'$, given that $\mu_i+\nu_i=\mu'_i+\nu'_i$, and the sum $\sum_{i=0}^k \rho_i$ from Property~Pr5 has the same value in $\mathcal S$ as in $\mathcal S'$, given that $\rho'_{i-1}+\rho'_{i}=\rho_{i-1}+\rho_{i}$; it follows that $\mathcal S'$ satisfies Property~Pr5.
\end{proof}

Lemmata~\ref{le:smaller-is-better} and~\ref{le:same-size-is-the-same} allow us either to conclude that $G$ admits no $(\mathcal E,\ell,\mu,\nu)$-representation or to select optimal values $\mu_2,\nu_2,\mu_3,\nu_3,\dots,\mu_{k-1},\nu_{k-1}$. This is the first step of the following main algorithmic tool.

\begin{lemma} \label{le:find-values}
	Suppose that, for every $uv$-subgraph $G_i$ of $G$, the following information is known: 
	\begin{enumerate}[(i)]
		\item whether $G_i$ is trivial or not;
		\item in case $G_i$ is not trivial, whether it lies inside or outside $\mathcal C^*_{uv}$ in $\mathcal E$; and
		\item in case $G_i$ is not trivial, whether it admits an $(\mathcal E_i,\ell_i,\mu_i,\nu_i)$-representation or not, for every pair $(\mu_i,\nu_i)$ with $\mu_i,\nu_i \in \{90\degree,180\degree\}$.
	\end{enumerate}  
	Then it is possible to determine in $O(k)$ time whether $G$ admits an $(\mathcal E,\ell,\mu,\nu)$-representation; in case it does, values $\rho_0,\rho_1,\mu_1,\nu_1,\rho_2,\mu_2,\nu_2,\dots,\rho_k,\mu_k,\nu_k$ satisfying Properties~Pr1--Pr5 of Lemma~\ref{le:structural-fixed-embedding} can be determined in $O(k)$ time.
\end{lemma}

\begin{proof}
	By Lemma~\ref{le:structural-fixed-embedding}, we have that $G$ admits an $(\mathcal E,\ell,\mu,\nu)$-representation if and only if there exist values $\rho_0,\rho_1,\mu_1,\nu_1,\rho_2,\mu_2,\nu_2,\dots,\rho_k,\mu_k,\nu_k$ satisfying Properties~Pr1--Pr5 of Lemma~\ref{le:structural-fixed-embedding}. We now present an algorithm that establishes the existence of such values. First, the algorithm fixes the values $\mu_2,\nu_2,\mu_3,\nu_3,\dots,\mu_{k-1},\nu_{k-1}$ without loss of generality; this is done by exploiting Lemmata~\ref{le:smaller-is-better} and~\ref{le:same-size-is-the-same}. Then the algorithm considers all the possible tuples $(\mu_1,\nu_1,\mu_k,\nu_k)$ and, for each such tuple, determines whether values $\rho_0,\rho_1,\dots,\rho_k$ exist such that Properties~Pr1--Pr5 are satisfied.
	
	First, we fix $\mu_i=\nu_i=0\degree$ for every $i\in \{2,\dots,k-1\}$ such that $G_i$ is a trivial $uv$-subgraph of $G$  (as required by Property~Pr2). 
	
	Second, we consider each non-trivial $uv$-subgraph $G_i$ of $G$ such that $i\in \{2,\dots,k-1\}$ and such that $G_{i-1}$ and $G_{i+1}$ are both trivial. By Properties~Pr1 and~Pr3, we have $\mu_i\leq \alpha_{i-1}:=360\degree - \ell(u_{i-1},f^{uv}_{\mathcal E}) - \ell(u_{i-1},f^{*}_{\mathcal E})$ and $\nu_i\leq \alpha_i:=360\degree - \ell(u_{i},f^{uv}_{\mathcal E}) - \ell(u_{i},f^{*}_{\mathcal E})$; note that $\alpha_{i-1} \leq 180\degree$ and $\alpha_i \leq 180\degree$. Hence, if $G_i$ admits no $(\mathcal E_i,\ell_i,\mu_i,\nu_i)$-representation with $\mu_i,\nu_i\in \{90\degree,180\degree\}$, with $\mu_i\leq \alpha_{i-1}$, and with $\nu_i\leq \alpha_i$, then we conclude that $G$ admits no $(\mathcal E,\ell,\mu,\nu)$-representation (by Properties Pr1,~Pr2, and~Pr3). Otherwise, we fix $\mu_i$ and $\nu_i$ so that $G_i$ admits an $(\mathcal E_i,\ell_i,\mu_i,\nu_i)$-representation, so that $\mu_i,\nu_i\in \{90\degree,180\degree\}$, so that $\mu_i\leq \alpha_{i-1}$ and $\nu_i\leq \alpha_i$, and so that $\mu_i+\nu_i$ is minimum. This is not a loss of generality by Lemma~\ref{le:same-size-is-the-same}. 
	
	Third, we consider each maximal sequence $G_i,G_{i+1},\dots,G_j$ of non-trivial $uv$-subgraphs of $G$ such that 
	$i,j\in \{1,\dots,k\}$ and such that $i<j$. For $h=i,\dots,j-1$, by Properties~Pr1 and~Pr3, we have $\nu_h+\mu_{h+1}\leq \alpha_h:=360\degree - \ell(u_{h},f^{uv}_{\mathcal E}) - \ell(u_h,f^{*}_{\mathcal E})$; note that $\alpha_h \leq 180\degree$. Further, by Property~Pr2, we have $\nu_h\geq 90\degree$ and $\mu_{h+1}\geq 90\degree$. Hence, if $\ell(u_{h},f^{uv}_{\mathcal E})>90\degree$ or if $\ell(u_h,f^{*}_{\mathcal E})>90\degree$, then we conclude that $G$ admits no $(\mathcal E,\ell,\mu,\nu)$-representation, otherwise  we can fix $\nu_i=\mu_{i+1}=\nu_{i+1}= \dots=\mu_j =90\degree$ without loss of generality (since the values $\nu_1$ and $\mu_k$ will be dealt with later, we actually only fix $\nu_i=90\degree$ and $\mu_j=90\degree$ if $i\geq 2$ and $j\leq k-1$, respectively). Now, if $G_h$ admits no $(\mathcal E_h,\ell_h,90\degree,90\degree)$-representation, for some $h\in \{i+1,i+2,\dots,j-1\}$, then we conclude that $G$ admits no $(\mathcal E,\ell,\mu,\nu)$-representation. If $i\geq 2$, we also fix the value of $\mu_i$ as follows (if $j\leq k-1$, the value of $\nu_j$ is fixed symmetrically to the one of $\mu_i$); note that $G_{i-1}$ is trivial, by the maximality of the sequence $G_i,G_{i+1},\dots,G_j$. By Properties~Pr1 and~Pr3, we have $\mu_i\leq \alpha_{i-1}:=360\degree - \ell(u_{i-1},f^{uv}_{\mathcal E}) - \ell(u_{i-1},f^{*}_{\mathcal E})$; note that $\alpha_{i-1} \leq 180\degree$.
	We check whether $G_i$ admits an $(\mathcal E_i,\ell_i,\mu_i,90\degree)$-representation with $\mu_i\in \{90\degree,180\degree\}$ and with $\mu_i\leq \alpha_{i-1}$. If that is not the case, then we conclude that $G$ admits no $(\mathcal E,\ell,\mu,\nu)$-representation, otherwise we fix $\mu_i$ so that $G_i$ admits an $(\mathcal E_i,\ell_i,\mu_i,90\degree)$-representation, so that $\mu_i\in \{90\degree,180\degree\}$, and so that $\mu_i$ is minimum. This is not a loss of generality by Lemma~\ref{le:smaller-is-better}.

	%
	
	The values of $\mu_2,\nu_2,\dots,\mu_{k-1},\nu_{k-1}$ have now been fixed without loss of generality.  We check each of the $3^4\in O(1)$ tuples $(\mu_1,\nu_1,\mu_k,\nu_k)$ such that $\mu_1,\nu_1,\mu_k, \nu_k\in \{0\degree,90\degree,180\degree\}$. We discard a tuple $(\mu_1,\nu_1,\mu_k,\nu_k)$~if:
	\begin{itemize}
		\item  $G_1$ is trivial and $\max\{\mu_1,\nu_1\} > 0\degree$ (by Property~Pr2);
		\item $G_k$ is trivial and $\max\{\mu_k,\nu_k\} > 0\degree$ (by Property~Pr2); 
		\item  $G_1$ is non-trivial and $\min\{\mu_1,\nu_1\} = 0\degree$ (by Property~Pr2);
		\item $G_k$ is non-trivial and $\min\{\mu_k,\nu_k\} = 0\degree$ (by Property~Pr2); 
		\item $G_1$ is non-trivial and does not admit an $(\mathcal E_1,\ell_1,\mu_1,\nu_1)$-representation (by Property~Pr2);
		\item $G_k$ is non-trivial and does not admit an $(\mathcal E_k,\ell_k,\mu_k,\nu_k)$-representation (by Property~Pr2);
		\item $\mu_1 + \ell(u_0,f^{uv}_{\mathcal E})>\mu$ (by Properties~Pr1 and~Pr4); 
		\item $\nu_k + \ell(u_k,f^{uv}_{\mathcal E})>\nu$ (by Properties~Pr1 and~Pr4);
		\item $\nu_1+\mu_2>360\degree - \ell(u_1,f^{uv}_{\mathcal E}) - \ell(u_1,f^{*}_{\mathcal E})$ (by Properties~Pr1 and~Pr3); or
		\item $\nu_{k-1}+\mu_k> 360\degree - \ell(u_{k-1},f^{uv}_{\mathcal E}) - \ell(u_{k-1},f^{*}_{\mathcal E})$ (by Properties~Pr1 and~Pr3).
	\end{itemize}
	
	If we discarded all the tuples $(\mu_1,\nu_1,\mu_k,\nu_k)$, then we conclude that $G$ admits no $(\mathcal E,\ell,\mu,\nu)$-representation. Otherwise, there is a constant number of sequences $\mu_1,\nu_1,\mu_2,\nu_2,\dots,\mu_{k},\nu_{k}$ such that $\mu_2,\nu_2,\dots,\mu_{k-1},\nu_{k-1}$ are the values fixed as above and such that the tuple $(\mu_1,\nu_1,\mu_k,\nu_k)$ was not discarded. We call {\em feasible} each of these sequences. 
	
	By the above arguments, we have that there exist values $\rho_0, \rho_1, \mu_1, \nu_1, \rho_2, \mu_2$, $\nu_2, \dots, \rho_k, \mu_k, \nu_k$ satisfying Properties~Pr1--Pr5 of Lemma~\ref{le:structural-fixed-embedding} only if the values $\mu_1,\nu_1,\mu_2,\nu_2$, $\dots$, $\mu_{k},\nu_{k}$ form a feasible sequence. Hence, in the following we check, for each feasible sequence $\mu_1,\nu_1,\mu_2,\nu_2,\dots,\mu_{k},\nu_{k}$, whether values $\rho_0,\rho_1,\dots,\rho_k$ exist so that the values $\rho_0,\rho_1,\mu_1,\nu_1,\rho_2,\mu_2,\nu_2,\dots,\rho_k,\mu_k,\nu_k$ satisfy Properties~Pr1--Pr5 of Lemma~\ref{le:structural-fixed-embedding}. 
	
	The values $\mu_1,\nu_1,\mu_2,\nu_2,\dots,\mu_{k},\nu_{k}$ determine lower and upper bounds for the values $\rho_0,\rho_1,\dots,\rho_k$. In particular, for each $i=1,\dots,k-1$, we have $\rho_i\geq {\mathcal L}_i :=\ell(u_i,f^{uv}_{\mathcal E})$ (by Property~Pr1) and $\rho_i\leq {\mathcal U}_i := 360\degree-\ell(u_i,f^*_{\cal E}) - \nu_i - \mu_{i+1}$ (by Property~Pr3); the choices we performed above for the values $\mu_1,\nu_1,\mu_2,\nu_2,\dots,\mu_{k},\nu_{k}$ guarantee that ${\mathcal L}_i\leq {\mathcal U}_i$, as we did not conclude that $G$ admits no $(\mathcal E,\ell,\mu,\nu)$-representation. Further, we have a tight bound of ${\mathcal L}_0:={\mathcal U}_0=\mu-\mu_1$ for $\rho_0$ and a tight bound of ${\mathcal L}_k:={\mathcal U}_k=\nu-\nu_k$ for $\rho_k$ (by Property~Pr4). 
	
	Let $\mathcal L=\sum_{i=0}^k {\mathcal L}_i$ and $\mathcal U=\sum_{i=0}^k {\mathcal U}_i$. We claim that there exist values $\rho_0,\rho_1,\dots,\rho_k$ such that the sequence $\rho_0,\rho_1,\mu_1,\nu_1,\rho_2,\mu_2,\nu_2,\dots,\rho_k,\mu_k,\nu_k$ satisfies Properties~Pr1--Pr5 of Lemma~\ref{le:structural-fixed-embedding} if and only if $\mathcal L + \sum_{i=1}^k {\sigma}_i\leq (k-1)\cdot 180\degree \leq \mathcal U + \sum_{i=1}^k {\sigma}_i$. 
	
	We now prove this claim. The necessity trivially comes from Property~Pr5, given that $\mathcal L\leq \sum_{i=0}^k {\rho}_i \leq\mathcal U$. For the sufficiency, let $\mathcal A:= (k-1)\cdot 180\degree - \mathcal L - \sum_{i=1}^k {\sigma}_i$; then $\mathcal A\geq 0$. We fix the values $\rho_0,\rho_1,\dots,\rho_k$ in this order. While $\mathcal A>0$, we fix the value $\rho_i$ to $\min\{{\mathcal U}_i,{\mathcal L}_i+\mathcal A\}$ and then decrease $\mathcal A$ by $\rho_i-{\mathcal L}_i$. Since $(k-1)\cdot 180\degree - \sum_{i=1}^k {\sigma}_i\leq {\mathcal U}=\sum_{i=0}^k {\mathcal U}_i$, either $\mathcal A= 0$ to start with, or there exists an index $j\in \{0,1,\dots,k\}$ such that, after fixing $\rho_j$ and decreasing $\mathcal A$ by $\rho_j-{\mathcal L}_j$, we have $\mathcal A=0$; then we fix the remaining values $\rho_{j+1},\rho_{j+2},\dots,\rho_k$, if any, to ${\mathcal L}_{j+1},{\mathcal L}_{j+2},\dots,{\mathcal L}_k$, respectively. Since ${\mathcal L}_i \leq \rho_i \leq {\mathcal U}_i$, we have that the sequence $\rho_0,\rho_1,\mu_1,\nu_1,\rho_2,\mu_2,\nu_2,\dots,\rho_k,\mu_k,\nu_k$ satisfies Properties~Pr1, Pr3, and Pr4. Property~Pr2 is satisfied independently of the choice of $\rho_0,\rho_1,\dots,\rho_k$, given that $\mu_1,\nu_1,\mu_2,\nu_2,\dots,\mu_{k},\nu_{k}$ is a feasible sequence. Finally, the sequence $\rho_0,\rho_1,\mu_1,\nu_1,\rho_2,\mu_2,\nu_2,\dots,\rho_k,\mu_k,\nu_k$ satisfies Property~Pr5 by the just described construction. 
	
	It remains to analyze the running time of the described algorithm. The first part of the algorithm either concludes that $G$ admits no $(\mathcal E,\ell,\mu,\nu)$-representation or fixes the values $\mu_2,\nu_2,\mu_3,\nu_3,\dots,\mu_{k-1},\nu_{k-1}$. This requires to check whether each $uv$-subgraph $G_i$ of $G$ is trivial or not, and to handle the $\Theta(1)$ pairs $(\mu_i,\nu_i)$ such that $G_i$ admits an $(\mathcal E_i,\ell_i,\mu_i,\nu_i)$-representation. By assumption, this information is known, hence this computation takes $O(1)$ time for $G_i$, and hence $O(k)$ time over all the $uv$-subgraphs of $G$. The second part of the algorithm considers the $O(1)$ tuples $(\mu_1,\nu_1,\mu_k,\nu_k)$ such that $\mu_1,\nu_1,\mu_k, \nu_k\in \{0\degree,90\degree,180\degree\}$; for each of them, we have that $O(1)$ conditions, each of which can be checked in $O(1)$ time, determine whether the tuple has to be discarded or not. Hence, this part of the algorithm has $O(1)$ running time. The third part of the algorithm is repeated $O(1)$ times, namely once for each feasible sequence $\mu_1,\nu_1,\mu_2,\nu_2,\dots,\mu_{k},\nu_{k}$. This part starts by determining lower bounds $\mathcal L_0,\mathcal L_1,\dots,\mathcal L_k$ and upper bounds $\mathcal U_0,\mathcal U_1,\dots,\mathcal U_k$ for the values $\rho_0,\rho_1,\dots,\rho_k$; this is done in $O(1)$ time per value $\rho_i$, hence in $O(k)$ time in total. After computing the sums $\mathcal L=\sum_{i=0}^k {\mathcal L}_i$, $\mathcal U=\sum_{i=0}^k {\mathcal U}_i$, and $\sum_{i=1}^k \sigma_i$ in $O(k)$ time (note that each sum has $O(k)$ terms and that it is known whether $G_i$ is internal to $\mathcal C^*_{uv}$ in $\cal E$, and hence whether $\sigma_i=0\degree$ or $\sigma_i=\mu_i+\nu_i$), the algorithm simply checks whether $\mathcal L + \sum_{i=1}^k {\sigma}_i\leq (k-1)\cdot 180\degree \leq \mathcal U + \sum_{i=1}^k {\sigma}_i$. Finally, the values $\rho_0,\rho_1,\dots,\rho_k$ are computed by executing $O(k)$ times the minimum between two values and a subtraction, which is done in $O(1)$ time. 
\end{proof}

We are now ready to prove the following.

\begin{theorem} \label{th:2-con-lowers}
	Let $G$ be an $n$-vertex $2$-connected outerplanar graph with a prescribed plane embedding $\mathcal E$. For each face $f$ of $\cal E$ and each vertex $w$ incident to $f$, let $\ell(w,f)$ be a value in $\{90\degree,180\degree,270\degree\}$. There is an $O(n)$-time algorithm that tests whether a function $\phi$ exists such that $(\cal E,\phi)$ is an $\ell$-constrained representation of $G$. Further, if $G$ admits such a representation, the algorithm constructs the function $\phi$ in $O(n)$ time.
\end{theorem}

\begin{proof}
	Assume that the input contains the following information:
	
	\begin{itemize}
		\item for each vertex $v$ of $G$, a circular list $L_{\mathcal E}(v)$, which represents the clockwise order of the edges incident to $v$ in $\mathcal E$;
		\item for each face $f$ of $\mathcal E$, the clockwise order of the edges along the boundary of $f$; and 
		\item for each face $f$ of $\cal E$ and each vertex $w$ incident to $f$, the value $\ell(w,f)\in \{90\degree,180\degree,270\degree\}$.
	\end{itemize}
	
	First, we compute in $O(n)$ time the outerplane embedding $\mathcal O$ of $G$~\cite{cnao-lta-85,d-iroga-07,ht-ept-74,m-laarogmog-79,w-rolt-87}. We now have, for each vertex $v$ of $G$, a circular list $L_{\mathcal O}(v)$, which represents the clockwise order of the edges incident to $v$ in $\mathcal O$. We also compute, for each face $f$ of $\mathcal O$, the clockwise order of the edges along the boundary of $f$; these orders can be recovered from the lists $L_{\mathcal O}(v)$ in $O(n)$ time. We next compute in $O(n)$ time the extended dual tree $\mathcal T$ of $\mathcal O$. Each internal node $s\in \mathcal T$ is associated with the cycle $\mathcal C_s$ delimiting the internal face of $\mathcal O$ dual to $s$.
	
	We now find an edge $uv$ incident to $f^*_{\mathcal O}$ and to $f^*_{\mathcal E}$; this can be done in $O(n)$ time by marking all the edges incident to $f^*_{\mathcal O}$ and to $f^*_{\mathcal E}$, and by then checking whether any edge is marked twice. If no such an edge exists, then by Lemma~\ref{le:no-outer-edge} we conclude that $G$ has no $\ell$-constrained representation $(\mathcal E,\phi)$. Otherwise, we root $G$ at $uv$ and we root $\mathcal T$ at the leaf $r^*$ such that the edge $rr^*$ of $T$ incident to $r^*$ is dual to $uv$. 
	
	Consider a non-leaf node $s$ of $\mathcal T$; let $p$ be the parent of $s$ and let $v^s_0,v^s_1,\dots,v^s_{k_s}$ be the clockwise order of the vertices along $\mathcal C_s$ in $\cal O$, where $v^s_0v^s_{k_s}$ is the edge dual to  $sp$. Further, we denote by $\mathcal T_s$ the subtree of $\mathcal T$ rooted at $s$, by $G_s$ the subgraph of $G$ defined as $\bigcup_{t\in \mathcal T_s}\mathcal C_t$, and by $\mathcal E_s$ the restriction of $\mathcal E$ to the vertices and edges of $G_s$. Note that $G_r=G$, where $r$ is the child of $r^*$ in $T$. For a leaf node $s\neq r^*$ of $\mathcal T$, with a slight overload of notation, we denote by $G_s$ the edge of $G$ that is dual to the edge of $T$ incident to $s$. For a node $s\notin \{r,r^*\}$ of $\mathcal T$, the graph $G_s$ is rooted at the edge $u_s v_s$ that is shared by $\mathcal C_s$ and $\mathcal C_p$; assume, w.l.o.g., that $u_s=v^s_0$ and $v_s=v^s_{k_s}$. Then, for a node $s$ of $\mathcal T$ with children $s_1, \dots, s_{k_s}$, the graphs $G_{s_1},\dots,G_{s_{k_s}}$ are the $u_sv_s$-subgraphs of $G_s$.
	
	We start by associating a boolean value $\omega_{ps}$ to each edge $ps$ of $T$, where $p$ is the parent of $s$ and both $p$ and $s$ are not leaves of $T$; the value $\omega_{ps}$ is {\sc true} if $\mathcal C_p$ and $\mathcal C_s$ lie one outside the other in $\mathcal E$, and it is {\sc false} if $\mathcal C_s$ lies inside $\mathcal C_p$ in $\mathcal E$; the planarity of $\mathcal E$ and the assumption that $uv$ is incident to $f^*_{\mathcal E}$ imply that $\mathcal C_p$ does not lie inside $\mathcal C_s$ in $\mathcal E$. The value $\omega_{ps}$ can be computed in $O(1)$ time by looking at the list $L_{\mathcal E}(u_s)$.

	We now perform a bottom-up visit of $\mathcal T$ which ends after processing $r$. There is nothing to be done on the leaves of $\mathcal T$. When processing an internal node $s$ of $\mathcal T$, we check whether $G_s$ admits an $(\mathcal E_s,\ell_s,\mu,\nu)$-representation, for each of the $9$ possible pairs $(\mu,\nu)$ with $\mu,\nu\in \{90\degree,180\degree,270\degree\}$. This can be done in $O(k_s)$ time by Lemma~\ref{le:find-values}; namely, for every $u_s v_s$-subgraph $G_{s_i}$ of $G_s$, the following information is known: 
	\begin{enumerate} [(i)]
		\item whether $G_{s_i}$ is trivial or not (indeed, $G_{s_i}$ is trivial if and only if $s_i$ is a leaf of $\mathcal T$);
		\item if $G_{s_i}$ is non-trivial, whether it lies inside or outside $G_s$ (this information is in the label $\omega_{ss_i}$); and
		\item if $G_{s_i}$ is non-trivial, whether it admits an $(\mathcal E_{s_i},\ell_{s_i},\mu_i,\nu_i)$-representation or not, for every pair $(\mu_i,\nu_i)$ with $\mu_i,\nu_i\in \{90\degree,180\degree\}$ (this information has been computed when processing $s_i$).
	\end{enumerate}
	Observe that the function $\ell_s$ does not need to be explicitly computed in order to apply Lemma~\ref{le:find-values}. Indeed, it suffices to recover the values $\ell(v^s_i,f^{u_sv_s}_{\mathcal E_s})$ and $\ell(v^s_i,f^*_{\mathcal E_s})$, for each $i=0,1,\dots,k_s$, where $f^{u_sv_s}_{\mathcal E_s}$ is the internal face of $\mathcal E_s$ incident to $u_sv_s$ and $f^*_{\mathcal E_s}$ is the outer face of $\mathcal E_s$. This can be done in $O(1)$ time per value, and hence in total $O(k_s)$ time, as follows; refer to Figure~\ref{fig:lower-values}. The value $\ell(v^s_i,f^{u_sv_s}_{\mathcal E_s})$ can be recovered in $O(1)$ time, as it is equal to $\ell(v^s_i,f^{u_sv_s}_{\mathcal E})$. The value $\ell(v^s_i,f^*_{\mathcal E_s})$ can also be recovered in $O(1)$ time. Namely, if $1\leq i\leq k_s-1$, then we have $\ell(v^s_i,f^*_{\mathcal E_s})=\ell(v^s_i,f^*_{\mathcal E})$. Further, if $i=0$ or if $i=k_s$ (that is, if $v^s_i=u_s$ or if $v^s_i=v_s$), then there might be up to three values which need to be added in order to obtain $\ell(v^s_i,f^*_{\mathcal E_s})$, namely $\ell(v^s_i,f^*_{\mathcal E})$ and the values $\ell(v^s_i,f)$ for the internal faces $f$ of ${\mathcal E}$ incident to $u_s$ and $v_s$ whose corresponding face in $\mathcal E_s$ is $f^*_{\mathcal E_s}$. These faces are delimited by edges of $G$ that are incident to $v^s_i$ and that are not in $G_s$. Each of these $O(1)$ values can be accessed in $O(1)$ time, hence $\ell(v^s_i,f^*_{\mathcal E_s})$ can be computed in $O(1)$ time also in this case.	
	\begin{figure}[htb]
		\centering
		\includegraphics[scale=0.7]{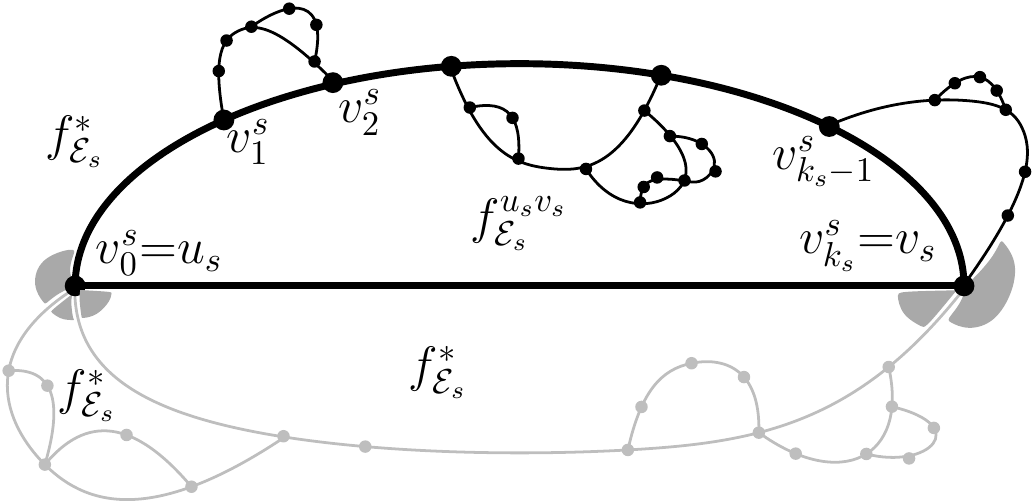}
		\caption{Illustration for the computation of the values $\ell(v^s_i,f^{u_sv_s}_{\mathcal E_s})$ and $\ell(v^s_i,f^*_{\mathcal E_s})$. The graph $G_s$ is black and the cycle $\mathcal C_s$ is thick. The gray vertices and edges belong to $G$ but not to $G_s$. The faces of $\mathcal E$ incident to $u_s$ and $v_s$ whose corresponding face in $\mathcal E_s$ is $f^*_{\mathcal E_s}$ are colored dark gray, close to their incidence with $u_s$ and $v_s$. }
		\label{fig:lower-values}
	\end{figure}
	
	If there is no pair $(\mu,\nu)$ with $\mu,\nu\in \{90\degree,180\degree,270\degree\}$ such that $G_s$ admits an $(\mathcal E_s,\ell_s,\mu,\nu)$-representation, then we conclude that $G$ has no $\ell$-constrained representation $(\cal E,\phi)$, \mbox{otherwise we continue the visit of $\mathcal T$.}
	
	After processing $r$, we have that $G$ has an $\ell$-constrained representation $(\cal E,\phi)$ if and only if there is a pair $(\mu,\nu)$ with $\mu,\nu\in \{90\degree,180\degree,270\degree\}$ such that $G_r=G$ admits an $(\mathcal E_r=\mathcal E,\ell,\mu,\nu)$-representation.  
	
	Concerning the running time, we have that the algorithm processes a node $s$ of $\mathcal T$ in $O(k_s)$ time, hence it takes $O(n)$ time over the entire tree $\mathcal T$. 
	
	Finally, we describe how to modify the algorithm so that it constructs in $O(n)$ time an $\ell$-constrained representation $(\cal E,\phi)$ of $G$. When we process an internal node $s$ of $\mathcal T$ during the bottom-up visit of $\mathcal T$, for each of the $9$ possible pairs $(\mu,\nu)$ with $\mu,\nu\in \{90\degree,180\degree,270\degree\}$ such that $G_s$ admits an $(\mathcal E_s,\ell_s,\mu,\nu)$-representation with plane embedding $\mathcal E_s$, we store the $O(k_s)$ values $\rho^s_0,\rho^s_1,\mu^s_1,\nu^s_1,\rho^s_2,\mu^s_2,\nu^s_2,\dots,\rho^s_{k_s},\mu^s_{k_s},\nu^s_{k_s}$ satisfying Properties~Pr1--Pr5 of Lemma~\ref{le:structural-fixed-embedding}; by Lemma~\ref{le:find-values}, these values can be found in $O(k_s)$ time. After processing $r$, we perform a top-down visit of $\mathcal T$. When we start processing a node $s$ of $\mathcal T$, we assume that a pair $(\mu,\nu)$ is already associated to $s$. Initially, for $r$, such a pair $(\mu,\nu)$ is the one for which we concluded that an $(\mathcal E,\ell,\mu,\nu)$-representation of $G$ exists. For each internal node $s$ of $\mathcal T$ and for each $i=0,1,\dots,k_s$, we set $\phi(v^s_i,f^{u_sv_s}_{\mathcal E_s})=\rho^s_i$ and we associate the pair $(\mu^s_i,\nu^s_i)$ to the child $s_i$ of $s$; clearly, this takes $O(k_s)$ time and hence $O(n)$ time for the entire tree $\mathcal T$. After the completion of the top-down visit of $\mathcal T$, for each vertex $w$, the value $\phi(w,f)$ has been determined for all the faces $f$ of $\cal E$ incident to $w$, except for one face $f_w$; this is the unique face of $\cal E$ incident to $w$ whose corresponding face in ${\mathcal E_s}$ is $f^*_{\mathcal E_s}$, where $s$ is the first node that is encountered in the top-down visit of $\mathcal T$ such that $w$ is a vertex of $\mathcal C_s$. We complete the rectilinear representation $(\mathcal E,\phi)$ of $G$ by setting, for each vertex $w$ of $G$, $\phi(w,f_w)=360\degree-\sum_f \phi(w,f)$, where the sum is over all the faces $f\neq f_w$ of $\mathcal E$ incident to $w$. By Tamassia's results~\cite{t-eggmnb-87},  a planar rectilinear drawing in the equivalence class $(\mathcal E,\phi)$ can be constructed in $O(n)$ time.
\end{proof}

%
%

\subsection{Simply-Connected Outerplanar Graphs} \label{se:fixed-cutvertices}

In this section we show how to test in linear time whether an $n$-vertex outerplanar graph $G$ admits a planar rectilinear drawing with a prescribed plane embedding $\mathcal E$. By Tamassia's results~\cite{t-eggmnb-87}, it suffices to test in $O(n)$ time whether $G$ admits a rectilinear representation $(\mathcal E,\phi)$. 

Consider the {\em block-cut-vertex tree} $T$ of $G$~\cite{h-gt-69,ht-aeagm-73}. We denote by $G_b$ the block corresponding to a B-node $b$, by $n_b$ the number of vertices of $G_b$, and by $\mathcal E_b$ and $\mathcal O_b$ the restrictions of $\cal E$ and $\cal O$ to $G_b$, respectively. The goal is now to reduce the problem of testing whether a function $\phi$ exists such that $G$ admits a rectilinear representation $(\mathcal E,\phi)$ to the problem of testing, for a suitable function $\ell_b$, whether a function $\phi_b$ exists for each individual non-trivial block $G_b$ of $G$ such that $(\mathcal E_b,\phi_b)$ is an $\ell_b$-constrained representation of $G_b$. 

Consider a cut-vertex $c$ of $G$. Since $G$ contains at most four edges incident to $c$ and since each non-trivial block incident to $c$ contains at least two edges incident to $c$, we have that $c$ is of one of the following types (refer to Fig.~\ref{fig:around-cutvertex}):
\begin{figure}[tb]\tabcolsep=4pt
	\centering
	\begin{tabular}{c c c c c}
		\includegraphics[scale=0.7]{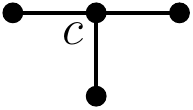} \hspace{1mm} &
		\includegraphics[scale=0.7]{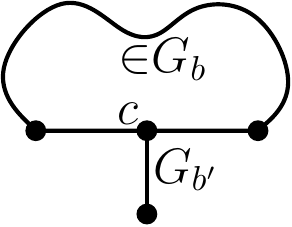} \hspace{1mm} &
		\includegraphics[scale=0.7]{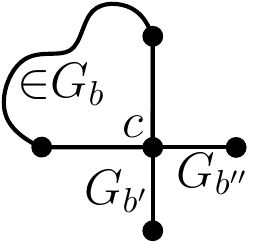} \hspace{1mm} &
		\includegraphics[scale=0.7]{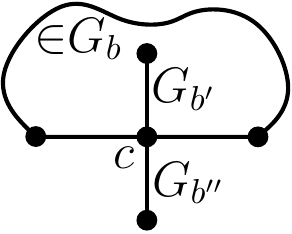} \hspace{1mm} &
		\includegraphics[scale=0.7]{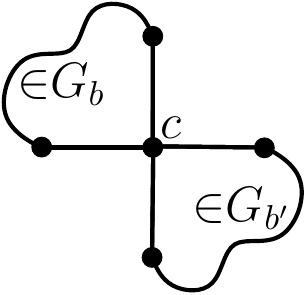} \\
		(1) \hspace{1mm} & (2) \hspace{1mm} & (3) \hspace{1mm} & (4) \hspace{1mm} & (5)\\
	\end{tabular}
	\caption{Blocks incident to a cut-vertex $c$ of Type~(1)--(5).}
	\label{fig:around-cutvertex}
\end{figure}
\begin{itemize}
	\item {\em Type (1)}: all the (at least two) blocks incident to $c$ are trivial;
	\item {\em Type (2)}: $c$ is incident to one non-trivial block $G_b$ and one trivial block $G_{b'}$;
	\item {\em Type (3)}:  $c$ is incident to one non-trivial block $G_b$ and two trivial blocks $G_{b'}$ and $G_{b''}$; further, the edges $G_{b'}$ and $G_{b''}$ are consecutive in the clockwise order of the edges incident to $c$ in $\mathcal E$;
	\item {\em Type (4)}:  $c$ is incident to one non-trivial block $G_{b}$ and two trivial blocks $G_{b'}$ and $G_{b''}$; further, the edges $G_{b'}$ and $G_{b''}$ are not consecutive in the clockwise order of the edges incident to $c$ in $\mathcal E$; and 
	\item {\em Type (5)}:  $c$ is incident to two non-trivial blocks $G_{b}$ and $G_{b'}$. 
\end{itemize}

Note that if $c$ is of Type (5), by the planarity of $\mathcal E$, the edges of $G_{b}$ incident to $c$ are consecutive in the clockwise order of the edges incident to $c$ in $\mathcal E$. 

Some lower bounds for the angles at the vertices of the non-trivial blocks of $G$ incident to $c$ can be established, according to the type of $c$.

\begin{itemize}
	\item If $c$ is of Type (2), let $f$ be the face of $\mathcal E_b$ that \emph{contains $G_{b'}$ in $\mathcal E$} (meaning that the cycle delimiting $f$ also delimits the face of $\mathcal E$ the edge $G_{b'}$ is incident to); then we let $\ell_b(c,f)=180\degree$.
	\item If $c$ is of Type (3), let $f$ be the face of $\mathcal E_b$ that contains $G_{b'}$ and $G_{b''}$ in $\mathcal E$; then we let $\ell_b(c,f)=270\degree$.	
	\item If $c$ is of Type (4), let $f'$ and $f''$ be the faces of $\mathcal E_b$ that contain $G_{b'}$ and $G_{b''}$ in $\mathcal E$, respectively; then we let $\ell_b(c,f')=180\degree$ and $\ell_b(c,f'')=180\degree$.
	\item Finally, if $c$ is of Type (5), let $f$ be the face of $\mathcal E_b$ that contains $G_{b'}$ in $\mathcal E$ and let $f'$ be the face of $\mathcal E_{b'}$ that contains $G_{b}$ in $\mathcal E$; then we let $\ell_b(c,f)=270\degree$ and $\ell_{b'}(c,f')=270\degree$.	
\end{itemize}

For every non-trivial block $G_b$ of $G$, for every face $f$ of $\mathcal E_b$, and for every vertex $w$ of $G_b$ incident to $f$ such that no value for $\ell_b(w,f)$ has been established yet, we let $\ell_b(w,f)=90\degree$. We have the following.

\begin{lemma} \label{le:from-graph-to-blocks}
	A function $\phi$ such that $G$ admits a rectilinear representation $(\mathcal E,\phi)$ exists if and only if, for every non-trivial block $G_b$ of $G$, a function $\phi_b$ exists such that $(\mathcal E_b,\phi_b)$ is an $\ell_b$-constrained representation of $G_b$. Moreover, the function $\phi$ can be computed from the functions $\phi_b$ in $O(n)$ time.
\end{lemma}

\begin{proof}
	$(\Longrightarrow)$ Suppose that $G$ admits a rectilinear representation $(\mathcal E,\phi)$. For any non-trivial block $G_b$ of $G$, let $(\mathcal E_b,\phi_b)$ be the restriction of $(\mathcal E,\phi)$ to $G_b$. We need to prove that $(\mathcal E_b,\phi_b)$ is an $\ell_b$-constrained representation of $G_b$. That is, for each vertex $w$ of $G_b$ and each face $f$ of $\mathcal E_b$ incident to $w$, we need to prove that $\phi_b(w,f)\geq \ell_b(w,f)$. This is trivially true if $\ell_b(w,f)=90\degree$, hence assume that $\ell_b(w,f)\geq 180\degree$, which implies that $w$ is a cut-vertex of $G$. We distinguish three cases based on the type of $c$; refer to Figure~\ref{fig:Fixed-Simply-necessity}.
	\begin{figure}[htb]\tabcolsep=4pt
		\centering
		\begin{tabular}{c c c}
			\includegraphics[scale=0.7]{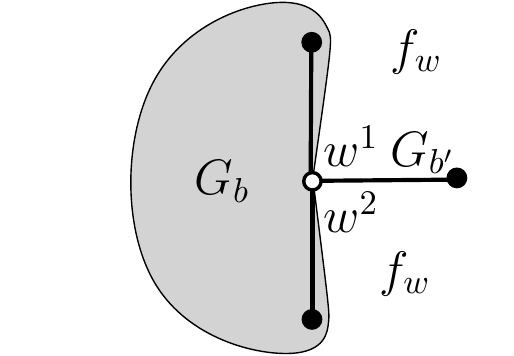} \hspace{1mm} &
			\includegraphics[scale=0.7]{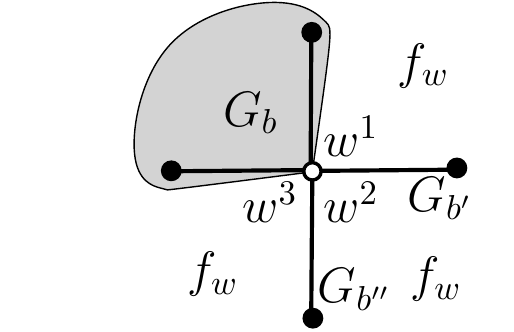} \hspace{1mm} &
			\includegraphics[scale=0.7]{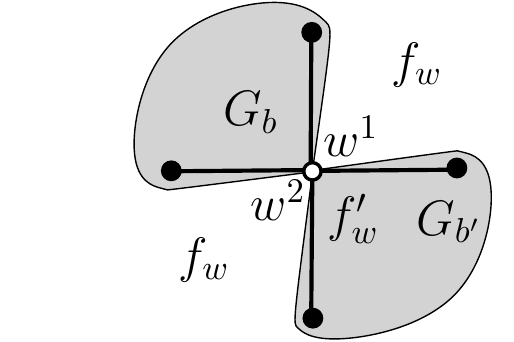}\\
			(2) and (4) \hspace{1mm} & (3) \hspace{1mm} & (5)
		\end{tabular}
		\caption{Illustration for the proof that $\phi_b(w,f)\geq \ell_b(w,f)$ if $c$ is of Type~(2) or~(4), if $c$ is of Type~(3), and if $c$ is of Type~(5). The empty disk represents $w$.}
		\label{fig:Fixed-Simply-necessity}
	\end{figure}
	\begin{itemize}
		\item If $c$ is of Type~(2) or~(4), we have $\ell_b(w,f)= 180\degree$ and there exists a trivial block $G_{b'}$ incident to $w$ such that $f$ contains $G_{b'}$. Let $f_w$ be the unique face of $\mathcal E$ whose corresponding face in $\mathcal E_b$ is $f$; further, let $w^1$ and $w^2$ be the occurrences of $w$ along the boundary of $f_w$. Since $(\mathcal E,\phi)$ is a rectilinear representation of $G$ and the degree of $w$ in $G$ is larger than $1$, we have $\phi(w^1,f_w)\geq 90\degree$ and $\phi(w^2,f_w)\geq 90\degree$. Since $\phi_b(w,f)=\phi(w^1,f_w)+\phi(w^2,f_w)$, it follows that $\phi_b(w,f)\geq \ell_b(w,f)$.
		\item If $c$ is of Type~(3), we have $\ell_b(w,f)= 270\degree$ and there exist two trivial blocks $G_{b'}$ and $G_{b''}$ incident to $w$ such that $f$ contains $G_{b'}$ and $G_{b''}$. Let $f_w$ be the unique face of $\mathcal E$ whose corresponding face in $\mathcal E_b$ is $f$ and let $w^1$, $w^2$, and $w^3$ be the occurrences of $w$ along the boundary of $f_w$. Since $(\mathcal E,\phi)$ is a rectilinear representation of $G$ and the degree of $w$ in $G$ is larger than $1$, we have $\phi(w^1,f_w)\geq 90\degree$, $\phi(w^2,f_w)\geq 90\degree$, and $\phi(w^3,f_w)\geq 90\degree$. Since $\phi_b(w,f)=\phi(w^1,f_w)+\phi(w^2,f_w)+\phi(w^3,f_w)$, it follows that $\phi_b(w,f)\geq \ell_b(w,f)$.
		\item If $c$ is of Type~(5), we have $\ell_b(w,f)= 270\degree$ and there exists a non-trivial block $G_{b'}$ different from $G_b$ and incident to $w$ such that $f$ contains $G_{b'}$. Let $f_w$ and $f'_w$ be the faces of $\mathcal E$ whose corresponding face in $\mathcal E_b$ is $f$, where $f_w$ is incident to edges of $G_b$, while $f'_w$ is not. Further, let $w^1$ and $w^2$ be the occurrences of $w$ along the boundary of $f_w$. Since $(\mathcal E,\phi)$ is a rectilinear representation of $G$ and the degree of $w$ in $G$ is larger than $1$, we have $\phi(w^1,f_w)\geq 90\degree$, $\phi(w^2,f_w)\geq 90\degree$, and $\phi(w,f'_w)\geq 90\degree$. Since $\phi_b(w,f)=\phi(w^1,f_w)+\phi(w^2,f_w)+\phi(w,f'_w)$, it follows that $\phi_b(w,f)\geq \ell_b(w,f)$.
	\end{itemize}

	$(\Longleftarrow)$ Suppose now that every non-trivial block $G_b$ of $G$ admits an $\ell_b$-constrained representation $(\mathcal E_b,\phi_b)$. We show that $G$ admits a rectilinear representation $(\mathcal E,\phi)$.
	
	Let $b^*$ be any B-node of $T$ whose corresponding block $G_{b^*}$ of $G$ contains an edge incident to $f^*_{\mathcal E}$, let $G':=G_{b^*}$, and let $\phi':=\phi_{b^*}$.
	
	Now assume that a subgraph $G'$ of $G$ has been defined satisfying the following properties: 
	\begin{enumerate}[(i)]
		\item $G'$ is connected, consists of the union of a set of blocks of $G$, and contains edges incident to $f^*_{\mathcal E}$; and
		\item for any cut-vertex $w$ of $G$, either all the blocks of $G$ containing $w$ belong to $G'$ or at most one block of $G$ containing $w$ belongs to $G'$.
	\end{enumerate}
	
	Denote by $\mathcal E'$ the restriction of $\mathcal E$ to $G'$. Assume that a rectilinear representation $(\mathcal E',\phi')$ has been constructed satisfying the following property:
	
	\begin{enumerate}
		\item[(iii)] for every non-trivial block $G_b$ of $G$ that belongs to $G'$, the restriction of $(\mathcal E',\phi')$ to $G_b$ is $(\mathcal E_b,\phi_b)$.
	\end{enumerate}
	
	These assumptions are initially met with $G'=G_{b^*}$ and $\phi':=\phi_{b^*}$. Namely, Properties~(i) and~(iii) are trivially satisfied. Further, Property~(ii) is satisfied since, for any cut-vertex $w$ of $G$, at most one block of $G$ containing $w$ belongs to $G'$. 
	
	Now consider a cut-vertex $w$ of $G$ that belongs to $G'$ and such that $G\setminus G'$ contains edges incident to $w$. By Property~(ii) there exists exactly one block $G_b$ of $G$  that belongs to $G'$ and that contains $w$. Let $G_{b_1},\dots,G_{b_k}$ be the blocks of $G$ containing $w$ and different from $G_b$. Denote by $G''$ the subgraph of $G$ composed of $G',G_{b_1},\dots,G_{b_k}$; further, denote by $\mathcal E''$ the restriction of $\mathcal E$ to $G''$. Then $G''$ clearly satisfies Property~(i). Further, $G''$ satisfies Property~(ii) since $G'$ does, since all the blocks of $G$ containing $w$ belong to $G''$, and since for every cut-vertex $z$ of $G$ in $G''\setminus G'$, exactly one block $G_{b_i}$ of $G$ containing $z$ belongs to $G''$.
	
	We define a function $\phi''$ such that $(\mathcal E'',\phi'')$ is a join of $(\mathcal E',\phi'),(\mathcal E_{b_1},\phi_{b_1}), \dots, (\mathcal E_{b_k},\phi_{b_k})$; by Lemma~\ref{le:preliminaries-subgraphs-composition}, this implies that $(\mathcal E'',\phi'')$ is a rectilinear representation of $G''$ satisfying Property~(iii). For ease of notation, let $G_{b_0}:=G'$, let $\mathcal E_{b_0}:=\mathcal E'$, and let $\phi_{b_0}:=\phi'$. 
	
	First, consider any vertex $v\neq w$ of $G''$; then $v$ belongs to a single graph $G_{b_j}$ among $G_{b_0},G_{b_1},\dots,G_{b_k}$. For any face $f$ of $\mathcal E''$ incident to $v$, let $f_j$ be the face of $\mathcal E_{b_j}$ corresponding to $f$. Then, for every occurrence $v^x$ of $v$ along the boundary of $f$, we define $\phi''(v^x,f)=\phi_{b_j}(v^x,f_j)$.

	
	Second, we define the function $\phi''$ for the occurrences of $w$ along the boundaries of faces of $\mathcal E''$. This is done differently according to the type of $w$.

	\begin{itemize}
		\item If $w$ is of Type~(1), let $w^0,w^1,\dots,w^k$ be the occurrences of $w$ along the boundary of the unique face $f$ of $\mathcal E''$ incident to $w$. Then we define $\phi''(w^x,f)=90\degree$, for $x=1,\dots,k$, and  $\phi''(w^0,f)=360\degree-\sum_{x=1}^{k}\phi''(w^x,f)$. 
		\item If $w$ is of Type~(2), then $k=1$. Note that $w$ is incident to a non-trivial block $G_b$ and to a trivial block $G_{b'}$ of $G$. Assume that $G_b$ belongs to $G_{b_0}$ and that $G_{b'}$ coincides with $G_{b_1}$, as the case in which $G_{b'}$ belongs to $G_{b_0}$ and $G_{b}$ coincides with $G_{b_1}$ is analogous. Let $f$ be the face of $\mathcal E''$ incident to $w$ and incident to $G_{b'}$, and let $f_0$ be the face of $\mathcal E_{b_0}$ corresponding to $f$. Further, let $w^0$ and $w^1$ be the occurrences of $w$ along the boundary of $f$. We let $\phi''(w^0,f)=90\degree$ and $\phi''(w^1,f)=\phi_{b_0}(w,f_0)-90\degree$. Finally, 
		for every face $g\neq f$ of $\mathcal E''$ incident to $w$, let $g_0$ be the corresponding face of $\mathcal E_{b_0}$. Then we let $\phi''(w,g)=\phi_{b_0}(w,g_0)$.
		\item Finally, if $w$ is of one of the Types~(3)--(5), let $\phi''(w^k,f)=90\degree$ for each face $f$ of $\mathcal E''$ incident to $w$ and for each occurrence $w^k$ of $w$ along the boundary of $f$. 
	\end{itemize} 
	This concludes the definition of $\phi''$. We have the following.
	
	\begin{claimx} \label{cl:function-phi}
		We have that $(\mathcal E'',\phi'')$ is a join of  $(\mathcal E',\phi'),(\mathcal E_{b_1},\phi_{b_1}), \dots, (\mathcal E_{b_k},\phi_{b_k})$.
	\end{claimx}
	
	\begin{proof}
		First, Property~(a) of a join is trivially satisfied by $(\mathcal E'',\phi'')$. Namely, $\mathcal E''$ is a plane embedding since, by definition, it is the restriction of the plane embedding $\mathcal E$ of $G$ to $G''$. 
		
		Second, we make some progress towards a proof that Property~(d) of a join is satisfied by $(\mathcal E'',\phi'')$. For $j=0,1,\dots,k$, by definition we have that $\mathcal E_{b_j}$ is the restriction of $\mathcal E''$ to $G_{b_j}$. Moreover, for every face $f$ of $\mathcal E''$ and for every occurrence $v^x$ of a vertex $v\neq w$ along the boundary of $f$, by construction we have $\phi''(v^x,f)=\phi_{b_j}(v^x,f_j)$, where $G_{b_j}$ is the graph among $G_{b_0},G_{b_1},\dots,G_{b_k}$ the vertex $v$ belongs to and $f_j$ is the face of $\mathcal E_{b_j}$ corresponding to $f$. It follows that, in order to prove that Property~(d) of a join is satisfied by $(\mathcal E'',\phi'')$, we only need to prove that, for $j=0,1,\dots,k$ and for each face $f_j$ of $\mathcal E_{b_j}$ incident to $w$, we have $\phi_{b_j}(w,f_j)=\sum_{f,x} \phi''(w^x,f)$, where the sum is over all the faces $f$ of $\mathcal E''$ incident to $w$ whose corresponding face in $\mathcal E_{b_j}$ is $f_{j}$ and over all the occurrences $w^x$ of $w$ along the boundary of~$f$. 
		
		We now distinguish five cases, according to the type of $w$.
		
		\begin{itemize}
			\item Suppose first that $w$ is of Type~(1). Recall that $w^0,w^1,\dots,w^k$ are the occurrences of $w$ along the boundary of the unique face $f$ of $\mathcal E''$ incident to $w$. By construction, we have $\phi''(w^x,f)=90\degree$, for $x=1,\dots,k$, and $\phi''(w^0,f)=360\degree-\sum_{x=1}^{k}\phi''(w^x,f)$. Then Property~(c) of a join follows immediately and Property~(b) of a join follows from the fact $1\leq k\leq 3$, given that the degree of $w$ in $G''$ is larger than or equal to $2$ and smaller than or equal to $4$. For $j=0,1,\dots,k$, there is a unique face $f_j$ of $\mathcal E_{b_j}$ incident to $w$; further, $\phi_{b_j}(w,f_j)=360\degree$, given that $w$ has degree $1$ in $G_{b_j}$. Finally, $f_{j}$ corresponds to $f$ and $\sum_{x=0}^k\phi''(w^x,f)=360\degree$, by construction. Property~(d) of a join follows.	
			\item Suppose next that $w$ is of Type~(2). Recall that $w$ is incident to a non-trivial block $G_b$ and to a trivial block $G_{b'}$ of $G$, that $f$ is the face of $\mathcal E''$ incident to $w$ and incident to $G_{b'}$, and that $w^0$ and $w^1$ are the occurrences of $w$ along the boundary of $f$.  Assume that $G_b$ belongs to $G_{b_0}$ and that $G_{b'}$ coincides with $G_{b_1}$, the other case is analogous. Let $f_0$ and $f_b$ be the faces of $\mathcal E_{b_0}$ and $\mathcal E_{b}$ corresponding to $f$, respectively.
			
			By construction, we have $\phi''(w^0,f)=90\degree$ and $\phi''(w^1,f)=\phi_{b_0}(w,f_0)-90\degree$. Further, for every face $g\neq f$ of $\mathcal E''$ incident to $w$, we have $\phi''(w,g)=\phi_{b_0}(w,g_0)$, where $g_0$ is the face of $\mathcal E_{b_0}$ corresponding to $g$. Hence, $\sum_{g,x}\phi''(w^x,g)=\sum_{g_0} \phi_{b_0}(w,g_0)$, where the first sum is over all the faces $g$ of $\mathcal E$ incident to $w$ and over all the occurrences $w^x$ of $w$ along the boundary of $g$, and the second sum is over all the faces $g_0$ of $\mathcal E_{b_0}$ incident to $w$.  Since $(\mathcal E_{b_0},\phi_{b_0})$ is a rectilinear representation of $G_{b_0}$, it follows that $\sum_{g,x}\phi''(w^x,g)=360\degree$, hence Property~(c) of a join is satisfied by $(\mathcal E'',\phi'')$.
			
			By construction, we have $\phi_{b_0}(w,f_0)=\phi''(w^0,f)+\phi''(w^1,f)$; note that $f$ is the unique face of $\mathcal E$ whose corresponding face in $\mathcal E_{b_0}$ is $f_0$. Further, for every face $g_0\neq f_0$ of $\mathcal E_{b_0}$ incident to $w$, we have $\phi_{b_0}(w,g_0)=\phi''(w,g)$, where $g$ is the unique face of $\mathcal E$ whose corresponding face in $\mathcal E_{b_0}$ is $g_0$. It follows that Property~(d) of a join is satisfied by $(\mathcal E'',\phi'')$. 
			
			Since $(\mathcal E_{b_0},\phi_{b_0})$ is a rectilinear representation of $G_{b_0}$, we have $\phi''(w,g)=\phi_{b_0}(w,g_0)\in \{90\degree,180\degree,270\degree\}$ for every face $g\neq f$ of $\mathcal E''$ incident to $w$. Further, $\phi''(w^0,f)=90\degree$, by construction. Finally, $\phi''(w^1,f)\in \{90\degree,180\degree\}$, given that $270\degree \geq \phi_{b_0}(w,f_0)=\phi_b(w,f_b)\geq \ell_b(w,f_b)=180\degree$, where $270\degree \geq \phi_{b_0}(w,f_0)$ follows from the fact that $w$ has degree greater than or equal to $2$ in $G_{b_0}$, $\phi_{b_0}(w,f_0)=\phi_b(w,f_b)$ follows from the fact that $(\mathcal E_{b_0},\phi_{b_0})$ satisfies Property~(iii), and $\phi_b(w,f_b)\geq \ell_b(w,f_b)$ follows from the fact that $(\mathcal E_{b},\phi_b)$ is an $\ell_b$-constrained representation of $G_b$. Hence, $(\mathcal E'',\phi'')$ satisfies Property~(b) of a join. 
			\item Suppose next that $w$ is of Type~(3). Then Properties~(b) and~(c) of a join are trivially satisfied by construction and since $w$ has degree $4$ in $G''$. We prove that $(\mathcal E,\phi)$ satisfies Property~(d) of a join. 	Recall that $w$ is incident to one non-trivial block $G_b$ and to two trivial blocks $G_{b'}$ and $G_{b''}$ of $G$. Assume that $G_b$ belongs to $G_{b_0}$ and that $G_{b'}$ and $G_{b''}$ coincide with $G_{b_1}$ and $G_{b_2}$, respectively; the other cases are analogous. Let $f$ be the face of $\mathcal E''$ that is incident to $w$ and that is incident to $G_{b_1}$ and $G_{b_2}$, and let $g$ be the other face of $\mathcal E''$ incident to $w$. Let $w^0$, $w^1$, and $w^2$ be the occurrences of $w$ along the boundary of $f$. By construction, $\phi''(w^0,f)=\phi''(w^1,f)=\phi''(w^2,f)=\phi''(w,g)=90\degree$.  Let $f_0$ and $g_0$ be the faces of $\mathcal E_{b_0}$ corresponding to $f$ and $g$, respectively; further, let $f_b$ and $g_b$ be the faces of $\mathcal E_{b}$ corresponding to $f$ and $g$, respectively. Then, as in the case in which $w$ is of Type~(2), we have $\phi_{b_0}(w,f_0)=\phi_{b}(w,f_b)\geq \ell_b(w,f_b)=270\degree$ and $\phi_{b_0}(w,g_0)=\phi_b(w,g_b)\geq \ell_b(w,g_b)=90\degree$. Hence, $\phi_{b_0}(w,f_0)=270\degree=\phi''(w^0,f)+\phi''(w^1,f)+\phi''(w^2,f)$ and $\phi_{b_0}(w,g_0)=90\degree=\phi''(w,g)$, as required. Further, for $j=1,2$, there is a unique face $f_j$ of $\mathcal E_{b_j}$ incident to $w$; then $\phi_{b_j}(w,f_j)=360\degree$, given that $w$ has degree $1$ in $G_{b_j}$. Finally, $f_{j}$ corresponds to $f$ and $\phi''(w^0,f)+\phi''(w^1,f)+\phi''(w^2,f)+\phi''(w,g)=360\degree$, by construction. Hence, $(\mathcal E'',\phi'')$ satisfies Property~(b) of a join. 
			\item Suppose next that $w$ is of Type~(4). Then Properties~(b) and~(c) of a join are trivially satisfied by construction and since $w$ has degree $4$ in $G''$. We prove that $(\mathcal E,\phi)$ satisfies Property~(d) of a join. Recall that $w$ is incident to one non-trivial block $G_b$ and to two trivial blocks $G_{b'}$ and $G_{b''}$ of $G$. Assume that $G_b$ belongs to $G_{b_0}$ and that $G_{b'}$ and $G_{b''}$ coincide with $G_{b_1}$ and $G_{b_2}$, respectively; the other cases are analogous. Let $f$ (let $g$) be the face of $\mathcal E''$ incident to $w$ and incident to $G_{b'}$ (resp.\ to $G_{b''}$), respectively. Let $w^0$ and $w^1$ (let $w^2$ and $w^3$) be the occurrences of $w$ along the boundary of $f$ (resp.\ of $g$). Let $f_0$ and $g_0$ be the faces of $\mathcal E_{b_0}$ corresponding to $f$ and $g$, respectively; further, let $f_b$ and $g_b$ be the faces of $\mathcal E_{b}$ corresponding to $f$ and $g$, respectively. Then, as in the case in which $w$ is of Type~(2), we have $\phi_{b_0}(w,f_0)=\phi_{b}(w,f_b)\geq \ell_b(w,f_b)=180\degree$ and $\phi_{b_0}(w,g_0)=\phi_b(w,g_b)\geq \ell_b(w,g_b)=180\degree$. Hence, $\phi_{b_0}(w,f_0)=180\degree=\phi''(w^0,f)+\phi''(w^1,f)$ and $\phi_{b_0}(w,g_0)=180\degree=\phi''(w^2,g)+\phi''(w^3,g)$, as required. Further, for $j=1,2$, there is a unique face $f_j$ of $\mathcal E_{b_j}$ incident to $w$; then $\phi_{b_j}(w,f_j)=360\degree$, given that $w$ has degree $1$ in $G_{b_j}$. Finally, $f_{j}$ corresponds to $f$ and $\phi''(w^0,f)+\phi''(w^1,f)+\phi''(w^2,g)+\phi''(w^3,g)=360\degree$, by construction. Property~(d) of a join follows.
			\item Suppose next that $w$ is of Type~(5). Then Properties~(b) and~(c) of a join are trivially satisfied by construction and since $w$ has degree $4$ in $G''$.  We prove that $(\mathcal E,\phi)$ satisfies Property~(d) of a join. Recall that $w$ is incident to two non-trivial blocks $G_b$ and $G_{b'}$ of $G$. Assume that $G_b$ belongs to $G_{b_0}$ and that $G_{b'}$ coincides with $G_{b_1}$; the other case is analogous. Let $f$, $g$, and $h$ be the faces of $\mathcal E''$ incident to $w$, where $f$ is incident to $G_{b_0}$ and $G_{b_1}$, $g$ is incident to $G_{b_0}$ and not to $G_{b_1}$, and $h$ is incident to $G_{b_1}$ and not to $G_{b_0}$. Let $w^0$ and $w^1$ be the occurrences of $w$ along the boundary of $f$. By construction, $\phi''(w^0,f)=\phi''(w^1,f)=\phi''(w,g)=\phi''(w,h)=90\degree$.  Let $f_0$ be the face of $\mathcal E_{b_0}$ corresponding to $f$ and $h$, and let $g_0$ be the face of $\mathcal E_{b_0}$ corresponding to $g$; further, let $f_b$ be the face of $\mathcal E_{b}$ corresponding to $f$ and $h$, and let $g_b$ be the face of $\mathcal E_{b}$ corresponding to $g$. Then, as in the case in which $w$ is of Type~(2), we have $\phi_{b_0}(w,f_0)=\phi_{b}(w,f_b)\geq \ell_b(w,f_b)=270\degree$ and $\phi_{b_0}(w,g_0)=\phi_b(w,g_b)\geq \ell_b(w,g_b)=90\degree$. Hence, $\phi_{b_0}(w,f_0)=270\degree=\phi''(w^0,f)+\phi''(w^1,f)+\phi''(w,h)$ and $\phi_{b_0}(w,g_0)=90\degree=\phi''(w,g)$, as required. An analogous proof can be exhibited for the faces of $\mathcal E_{b_1}$ and Property~(d) of a join follows. 
		\end{itemize} 
		
		This completes the proof of the claim.
	\end{proof}
	
	Claim~\ref{cl:function-phi} and Lemma~\ref{le:preliminaries-subgraphs-composition} imply that $(\mathcal E'',\phi'')$ is a rectilinear representation of $G''$ satisfying Property~(iii). The repetition of the described augmentation eventually leads to the construction of a rectilinear representation of $G$. 
	
	When constructing $(\mathcal E'',\phi'')$ from $(\mathcal E',\phi'),(\mathcal E_{b_1},\phi_{b_1}), \dots, (\mathcal E_{b_k},\phi_{b_k})$, only $O(1)$ values have to be defined, namely the angles incident to the cut-vertex $w$. The rules for the definition of such angles described before Claim~\ref{cl:function-phi}  can clearly be applied in $O(1)$ time. Hence, $\phi$ can be computed in total $O(n)$ time. 
\end{proof}

We thus have the following.

\begin{theorem} \label{th:fixed}
	Let $G$ be an $n$-vertex outerplanar graph with a prescribed plane embedding $\mathcal E$. There is an $O(n)$-time algorithm which tests whether $G$ admits a planar rectilinear drawing with plane embedding $\mathcal E$; in the positive case, the algorithm constructs such a drawing in $O(n)$ time.
\end{theorem}

\begin{proof}
	First, we compute in $O(n)$ time the block-cut-vertex tree $T$ of $G$~\cite{h-gt-69,ht-aeagm-73}, labeling each edge of $G$ with the block it belongs to. A visit of $\mathcal E$ then suffices to construct the plane embedding $\mathcal E_b$ of each block $G_b$ of $G$ in total $O(n)$ time. Further, we label each cut-vertex of $G$ with a label in $\{\textrm{(1)},\dots,\textrm{(5)}\}$, according to its type; this can be done in $O(1)$ time per cut-vertex by looking at the labels of the edges incident to it. Now, for every face $f$ of $\mathcal E_b$ and for every vertex $w$ incident to $f$, we set the value $\ell_b(w,f)$; this is done in $O(1)$ time per pair $(w,f)$, and hence in total $O(n)$ time, depending on the label in $\{\textrm{(1)},\dots,\textrm{(5)}\}$ attached to $w$, as described before Lemma~\ref{le:from-graph-to-blocks}. 
	
	By Tamassia's characterization~\cite{t-eggmnb-87}, we have that $G$ admits a planar rectilinear drawing with plane embedding $\mathcal E$ if and only if a function $\phi$ exists such that $(\mathcal E,\phi)$ is a rectilinear representation. By Lemma~\ref{le:from-graph-to-blocks}, such a function $\phi$ exists if and only if, for every non-trivial block $G_b$ of $G$, a function $\phi_b$ exists such that $(\mathcal E_b,\phi_b)$ is an $\ell_b$-constrained representation of $G_b$. By Theorem~\ref{th:2-con-lowers}, for every non-trivial block $G_b$ of $G$ with $n_b$ vertices, we can test in $O(n_b)$ time whether a function $\phi_b$ exists such that $(\mathcal E_b,\phi_b)$ is an $\ell_b$-constrained representation of $G_b$; moreover, in case such a function $\phi_b$ exists, it can be constructed in $O(n_b)$ time. Hence, such a computation takes $O(n)$ time over all the non-trivial blocks of $G$. If the test for the existence of a function $\phi_b$ such that $(\mathcal E_b,\phi_b)$ is an $\ell_b$-constrained representation of $G_b$ was successful for every non-trivial block $G_b$ of $G$, by Lemma~\ref{le:from-graph-to-blocks} we can construct in $O(n)$ time a function $\phi$ such that $(\mathcal E,\phi)$ is a rectilinear representation of $G$. Finally, a planar rectilinear drawing with plane embedding $\mathcal E$ can be constructed from $(\mathcal E,\phi)$ in $O(n)$ time~\cite{t-eggmnb-87}.
\end{proof}

\subsection{Outerplane Embedding}

Before moving to the variable embedding scenario, we show how Theorem~\ref{th:fixed} can be used in order to test in linear time whether an outerplanar graph $G$ admits an outerplanar rectilinear drawing. 

Note that a $2$-connected outerplanar graph has a unique outerplane embedding, up to a reflection; hence, if $G$ is $2$-connected, the problem of testing whether it admits an outerplanar rectilinear drawing is just a special case of the problem of testing whether it admits a planar rectilinear drawing with a prescribed plane embedding. If $G$ is not $2$-connected, we can also reduce the problem to the rectilinear-planarity testing with prescribed plane embedding, as a consequence of the following lemma.

\begin{lemma} \label{le:equivalent-embeddings}
	Let $\mathcal O_1$ and $\mathcal O_2$ be two outerplane embeddings of an outerplanar graph $G$. There exists a function $\phi_1$ such that $(\mathcal O_1,\phi_1)$ is a rectilinear representation of $G$ if and only if there exists a function $\phi_2$ such that $(\mathcal O_2,\phi_2)$ is a rectilinear representation of $G$.
\end{lemma}

\begin{proof}
	Assume that there exists a function $\phi_1$ such that $(\mathcal O_1,\phi_1)$ is a rectilinear representation of $G$; we prove that there exists a function $\phi_2$ such that $(\mathcal O_2,\phi_2)$ is a rectilinear representation of $G$. The reverse direction is analogous. In order to prove that $(\mathcal O_2,\phi_2)$ is a rectilinear representation of $G$, we prove that it satisfies conditions (1) and (2) of Tamassia's characterization~\cite{t-eggmnb-87}.
	
	First, every internal face $f_2$ of $\mathcal O_2$ is delimited by a simple cycle $\mathcal C$, given that $\mathcal O_2$ is an outerplane embedding; let $G_b$ be the  block  of $G$ containing $\mathcal C$. Since $G_b$ has a unique outerplane embedding, the restrictions of $\mathcal O_1$ and $\mathcal O_2$ to $G_b$ coincide, hence $\mathcal O_1$ contains an internal face $f_1$ that is delimited by $\mathcal C$. Thus, for every vertex $w$ incident to $f_2$, we can set $\phi_2(w,f_2)=\phi_1(w,f_1)$. Then $\sum_w (2-\phi_2(w,f_2)/90\degree) = +4$, given that $\sum_w (2-\phi_1(w,f_1)/90\degree) = +4$, where the sums are over all the vertices $w$ incident to $f_2$ and $f_1$, respectively. 
	
	Second, every vertex $w$ that is not a cut-vertex of $G$ has a unique occurrence along the boundary of $f^*_{\mathcal O_1}$ and $f^*_{\mathcal O_2}$; then we let $\phi_2(w,f^*_{\mathcal O_2})=\phi_1(w,f^*_{\mathcal O_1})$. It follows that $\sum_f \phi_2(w,f)=360\degree$, given that $\sum_f \phi_1(w,f)=360\degree$, where the sums are over all the faces $f$ of $\mathcal O_2$ and $\mathcal O_1$, respectively, incident to $w$.
	
	Finally, let $w$ be a cut-vertex of $G$, let $w^1_1,w^2_1,\dots,w^h_1$ be the occurrences of $w$ along the boundary of $f^*_{\mathcal O_1}$, labeled in arbitrary order, and let $w^1_2,w^2_2,\dots,w^k_2$ be the occurrences of $w$ along the boundary of $f^*_{\mathcal O_2}$, also labeled in arbitrary order. The key observation here is that $h=k$; that is, the number of occurrences of $w$ along the boundary of $f^*_{\mathcal O_1}$ is the same as the number of occurrences of $w$ along the boundary of $f^*_{\mathcal O_2}$. Hence, we can set $\phi_2(w^i_2,f^*_{\mathcal O_2})=\phi_1(w^i_1,f^*_{\mathcal O_1})$, for $i=1,\dots,k$. Then $\sum_f \phi_2(w,f)=360\degree$, given that $\sum_f \phi_1(w,f)=360\degree$, where the sums are over all the faces $f$ of $\mathcal O_2$ and $\mathcal O_1$, respectively, incident to $w$. Further, $\sum_w (2-\phi_2(w,f^*_{\mathcal O_2})/90\degree) = -4$, given that $\sum_w (2-\phi_1(w,f^*_{\mathcal O_1})/90\degree) = -4$, where the sums are over all the vertices $w$ incident to $f^*_{\mathcal O_2}$ and $f^*_{\mathcal O_1}$, respectively. 
	
	This concludes the proof that $(\mathcal O_2,\phi_2)$ satisfies conditions (1) and (2) of Tamassia's characterization~\cite{t-eggmnb-87}, and hence concludes the proof of the lemma.
\end{proof}

We get the following.

\begin{theorem} \label{th:outerplanar}
	Let $G$ be an $n$-vertex outerplanar graph. There is an $O(n)$-time algorithm which tests whether $G$ admits an outerplanar rectilinear drawing; in the positive case, the algorithm constructs such a drawing in $O(n)$ time.
\end{theorem}

\begin{proof}
	Let $\mathcal O$ be any outerplane embedding of $G$. By Lemma~\ref{le:equivalent-embeddings}, we have that $G$ admits an outerplanar rectilinear drawing if and only if it admits an outerplanar rectilinear drawing with plane embedding $\mathcal O$. By Theorem~\ref{th:fixed}, there is an $O(n)$-time algorithm which tests whether $G$ admits a planar rectilinear drawing with plane embedding $\mathcal O$; further, in the positive case, the algorithm constructs such a drawing in $O(n)$ time. The theorem follows.\end{proof}

\section{Variable Embedding} \label{se:variable}

In this section, we deal with an $n$-vertex outerplanar graph $G$ with a variable embedding. More in detail, in Section~\ref{sse:var-2con-edge}, we assume that $G$ is $2$-connected and we assume that an edge of $G$ is prescribed to be incident to the outer face of the planar rectilinear drawing we seek. In Section~\ref{sse:var-2con}, we get rid of the second assumption. In Section~\ref{sse:var}, we get rid of the first assumption.

\subsection{$2$-Connected Outerplanar Graphs with an Edge on the Outer Face} \label{sse:var-2con-edge}

When dealing with a $2$-connected outerplanar graph $G$ with a fixed plane embedding $\mathcal E$, for each face $f$ of $\cal E$ and each vertex $w$ incident to $f$, we required the angle $\phi(w,f)$ of the sought rectilinear representation $(\mathcal E,\phi)$ of $G$ to be larger than a certain lower bound $\ell(w,f)$; this was used in order to ensure that blocks incident to $w$ that are required by $\mathcal E$ to lie inside $f$ would find room to be inserted inside $f$. As the plane embedding is now arbitrary, we do not have anymore a requirement that each angle incident to a vertex should be larger than a certain lower bound, given that distinct blocks incident to a cut-vertex can be freely inserted inside those faces that are large enough to accommodate them. There is one exception, though: if two non-trivial blocks $G_1$ and $G_2$ share a cut-vertex $w$, then in any rectilinear representation of $G_1$, each angle incident to $w$ has to be \emph{either} $90\degree$ \emph{or} $270\degree$, in order to allow for a placement of $G_2$. This constraint leads to the definition of the problem that we show how to solve in this section.

Let $G$ be a $2$-connected outerplanar graph and let $\chi$ be a set of degree-$2$ vertices of $G$. 
A \emph{$\chi$-constrained} representation of $G$ is a rectilinear representation $(\mathcal E,\phi)$ of $G$ such that, for every vertex $x\in \chi$ and every face $f$ of $\mathcal E$ incident to $x$, we have either $\phi(x,f)=90\degree$ or $\phi(x,f)=270\degree$. Given an edge $uv$ incident to the outer face $f^*_{\mathcal O}$ of the outerplane embedding $\mathcal O$ of $G$ and given two values $\mu,\nu \in \{90\degree,180\degree,270\degree\}$, a $\chi$-constrained representation $(\mathcal E,\phi)$ in which $uv$ is incident to $f^*_{\mathcal E}$ and the angles $\phi^{\mathrm{int}}(u)$ and $\phi^{\mathrm{int}}(v)$ are equal to $\mu$ and $\nu$, respectively, is called a \emph{$(\chi,\mu,\nu)$-representation} of $G$. This is formalized in the following.

\begin{definition}
	A $(\chi,\mu,\nu)$-representation of $G$ is a $\chi$-constrained representation $(\mathcal E,\phi)$ of $G$ such that $uv$ is incident to $f^*_{\mathcal E}$, $\phi^{\mathrm{int}}(u)=\mu$, and $\phi^{\mathrm{int}}(v)=\nu$.
\end{definition}

We show how to test, for any $\mu,\nu\in\{90\degree,180\degree,270\degree\}$, whether $G$ admits a $(\chi,\mu,\nu)$-representation. Let ${\mathcal C}^*_{uv}$ be the cycle delimiting the internal face of $\mathcal O$ incident to $uv$. Assume,  w.l.o.g.\ up to a reflection of $\mathcal O$, that $u$ immediately follows $v$ in the clockwise order of the edges along ${\mathcal C}^*_{uv}$ in $\mathcal O$. Let $u=u_0,u_1,\dots,u_k=v$ be the clockwise order of the vertices of ${\mathcal C}^*_{uv}$ in $\mathcal O$. For $i=1,\dots,k$, let $G_i$ be the $uv$-subgraph of $G$ with root $u_{i-1}u_i$. The assumption that $uv$ is incident to the outer face of the desired plane embedding $\mathcal E$ ensures that ${\mathcal C}^*_{uv}$ lies in the outer face of each $uv$-subgraph $G_i$ of $G$ in $\mathcal E$ (and thus any two distinct $uv$-subgraphs $G_i$ and $G_j$ of $G$ lie in the outer face of each other in $\mathcal E$). Conversely, each $uv$-subgraph $G_i$ of $G$ might lie inside or outside ${\mathcal C}^*_{uv}$ in $\mathcal E$; in fact, as we shall see, determining which $uv$-subgraphs of $G$ lie inside ${\mathcal C}^*_{uv}$  and which $uv$-subgraphs of $G$ lie outside ${\mathcal C}^*_{uv}$ is the main challenge towards the construction of $\mathcal E$. Note that if $u_i$ belongs to $\chi$, for any $i=1,\dots, k-1$, then both $G_i$ and $G_{i+1}$ are trivial $uv$-subgraphs of $G$; analogously, if $u_0$ or $u_k$ belongs to $\chi$, then $G_1$ or $G_k$ is a trivial $uv$-subgraph of $G$, respectively.

Lemma~\ref{le:structural-fixed-embedding} finds an immediate extension to the current setting, as formalized below. For any $i=1,\dots,k$, let $\chi_i$ be the intersection between $\chi$ and the vertex set of $G_i$. An \emph{in-out assignment} is an assignment of each non-trivial $uv$-subgraph $G_i$ of $G$ either to the inside or to the outside of ${\mathcal C}^*_{uv}$.

\begin{lemma} \label{le:structural-variable-embedding}
	For any $\mu,\nu\in \{90\degree,180\degree,270\degree\}$, we have that $G$ admits a $(\chi,\mu,\nu)$-representation if and only if there exist an in-out assignment $\mathcal A$ and values $\rho_0, \rho_1, \mu_1, \nu_1, \rho_2, \mu_2, \nu_2, \dots, \rho_k, \mu_k, \nu_k$ in $\{0\degree,90\degree,180\degree,270\degree\}$ so that the following properties are satisfied:
	
	\begin{enumerate}
		\item[$\mathcal V1$:] for $i=0,\dots,k$, we have $\rho_i \geq 90\degree$; further, if $u_i\in \chi$, then either $\rho_i =90\degree$ or $\rho_i =270\degree$;
		\item[$\mathcal V2$:] for $i=1,\dots,k$, if $G_i$ is trivial then $\mu_i=\nu_i=0\degree$, otherwise $\mu_i,\nu_i \in \{90\degree,180\degree\}$ and $G_i$ admits a $(\chi_i,\mu_i,\nu_i)$-representation;
		\item[$\mathcal V3$:] for $i=1,\dots,k-1$, we have that $\nu_i+\rho_i+\mu_{i+1}\leq 270\degree$;
		\item[$\mathcal V4$:] $\rho_0+\mu_1=\mu$ and $\rho_k+\nu_k=\nu$; and
		\item[$\mathcal V5$:] for $i=1,\dots,k$, if $G_i$ is trivial or if it is assigned by $\mathcal A$ to the outside of ${\mathcal C}^*_{uv}$, then let $\sigma_i=0\degree$, otherwise let $\sigma_i=\mu_i+\nu_i$; then we have $\sum_{i=0}^k \rho_i + \sum_{i=1}^k \sigma_i = (k-1) \cdot 180\degree$.
	\end{enumerate}
\end{lemma}

\begin{proof}
	$(\Longrightarrow)$ We first prove the necessity. Suppose that $G$ admits a $(\chi,\mu,\nu)$-representation $(\mathcal E,\phi)$. For each face $f$ of $\mathcal E$ and each vertex $w$ incident to $f$, define $\ell(w,f)=90\degree$. Since $(\mathcal E,\phi)$ is a rectilinear representation, for each face $f$ of $\mathcal E$ and each vertex $w$ incident to $f$, we have $\phi(w,f)\geq 90\degree$. Hence $(\mathcal E,\phi)$ is an $(\mathcal E,\ell,\mu,\nu)$-representation. By Lemma~\ref{le:structural-fixed-embedding}, values $\rho_0, \rho_1, \mu_1, \nu_1, \rho_2, \mu_2, \nu_2, \dots, \rho_k, \mu_k, \nu_k$ in $\{0\degree,90\degree,180\degree,270\degree\}$ exist satisfying Properties~Pr1--Pr5. We show that these values, together with the in-out assignment $\mathcal A$ that assigns a non-trivial $uv$-subgraph $G_i$ of $G$ to the inside of ${\mathcal C}^*_{uv}$ if and only if $G_i$ lies inside ${\mathcal C}^*_{uv}$ in $\mathcal E$, satisfy properties~$\mathcal V1$--$\mathcal V5$. 
	
	\begin{itemize}
		\item Concerning Property~$\mathcal V1$, we have that Property~Pr1 of Lemma~\ref{le:structural-fixed-embedding} implies that $\rho_i\geq 90\degree$, for $i=0,\dots,k$. Further, consider any vertex $u_i \in \chi$. By definition of $(\chi,\mu,\nu)$-representation, we have that either $\phi(u_i,f^{uv}_{\mathcal E})=90\degree$ or $\phi(u_i,f^{uv}_{\mathcal E})=270\degree$, where $f^{uv}_{\mathcal E}$ denotes the internal face of $\mathcal E$ incident to the edge $uv$. By construction (see the proof of necessity in Lemma~\ref{le:structural-fixed-embedding}), we have $\rho_i=\phi(u_i,f^{uv}_{\mathcal E})$, hence we have that either $\rho_i=90\degree$ or $\rho_i=270\degree$, as required. 
		\item By Property~Pr2 of Lemma~\ref{le:structural-fixed-embedding}, we have that $\mu_i=\nu_i=0\degree$ if $G_i$ is trivial, while $\mu_i,\nu_i \in \{90\degree,180\degree\}$ if $G_i$ is non-trivial. Consider the restriction $(\mathcal E_i,\phi_i)$ of $(\mathcal E,\phi)$ to $G_i$. Since $(\mathcal E,\phi)$ is a $(\chi,\mu,\nu)$-representation, it follows that $\phi_i(w,f)\in \{90\degree,270\degree\}$, for every vertex $w\in \chi_i$ and for every face $f$ of $\mathcal E_i$ incident to $w$. By construction  (see the proof of necessity in Lemma~\ref{le:structural-fixed-embedding}), we have $\mu_i= \sum_f \phi(u_{i-1},f)$, where the sum is over all the internal faces $f$ of $\mathcal E_i$ incident to $u_{i-1}$, and $\nu_i= \sum_f \phi(u_{i},f)$, where the sum is over all the internal faces $f$ of $\mathcal E_i$ incident to $u_{i}$. It follows that $(\mathcal E_i,\phi_i)$ is a $(\chi_i,\mu_i,\nu_i)$-representation of $G_i$, as required.	 
		\item Property~$\mathcal V3$ directly follows from Property~Pr3 of Lemma~\ref{le:structural-fixed-embedding} and from the fact that $\ell(u_i,f^*_{\mathcal E})=90\degree$, for $i=0,\dots,k$.
		\item Property~$\mathcal V4$ directly follows from Property~Pr4 of Lemma~\ref{le:structural-fixed-embedding}.
		\item Property~$\mathcal V5$ follows from Property~Pr5 of Lemma~\ref{le:structural-fixed-embedding} and from the fact that, for $i=1,\dots,k$, the value $\sigma_i$ in Property~Pr5 of Lemma~\ref{le:structural-fixed-embedding} coincides with the value $\sigma_i$ in Property~$\mathcal V5$; this is immediate if $G_i$ is trivial, and it follows from the fact that $\mathcal A$ assigns $G_i$ to the inside of ${\mathcal C}^*_{uv}$ if and only if $G_i$ lies inside ${\mathcal C}^*_{uv}$ in $\mathcal E$ if $G_i$ is non-trivial.
	\end{itemize}
	
	$(\Longleftarrow)$ We next prove the sufficiency. Assume that an in-out assignment $\mathcal A$ and values $\rho_0,\rho_1,\mu_1,\nu_1,\dots,\rho_k,\mu_k,\nu_k$ exist satisfying Properties $\mathcal V1$--$\mathcal V5$. For $i=1,\dots,k$, if $G_i$ is not trivial, let $(\mathcal E_i,\phi_i)$ be a $(\chi_i,\mu_i,\nu_i)$-representation of $G_i$; this exists by Property~$\mathcal V2$. 
	
	We define a plane embedding $\mathcal E$ of $G$ as follows. First, the vertices $u_0,u_1,\dots,u_k$ appear in this clockwise order along $\mathcal C^*_{uv}$ in $\mathcal E$. Second, the embedding of each non-trivial $uv$-subgraph $G_i$ of $G$ in $\mathcal E$ is either $\mathcal E_i$ or the reflection of $\mathcal E_i$. Further, each non-trivial $uv$-subgraph $G_i$ of $G$ is embedded inside or outside $\mathcal C^*_{uv}$ according to the in-out assignment $\mathcal A$; this is what might cause the reflection of $\mathcal E_i$ (indeed, if $G_i$ is embedded inside or outside $\mathcal C^*_{uv}$, then $u_{i-1}$ needs to immediately precede or follow $u_i$ in the clockwise order of the vertices along the boundary of the outer face of $\mathcal E_i$, respectively). 
	
	For each face $f$ of $\mathcal E$ and each vertex $w$ of $G$ incident to $f$, let $\ell(w,f)=90\degree$ and let $\ell_i$ be the restriction of $\ell$ to the vertices and edges of $G_i$. Then the values $\rho_0,\rho_1,\mu_1,\nu_1,\dots,\rho_k,\mu_k,\nu_k$ satisfy Properties~Pr1--Pr5 of Lemma~\ref{le:structural-fixed-embedding}, namely:
	
	\begin{itemize}
		\item Property~Pr1 of Lemma~\ref{le:structural-fixed-embedding} directly follows from Property~$\mathcal V1$ and from the fact that $\ell(u_i,f^{uv}_{\mathcal E})=90\degree$, for $i=0,1,\dots,k$, where $f^{uv}_{\mathcal E}$ is the internal face of $\mathcal E$ incident to $uv$.
		\item By Property~$\mathcal V2$, we have that $\mu_i=\nu_i=0\degree$ if $G_i$ is trivial, while $\mu_i,\nu_i \in \{90\degree,180\degree\}$ if $G_i$ is non-trivial. Further, each non-trivial $uv$-subgraph $G_i$ of $G$ admits a $(\chi_i,\mu_i,\nu_i)$-representation $(\mathcal E_i,\phi_i)$. We prove that $(\mathcal E_i,\phi_i)$ is an $(\mathcal E_i,\ell_i, \mu_i,\nu_i)$-representation of $G_i$. Since $(\mathcal E_i,\phi_i)$ is a rectilinear representation, for each internal face $f$ of $\mathcal E_i$ and for each vertex $w$ incident to $f$, we have $\phi_i(w,f)\geq 90\degree=\ell_i(w,f)$. Further, we have $\phi_i(u_{i-1},f^*_{\mathcal E_i})=360\degree-\phi_i^{\mathrm{int}}(u_{i-1})=360\degree-\mu_i\geq 360\degree-270\degree+\rho_{i-1}+\nu_{i-1}=90\degree+\rho_{i-1}+\nu_{i-1}$, where we exploited Property~$\mathcal V3$. Since $\ell_i(u_{i-1},f^*_{\mathcal E_{i}})=180\degree$ if $G_{i-1}$ is trivial and $\ell_i(u_{i-1},f^*_{\mathcal E_{i}})=270\degree$ if $G_{i-1}$ is non-trivial, the inequality $\phi_i(u_{i-1},f^*_{\mathcal E_i})\geq \ell_i(u_{i-1},f^*_{\mathcal E_{i}})$ follows from the fact that $\rho_{i-1}\geq 90\degree$, by Property~$\mathcal V1$, and from the fact that $\nu_{i-1}\geq 90\degree$ if $G_{i-1}$ is non-trivial, by Property~$\mathcal V2$. An analogous proof shows that $\phi_i(u_{i},f^*_{\mathcal E_i})\geq \ell_i(u_{i},f^*_{\mathcal E_{i}})$. Hence, $(\mathcal E_i,\phi_i)$ is an $(\mathcal E_i,\ell_i, \mu_i,\nu_i)$-representation of $G_i$, and Property~Pr2 of Lemma~\ref{le:structural-fixed-embedding} follows.
		\item Property~Pr3 of Lemma~\ref{le:structural-fixed-embedding} directly follows from Property~$\mathcal V3$ and from the fact that $\ell(u_i,f^*_{\mathcal E})=90\degree$, for $i=0,1,\dots,k$.
		\item Property~Pr4 of Lemma~\ref{le:structural-fixed-embedding} directly follows from Property~$\mathcal V4$.
		\item Property~Pr5 of Lemma~\ref{le:structural-fixed-embedding} follows from Property~$\mathcal V5$ and from the fact that, for $i=1,\dots,k$, the value $\sigma_i$ in Property~Pr5 of Lemma~\ref{le:structural-fixed-embedding} coincides with the value $\sigma_i$ in Property~$\mathcal V5$; this is immediate if $G_i$ is trivial, and it follows from the fact that $G_i$ lies inside ${\mathcal C}^*_{uv}$ in $\mathcal E$ if and only if $\mathcal A$ assigns $G_i$ to the inside of ${\mathcal C}^*_{uv}$ if $G_i$ is non-trivial.
	\end{itemize}
	
	By Lemma~\ref{le:structural-fixed-embedding} we have that $G$ admits an $(\mathcal E,\ell,\mu,\nu)$-representation $(\mathcal E,\phi)$. In order to prove that $(\mathcal E,\phi)$ is indeed a $(\chi,\mu,\nu)$-representation of $G$, it suffices to prove that, for each vertex $w\in \chi$ and each face $f$ incident to $w$, either $\phi(w,f)=90\degree$ or $\phi(w,f)=270\degree$ holds true. 
	
	\begin{itemize}
		\item If $w=u_i$, for some $i\in \{0,1,\dots,k\}$, and $f=f^{uv}_{\mathcal E}$, then by construction (see the proof of sufficiency in Lemma~\ref{le:structural-fixed-embedding}), we have $\phi(w,f)=\rho_i$. By Property~$\mathcal V1$ we have that $\rho_i=90\degree$ or $\rho_i=270\degree$, hence either $\phi(w,f)=90\degree$ or $\phi(w,f)=270\degree$ holds true. 
		\item Analogously, if $w=u_i$, for some $i\in \{0,1,\dots,k\}$, and $f=f^*_{\mathcal E}$, then by construction (see the proof of sufficiency in Lemma~\ref{le:structural-fixed-embedding}) and since $w$ has degree $2$ in $G$, we have $\phi(w,f)=360\degree-\rho_i$. By Property~$\mathcal V1$ we have that either $\rho_i=90\degree$ or $\rho_i=270\degree$, hence $\phi(w,f)=270\degree$ or $\phi(w,f)=90\degree$ holds true. 
		\item If $w\in \chi_i$, for some non-trivial $uv$-subgraph $G_i$ of $G$, then by construction (see the proof of sufficiency in Lemma~\ref{le:structural-fixed-embedding}) we have that $\phi(w,f)=\phi_i(w,f')$, where $f'$ is the face of $\mathcal E_i$ corresponding to $f$. Since $(\mathcal E_i,\phi_i)$ is a $(\chi_i,\mu_i,\nu_i)$-representation of $G_i$, it follows that $\phi_i(w,f')=90\degree$ or $\phi_i(w,f')=270\degree$ holds true. Consequently, either $\phi(w,f)=270\degree$ or $\phi(w,f)=90\degree$ holds true, as well.
	\end{itemize}
	
	This concludes the proof of the lemma.
\end{proof}

Property $\mathcal V2$ of Lemma~\ref{le:structural-variable-embedding} implies that, for every trivial $uv$-subgraph $G_i$ of $G$, the values $\mu_i$ and $\nu_i$ can be set equal to $0\degree$ without loss of generality. The values $\mu_i$ and $\nu_i$ can also be chosen ``optimally'' for every non-trivial $uv$-subgraph $G_i$ of $G$, except for $G_1$ and $G_k$, due to the following two lemmas, which play the same role as Lemmata~\ref{le:smaller-is-better} and~\ref{le:same-size-is-the-same} do for the fixed embedding setting.

\begin{lemma}\label{le:smaller-is-better-variable}
	Suppose that there exist an in-out assignment $\mathcal A$ and a sequence $\mathcal S$ of values $\rho_0,\rho_1,\mu_1,\nu_1,\dots,\rho_k$, $\mu_k,\nu_k$ such that Properties~$\mathcal V1$--$\mathcal V5$ of Lemma~\ref{le:structural-variable-embedding} are satisfied. Further, suppose that $G_i$ is a non-trivial $uv$-subgraph of $G$, for some $i\in \{2,\dots,k-1\}$. 
	
	If $G_i$ admits a $(\chi_i,\mu^*_i,\nu^*_i)$-representation, where $\mu^*_i,\nu^*_i\in \{90\degree,180\degree\}$, $\mu^*_i\leq \mu_i$, and $\nu^*_i\leq \nu_i$, then there exist values $\rho^*_{i-1}$ and $\rho^*_i$ such that the in-out assignment $\mathcal A$ and the sequence $\mathcal S^*$ obtained from $\mathcal S$ by replacing the values $\rho_{i-1}$, $\rho_i$, $\mu_i$, and $\nu_i$ with $\rho^*_{i-1}$, $\rho^*_i$, $\mu^*_i$, and $\nu^*_i$, respectively, satisfy Properties~$\mathcal V1$--$\mathcal V5$ of Lemma~\ref{le:structural-variable-embedding}. 
\end{lemma}

\begin{proof}
	We distinguish the case in which $\mathcal A$ assigns $G_i$ to the outside of $\mathcal C^*_{uv}$ from the one in which $\mathcal A$ assigns $G_i$ to the inside of $\mathcal C^*_{uv}$. In the former case, we let $\rho^*_{i-1}=\rho_{i-1}$ and $\rho^*_{i}=\rho_{i}$, while in the latter case, we let $\rho^*_{i-1}=\rho_{i-1}+\mu_{i}-\mu^*_{i}$ and $\rho^*_{i}=\rho_{i}+\nu_i-\nu^*_i$. The proof that the resulting sequence $\mathcal S^*$, together with the in-out assignment $\mathcal A$, satisfies Properties~$\mathcal V1$--$\mathcal V5$ of Lemma~\ref{le:structural-variable-embedding} follows almost verbatim the proof of Lemma~\ref{le:smaller-is-better}, with Properties Pr1--Pr5 replaced by Properties~$\mathcal V1$--$\mathcal V5$, respectively, and with the values $\ell(u_{i-1},f^{uv}_{\mathcal E})$ and $\ell(u_{i},f^{uv}_{\mathcal E})$ replaced by $90\degree$.
\end{proof}

\begin{lemma}\label{le:same-size-is-the-same-variable}
	Suppose that there exist an in-out assignment $\mathcal A$ and a sequence $\mathcal S$ of values $\rho_0,\rho_1,\mu_1,\nu_1,\dots,\rho_k$, $\mu_k,\nu_k$ such that Properties~$\mathcal V1$--$\mathcal V5$ of Lemma~\ref{le:structural-variable-embedding} are satisfied. Further, suppose that $G_i$ is a non-trivial $uv$-subgraph of $G$, for some $i\in \{2,\dots,k-1\}$, and that $G_{i-1}$ and $G_{i+1}$ are both trivial $uv$-subgraphs of $G$.
	
	If $G_i$ admits a $(\chi_i,\mu^*_i,\nu^*_i)$-representation, where $\mu^*_i,\nu^*_i\in \{90\degree,180\degree\}$ and $\mu^*_i+\nu^*_i\leq \mu_i+\nu_i$, then there exist values $\rho^*_{i-1}$ and $\rho^*_i$ such that the in-out assignment $\mathcal A$ and the sequence $\mathcal S^*$ obtained from $\mathcal S$ by replacing the values $\rho_{i-1}$, $\rho_i$, $\mu_i$, and $\nu_i$ with $\rho^*_{i-1}$, $\rho^*_i$, $\mu^*_i$, and $\nu^*_i$, respectively, satisfy Properties~$\mathcal V1$--$\mathcal V5$ of Lemma~\ref{le:structural-variable-embedding}. 
\end{lemma}

\begin{proof}
	If $\mu^*_i\leq \mu_i$ and $\nu^*_i\leq \nu_i$, then the statement follows by Lemma~\ref{le:smaller-is-better-variable}. We can hence assume that $\mu^*_i>\mu_i$ and $\nu^*_i<\nu_i$, or that $\mu^*_i<\mu_i$ and $\nu^*_i>\nu_i$. Assume the former, as the proof for the latter is symmetric. Recall that $\mu_i,\nu_i\in \{90\degree,180\degree\}$, since $\mathcal S$ satisfies Property~$\mathcal V2$, and that $\mu^*_i,\nu^*_i\in \{90\degree,180\degree\}$, by assumption. Hence, we have $\mu_i=90\degree$, $\nu_i=180\degree$, $\mu^*_i=180\degree$, and $\nu^*_i=90\degree$. Further, since $G_{i-1}$ and $G_{i+1}$ are both trivial, we have $\nu_{i-1}=\mu_{i+1}=0\degree$. 
	We let $\rho^*_{i-1}=\rho_{i}$ and $\rho^*_{i}=\rho_{i-1}$. The proof that the resulting sequence $\mathcal S^*$, together with the in-out assignment $\mathcal A$, satisfies Properties~$\mathcal V1$--$\mathcal V5$ of Lemma~\ref{le:structural-variable-embedding} follows almost verbatim the proof of Lemma~\ref{le:same-size-is-the-same}, with Properties Pr1--Pr5 replaced by Properties~$\mathcal V1$--$\mathcal V5$, respectively, and with the values $\ell(u_{i-1},f^{uv}_{\mathcal E})$, $\ell(u_{i},f^{uv}_{\mathcal E})$, $\ell(u_{i},f^*_{\mathcal E})$, and $\ell(u_{i-1},f^*_{\mathcal E})$ replaced by $90\degree$.
\end{proof}

We now provide an algorithm that establishes in $O(k)$ time whether a $(\chi,\mu,\nu)$-representation of $G$ exists, and in case it does, it determines an in-out assignment $\mathcal A$ and values $\rho_0, \rho_1, \mu_1, \nu_1, \rho_2, \mu_2, \nu_2, \dots, \rho_k, \mu_k, \nu_k$ satisfying Properties~$\mathcal V1$--$\mathcal V5$ of Lemma~\ref{le:structural-variable-embedding}. 

Similarly to the algorithm in the proof of Lemma~\ref{le:find-values}, we assume that, for every $uv$-subgraph $G_i$ of $G$, it is known whether $G_i$ is trivial or not and, in case $G_i$ is non-trivial, whether it admits a $(\chi_i,\mu_i,\nu_i)$-representation or not, for every pair $(\mu_i,\nu_i)$ with $\mu_i,\nu_i \in \{90\degree,180\degree\}$. Differently from the algorithm in the proof of Lemma~\ref{le:find-values}, no plane embedding for $G$ is fixed, hence it is not specified whether each non-trivial $uv$-subgraph of $G$ lies inside or outside $\mathcal C^*_{uv}$; indeed, determining an in-out assignment $\mathcal A$, i.e., an assignment of each $uv$-subgraph of $G$ either to the inside or to the outside of $\mathcal C^*_{uv}$, is one of the main challenges we face. 

For outerplanar graphs with a fixed plane embedding, our $O(k)$-time algorithm that decides whether an $(\mathcal E, \ell,\mu,\nu)$-representation exists, given the values $\mu_i, \nu_i \in \{90\degree,180\degree\}$ for which each non-trivial subgraph $G_i$ admits an $(\mathcal E_i, \ell_i,\mu_i,\nu_i)$-representation, was hidden in the proof of Lemma~\ref{le:find-values}. Here, for outerplanar graphs with variable embedding, we need to expose such an algorithm and argue that some parts of it actually run in $O(1)$ time. In fact, this will be needed in Section~\ref{sse:var-2con}, in order to remove the assumption that the plane embedding we seek has a prescribed edge incident to the outer face.

At a high-level, the algorithm starts by fixing the values $\mu_2,\nu_2,\mu_3,\nu_3,\dots,\mu_{k-1},\nu_{k-1}$ without loss of generality. Then it considers all the possible tuples $(\mu_1,\nu_1,\mu_k,\nu_k,\rho_0,\rho_k)$; after getting rid of some of them, for each of the remaining tuples it determines whether an in-out assignment $\mathcal A$ and values $\rho_1,\rho_2,\dots,\rho_{k-1}$ exist so that Properties~$\mathcal V1$--$\mathcal V5$ are satisfied. We start with the following.

\begin{lemma} \label{le:fix-values}
	Let $i\in \{2,\dots,k-1\}$ and suppose that the following information is known:
	\begin{itemize}
		\item whether each of $G_{i-1}$, $G_i$, and $G_{i+1}$ is trivial or not; and
		\item for $j=i-1,i,i+1$, in case $G_{j}$ is not trivial, whether it admits a $(\chi_{j},\mu_{j},\nu_{j})$-representation or not, for every pair $(\mu_{j},\nu_{j})$ with $\mu_{j},\nu_{j} \in \{90\degree,180\degree\}$.
	\end{itemize} 
	Then in $O(1)$ time we can either correctly conclude that $G$ has no $(\chi,\mu,\nu)$-representation, or we can find two values $\mu^*_i$ and $\nu^*_i$ so that the following is true. Suppose that an in-out assignment $\mathcal A$ and a sequence $\mathcal S$ of values $\rho_0,\rho_1,\mu_1,\nu_1,\dots,\rho_k,\mu_k,\nu_k$ exist such that Properties~$\mathcal V1$--$\mathcal V5$ of Lemma~\ref{le:structural-variable-embedding} are satisfied. Then there exist values $\rho^*_{i-1}$ and $\rho^*_i$ such that the in-out assignment $\mathcal A$ and the sequence $\mathcal S^*$ obtained from $\mathcal S$ by replacing the values $\rho_{i-1}$, $\rho_i$, $\mu_i$, and $\nu_i$ with $\rho^*_{i-1}$, $\rho^*_i$, $\mu^*_i$, and $\nu^*_i$, respectively, also satisfy Properties~$\mathcal V1$--$\mathcal V5$ of Lemma~\ref{le:structural-variable-embedding}. 
\end{lemma}

\begin{proof}
	First, if $G_i$ is trivial, then, by Property~$\mathcal V2$, we have $\mu_i=\nu_i=0\degree$. Hence, it suffices to set $\mu^*_i=\nu^*_i=0\degree$, $\rho^*_{i-1}=\rho_{i-1}$, and $\rho^*_{i}=\rho_{i}$. 
	
	Second, suppose that $G_i$ is non-trivial and $G_{i-1}$ and $G_{i+1}$ are both trivial. If $G_i$ admits no $(\chi_i,\mu_i,\nu_i)$-representation with $\mu_i,\nu_i\in \{90\degree,180\degree\}$, then by Property~$\mathcal V2$ we can conclude that $G$ admits no $(\chi,\mu,\nu)$-representation. Otherwise, by Lemma~\ref{le:same-size-is-the-same-variable}, it suffices to set $\mu^*_i$ and $\nu^*_i$ so that $G_i$ admits a $(\chi_i,\mu^*_i,\nu^*_i)$-representation, so that $\mu^*_i,\nu^*_i\in \{90\degree,180\degree\}$, so that $\mu^*_i+\nu^*_i$ is minimum, and so that $\mu^*_i$ is minimum subject to the previous constraint. 
	
	Third, suppose that $G_{i-1}$, $G_{i}$, and $G_{i+1}$ are all non-trivial. By Properties~$\mathcal V1$ and~$\mathcal V3$, we have $\nu_{i-1}+\mu_{i}\leq 180\degree$ and $\nu_{i}+\mu_{i+1}\leq 180\degree$. Further, by Property~$\mathcal V2$, we have $\nu_{i-1}\geq 90\degree$, $\mu_i\geq 90\degree$, $\nu_i\geq 90\degree$, and $\mu_{i+1}\geq 90\degree$. Hence, if $G_i$ admits no $(\chi_i,90\degree,90\degree)$-representation, then by Property~$\mathcal V2$ we can conclude that $G$ admits no $(\chi,\mu,\nu)$-representation. Otherwise, it suffices to set $\mu^*_i=\nu^*_i=90\degree$.
	
	Finally,  suppose that $G_{i-1}$ and $G_{i}$ are non-trivial, while $G_{i+1}$ is trivial; the case in which $G_{i}$ and $G_{i+1}$ are non-trivial, while $G_{i-1}$ is trivial is symmetric. By Properties~$\mathcal V1$ and~$\mathcal V3$, we have $\nu_{i-1}+\mu_{i}\leq 180\degree$. Further, by Property~$\mathcal V2$, we have $\nu_{i-1}\geq 90\degree$ and $\mu_i\geq 90\degree$. Hence, if $G_i$ admits no $(\chi_i,90\degree,\nu_i)$-representation with $\nu_i\in \{90\degree,180\degree\}$, then by Property~$\mathcal V2$ we can conclude that $G$ admits no $(\chi,\mu,\nu)$-representation. Otherwise, by Lemma~\ref{le:smaller-is-better-variable}, it suffices to set $\mu^*_i=90\degree$ and $\nu^*_i$ so that $G_i$ admits a $(\chi_i,90\degree,\nu^*_i)$-representation, so that $\nu^*_i\in \{90\degree,180\degree\}$, and so that $\nu^*_i$ is minimum. 
\end{proof}

By Lemma~\ref{le:fix-values}, for each $i=2,\dots,k-1$, in $O(1)$ time we can either conclude that $G$ admits no $(\chi,\mu,\nu)$-representation or, without loss of generality, we can fix the values $\mu_i$ and $\nu_i$, in the sequence of values from Lemma~\ref{le:structural-variable-embedding} whose existence we are trying to establish, to $\mu^*_i$ and $\nu^*_i$, respectively. For $i=2,\dots,k-1$, we say that $(\mu^*_i,\nu^*_i)$ is the \emph{optimal pair} for $G_i$ and we say that $\mu^*_2,\nu^*_2,\dots,\mu^*_{k-1},\nu^*_{k-1}$ is the \emph{optimal sequence} for $G$. 

A fact that we are going to use later is that the values $\mu^*_i$ and $\nu^*_i$ of the optimal pair for $G_i$ do not depend on $\mu$ and $\nu$; that is, if we were trying to establish the existence of a $(\chi,\mu',\nu')$-representation of $G$, where $\mu',\nu'\in \{90\degree,180\degree,270\degree\}$, and, possibly, $\mu'\neq \mu$ or $\nu'\neq \nu$, then $(\mu^*_i,\nu^*_i)$ would still be the optimal pair for $G_i$. Further, if an edge $u_{x}u_{x+1}$ of $\mathcal C^*_{uv}$ incident to $f^*_{\mathcal O}$ was prescribed to be incident to the outer face of the sought rectilinear representation of $G$ in place of $uv$,  then $(\mu^*_i,\nu^*_i)$ would still be the optimal pair for $G_i$, as long as $\{x,x+1\}\cap\{i-1,i\}=\emptyset$. These two observations easily descend from the fact that $\mu^*_i$ and $\nu^*_i$ are fixed solely based on whether, for $j=i-1,i,i+1$, the graph $G_j$ is trivial or not and, in case it is non-trivial, based on the values $\mu_j,\nu_j\in \{90\degree,180\degree,270\degree\}$ for which it admits a $(\chi_j,\mu_j,\nu_j)$-representation.

%

If we did not conclude that $G$ admits no $(\chi,\mu,\nu)$-representation, then by means of Lemma~\ref{le:fix-values} we determined in $O(k)$ time the optimal sequence $\mu^*_2,\nu^*_2,\dots,\mu^*_{k-1},\nu^*_{k-1}$ for $G$. We now also determine possible values for $\mu_1$, $\nu_1$, $\mu_k$, $\nu_k$, $\rho_0$, and $\rho_k$. This is done in the following procedure. We consider each of the $3^6\in O(1)$ tuples $(\mu_1,\nu_1,\mu_k,\nu_k,\rho_0,\rho_k)$ such that $\mu_1,\nu_1,\mu_k, \nu_k\in \{0\degree,90\degree,180\degree\}$ and $\rho_0,\rho_k \in \{90\degree,180\degree,270\degree\}$. We discard a tuple $(\mu_1,\nu_1,\mu_k,\nu_k,\rho_0,\rho_k)$ if:

\begin{enumerate}[(C1)]
	\item $G_1$ is trivial and $\max\{\mu_1,\nu_1\} > 0\degree$, or $G_1$ is non-trivial and $\min\{\mu_1,\nu_1\} = 0\degree$; 
	\item $G_k$ is trivial and $\max\{\mu_k,\nu_k\} > 0\degree$, or $G_k$ is non-trivial and $\min\{\mu_k,\nu_k\} = 0\degree$; 
	\item $G_1$ is non-trivial and does not admit a $(\chi_1,\mu_1,\nu_1)$-representation, or $G_k$ is non-trivial and does not admit a $(\chi_k,\mu_k,\nu_k)$-representation; 
	\item $\mu_1 + \rho_0 \neq \mu$, or $\nu_k + \rho_k \neq \nu$; 
	\item $u_0\in \chi$ and $\rho_0 \notin \{90\degree,270\degree\}$, or $u_k\in \chi$ and $\rho_k \notin \{90\degree,270\degree\}$; or
	\item $\nu_1+\mu^*_2>180\degree$, or $\nu^*_{k-1}+\mu_k> 180\degree$. 
\end{enumerate}

We prove that this procedure does not erroneously get rid of any tuple which might appear in a solution.

\begin{lemma} \label{le:initial-final-values}
	Suppose that an in-out assignment $\mathcal A$ and a sequence of values $\rho_0,\rho_1,\mu_1,\nu_1,\rho_2,\mu^*_2,\nu^*_2,\dots,\rho_{k-1}$, $\mu^*_{k-1},\nu^*_{k-1},\rho_k,\mu_k,\nu_k$ exist such that Properties~$\mathcal V1$--$\mathcal V5$ of Lemma~\ref{le:structural-variable-embedding} are satisfied, where $\mu^*_2,\nu^*_2,\dots,\mu^*_{k-1},\nu^*_{k-1}$ is the optimal sequence for $G$. Then the tuple $(\mu_1,\nu_1,\mu_k,\nu_k,\rho_0,\rho_k)$ has not been discarded. 
\end{lemma} 	 

\begin{proof}
	Suppose, for a contradiction, that the tuple $(\mu_1,\nu_1,\mu_k,\nu_k,\rho_0,\rho_k)$ has been discarded. If this happened because of Condition~(C1), (C2), or (C3) we have that $\mathcal A$ and $\mathcal S$ do not satisfy Property~$\mathcal V2$, a contradiction. If the tuple has been discarded because of Condition~(C4), we have that $\mathcal A$ and $\mathcal S$ do not satisfy Property~$\mathcal V4$, a contradiction. If the tuple has been discarded because of Condition~(C5), we have that $\mathcal A$ and $\mathcal S$ do not satisfy Property~$\mathcal V1$, a contradiction. Finally, if the tuple has been discarded because of Condition~(C6), we have that $\mathcal A$ and $\mathcal S$ do not satisfy Property~$\mathcal V1$ or~$\mathcal V3$, a contradiction.  
\end{proof}

If we discarded all the tuples $(\mu_1,\nu_1,\mu_k,\nu_k,\rho_0,\rho_k)$, we conclude that $G$ admits no $(\chi,\mu,\nu)$-representation. Otherwise, there are a constant number of sequences $\mu_1,\nu_1,\mu^*_2,\nu^*_2,\dots,\mu^*_{k-1},\nu^*_{k-1},\mu_{k},\nu_{k},\rho_0,\rho_k$ such that $\mu^*_2,\nu^*_2,\dots,\mu^*_{k-1},\nu^*_{k-1}$ is the optimal sequence for $G$ and the tuple $(\mu_1,\nu_1,\mu_k,\nu_k,\rho_0,\rho_k)$ was not discarded. We say that each of these~sequences is {\em promising} for $(G,\mu,\nu)$. The notation includes $\mu$ and $\nu$ to highlight the fact that whether a sequence is promising depends from $\mu$ and $\nu$.

\begin{lemma} \label{le:promising-sequences}
	The promising sequences for $(G,\mu,\nu)$ can be constructed in $O(1)$ time from the optimal sequence for~$G$.
\end{lemma}


\begin{proof}
	There are $O(1)$ tuples $(\mu_1,\nu_1,\mu_k,\nu_k,\rho_0,\rho_k)$ such that $\mu_1,\nu_1,\mu_k, \nu_k\in \{0\degree,90\degree,180\degree\}$ and $\rho_0,\rho_k \in \{90\degree,180\degree,270\degree\}$. For each of them, we can check in $O(1)$ time whether any of Conditions~(C1)--(C6) is satisfied and, in the positive case, discard the tuple. For each tuple that was not discarded, we obtain a promising sequence for $(G,\mu,\nu)$.
\end{proof}

We say that a promising sequence for $(G,\mu,\nu)$ is \emph{extensible} if there exist an in-out assignment $\mathcal A$ and values $\rho_1,\rho_2,\dots,\rho_{k-1}$ that, together with the promising sequence for $G$, satisfy Properties~$\mathcal V1$--$\mathcal V5$ of Lemma~\ref{le:structural-variable-embedding}. The discussion so far leads to the following.

\begin{lemma} \label{le:extensible-solution}
	There is a $(\chi,\mu,\nu)$-representation of $G$ if and only if there is a promising sequence for $(G,\mu,\nu)$ which is extensible. 
\end{lemma} 

\begin{proof}
	If a promising sequence for $(G,\mu,\nu)$ is extensible then, by Lemma~\ref{le:structural-variable-embedding}, we have that $G$ admits a $(\chi,\mu,\nu)$-representation.
	Conversely, if $G$ admits a $(\chi,\mu,\nu)$-representation, then there exist an in-out assignment $\mathcal A$ and a sequence $\mathcal S$ of values such that Properties~$\mathcal V1$--$\mathcal V5$ of Lemma~\ref{le:structural-variable-embedding} are satisfied. By repeated applications of Lemma~\ref{le:fix-values}, we have that there exists a sequence $\rho_0,\rho_1,\mu_1,\nu_1,\rho_2,\mu^*_2,\nu^*_2,\dots,\rho_{k-1},\mu^*_{k-1},\nu^*_{k-1},\rho_k,\mu_k,\nu_k$ of values that, together with $\mathcal A$, satisfies Properties~$\mathcal V1$--$\mathcal V5$ of Lemma~\ref{le:structural-variable-embedding}. By Lemma~\ref{le:initial-final-values}, we have that the tuple $(\mu_1,\nu_1,\mu_k,\nu_k,\rho_0,\rho_k)$ has not been discarded, hence $\mu_1,\nu_1,\mu^*_2,\nu^*_2,\dots,\mu^*_{k-1},\nu^*_{k-1},\mu_{k},\nu_{k},\rho_0,\rho_k$ is an extensible promising sequence for $(G,\mu,\nu)$.
\end{proof}

We now treat each promising sequence for $(G,\mu,\nu)$ independently and show an algorithm that determines whether a promising sequence for $(G,\mu,\nu)$ is extensible. Our algorithm will be accompanied by a characterization of the extensible promising sequences for $(G,\mu,\nu)$. This will be used in Section~\ref{sse:var-2con}, where we will remove the assumption that the edge $uv$ is incident to the outer face of the sought rectilinear representation.  

Consider a promising sequence for $(G,\mu,\nu)$. For ease of notation, we drop the star from the values of the optimal sequence for $G$ and denote the promising sequence for $(G,\mu,\nu)$ as $\mu_1,\nu_1,\mu_2,\nu_2,\dots,\mu_{k},\nu_{k},\rho_0,\rho_k$. Note that, for any in-out assignment $\mathcal A$ and any choice of the values $\rho_1,\rho_2,\dots,\rho_{k-1}$, Properties~$\mathcal V2$ and $\mathcal V4$ of Lemma~\ref{le:structural-variable-embedding} are satisfied, given that $\mu_1,\nu_1,\mu_2,\nu_2,\dots,\mu_{k},\nu_{k},\rho_0,\rho_k$ is a promising sequence for $(G,\mu,\nu)$.

Our goal is now to determine whether an in-out assignment $\mathcal A$ and a choice of the values $\rho_1,\rho_2,\dots,\rho_{k-1}$ exist so that Property~$\mathcal V5$ is satisfied, i.e., so that $\sum_{i=0}^{k} \rho_i + \sum_{i=1}^k \sigma_i = (k-1) \cdot 180\degree$; recall that, for $i=1,\dots,k$, the value $\sigma_i$ is equal to $0\degree$ if $G_i$ is trivial or if $\mathcal A$ assigns $G_i$ to the outside of ${\mathcal C}^*_{uv}$, while it is equal to $\mu_i+\nu_i$ if $\mathcal A$ assigns $G_i$ to the inside of ${\mathcal C}^*_{uv}$. The previous equality turns into $\sum_{i=1}^{k-1} (\frac{\rho_i}{90\degree}-1) + \sum_{i=1}^k \frac{\sigma_i}{90\degree} = (k-1) -\frac{\rho_0+\rho_k}{90\degree}$. Let the \emph{target value} be defined as $t:= (k-1) -\frac{\rho_0+\rho_k}{90\degree}$ and note that $t$ is fixed, given that $\rho_0$ and $\rho_k$ have already been set. Hence, our goal is to determine whether an in-out assignment $\mathcal A$ and a choice of the values $\rho_1,\rho_2,\dots,\rho_{k-1}$ exist so that $\sum_{i=1}^{k-1} (\frac{\rho_i}{90\degree}-1) + \sum_{i=1}^k \frac{\sigma_i}{90\degree}$ is equal to the target value. The choice of $\rho_1,\rho_2,\dots,\rho_{k-1}$ is however not arbitrary, as it needs to comply with Properties~$\mathcal V1$ and~$\mathcal V3$. In the following we make this precise and we introduce some notation.

\begin{itemize}
	\item Each non-trivial $uv$-subgraph $G_i$ of $G$ is of one of three types: We say that $G_i$ is a \emph{$2$-component}, a \emph{$3$-component}, or a \emph{$4$-component} if $\mu_i+\nu_i=180\degree$, if $\mu_i+\nu_i=270\degree$, or if $\mu_i+\nu_i=360\degree$. Note that if an $x$-component $G_i$ is assigned to the inside of ${\mathcal C}^*_{uv}$, then $\frac{\sigma_i}{90\degree}=x$. Hence, assigning an $x$-component $G_i$ to the inside of ${\mathcal C}^*_{uv}$ contributes $x$ units towards the target value. Let $a$, $b$, and $c_1$ be the number of $4$-components, $3$-components, and $2$-components of $G$, respectively. 
	\item Consider any vertex $u_i \in \chi$ with $i\in \{1,\dots,k-1\}$. In order to satisfy Property~$\mathcal V1$, we need to set either $\rho_i=90\degree$ or $\rho_i=270\degree$. Hence, $(\frac{\rho_i}{90\degree}-1)$ contributes either $0$ or  $2$ units towards the target value, respectively. Let $c_2$ be the number of vertices $u_i \in \chi$ with $i\in \{1,\dots,k-1\}$. Further, let $c=c_1+c_2$.
	\item Consider any vertex $u_i \notin \chi$ with $i\in \{1,\dots,k-1\}$. In order to satisfy Properties~$\mathcal V1$ and~$\mathcal V3$, we need to choose for $\rho_i$ a value which is at least $90\degree$ and at most $270\degree-\nu_i-\mu_{i+1}$. The algorithm described in the proof of Lemma~\ref{le:fix-values} fixes $\nu_i=90\degree$ and $\mu_{i+1}=90\degree$ if both $G_i$ and $G_{i+1}$ are non-trivial, fixes $\nu_i=0\degree$ and $\mu_{i+1}\in\{90\degree,180\degree\}$ if $G_i$ is trivial and $G_{i+1}$ is non-trivial, fixes $\nu_i\in\{90\degree,180\degree\}$ and $\mu_{i+1}=0\degree$ if $G_i$ is non-trivial and $G_{i+1}$ is trivial, and fixes $\nu_i=0\degree$ and $\mu_{i+1}=0\degree$ if both $G_i$ and $G_{i+1}$ are trivial. In any case, we have $270\degree-\nu_i-\mu_{i+1}\geq 90\degree$. Let $d_i:=2-\frac{\nu_i+\mu_{i+1}}{90\degree}$ and note that $0\leq d_i\leq 2$. Hence, we can choose $\rho_i$ so that $(\frac{\rho_i}{90\degree}-1)$ contributes any integer value between $0$ and $d_i$ units towards the target value. Let $d:=\sum d_i$, where the sum is over all the vertices $u_i\notin \chi$ with $i\in \{1,\dots,k-1\}$.
\end{itemize}  

The following lemma allows us to shift our attention from the extensibility of $\mu_1,\nu_1,\mu_2,\nu_2,\dots,\mu_{k},\nu_{k},\rho_{0},\rho_{k}$ to a numerical problem on the values $a$, $b$, $c$, $d$, and $t$.

\begin{lemma}\label{le:extensible-values}
	The promising sequence $\mu_1,\nu_1,\mu_2,\nu_2,\dots,\mu_{k},\nu_{k},\rho_{0},\rho_{k}$ for $(G,\mu,\nu)$ is extensible if and only if integer values $0\leq a'\leq a$, $0\leq b'\leq b$, $0\leq c'\leq c$, and $0\leq d'\leq d$ exist such that $4a'+3b'+2c'+d'=t$. Moreover, assume that such values $a'$, $b'$, $c'$, and $d'$ exist and are known; then an in-out assignment $\mathcal A$ and values $\rho_1,\rho_2,\dots,\rho_{k-1}$ that, together with the promising sequence $\mu_1,\nu_1,\mu_2,\nu_2,\dots,\mu_{k},\nu_{k},\rho_{0},\rho_{k}$ for $(G,\mu,\nu)$, satisfy Properties~$\mathcal V1$--$\mathcal V5$ of Lemma~\ref{le:structural-variable-embedding} can be determined in $O(k)$ time. 
\end{lemma}

\begin{proof}
	$(\Longrightarrow)$ Suppose that the promising sequence $\mu_1,\nu_1,\mu_2,\nu_2,\dots,\mu_{k},\nu_{k},\rho_{0},\rho_{k}$ for $(G,\mu,\nu)$ is extensible, i.e., there exist an in-out assignment $\mathcal A$ and values $\rho_1,\rho_2,\dots,\rho_{k-1}$ that, together with the promising sequence for $(G,\mu,\nu)$, satisfy Properties~$\mathcal V1$--$\mathcal V5$ of Lemma~\ref{le:structural-variable-embedding}. Let $a'$, $b'$, and $c'_1$ be the number of $4$-components, $3$-components, and $2$-components of $G$ that are assigned to the inside of ${\mathcal C}^*_{uv}$ by $\mathcal A$. Further, let $c'_2$ be the number of vertices $u_i \in \chi$ with $i\in \{1,\dots,k-1\}$ for which $\rho_i=270\degree$ and let $c'=c'_1+c'_2$. Finally, for each $i\in \{1,\dots,k-1\}$ with $u_i \notin \chi$, let $d'_i=\frac{\rho_i}{90\degree}-1$ and let $d':=\sum d'_i$, where the sum is over all the vertices $u_i\notin \chi$ with $i\in \{1,\dots,k-1\}$.
	
	Clearly, $a'$, $b'$, and $c'$ are integer and non-negative; by Property~$\mathcal V1$ of Lemma~\ref{le:structural-variable-embedding}, for each $i\in \{1,\dots,k-1\}$ with $u_i \notin \chi$, we have that $d'_i=\frac{\rho_i}{90\degree}-1$ is integer and non-negative, and hence so is $d'$. Further, we have $a'\leq a$, given that $a$ represents the number of $4$-components of $G$, while $a'$ represents the number of $4$-components of $G$ that are assigned to the inside of ${\mathcal C}^*_{uv}$ by $\mathcal A$. Similarly, $b'\leq b$ and $c'_1\leq c_1$. Moreover, $c'_2\leq c_2$ (and hence $c'\leq c$), given that $c_2$ represents the number of vertices $u_i\in \chi$ with $i\in \{1,\dots,k-1\}$, while $c'_2$ represents the number of vertices $u_i\in \chi$ with $i\in \{1,\dots,k-1\}$ for which $\rho_i=270\degree$. Finally, in order to prove that $d'\leq d$, it suffices to prove that $d'_i\leq d_i$, for each $i\in \{1,\dots,k-1\}$ with $u_i \notin \chi$. By construction, $d'_i= \frac{\rho_i}{90\degree}-1$. By Property~$\mathcal V3$ of Lemma~\ref{le:structural-variable-embedding}, we have $\rho_i\leq 270\degree-\nu_i-\mu_{i+1}$, hence $d'_i\leq \frac{270\degree-\nu_i-\mu_{i+1}}{90\degree}-1=2-\frac{\nu_i+\mu_{i+1}}{90\degree}=d_i$.
	
	It remains to prove that $4a'+3b'+2c'+d'=t$. By Property~$\mathcal V5$ of Lemma~\ref{le:structural-variable-embedding}, we have $\sum_{i=1}^{k-1} (\frac{\rho_i}{90\degree}-1) + \sum_{i=1}^k \frac{\sigma_i}{90\degree} = t$. We show that: (i) $\sum_{i=1}^{k-1} (\frac{\rho_i}{90\degree}-1)=2c'_2+d'$; and (ii) $\sum_{i=1}^k \frac{\sigma_i}{90\degree}=4a'+3b'+2c'_1$. The two equalities imply that $4a'+3b'+2c'+d'=t$.
	
	\begin{enumerate}[(i)]
		\item We partition the set $\{1,\dots,k-1\}$ into three subsets. The subset $\mathcal I_1$ consists of those values $i$ such that  $u_i \in \chi$ and $\rho_i=90\degree$. The subset $\mathcal I_2$ consists of those values $i$ such that  $u_i \in \chi$ and $\rho_i=270\degree$. The subset $\mathcal I_3$ consists of those values $i$ such that  $u_i \notin \chi$. Then we have $\sum_{i=1}^{k-1} (\frac{\rho_i}{90\degree}-1)=\sum_{i\in {\mathcal I_1}} (\frac{\rho_i}{90\degree}-1)+\sum_{i\in {\mathcal I_2}} (\frac{\rho_i}{90\degree}-1)+\sum_{i\in {\mathcal I_3}} (\frac{\rho_i}{90\degree}-1)$. First, we have $\sum_{i\in {\mathcal I_1}} (\frac{\rho_i}{90\degree}-1)=\sum_{i\in {\mathcal I_1}} 0 = 0$. Second, we have $\sum_{i\in {\mathcal I_2}} (\frac{\rho_i}{90\degree}-1)=\sum_{i\in {\mathcal I_2}} 2 = 2|\mathcal I_2|$; by definition, we have $|\mathcal I_2|=c'_2$, hence $\sum_{i\in {\mathcal I_2}} (\frac{\rho_i}{90\degree}-1)=2c'_2$. Third, we have $\sum_{i\in {\mathcal I_3}} (\frac{\rho_i}{90\degree}-1)=\sum_{i\in {\mathcal I_3}} d'_i=d'$. Hence, $\sum_{i=1}^{k-1} (\frac{\rho_i}{90\degree}-1)=2c'_2+d'$.
		\item We partition the set $\{1,\dots,k\}$ into four subsets. The subset $\mathcal J_1$ consists of those values $i$ such that $G_i$ is trivial or is assigned to the outside of ${\mathcal C}^*_{uv}$ by $\mathcal A$. The subset $\mathcal J_2$ (resp., $\mathcal J_3$, $\mathcal J_4$) consists of those values $i$ such that $G_i$ is a $2$-component (resp.\ $3$-component, $4$-component) assigned to the inside of ${\mathcal C}^*_{uv}$ by $\mathcal A$. Then we have $\sum_{i=1}^k \frac{\sigma_i}{90\degree}=\sum_{i\in \mathcal J_1} \frac{\sigma_i}{90\degree}+\sum_{i\in \mathcal J_2} \frac{\sigma_i}{90\degree}+\sum_{i\in \mathcal J_3} \frac{\sigma_i}{90\degree}+\sum_{i\in \mathcal J_4} \frac{\sigma_i}{90\degree}$. First, we have $\sum_{i\in \mathcal J_1} \frac{\sigma_i}{90\degree}=\sum_{i\in \mathcal J_1} 0 = 0$. Second, we have $\sum_{i\in \mathcal J_2} \frac{\sigma_i}{90\degree}=\sum_{i\in \mathcal J_2} \frac{\mu_i+\nu_i}{90\degree}=\sum_{i\in \mathcal J_2} 2$, where the last equality uses that $\mu_i+\nu_i=180\degree$ for a $2$-component; by definition, we have $|\mathcal J_2|=c'_1$, hence $\sum_{i\in \mathcal J_2} \frac{\sigma_i}{90\degree}=2c'_1$. Analogously, and since $\mu_i+\nu_i=270\degree$ and $\mu_i+\nu_i=360\degree$ for a $3$-component and a $4$-component, respectively, we have that $\sum_{i\in \mathcal J_3} \frac{\sigma_i}{90\degree}=3b'$ and $\sum_{i\in \mathcal J_4} \frac{\sigma_i}{90\degree}=4a'$. Hence, $\sum_{i=1}^k \frac{\sigma_i}{90\degree}=4a'+3b'+2c'_1$.
	\end{enumerate} 
	
	($\Longleftarrow$) Suppose that integer values $0\leq a'\leq a$, $0\leq b'\leq b$, $0\leq c'\leq c$, and $0\leq d'\leq d$ exist such that $4a'+3b'+2c'+d'=t$. Choose any non-negative integers $c'_1\leq c_1$ and $c'_2\leq c_2$ such that $c'_1+c'_2=c'$; these exist since $0\leq c'\leq c$ and $c_1+c_2=c$. 
	
	We determine an in-out assignment $\mathcal A$ as follows. We assign any $a'$ $4$-components of $G$ to the inside of ${\mathcal C}^*_{uv}$, and we assign the remaining $a-a'$ $4$-components of $G$ to the outside of ${\mathcal C}^*_{uv}$. Similarly, we assign any $b'$ $3$-components of $G$ and any $c'_1$ $2$-components of $G$ to the inside of ${\mathcal C}^*_{uv}$, and we assign the remaining $b-b'$ $3$-components and the remaining $c_1-c'_1$ $2$-components of $G$ to the outside of ${\mathcal C}^*_{uv}$. 
	
	We determine the values $\rho_1,\dots,\rho_{k-1}$ as follows. 
	
	Let $\chi'$ be any subset of $\chi-\{u_0,u_k\}$ such that $|\chi'|=c'_2$ (recall that $|\chi-\{u_0,u_k\}|=c_2\geq c'_2$). For any vertex $u_i$ in $\chi'$ we set $\rho_i=270\degree$; further, for any vertex $u_i$ in $\chi\setminus \chi'$ with $i\in \{1,\dots,k-1\}$, we set $\rho_i=90\degree$. 
	
	Finally, let $u_{\tau(1)},u_{\tau(2)},\dots,u_{\tau(h)}$ be any order of the vertices in $\{u_1,\dots,u_{k-1}\}\setminus \chi$. Choose any non-negative values $\rho_{\tau(1)},\dots,\rho_{\tau(h)}$ such that $\rho_{\tau(i)}\leq 90\degree \cdot (1+d_{\tau(i)})$, for $i=1,\dots,h$, and such that $\sum_{i=1}^h (\frac{\rho_{\tau(i)}}{90\degree}-1)=d'$. These values exist since $\sum_{i=1}^h d_{\tau(i)}=d$ and $d'\leq d$. Algorithmically, the values $\rho_{\tau(1)},\dots,\rho_{\tau(h)}$ can be determined as follows. Let $\delta_0=d'$. For $i=1,\dots,h$, we set $\rho_{\tau(i)}$ equal to $90\degree \cdot \min \{(1+\delta_{i-1}),(1+d_{\tau(i)})\}$ and we let $\delta_i=\delta_{i-1}-(\frac{\rho_{\tau(i)}}{90\degree}-1)$. Note that $\delta_i\geq 0$ for $i=0,\dots,h$. Namely, we have $\delta_0\geq 0$, given that $d'\geq 0$. Further, suppose that $\delta_{i-1}\geq 0$. We have $\rho_{\tau(i)}\leq 90\degree \cdot  (1+\delta_{i-1})$ and hence $\delta_i\geq \delta_{i-1}-(\frac{90\degree \cdot (1+\delta_{i-1})}{90\degree}-1)=0$. 
	
	Clearly, the described algorithm for the determination of the in-out assignment $\mathcal A$ and of the values $\rho_1,\dots,\rho_{k-1}$ runs in $O(k)$ time. We now prove that $\mathcal A$ and $\rho_1,\dots,\rho_{k-1}$, together with the promising sequence $\mu_1,\nu_1,\mu_2,\nu_2,\dots,\mu_{k},\nu_{k},\rho_{0},\rho_{k}$ for $(G,\mu,\nu)$, satisfy Properties~$\mathcal V1$--$\mathcal V5$ of Lemma~\ref{le:structural-variable-embedding}. 
	
	\begin{itemize}
		\item As remarked earlier, Properties~$\mathcal V2$ and~$\mathcal V4$ are satisfied independently of the choice of $\mathcal A$ and $\rho_1,\dots,\rho_{k-1}$.
		\item By construction, for every vertex $u_i\in \chi$ with $i\in \{1,\dots,k-1\}$ we have either $\rho_i=270\degree$ or $\rho_i=90\degree$ (depending on whether $u_i\in \chi'$ or not). Further, for $i=1,\dots,h$, we have $\rho_{\tau(i)}= 90\degree \cdot \min \{(1+\delta_{i-1}),(1+d_{\tau(i)})\geq 90\degree$. The last inequality follows from $\delta_{i-1}\geq 0$ and $d_{\tau(i)}\geq 0$. Hence, Property~$\mathcal V1$ is satisfied.
		\item For every vertex $u_i\in \chi$ with $i\in \{1,\dots,k-1\}$, we have $\nu_i=\mu_{i+1}=0\degree$ and we have either $\rho_i=90\degree$ or $\rho_i=270\degree$, hence $\rho_i \leq 270\degree -\nu_i- \mu_{i+1}$. Further, for $i=1,\dots,h$, we have $\rho_{\tau(i)}\leq 90\degree \cdot (1+d_{\tau(i)})= 90\degree \cdot (3-\frac{\nu_{\tau(i)}+\mu_{\tau(i)+1}}{90\degree})=270\degree -\nu_{\tau(i)}- \mu_{\tau(i)+1}$. Property~$\mathcal V3$ follows.  
		\item Finally, we deal with Property~$\mathcal V5$. By hypothesis, we have $4a'+3b'+2c'+d'=t$. We show that: (i) $\sum_{i=1}^{k-1} (\frac{\rho_i}{90\degree}-1)=2c'_2+d'$; and (ii) $\sum_{i=1}^k \frac{\sigma_i}{90\degree}=4a'+3b'+2c'_1$. The two equalities imply that $\sum_{i=1}^{k-1} (\frac{\rho_i}{90\degree}-1) + \sum_{i=1}^k \frac{\sigma_i}{90\degree} = t$, as requested. 
		
		\begin{enumerate}[(i)]
			\item As in the proof of necessity, we partition the set $\{1,\dots,k-1\}$ into three subsets $\mathcal I_1$, $\mathcal I_2$, and $\mathcal I_3$ respectively consisting of those values $i$ such that  $u_i \in \chi$ and $\rho_i=90\degree$, of those values $i$ such that  $u_i \in \chi$ and $\rho_i=270\degree$, and of those values $i$ such that  $u_i \notin \chi$; then $\sum_{i=1}^{k-1} (\frac{\rho_i}{90\degree}-1)=\sum_{i\in {\mathcal I_1}} (\frac{\rho_i}{90\degree}-1)+\sum_{i\in {\mathcal I_2}} (\frac{\rho_i}{90\degree}-1)+\sum_{i\in {\mathcal I_3}} (\frac{\rho_i}{90\degree}-1)$. As in the proof of necessity, we have $\sum_{i\in {\mathcal I_1}} (\frac{\rho_i}{90\degree}-1)=\sum_{i\in {\mathcal I_1}} 0 = 0$. Further, we have $\sum_{i\in {\mathcal I_2}} (\frac{\rho_i}{90\degree}-1)=\sum_{i\in {\mathcal I_2}} 2 = 2|\mathcal I_2|$; by construction, we have $|\mathcal I_2|=c'_2$, hence $\sum_{i\in {\mathcal I_2}} (\frac{\rho_i}{90\degree}-1)=2c'_2$. It remains to prove that $\sum_{i\in {\mathcal I_3}} (\frac{\rho_i}{90\degree}-1)=d'$, which implies that $\sum_{i=1}^{k-1} (\frac{\rho_i}{90\degree}-1)=2c'_2+d'$. The proof is as follows. 
			
			Firstly, we have that if, for some index $i\in \{1,\dots,h\}$, the minimum between $(1+\delta_{i-1})$ and $(1+d_{\tau(i)})$ is $(1+\delta_{i-1})$, then $\delta_i=\delta_{i+1}=\cdots=\delta_{h}=0$. Namely, if $\min \{(1+\delta_{i-1}),(1+d_{\tau(i)})\}=(1+\delta_{i-1})$, then $\delta_i=\delta_{i-1}-(\frac{90\degree \cdot (1+\delta_{i-1})}{90\degree}-1)=0$. Further, if $\delta_{j-1}=0$, for some $j\in \{1,\dots,h\}$, then $\min \{(1+\delta_{j-1}),(1+d_{\tau(j)})\}=(1+\delta_{j-1})$, given that $d_{\tau(j)}\geq 0$, hence $\delta_{j}=0$.
			
			Secondly, we have that there exists an index $i\in \{1,\dots,h\}$ such that the minimum between $(1+\delta_{i-1})$ and $(1+d_{\tau(i)})$ is $(1+\delta_{i-1})$. Namely, suppose that, for $i=1,\dots,h-1$, we have $(1+d_{\tau(i)})<(1+\delta_{i-1})$. Then we have $\delta_{h-1}=\delta_0-((\frac{\rho_{\tau(1)}}{90\degree}-1)+\cdots+(\frac{\rho_{\tau(h-1)}}{90\degree}-1))=d'-(d_{\tau(1)}+\cdots+d_{\tau(h-1)})=d'-(d-d_{\tau(h)})\geq d_{\tau(h)}$, where the last inequality is given by $d'\leq d$. Hence, the minimum between $(1+\delta_{h-1})$ and $(1+d_{\tau(h)})$ is $(1+\delta_{h-1})$.
			
			On one hand, the two facts above imply that $\delta_h=0$. On the other hand, we have $\delta_{h}=\delta_0-((\frac{\rho_{\tau(1)}}{90\degree}-1)+\cdots+(\frac{\rho_{\tau(h)}}{90\degree}-1))$, and hence $\sum_{i\in {\mathcal I_3}} (\frac{\rho_i}{90\degree}-1)=\sum_{i=1}^h (\frac{\rho_{\tau(i)}}{90\degree}-1) = \delta_0=d'$.
			\item The proof that $\sum_{i=1}^k \frac{\sigma_i}{90\degree}=4a'+3b'+2c'_1$ is the same as in the proof of necessity, with the only difference that the number of $4$-, $3$-, and $2$-components that are assigned to the inside of ${\mathcal C}^*_{uv}$ is $a'$, $b'$, and $c'_1$ by construction, and not by definition.
			
		\end{enumerate}
	\end{itemize} 
	
	This concludes the proof of the lemma.
\end{proof}

Next, we show that values $a'$, $b'$, $c'$, and $d'$ as in Lemma~\ref{le:extensible-values} exist if and only if the values $a$, $b$, $c$, $d$, and $t$ satisfy certain conditions. In order to keep such a characterization as simple as possible, we first get rid of the cases in which $t$ is ``very small'' or ``very large'' by means of exhaustive search. 

\begin{lemma} \label{le:small-large-t}
	Suppose that $t\leq \constant$ or that $4a+3b+2c+d\leq t+\constant$. Then it is possible to test in $O(1)$ time whether integer values $0\leq a'\leq a$, $0\leq b'\leq b$, $0\leq c'\leq c$, and $0\leq d'\leq d$ exist such that $4a'+3b'+2c'+d'=t$. Further, if such values exist, they can be determined in $O(1)$ time. 
\end{lemma}	

\begin{proof}
	We distinguish three cases.	
	
	In Case~1, we have $t\leq \constant$ and $4a+3b+2c+d> t+\constant$. Then, for any non-negative integer values $a'$, $b'$, $c'$, and $d'$ such that $4a'+3b'+2c'+d'=t$, we have $a'\leq 2$, $b'\leq 3$, $c'\leq 5$, and $d'\leq 10$. We consider every tuple $(a',b',c',d')$ of integers such that $0\leq a'\leq \min\{2,a\}$, $0\leq b'\leq \min\{3,b\}$, $0\leq c'\leq \min\{5,c\}$, and $0\leq d'\leq \min\{10,d\}$. Note that there are $O(1)$ such tuples. For each tuple, we check in $O(1)$ time whether $4a'+3b'+2c'+d'=t$. If a tuple $(a',b',c',d')$ satisfies $4a'+3b'+2c'+d'=t$, then the values $a'$, $b'$, $c'$, and $d'$ have been determined in $O(1)$ time. If no  tuple $(a',b',c',d')$ satisfies $4a'+3b'+2c'+d'=t$, then we conclude that the required values $a'$, $b'$, $c'$, and $d'$ do not exist.
	
	In Case~2, we have $4a+3b+2c+d\leq t+\constant$ and $t>\constant$. Then, for any non-negative integer values $a'$, $b'$, $c'$, and $d'$ such that $4a'+3b'+2c'+d'=t$, we have $a'\geq a-2$, $b'\geq b-3$, $c'\geq c-5$, and $d'\geq d-10$. We consider every tuple $(a',b',c',d')$ of integers such that $\max\{0,a-2\}\leq a'\leq a$, $\max\{0,b-3\}\leq b'\leq b$,  $\max\{0,c-5\}\leq c'\leq c$, and $\max\{0,d-10\}\leq d'\leq d$. Note that there are $O(1)$ such tuples. As in Case~1, for each tuple $(a',b',c',d')$, we check in $O(1)$ time whether $4a'+3b'+2c'+d'=t$; this leads either to the determination of a tuple $(a',b',c',d')$ such that $4a'+3b'+2c'+d'=t$, or to the conclusion that the required values $a'$, $b'$, $c'$, and $d'$ do not exist. 
	
	In Case~3, we have $t\leq \constant$ and $4a+3b+2c+d\leq t+\constant$. Then the constraints of both Cases~1 and~2 apply. Hence, we need to consider every tuple $(a',b',c',d')$ of integers such that $\max\{0,a-2\}\leq a' \leq \min\{2,a\}$, $\max\{0,b-3\}\leq b'\leq \min\{3,b\}$, $\max\{0,c-5\}\leq c'\leq \min\{5,c\}$, and $\max\{0,d-10\}\leq d'\leq \min\{10,d\}$. Note that there are $O(1)$ such tuples. As in Cases~1 and~2, for each tuple $(a',b',c',d')$, we check in $O(1)$ time whether $4a'+3b'+2c'+d'=t$; this leads either to the determination of a tuple $(a',b',c',d')$ such that $4a'+3b'+2c'+d'=t$, or to the conclusion that the required values $a'$, $b'$, $c'$, and $d'$ do not exist. 
\end{proof}

We now deal with the case in which $t$ is neither too small nor too large.

\begin{lemma} \label{le:medium-t}
	Suppose that $t>\constant$ and that $4a+3b+2c+d>t+\constant$. Then integer values $0\leq a'\leq a$, $0\leq b'\leq b$, $0\leq c'\leq c$, and $0\leq d'\leq d$ exist such that $4a'+3b'+2c'+d'=t$ if and only if the formula $\mathcal F:= \mathcal F_3 \vee \mathcal F_2 \vee \mathcal F_1 \vee \mathcal F_0$ is satisfied, where:
	
	\begin{itemize}
		\item $\mathcal F_3 := (d\geq 3)$;
		\item $\mathcal F_2 := (d=2) \wedge (\mathcal F_{2,0} \vee \mathcal F_{2,1} \vee \mathcal F_{2,2})$, where 
		\begin{itemize}
			\item $\mathcal F_{2,0}:=(c>0)$,
			\item $\mathcal F_{2,1}:=(b>0)$, and
			\item $\mathcal F_{2,2}:=(t \not \equiv 3 \mod 4)$;
		\end{itemize}  
		\item $\mathcal F_1 := (d=1) \wedge (\mathcal F_{1,0} \vee \mathcal F_{1,1} \vee \mathcal F_{1,2} \vee \mathcal F_{1,3} \vee \mathcal F_{1,4} \vee \mathcal F_{1,5})$, where 
		\begin{itemize}
			\item $\mathcal F_{1,0}:=(c>0)$, 
			\item $\mathcal F_{1,1}:= (a=0) \wedge (t \not \equiv 2 \mod 3)$, 
			\item $\mathcal F_{1,2}:=(b=0) \wedge ((t \equiv 0 \mod 4)\vee(t \equiv 1 \mod 4))$, 
			\item $\mathcal F_{1,3}:=(b=1) \wedge ((t \equiv 0 \mod 4)\vee(t \equiv 1 \mod 4)\vee(t \equiv 3 \mod 4))$, 
			\item $\mathcal F_{1,4}:= (b=2)$, and 
			\item $\mathcal F_{1,5}:=(a\geq 1) \wedge (b\geq 3)$; 
		\end{itemize} 
		\item $\mathcal F_0 := (d=0) \wedge (\mathcal F_{0,0} \vee \mathcal F_{0,1} \vee \mathcal F_{0,2} \vee \mathcal F_{0,3} \vee \mathcal F_{0,4} \vee \mathcal F_{0,5} \vee \mathcal F_{0,6} \vee \mathcal F_{0,7} \vee \mathcal F_{0,7})$, where 
		\begin{itemize}
			\item $\mathcal F_{0,0}:=(c=0) \wedge (b=0) \wedge (t\equiv 0 \mod 4)$, 
			\item $\mathcal F_{0,1}:=(c=0) \wedge (b=1) \wedge ((t\equiv 0 \mod 4)\vee (t\equiv 3 \mod 4))$,  
			\item $\mathcal F_{0,2}:=(c=0) \wedge (b=2) \wedge ((t\equiv 0 \mod 4)\vee (t\equiv 2 \mod 4)\vee (t\equiv 3 \mod 4))$,
			\item $\mathcal F_{0,3}:=(c=0) \wedge (a=0) \wedge (t\equiv 0 \mod 3)$,
			\item $\mathcal F_{0,4}:=(c=0) \wedge (a=1) \wedge ((t\equiv 0 \mod 3)\vee(t\equiv 1 \mod 3))$,
			\item $\mathcal F_{0,5}:=(c=0) \wedge (b\geq 3) \wedge (a\geq 2)$,
			\item $\mathcal F_{0,6}:=(c\geq 1) \wedge (b=0) \wedge (t\equiv 0 \mod 2)$,
			\item $\mathcal F_{0,7}:=(c=1) \wedge (a=0) \wedge ((t\equiv 0 \mod 3)\vee (t\equiv 2 \mod 3))$,
			\item $\mathcal F_{0,8}:=(c\geq 1) \wedge (b\geq 1)  \wedge (a\geq 1)$, and
			\item $\mathcal F_{0,9}:=(c\geq 2) \wedge (b\geq 1)$.
		\end{itemize}
	\end{itemize}
	
	Consequently, it can be tested in $O(1)$ time whether integer values $a'$, $b'$, $c'$, and $d'$ as above exist. Further, if such values exist, they can be determined in $O(1)$ time. 
\end{lemma}	

\begin{proof}
	($\Longrightarrow$) We first prove the necessity of the characterization. That is, we show that, if $\mathcal F$ is not satisfied, then integer values $0\leq a'\leq a$, $0\leq b'\leq b$, $0\leq c'\leq c$, and $0\leq d'\leq d$ such that $4a'+3b'+2c'+d'=t$ do not exist. Since $\mathcal F$ is not satisfied, neither of $\mathcal F_0$--$\mathcal F_3$ is satisfied. Since $\mathcal F_3$ is not satisfied, we have $d\leq 2$. We distinguish three cases according to whether $d=2$, $d=1$, and $d=0$.
	
	\begin{itemize}
		\item Suppose first that $d=2$. Since $\mathcal F_2$ is not satisfied, we have $b=0$, $c=0$, and $t\equiv 3 \mod 4$. Then $b'=c'=0$, while $d'\in \{0,1,2\}$. Hence, for any integer value of $a'$, we have that $4a'+3b'+2c'+d'=4a'+d'$ is equivalent to either $0$, $1$, or $2$ modulo $4$, while $t\equiv 3 \mod 4$, hence $4a'+3b'+2c'+d'\neq t$.
		\item Suppose next that $d=1$. Since $\mathcal F_1$ is not satisfied, neither of $\mathcal F_{1,0}$--$\mathcal F_{1,5}$ is satisfied. Since $\mathcal F_{1,0}$ is not satisfied, we have $c=0$; further, since $\mathcal F_{1,5}$ is not satisfied, we have $a=0$ or $b\leq 2$. Since $\mathcal F_{1,4}$ is not satisfied, we have $b\neq 2$. Hence, we can assume that $a=0$, or $b=0$, or $b=1$.
		\begin{itemize}
			\item Suppose that $a=0$. Since $\mathcal F_{1,1}$ is not satisfied, we have $t \equiv 2 \mod 3$. Since $a'=c'=0$, we have $4a'+3b'+2c'+d'=3b'+d'$; further, since $d'\in \{0,1\}$, we have that $3b'+d'$ is equivalent to either $0$ or $1$ modulo $3$, hence it can not be equal to $t$, which is equivalent to  $2$ modulo $3$.
			\item Suppose that $b=0$. Since $\mathcal F_{1,2}$ is not satisfied, we have $t \equiv 2 \mod 4$ or $t \equiv 3 \mod 4$. Since $b'=c'=0$, we have $4a'+3b'+2c'+d'=4a'+d'$; further, since $d'\in \{0,1\}$, we have that $4a'+d'$ is equivalent to $0$ or $1$ modulo $4$, hence it can not be equal to $t$, which is equivalent to $2$ or $3$ modulo $4$.
			\item Suppose that $b=1$. Since $\mathcal F_{1,3}$ is not satisfied, we have $t \equiv 2 \mod 4$. Recall that $c'=0$ and $d'\in \{0,1\}$, given that $c=0$ and $d=1$. If $b'=0$, then we have  $4a'+3b'+2c'+d'=4a'+d'$, which is equivalent to $0$ or $1$ modulo $4$, hence it can not be equal to $t$, which is equivalent to $2$ modulo $4$. Further, if $b'=1$, then we have  $4a'+3b'+2c'+d'=4a'+3+d'$, which is equivalent to $0$ or $3$ modulo $4$, hence it can not be equal to $t$, which is equivalent to $2$ modulo $4$.
		\end{itemize}
		
		\item Suppose finally that $d=0$. Since $\mathcal F_0$ is not satisfied, neither of $\mathcal F_{0,0}$--$\mathcal F_{0,9}$ is satisfied. 
		
		Suppose first that $c=0$. Since $\mathcal F_{0,5}$ is not satisfied, we have $b\leq 2$ or $a\leq 1$. We distinguish five cases.
		\begin{itemize}
			\item If $b=0$, then since $\mathcal F_{0,0}$ is not satisfied, we have $t \equiv 1 \mod 4$ or $t \equiv 2 \mod 4$ or $t \equiv 3 \mod 4$. Since $b'=c'=d'=0$, we have $4a'+3b'+2c'+d'=4a'$, which is equivalent to $0$ modulo $4$, hence it is not equal to $t$.			
			\item If $b=1$, then since $\mathcal F_{0,1}$ is not satisfied, we have $t \equiv 1 \mod 4$ or $t \equiv 2 \mod 4$. Since $c'=d'=0$, we have $4a'+3b'+2c'+d'=4a'+3b'$. Since $b'\in\{0,1\}$, it follows that $4a'+3b'$ is equivalent to $0$ or $3$ modulo $4$, hence it is not equal to $t$.			
			\item If $b=2$, then since $\mathcal F_{0,2}$ is not satisfied, we have $t \equiv 1 \mod 4$. Since $c'=d'=0$, we have $4a'+3b'+2c'+d'=4a'+3b'$. Since $b'\in\{0,1,2\}$, it follows that $4a'+3b'$ is equivalent to $0$, or $2$, or $3$ modulo $4$, hence it is not equal to $t$.
			\item If $a=0$, then since $\mathcal F_{0,3}$ is not satisfied, we have $t \equiv 1 \mod 3$ or $t \equiv 2 \mod 3$. Since $a'=c'=d'=0$, we have $4a'+3b'+2c'+d'=3b'$, which is equivalent to $0$ modulo $3$, hence it is not equal to $t$.
			\item If $a=1$, then since $\mathcal F_{0,4}$ is not satisfied, we have $t \equiv 2 \mod 3$. Since $c'=d'=0$, we have $4a'+3b'+2c'+d'=4a'+3b'$. Since $a'\in\{0,1\}$, it follows that $4a'+3b'$ is equivalent to $0$ or $1$ modulo $3$, hence it is not equal to $t$.			
		\end{itemize}
		
		Suppose next that $c\geq 1$. Since $\mathcal F_{0,8}$ is not satisfied, we have $a=0$ or $b=0$. Further, since $\mathcal F_{0,9}$ is not satisfied, we have $c\leq 1$ or $b=0$. We hence distinguish two cases, namely the one in which $b=0$ and the one in which $a=0$ and $c=1$.
		
		\begin{itemize}
			\item If $b=0$, then since $\mathcal F_{0,6}$ is not satisfied, we have $t\equiv 1 \mod 2$. Since $b'=d'=0$, we have $4a'+3b'+2c'+d'=4a'+2c'$, which is equivalent to $0$ modulo $2$, hence it is not equal to $t$.
			\item If $a=0$ and $c=1$, then since $\mathcal F_{0,7}$ is not satisfied, we have $t \equiv 1 \mod 3$. Since $a'=d'=0$, we have $4a'+3b'+2c'+d'=3b'+2c'$; further, since $c'\in \{0,1\}$, we have that $3b'+2c'$ is equivalent to $0$ or $2$ modulo $3$, hence it is not equal to $t$.
		\end{itemize}
	\end{itemize}
	

	($\Longleftarrow$)  We now prove the sufficiency of the characterization. That is, we show that, if $\mathcal F$ is satisfied, then there exist integer values $0\leq a'\leq a$, $0\leq b'\leq b$, $0\leq c'\leq c$, and $0\leq d'\leq d$ such that $4a'+3b'+2c'+d'=t$. The proof distinguishes four cases, according to which of $\mathcal F_3$, $\mathcal F_2$, $\mathcal F_1$, and $\mathcal F_0$ is satisfied.
	
	Before discussing these cases, we define a notion of \emph{greedy values} for $a'$, $b'$, $c'$, and $d'$. Roughly speaking, these are the values obtained by taking $a'$ as large as possible so that $4a'$ does not exceed $t$, then $b'$ as large as possible so that $3b'$ does not exceed $t-4a'$, then $c'$ as large as possible so that $2c'$ does not exceed $t-4a'-3b'$, and finally $d'$ as large as possible so that it does not exceed $t-4a'-3b'-2c'$. Formally, the greedy values for $a'$, $b'$, $c'$, and $d'$ are defined as $a'=\min\{a,\lfloor \frac{t}{4}\rfloor\}$, $b'=\min\{b,\lfloor \frac{t-4a'}{3}\rfloor\}$, $c'=\min\{c,\lfloor \frac{t-4a'-3b'}{2}\rfloor\}$, and $d'=\min\{d,t-4a'-3b'-2c'\}$. 
	
	Note that the greedy values for $a'$, $b'$, $c'$, and $d'$ satisfy $0\leq a'\leq a$, $0\leq b'\leq b$, $0\leq c'\leq c$, and $0\leq d'\leq d$. Namely, by construction, we have $a'\leq a$, $b'\leq b$, $c'\leq c$, and $d'\leq d$. Further, since $a\geq 0$, $b\geq 0$, $c\geq 0$, and $d\geq 0$, in order to prove that $a'\geq 0$, $b'\geq 0$, $c'\geq 0$, and $d'\geq 0$, it suffices to prove that $t\geq 0$, that $t-4a'\geq 0$, that $t-4a'-3b'\geq 0$, and that $t-4a'-3b'-2c'\geq 0$. 
	
	The first inequality, $t\geq 0$, is true by the assumption $t>10$. Further, $a'\leq \lfloor \frac{t}{4}\rfloor$, hence $t-4a'\geq t-4\lfloor \frac{t}{4}\rfloor \geq 0$. Analogously, $b' \leq \lfloor \frac{t-4a'}{3}\rfloor$, hence $t-4a'-3b'\geq t-4a'-3\lfloor \frac{t-4a'}{3}\rfloor \geq 0$. Finally, $c'\leq \lfloor \frac{t-4a'-3b'}{2}\rfloor$, hence $t-4a'-3b'-2c'\geq t-4a'-3b'-2\lfloor \frac{t-4a'-3b'}{2}\rfloor\geq 0$.

	
	{\bf Suppose first that $\bf \mathcal F_3$ is satisfied}, that is, $d\geq 3$. We use the greedy values for $a'$, $b'$, $c'$, and $d'$. We prove that $\min\{d,t-4a'-3b'-2c'\}=t-4a'-3b'-2c'$, which implies that $4a'+3b'+2c'+d'=t$. Suppose, for a contradiction, that $d<t-4a'-3b'-2c'$. Since $d\geq 3$, we have $4a'+3b'+2c'\leq t-4$. From this inequality, it follows that (A) $a'=\min\{a,\lfloor \frac{t}{4}\rfloor\}=a$, (B) $b'=\min\{b,\lfloor \frac{t-4a'}{3}\rfloor\}=b$, and (C) $c'=\min\{c,\lfloor \frac{t-4a'-3b'}{2}\rfloor\}=c$. Namely, $4a'+3b'+2c'\leq t-4$ implies $a' \leq \frac{t}{4} -1 <\lfloor \frac{t}{4}\rfloor$, which implies (A). Analogously, $4a'+3b'+2c'\leq t-4$ implies $b' \leq \frac{t-4a'}{3}-\frac{4}{3}<\lfloor \frac{t-4a'}{3}\rfloor$, which implies (B), and it implies $c' \leq \frac{t-4a'-3b'}{2}-2<\lfloor \frac{t-4a'-3b'}{2}\rfloor$, which implies (C). Finally, by the assumption $d<t-4a'-3b'-2c'$, (A), (B), and (C) imply that $4a+3b+2c+d<t$, which contradicts the assumption $4a+3b+2c+d>t+\constant$.
	
	
	{\bf Suppose next that $\bf \mathcal F_2$ is satisfied}, that is, $d=2$ and (at least) one of $\mathcal F_{2,0}$, $\mathcal F_{2,1}$, and $\mathcal F_{2,2}$ holds true. We use the greedy values for $a'$, $b'$, $c'$, and $d'$. We prove that $\min\{d,t-4a'-3b'-2c'\}=t-4a'-3b'-2c'$. This implies that $4a'+3b'+2c'+d'=t$. Suppose, for a contradiction, that $d<t-4a'-3b'-2c'$. Since $d=2$, we have $4a'+3b'+2c'\leq t-3$. From this inequality, it follows that (B) $b'=\min\{b,\lfloor \frac{t-4a'}{3}\rfloor\}=b$ and (C) $c'=\min\{c,\lfloor \frac{t-4a'-3b'}{2}\rfloor\}=c$. Namely, the inequality $4a'+3b'+2c'\leq t-3$ implies $b' \leq \frac{t-4a'}{3}-1<\lfloor \frac{t-4a'}{3}\rfloor$, which implies (B), and it implies $c' \leq \frac{t-4a'-3b'}{2}-\frac{3}{2}<\lfloor \frac{t-4a'-3b'}{2}\rfloor$, which implies (C). Now, if $a'=\min\{a,\lfloor \frac{t}{4}\rfloor\}=a$, then the assumption $d<t-4a'-3b'-2c'$, together with (B) and (C), implies that $4a+3b+2c+d<t$, which contradicts the assumption $4a+3b+2c+d>t+\constant$. We can hence assume that $a'=\min\{a,\lfloor \frac{t}{4}\rfloor\}=\lfloor \frac{t}{4}\rfloor$, which implies that $t-4a' \in\{0,1,2,3\}$. In order to reach a contradiction, we distinguish three cases.
	\begin{itemize}
		\item[($\mathcal F_{2,0}$):] Suppose first that $c>0$. Since $t-4a' \in\{0,1,2,3\}$ and $c>0$, it follows that $t-4a'-3b'-2c' \in \{0,1\}$ (as if $t-4a'-3b'\in \{2,3\}$, then $c'=1$). Since $d=2$, this contradicts $d<t-4a'-3b'-2c'$.
		\item[($\mathcal F_{2,1}$):] Suppose next that $b>0$. Since $t-4a' \in\{0,1,2,3\}$ and $b>0$, it follows that $t-4a'-3b'-2c' \in \{0,1,2\}$ (as if $t-4a'=3$, then $b'=1$). Since $d=2$, this contradicts $d<t-4a'-3b'-2c'$.
		\item[($\mathcal F_{2,2}$):] Suppose finally that $t \not \equiv 3 \mod 4$. By the previous two cases, we can also assume that $b=c=0$. Hence, $t-4a'-3b'-2c'=t-4a'$. The value $t-4a'$ can not be $3$, given that  $t \not \equiv 3 \mod 4$. Hence, $t-4a'\in \{0,1,2\}$. Since $d=2$, this contradicts $d<t-4a'-3b'-2c'$.
	\end{itemize}
	
	
	{\bf Suppose next that $\bf \mathcal F_1$ is satisfied}, that is, $(d=1)$ and (at least) one of $\mathcal F_{1,0}$--$\mathcal F_{1,5}$ holds true. We accordingly distinguish six cases. 
	
	In cases ($\mathcal F_{1,0}$)--($\mathcal F_{1,3}$), we use the greedy values for $a'$, $b'$, $c'$, and $d'$, and we prove that $\min\{d,t-4a'-3b'-2c'\}=t-4a'-3b'-2c'$, which implies that $4a'+3b'+2c'+d'=t$. Suppose, for a contradiction, that $d<t-4a'-3b'-2c'$. Since $d=1$, we have $4a'+3b'+2c'\leq t-2$. From this inequality, it follows that (C) $c'=\min\{c,\lfloor \frac{t-4a'-3b'}{2}\rfloor\}=c$. Namely, the inequality $4a'+3b'+2c'\leq t-2$ implies $c' \leq \frac{t-4a'-3b'}{2}-1<\lfloor \frac{t-4a'-3b'}{2}\rfloor$, which implies (C).
	
	\begin{itemize}
		\item[($\mathcal F_{1,0}$):] If $c>0$, then we further distinguish three cases.
		
		\begin{itemize}
			\item If $b'=\min\{b,\lfloor \frac{t-4a'}{3}\rfloor\}=\lfloor \frac{t-4a'}{3}\rfloor$, then $t-4a'-3b' \in \{0,1,2\}$, hence $t-4a'-3b'-2c' \in \{0,1\}$ (as if $t-4a'-3b'=2$, then $c'=1$). Since $d=1$, this contradicts $d<t-4a'-3b'-2c'$.
			\item If $a'=\min\{a,\lfloor \frac{t}{4}\rfloor\}=\lfloor \frac{t}{4}\rfloor$, then $t-4a' \in\{0,1,2,3\}$. It follows that $t-4a'-3b'-2c' \in \{0,1\}$ (as if $t-4a'-3b'\in \{2,3\}$, then $c'=1$). Since $d=1$, this contradicts $d<t-4a'-3b'-2c'$.
			\item Finally, if $a'=\min\{a,\lfloor \frac{t}{4}\rfloor\}=a$ and $b'=\min\{b,\lfloor \frac{t-4a'}{3}\rfloor\}=b$, then $d<t-4a'-3b'-2c'$, together with $a'=a$, $b'=b$, and $c'=c$, implies that $4a+3b+2c+d<t$, which contradicts the assumption $4a+3b+2c+d>t+\constant$.
		\end{itemize}
		In the upcoming cases ($\mathcal F_{1,1}$)--($\mathcal F_{1,5}$), we hence assume that $c=0$.

		\item[($\mathcal F_{1,1}$):] If $a=0$ and $t \not \equiv 2 \mod 3$, then we further distinguish two cases. 
		\begin{itemize}
			\item If $b'=\min\{b,\lfloor \frac{t-4a'}{3}\rfloor\}=b$, then $d<t-4a'-3b'-2c'$, together with $a'=a$, $b'=b$, and $c'=c$ (the first and third equalities come from $a=0$ and $c=0$, respectively), implies that $4a+3b+2c+d<t$, which contradicts the assumption $4a+3b+2c+d>t+\constant$. 
			\item Otherwise, we have $b'=\min\{b,\lfloor \frac{t-4a'}{3}\rfloor\}=\lfloor \frac{t-4a'}{3}\rfloor$, hence $t-4a'-3b' \in \{0,1,2\}$. Since $a=0$ and $c=0$, we have $t-4a'-3b'-2c'=t-3b'$. The value $t-3b'$ can not be $2$ given that $t \not \equiv 2 \mod 3$. Hence, $t-3b' \in \{0,1\}$. Since $d=1$, this contradicts $d<t-4a'-3b'-2c'$. 
		\end{itemize}
		
		\item[($\mathcal F_{1,2}$):] We have $b=0$, and $t \equiv 0 \mod 4$ or $t \equiv 1 \mod 4$. We first show that $\min\{a,\lfloor \frac{t}{4}\rfloor\}=\lfloor \frac{t}{4}\rfloor$, given that $b\leq 2$; hence this equality holds true also for the cases ($\mathcal F_{1,3}$) and ($\mathcal F_{1,4}$). Indeed, if $a<\lfloor \frac{t}{4}\rfloor$, then we have $4a+3b+2c+d<4\lfloor\frac{t}{4}\rfloor+6+0+1=4\lfloor\frac{t}{4}\rfloor+7$. However, this contradicts the assumption that $4a+3b+2c+d> t+\constant$. 
		
		Then we have $a'=\min\{a,\lfloor \frac{t}{4}\rfloor\}=\lfloor \frac{t}{4}\rfloor$, which implies that $t-4a' \in\{0,1,2,3\}$. Since $b=0$ and $c=0$, we have $t-4a'-3b'-2c'=t-4a'$. The value $t-4a'$ can not be $2$ or $3$ given that $t \not \equiv 2 \mod 4$ and $t \not \equiv 3 \mod 4$. Hence, $t-4a' \in \{0,1\}$. Since $d=1$, this contradicts $d<t-4a'-3b'-2c'$. 
		
		\item[($\mathcal F_{1,3}$):] If $b=1$ and $t \not \equiv 2 \mod 4$, then we have $a'=\min\{a,\lfloor \frac{t}{4}\rfloor\}=\lfloor \frac{t}{4}\rfloor$, which implies that $t-4a' \in\{0,1,2,3\}$. The value $t-4a'$ can not be $2$ given that $t \not \equiv 2 \mod 4$. If $t-4a' \in \{0,1\}$, then $b'=0$, hence $t-4a'-3b'-2c'=t-4a'\in \{0,1\}$. Further, if $t-4a'=3$, then $b'=1$, hence $t-4a'-3b'-2c'=t-4a'-3=0$. In both cases, since $d=1$, this contradicts $d<t-4a'-3b'-2c'$. 
		
		\item[($\mathcal F_{1,4}$):] If $b=2$, then we do not always use the greedy values for $a'$, $b'$, $c'$, and $d'$. Recall that $\min\{a,\lfloor \frac{t}{4}\rfloor\}=\lfloor \frac{t}{4}\rfloor$. We distinguish two cases, based on the value of $t-4\lfloor\frac{t}{4}\rfloor$.
		\begin{itemize}
			\item If $t-4\lfloor\frac{t}{4}\rfloor=2$, then we set $a' = \lfloor\frac{t}{4}\rfloor -1$, so that $t-4a'=6$; further, we set $b'=2$, and $c'=d'=0$. Thus, $4a'+3b'+2c'+d'=t$.
			\item If $t-4\lfloor\frac{t}{4}\rfloor\in \{0,1,3\}$, then we use the greedy values for $a'$, $b'$, $c'$, and $d'$. Thus, we have $a' = \lfloor\frac{t}{4}\rfloor$, and hence $t-4a'\in \{0,1,3\}$; further $b'=\min\{b,\lfloor\frac{t-4a'}{3}\rfloor\}$, which implies that the value of $t-4a'-3b'$ is either $0$ or $1$. In the former case we set $c'=d'=0$, while in the latter case we set $c'=0$ and $d'=1$. In both cases we get $4a'+3b'+2c'+d'=t$.
		\end{itemize}
		
		\item[($\mathcal F_{1,5}$):] If $a\geq 1$ and $b\geq 3$, then we distinguish three cases.
		\begin{itemize}
			\item If $t\equiv 0 \mod 3$, then we set $a'$ to the largest multiple of $3$ such that $4a'$ does not exceed $t$; that is, $a'=\max \{a'' : a''\leq \min \{a,\lfloor \frac{t}{4} \rfloor\} \wedge a''\equiv 0 \mod 3\}$; then $t-4a' \equiv 0 \mod 3$. By setting $b'=\frac{t-4a'}{3}$ and $c'=d'=0$, we get that $4a'+3b'+2c'+d'=t$. 
			
			It remains to prove that $b'\leq b$. Suppose, for a contradiction, that $b'>b$. If $\min \{a,\lfloor \frac{t}{4} \rfloor\}=a$, then we have $a'\geq a-2$. Further, we have $d'=d-1$. Hence $t=4a'+3b'+2c'+d'>4(a-2)+3b+2c+(d-1)$, that is, $4a+3b+2c+d<t+9$, which contradicts $4a+3b+2c+d>t+\constant$. Otherwise, we have $\min \{a,\lfloor \frac{t}{4} \rfloor\}=\lfloor \frac{t}{4} \rfloor$. Since $t\equiv 0 \mod 3$, we have either $t=12t'$, or $t=12t'+3$, or $t=12t'+6$, or $t=12t'+9$, for some non-negative integer $t'$. Hence, in all four cases we have $a'=3t'$, and thus $t-4a'$ is either $0$, or $3$, or $6$, or $9$. Since $b'=\frac{t-4a'}{3}$, we have $b'\leq 3\leq b$, a contradiction. 
			
			\item If $t\equiv 1 \mod 3$, then we let $a'$ be the largest integer equivalent to $1$ modulo $3$ and such that $4a'$ does not exceed $t$; that is, $a'=\max \{a'' : a''\leq \min \{a,\lfloor \frac{t}{4} \rfloor\} \wedge a''\equiv 1 \mod 3\}$. Note that $a'\geq 1$; indeed, $a\geq 1$ by assumption, and $\lfloor \frac{t}{4} \rfloor\geq 1$ given that $t>\constant$. Further, since $a'=3a^*+1$, for some integer $a^*$, we have $4a'=12a^*+4\equiv 1 \mod 3$. Hence, $t-4a' \equiv 0 \mod 3$. By setting $b'=\frac{t-4a'}{3}$ and $c'=d'=0$, we get that $4a'+3b'+2c'+d'=t$.
			
			It remains to prove that $b'\leq b$. Suppose, for a contradiction, that $b'>b$. If $\min \{a,\lfloor \frac{t}{4} \rfloor\}=a$, a contradiction is reached as in the case in which $t\equiv 0 \mod 3$. Otherwise, we have $\min \{a,\lfloor \frac{t}{4} \rfloor\}=\lfloor \frac{t}{4} \rfloor$. Since $t\equiv 1 \mod 3$ and $t>\constant$, we have either $t=12t'+4$, or $t=12t'+7$, or $t=12t'+10$, or $t=12t'+13$, for some non-negative integer $t'$. Hence, in all four cases we have $a'=3t'+1$, and thus $t-4a'$ is either $0$, or $3$, or $6$, or $9$. Since $b'=\frac{t-4a'}{3}$, we have $b'\leq 3\leq b$, a contradiction. 
			
			\item If $t\equiv 2 \mod 3$, then we let $a'$ be the largest integer equivalent to $1$ modulo $3$ and such that $4a'$ does not exceed $t$; that is, $a'=\max \{a'' : a''\leq \min \{a,\lfloor \frac{t}{4} \rfloor\} \wedge a''\equiv 1 \mod 3\}$. As in the case in which $t\equiv 1 \mod 3$, we have $a'\geq 1$ and $4a'\equiv 1 \mod 3$. Hence, $t-4a' \equiv 1 \mod 3$. By setting $b'=\frac{t-4a'-1}{3}$, $c'=0$, and $d'=1$, we get that $4a'+3b'+2c'+d'=t$.
			
			It remains to prove that $b'\leq b$. Suppose, for a contradiction, that $b'>b$. If $\min \{a,\lfloor \frac{t}{4} \rfloor\}=a$, then we have $a'\geq a-2$. Further, we have $d'=d$. Hence $t=4a'+3b'+2c'+d'>4(a-2)+3b+2c+d$, that is, $4a+3b+2c+d<t+8$, which contradicts $4a+3b+2c+d>t+\constant$. Otherwise, we have $\min \{a,\lfloor \frac{t}{4} \rfloor\}=\lfloor \frac{t}{4} \rfloor$. Since $t\equiv 2 \mod 3$ and $t>\constant$, we have either $t=12t'+5$, or $t=12t'+8$, or $t=12t'+11$, or $t=12t'+14$, for some non-negative integer $t'$. Hence, in all four cases we have $a'=3t'+1$, and thus $t-4a'-1$ is either $0$, or $3$, or $6$, or $9$. Since $b'=\frac{t-4a'-1}{3}$, we have $b'\leq 3\leq b$, a contradiction. 
		\end{itemize} 
	\end{itemize}
	
	
	{\bf Suppose finally that $\bf \mathcal F_0$ is satisfied}, that is, $(d=0)$ and (at least) one of $\mathcal F_{0,0}$--$\mathcal F_{0,9}$ holds true. We accordingly distinguish ten cases. 
	
	\begin{itemize}
		\item[($\mathcal F_{0,0}$):] If $c=0$, $b=0$, and $t \equiv 0 \mod 4$, then we use the greedy values for $a'$, $b'$, $c'$, and $d'$. Since $t \equiv 0 \mod 4$, we have $\lfloor \frac{t}{4}\rfloor =\frac{t}{4}$; further, $\frac{t}{4}<a$, as otherwise $4a+3b+2c+d=4a\leq t$, while $4a+3b+2c+d>t+\constant$. Hence, we have $a'=\frac{t}{4}$, $b'=c'=d'=0$, and $4a'+3b'+2c'+d'= t$.
		\item[($\mathcal F_{0,1}$):] Suppose that $c=0$, $b=1$, and $t \equiv 0 \mod 4$ or $t \equiv 3 \mod 4$. Then $\min\{a,\lfloor \frac{t}{4}\rfloor\}=\lfloor \frac{t}{4}\rfloor$, as if $a<\lfloor \frac{t}{4}\rfloor\leq \frac{t}{4}$, then $4a+3b+2c+d=4a+3\leq t+3$, while $4a+3b+2c+d>t+\constant$. We use the greedy values for $a'$, $b'$, $c'$, and $d'$. 
		If $t \equiv 0 \mod 4$, then $t-4\lfloor \frac{t}{4}\rfloor=t-4\frac{t}{4}=0$, hence the greedy values $a'=\frac{t}{4}$, $b'=0$, $c'=0$, and $d'=0$ satisfy $4a'+3b'+2c'+d'= t$. If $t \equiv 3 \mod 4$, then $t-4\lfloor \frac{t}{4}\rfloor=3$, hence the greedy values $a'=\lfloor \frac{t}{4}\rfloor$, $b'=1$, $c'=0$, and $d'=0$ satisfy $4a'+3b'+2c'+d'= t$.
		\item[($\mathcal F_{0,2}$):] Suppose that $c=0$, $b=2$, and $t \equiv 0 \mod 4$, or $t \equiv 2 \mod 4$, or $t \equiv 3 \mod 4$. Then $\min\{a,\lfloor \frac{t}{4}\rfloor\}=\lfloor \frac{t}{4}\rfloor$, as if $a<\lfloor \frac{t}{4}\rfloor\leq \frac{t}{4}$, then $4a+3b+2c+d=4a+6\leq t+6$, while $4a+3b+2c+d>t+\constant$. If $t \equiv 0 \mod 4$ or $t \equiv 3 \mod 4$, then we use the greedy values for $a'$, $b'$, $c'$, and $d'$; as in the case ($\mathcal F_{0,1}$), such values satisfy $4a'+3b'+2c'+d'= t$. If $t \equiv 2 \mod 4$, then $t-4\lfloor \frac{t}{4}\rfloor=2$ and thus $t-4(\lfloor \frac{t}{4}\rfloor-1)=6$, hence the values $a'=\lfloor \frac{t}{4}\rfloor-1$, $b'=2$, $c'=0$, and $d'=0$ satisfy $4a'+3b'+2c'+d'= t$; note that $a'=\lfloor \frac{t}{4}\rfloor-1>0$, given that $t>\constant$.
		\item[($\mathcal F_{0,3}$):] If $c=0$, $a=0$, and $t \equiv 0 \mod 3$, then we use the greedy values for $a'$, $b'$, $c'$, and $d'$. Since $t \equiv 0 \mod 3$, we have $\lfloor \frac{t}{3}\rfloor =\frac{t}{3}$; further, $\frac{t}{3}<b$, as otherwise $4a+3b+2c+d=3b\leq t$, while $4a+3b+2c+d>t+\constant$. Hence, we have $b'=\frac{t}{3}$, $a'=c'=d'=0$, and $4a'+3b'+2c'+d'= t$.
		\item[($\mathcal F_{0,4}$):] Suppose that $c=0$, $a=1$, and $t \equiv 0 \mod 3$ or $t \equiv 1 \mod 3$. If $t \equiv 0 \mod 3$, then we set $a'=c'=d'=0$, and $b'=\frac{t}{3}$, hence $4a'+3b'+2c'+d'= t$. Note that $b'\leq b$, as if $b'>b$, then  $4a+3b+2c+d<4+3b'=t+4$, while $4a+3b+2c+d>t+\constant$. If $t \equiv 1 \mod 3$, then we use the greedy values for $a'$, $b'$, $c'$, and $d'$. Since $t>\constant$, we have $a'=\min\{a,\lfloor \frac{t}{4}\rfloor\}=1$. Further, $\min \{b,\lfloor \frac{t-4a'}{3}\rfloor\}=\lfloor \frac{t-4a'}{3}\rfloor$, as if $b<\lfloor \frac{t-4a'}{3}\rfloor$, then $4a+3b+2c+d<4+(t-4a')=t$, while $4a+3b+2c+d>t+\constant$. Since $t \equiv 1 \mod 3$, we have $t-4 \equiv 0 \mod 3$. Hence, $b'=\lfloor \frac{t-4a'}{3}\rfloor=\frac{t-4a'}{3}$, $c'=0$, $d'=0$, and $4a'+3b'+2c'+d'= t$.
		\item[($\mathcal F_{0,5}$):] Suppose that $c=0$, $b\geq 3$, and $a\geq 2$. We distinguish three cases, according to whether $t \equiv 0 \mod 3$, $t \equiv 1 \mod 3$, and $t \equiv 2 \mod 3$, similarly to case~($\mathcal F_{1,5}$).
		
		\begin{itemize}
			\item If $t\equiv 0 \mod 3$, then we set $a'=\max \{a'' : a''\leq \min \{a,\lfloor \frac{t}{4} \rfloor\} \wedge a''\equiv 0 \mod 3\}$; then $t-4a' \equiv 0 \mod 3$. By setting $b'=\frac{t-4a'}{3}$ and $c'=d'=0$, we get that $4a'+3b'+2c'+d'=t$. 
			
			It remains to prove that $b'\leq b$. Suppose, for a contradiction, that $b'>b$. If $\min \{a,\lfloor \frac{t}{4} \rfloor\}=a$, then we have $a'\geq a-2$. Further, we have $c'=c=0$ and $d'=d=0$. Hence $t=4a'+3b'+2c'+d'>4(a-2)+3b+2c+d$, that is, $4a+3b+2c+d<t+8$, which contradicts $4a+3b+2c+d>t+\constant$. Otherwise, we have $\min \{a,\lfloor \frac{t}{4} \rfloor\}=\lfloor \frac{t}{4} \rfloor$. Since $t\equiv 0 \mod 3$, we have either $t=12t'$, or $t=12t'+3$, or $t=12t'+6$, or $t=12t'+9$, for some non-negative integer $t'$. Hence, in all four cases we have $a'=3t'$, and thus $t-4a'$ is either $0$, or $3$, or $6$, or $9$. Since $b'=\frac{t-4a'}{3}$, we have $b'\leq 3\leq b$, a contradiction. 
			
			\item If $t\equiv 1 \mod 3$, then we set $a'=\max \{a'' : a''\leq \min \{a,\lfloor \frac{t}{4} \rfloor\} \wedge a''\equiv 1 \mod 3\}$. Note that $a'\geq 1$; indeed, $a\geq 1$ by assumption, and $\lfloor \frac{t}{4} \rfloor\geq 1$ given that $t>\constant$. Further, since $a'=3a^*+1$, for some integer $a^*$, we have $4a'=12a^*+4\equiv 1 \mod 3$. Hence, $t-4a' \equiv 0 \mod 3$. By setting $b'=\frac{t-4a'}{3}$ and $c'=d'=0$, we get that $4a'+3b'+2c'+d'=t$.
			
			It remains to prove that $b'\leq b$. Suppose, for a contradiction, that $b'>b$. If $\min \{a,\lfloor \frac{t}{4} \rfloor\}=a$, a contradiction is reached as in the case in which $t\equiv 0 \mod 3$. Otherwise, we have $\min \{a,\lfloor \frac{t}{4} \rfloor\}=\lfloor \frac{t}{4} \rfloor$. Since $t\equiv 1 \mod 3$ and $t>\constant$, we have either $t=12t'+4$, or $t=12t'+7$, or $t=12t'+10$, or $t=12t'+13$, for some non-negative integer $t'$. Hence, in all four cases we have $a'=3t'+1$, and thus $t-4a'$ is either $0$, or $3$, or $6$, or $9$. Since $b'=\frac{t-4a'}{3}$, we have $b'\leq 3\leq b$, a contradiction. 
			
			\item If $t\equiv 2 \mod 3$, then we set $a'=\max \{a'' : a''\leq \min \{a,\lfloor \frac{t}{4} \rfloor\} \wedge a''\equiv 2 \mod 3\}$. Note that $a'\geq 2$; indeed, $a\geq 2$ by assumption, and $\lfloor \frac{t}{4} \rfloor\geq 2$ given that $t>\constant$. Further, since $a'=3a^*+2$, for some integer $a^*$, we have $4a'=12a^*+8\equiv 2 \mod 3$. Hence, $t-4a' \equiv 0 \mod 3$. By setting $b'=\frac{t-4a'}{3}$ and $c'=d'=0$, we get that $4a'+3b'+2c'+d'=t$.
			
			It remains to prove that $b'\leq b$. Suppose, for a contradiction, that $b'>b$. If $\min \{a,\lfloor \frac{t}{4} \rfloor\}=a$, a contradiction is reached as in the case in which $t\equiv 0 \mod 3$. Otherwise, we have $\min \{a,\lfloor \frac{t}{4} \rfloor\}=\lfloor \frac{t}{4} \rfloor$. Since $t\equiv 2 \mod 3$ and $t>\constant$, we have either $t=12t'+8$, or $t=12t'+11$, or $t=12t'+14$, or $t=12t'+17$, for some non-negative integer $t'$. Hence, in all four cases we have $a'=3t'+2$, and thus $t-4a'$ is either $0$, or $3$, or $6$, or $9$. Since $b'=\frac{t-4a'}{3}$, we have $b'\leq 3\leq b$, a contradiction. 
		\end{itemize} 
		\item[($\mathcal F_{0,6}$):] If $b=0$, $c\geq 1$, and $t \equiv 0 \mod 2$, then we use the greedy values for $a'$, $b'$, $c'$, and $d'$. We prove that $c'=\min\{c,\lfloor \frac{t-4a'-3b'}{2} \rfloor\}=\lfloor \frac{t-4a'-3b'}{2}\rfloor=\frac{t-4a'}{2}$, where we exploited $b'=0$ and $t \equiv 0 \mod 2$. Note that $c'=\frac{t-4a'}{2}$ implies that $4a'+3b'+2c'+d'=t$. Suppose, for a contradiction, that $c<\frac{t-4a'}{2}$.
		
		If $a'=\min \{a,\lfloor \frac{t}{4} \rfloor\}=a$, then we get $4a+3b+2c+d< 4a' + 2 \frac{t-4a'}{2} = t$, which contradicts $4a+3b+2c+d> t+\constant$. 	If $a'=\min \{a,\lfloor \frac{t}{4} \rfloor\}=\lfloor \frac{t}{4} \rfloor$, then we get $t-4a'\in\{0,2\}$, where we exploited $t \equiv 0 \mod 2$, hence $c<\frac{t-4a'}{2}\leq 1$, a contradiction to $c\geq 1$.
		\item[($\mathcal F_{0,7}$):] If $c=1$, $a=0$, and $t \equiv 0 \mod 3$ or $t \equiv 2 \mod 3$, then we use the greedy values for $a'$, $b'$, $c'$, and $d'$. Since $a=0$, we have $a'=0$. Further, $b'=\min\{b,\lfloor \frac{t}{3}\rfloor\}=\lfloor \frac{t}{3}\rfloor$, as if $b<\lfloor \frac{t}{3}\rfloor$, then $4a+3b+2c+d<3\lfloor \frac{t}{3}\rfloor+2\leq t+2$, while $4a+3b+2c+d>t+\constant$. If $t \equiv 0 \mod 3$, then $t-3b'=t-3\lfloor \frac{t}{3}\rfloor=0$, hence $c'=d'=0$ and $4a'+3b'+2c'+d'= t$. If $t \equiv 2 \mod 3$, then  $t-3b'=t-3\lfloor \frac{t}{3}\rfloor=2$. Hence, $c'=1$, $d'=0$, and $4a'+3b'+2c'+d'= t$. 
		
		\item[($\mathcal F_{0,8}$):] Suppose that $a\geq 1$, that $b\geq 1$, and that $c\geq 1$. We first choose $b'$ so that $t-3b'$ is even. Namely, if $t\equiv 0 \mod 2$, then we set $b'=\max \{b'' : b''\leq \min \{b,\lfloor \frac{t}{3} \rfloor\} \wedge b''\equiv 0 \mod 2\}$, while if $t\equiv 1 \mod 2$, then we set $b'=\max \{b'' : b''\leq \min \{b,\lfloor \frac{t}{3} \rfloor\} \wedge b''\equiv 1 \mod 2\}$. In both cases we have $b'\geq 0$, given that $b\geq 1$ and $\lfloor \frac{t}{3} \rfloor\geq 1$ (since $t>\constant$). Further, $t-3b' \equiv 0 \mod 2$.	
		
		If $\min \{b,\lfloor \frac{t}{3} \rfloor\}=\lfloor \frac{t}{3} \rfloor$, then $t-3b'\in \{0,2,4\}$. Namely, if $t\equiv 0 \mod 2$, then we have $t=6t'$, or $t=6t'+2$, or $t=6t'+4$, for some non-negative integer $t'$; in all three cases, we have $b'=2t'$ and thus $t-3b'$ is either $0$, or $2$, or $4$. Analogously, if $t\equiv 1 \mod 2$, then we have $t=6t'+3$, or $t=6t'+5$, or $t=6t'+7$, for some non-negative integer $t'$; in all three cases, we have $b'=2t'+1$ and thus $t-3b'$ is either $0$, or $2$, or $4$. Now, if $t-3b'$ is $0$, or $2$, or $4$, it suffices to choose $a'=c'=d'=0$, or $a'=d'=0$ and $c'=1$, or $c'=d'=0$ and $a'=1$, respectively, in order to get $4a'+3b'+2c'+d'= t$. Note that $a'\leq 1\leq a$ and $c'\leq 1\leq c$ in all three cases.
		
		Assume next that $b<\lfloor \frac{t}{3} \rfloor$, which implies that $b\leq b'+1$. We now set $a'=\min\{a,\lfloor\frac{t-3b'}{4}\rfloor\}$. If $a'=\min\{a,\lfloor\frac{t-3b'}{4}\rfloor\}=\lfloor\frac{t-3b'}{4}\rfloor$, then $t-3b'-4a' \in \{0,2\}$, given that $t-3b' \equiv 0 \mod 2$. Hence, if $t-3b'-4a'=0$, we set $c'=d'=0$ and get that $4a'+3b'+2c'+d'= t$, while if  $t-3b'-4a'=2$, we set $c'=1$, $d'=0$, and get that $4a'+3b'+2c'+d'= t$. In both cases, we have $c'\leq 1\leq c$.  
		
		Assume next that $a<\lfloor\frac{t-3b'}{4}\rfloor$, which implies that $a=a'$. We set $c'=\frac{t-3b'-4a'}{2}$, which implies that $4a'+3b'+2c'+d'= t$. It remains to prove that $c'\leq c$. Suppose, for a contradiction, that $c<\frac{t-3b'-4a'}{2}$. Thus, we get $4a+3b+2c+d<4a'+3(b'+1)+2(\frac{t-3b'-4a'}{2})=t+3$, which contradicts  $4a+3b+2c+d>t+\constant$.
		
		\item[($\mathcal F_{0,9}$):] Suppose that $c\geq 2$ and $b\geq 1$. We can also assume $a=0$, as otherwise Case~($\mathcal F_{0,8}$) applies. We first choose $b'$ as in Case~($\mathcal F_{0,8}$), ensuring that $t-3b' \equiv 0 \mod 2$. 
		
		If $\min \{b,\lfloor \frac{t}{3} \rfloor\}=\lfloor \frac{t}{3} \rfloor$, then $t-3b'\in \{0,2,4\}$; this can be proved as in Case~($\mathcal F_{0,8}$). Now, if $t-3b'$ is $0$, or $2$, or $4$, it suffices to choose $c'=0$, $c'=1$, or $c'=2$, respectively, together with $a'=d'=0$, in order to get $4a'+3b'+2c'+d'= t$. Note that $c'\leq 2\leq c$.
		
		If $b<\lfloor \frac{t}{3} \rfloor$, we can set $a'=d'=0$ and $c'=\frac{t-3b'}{2}$, which implies that $4a'+3b'+2c'+d'= t$. The proof that $c'\leq c$ is the same as in Case~($\mathcal F_{0,8}$). 
	\end{itemize}
	
	This concludes the proof of the sufficiency. Since $\mathcal F$ has $O(1)$ size, it can be tested in $O(1)$ time whether given values $a$, $b$, $c$, $d$, and $t$ satisfy it. Further, the above proof of sufficiency is constructive and allows one to find values $a'$, $b'$, $c'$, and $d'$ satisfying the requirements of the lemma in $O(1)$ time. 
\end{proof}

Lemmata~\ref{le:extensible-values}--\ref{le:medium-t} imply the following.  

\begin{lemma} \label{le:find-values-variable}
	Let $\mu_1,\nu_1,\dots,\mu_k,\nu_k,\rho_0,\rho_k$ be a promising sequence for $(G,\mu,\nu)$. Further, let $a$ and $b$ denote the number of $4$- and $3$-components of $G$, respectively, let $c$ denote the number of $2$-components of $G$ plus the number of vertices $u_i \in \chi$ with $i\in \{1,\dots,k-1\}$, let $d=\sum_{i=1}^{k-1} (2-\frac{\nu_i+\mu_{i+1}}{90\degree})$, and let $t=(k-1)-\frac{\rho_0+\rho_k}{90\degree}$. Assume that the values $\mu_1,\nu_1,\dots,\mu_k,\nu_k,\rho_0,\rho_k,a,b,c,d,t$ are known.
	
	Then it is possible to determine in $O(1)$ time whether the promising sequence for $(G,\mu,\nu)$ is extensible. In the positive case, it is possible to determine in $O(k)$ time an in-out assignment $\mathcal A$ and values $\rho_1,\dots,\rho_{k-1}$ that, together with the promising sequence $\mu_1,\nu_1,\dots,\mu_k,\nu_k,\rho_0,\rho_k$, satisfy Properties~$\mathcal V_1$--$\mathcal V_5$ of Lemma~\ref{le:structural-variable-embedding}. 
\end{lemma}

\begin{proof}
	By Lemma~\ref{le:extensible-values}, the promising sequence for $(G,\mu,\nu)$ is extensible if and only if there exist integer values $0\leq a'\leq a$, $0\leq b'\leq b$, $0\leq c'\leq c$, and $0\leq d'\leq d$ such that $4a'+3b'+2c'+d'=t$. By Lemmata~\ref{le:small-large-t} and~\ref{le:medium-t}, it is possible to test in $O(1)$ time whether such values $a'$, $b'$, $c'$, and $d'$ exist. Again by Lemmata~\ref{le:small-large-t} and~\ref{le:medium-t}, if such values exist, then they can be determined in $O(1)$ time; then, by Lemma~\ref{le:extensible-values}, it is possible to determine in $O(k)$ time an in-out assignment $\mathcal A$ and values $\rho_1,\dots,\rho_{k-1}$ that, together with the promising sequence $\mu_1,\nu_1,\dots,\mu_k,\nu_k,\rho_0,\rho_k$ for $(G,\mu,\nu)$, satisfy Properties~$\mathcal V_1$--$\mathcal V_5$ of Lemma~\ref{le:structural-variable-embedding}. 
\end{proof}

We now have the following.

\begin{theorem} \label{th:2-con-variable-edge}
	Let $G$ be an $n$-vertex $2$-connected outerplanar graph, let $uv$ be an edge incident to the outer face $f^*_{\mathcal O}$ of the outerplane embedding of $G$, and let $\chi$ be a subset of the degree-$2$ vertices of $G$. There is an $O(n)$-time algorithm which tests, for any values $\mu,\nu \in \{90\degree,180\degree,270\degree\}$, whether $G$ admits a $(\chi,\mu,\nu)$-representation; further, in the positive case, the algorithm constructs such a representation in $O(n)$ time.
\end{theorem}

\begin{proof}
	First, we mark the vertices in $\chi$, so that given a vertex of $G$, we can check in $O(1)$ time whether it is in $\chi$ or not. Next, we compute the outerplane embedding $\mathcal O$ of $G$ in $O(n)$ time~\cite{cnao-lta-85,d-iroga-07,ht-ept-74,m-laarogmog-79,w-rolt-87}. We now have, for each vertex $v$ of $G$, a circular list $L_{\mathcal O}(v)$, which represents the clockwise order of the edges incident to $v$ in $\mathcal O$. We also compute, for each face $f$ of $\mathcal O$, the clockwise order of the edges along the boundary of $f$; these sets can be recovered from the lists $L_{\mathcal O}(v)$ in $O(n)$ time. We next compute in $O(n)$ time the extended dual tree $\mathcal T$ of~$\mathcal O$. Each internal node $s\in \mathcal T$ is associated with the cycle $\mathcal C_s$ delimiting the internal face of $\mathcal O$ corresponding to $s$. We root $G$ at $uv$ and we root $\mathcal T$ at the leaf $r^*$ such that the edge of $\mathcal T$ incident to $r^*$ is dual to $uv$. Let $rr^*$ be such an edge and note that $r$ is the only child of $r^*$ in $\mathcal T$.
	
	For any non-leaf node $s$ of $\mathcal T$, let $v^s_0,v^s_1,\dots,v^s_{k_s}$ be the clockwise order of the vertices along $\mathcal C_s$ in $\cal O$; further, we denote by $\mathcal T_s$ the subtree of $\mathcal T$ rooted at $s$ and by $G_s$ the subgraph of $G$ defined as $\bigcup_{t\in \mathcal T_s}\mathcal C_t$. Note that $G_{r}=G$. For a leaf node $s\neq r^*$ of $\mathcal T$, with a slight overload of notation we denote by $G_s$ the edge of $G$ corresponding to $s$. Consider any internal node $s$ of $\mathcal T$. Then the graph $G_s$ is rooted at the edge $u_s v_s$ that is shared by $\mathcal C_s$ and $\mathcal C_p$, where $p$ is the parent of $s$ in $\mathcal T$; then we assume, w.l.o.g., that $u_s=v^s_0$ and $v_s=v^s_{k_s}$. Note that, for any node $s\neq r^*$ of $\mathcal T$ and in any rectilinear representation $(\mathcal E,\phi)$ of $G$ in which the edge $uv$ is incident to the outer face of $\mathcal E$, the edge $u_s v_s$ is incident to the outer face of the restriction of $(\mathcal E,\phi)$ to $G_s$. Hence, if $s_1, \dots, s_{k_s}$ denote the children of $s$, then the graphs $G_{s_1},\dots,G_{s_{k_s}}$ are the $u_sv_s$-subgraphs of $G_s$.
	
	We now perform a bottom-up visit of $\mathcal T$ which ends after visiting the child $r$ of the root $r^*$. For each leaf $s\neq r^*$ of $\mathcal T$ with parent $p$, we associate a set $\mathcal N_{s\rightarrow p}=\{(0\degree,0\degree)\}$ to the edge $sp$. When processing an internal node $s$ of $\mathcal T$ with parent $p$, we compute a set $\mathcal N_{s\rightarrow p}$ which contains all the pairs $(\mu_s,\nu_s)$ with $\mu_s,\nu_s\in \{90\degree,180\degree,270\degree\}$ such that $G_s$ admits a $(\chi_s,\mu_s,\nu_s)$-representation. Note that $\mathcal N_{s\rightarrow p}$ contains at most $9$ such pairs. This can be done in $O(k_s)$ time as follows.
	
	We independently consider each of the $9$ pairs $(\mu_s,\nu_s)$ with $\mu_s,\nu_s\in \{90\degree,180\degree,270\degree\}$. Since we already visited the children of $s$ in $\mathcal T$, for every $u_s v_s$-subgraph $G_{s_i}$ of $G_s$, we have already computed the set $\mathcal N_{s_i\rightarrow s}$; in particular, $G_{s_i}$ is trivial if and only if $\mathcal N_{s_i\rightarrow s}=\{(0\degree,0\degree)\}$. By repeated applications of Lemma~\ref{le:fix-values}, we can construct the optimal sequence for $G_s$ in total $O(k_s)$ time. Then, by Lemma~\ref{le:promising-sequences}, we can construct the promising sequences for $(G_s,\mu_s,\nu_s)$ in total $O(k_s)$ time (if there is no promising sequence for $(G_s,\mu_s,\nu_s)$, by Lemma~\ref{le:extensible-solution} we can conclude that $G_s$ admits no $(\chi_s,\mu_s,\nu_s)$-representation). 
	
	We independently consider each promising sequence $\mu^s_1,\nu^s_1,\dots,\mu^s_{k_s},\nu^s_{k_s},\rho^s_0,\rho^s_{k_s}$ for $(G_s,\mu_s,\nu_s)$; recall that there are $O(1)$ such sequences. We compute in $O(k_s)$ time the values  $a_s$, $b_s$, $c_s$, $d_s$, and $t_s$, where $a_s$ and $b_s$ denote the number of $4$- and $3$-components of $G_s$, respectively, $c_s$ denotes the number of $2$-components of $G_s$ plus the number of vertices $v^s_i \in \chi$ with $i\in \{1,\dots,k_s-1\}$, $d_s=\sum_{i=1}^{k_s-1} (2-\frac{\nu^s_i+\mu^s_{i+1}}{90\degree})$, and $t_s=(k_s-1)-\frac{\rho^s_0+\rho^s_k}{90\degree}$. Note that the number of vertices $v^s_i \in \chi$ with $i\in \{1,\dots,k_s-1\}$ can be computed in $O(k_s)$ time as each of the vertices $v^s_i$ is marked if it belongs to $\chi$. By Lemma~\ref{le:find-values-variable}, it is possible to determine in $O(1)$ time whether the promising sequence $\mu^s_1,\nu^s_1,\dots,\mu^s_{k_s},\nu^s_{k_s},\rho^s_0,\rho^s_{k_s}$ for $(G_s,\mu_s,\nu_s)$ is extensible, i.e., whether there exist an in-out assignment $\mathcal A_s$ and values $\rho^s_1,\rho^s_2,\dots,\rho^s_{k_s-1}$ that, together with the promising sequence $\mu^s_1,\nu^s_1,\dots,\mu^s_{k_s},\nu^s_{k_s},\rho^s_0,\rho^s_{k_s}$, satisfy Properties~$\mathcal V1$--$\mathcal V5$ of Lemma~\ref{le:structural-variable-embedding}. By Lemma~\ref{le:extensible-solution}, we can conclude that $G_s$ admits a  $(\chi_s,\mu_s,\nu_s)$-representation if and only if at least one of the promising sequences for $(G_s,\mu_s,\nu_s)$ is extensible; in the positive case we add the pair $(\mu_s,\nu_s)$ to $\mathcal N_{s\rightarrow p}$. Observe that the set $\chi_s$ does not need to be explicitly computed when processing $s$. 
	
	If there is no pair $(\mu_s,\nu_s)$ with $\mu_s,\nu_s\in \{90\degree,180\degree,270\degree\}$ such that $G_s$ admits a $(\chi_s,\mu_s,\nu_s)$-representation, that is, if $\mathcal N_{s\rightarrow p}=\emptyset$, then we conclude that $G$ has no $(\chi,\mu,\nu)$-representation, otherwise we continue the visit of $\mathcal T$. After processing the child $r$ of the root $r^*$ of $\mathcal T$, we either conclude that $G$ has a $(\chi,\mu,\nu)$-representation or not, depending on whether $(\mu,\nu)\in \mathcal N_{r\rightarrow r^*}$ or $(\mu,\nu)\notin \mathcal N_{r\rightarrow r^*}$, respectively.  
	
	The algorithm processes a node $s$ of $\mathcal T$ in $O(k_s)$ time. Hence, the algorithm takes $O(n)$ time over the entire tree $\mathcal T$. 
	
	Finally, we describe how to modify the described algorithm so that it constructs in $O(n)$ time a $(\chi,\mu,\nu)$-representation $(\cal E,\phi)$ of $G$. When we process an internal node $s$ of $\mathcal T$ in the bottom-up visit of $\mathcal T$, for each of the at most $9$ pairs $(\mu_s,\nu_s)$ with $\mu_s,\nu_s\in \{90\degree,180\degree,270\degree\}$ such that $G_s$ admits a $(\chi_s,\mu_s,\nu_s)$-representation, we store the $O(k_s)$ values $\rho^s_0,\rho^s_1,\mu^s_1,\nu^s_1,\rho^s_2,\mu^s_2,\nu^s_2,\dots,\rho^s_{k_s},\mu^s_{k_s},\nu^s_{k_s}$ and the in-out assignment $\mathcal A_s$ satisfying Properties~$\mathcal V1$--$\mathcal V5$ of Lemma~\ref{le:structural-variable-embedding}; by Lemmata~\ref{le:fix-values},~\ref{le:promising-sequences}, and~\ref{le:find-values-variable}, these can be found in $O(k_s)$ time. The in-out assignment $\mathcal A_s$ can be represented by storing, for each child $s_i$ of $s$ in $\mathcal T$ such that $G_{s_i}$ is a non-trivial subgraph of $G_s$, a boolean value $\omega_{s_i\rightarrow s}$ which is {\sc true} if $G_{s_i}$ is assigned to the outside of $\mathcal C_s$, and {\sc false}~otherwise.
	The bottom-up visit of $\mathcal T$ ends once it reaches $r$. Then we perform two top-down visits of $\mathcal T$, each starting at the child $r$ of $r^*$. 
	
	The first top-down visit results in the construction of a plane embedding $\mathcal E$ of $G$ such that $G$ admits a $(\chi,\mu,\nu)$-representation $(\mathcal E,\phi)$, for some function $\phi$ to be determined in the second top-down visit. When we start processing a node $s$ of $\mathcal T$ in the first top-down visit, we assume that (see Figure~\ref{fig:variable-augment}(a)):
	
	\begin{enumerate} 
		\item a pair $(\mu_s,\nu_s)$ has been associated to $s$ such that $G_s$ admits a $(\chi_s,\mu_s,\nu_s)$-representation; this representation will coincide with the restriction of $(\mathcal E,\phi)$ to $G_s$;
		\item the cycle $\mathcal C_s$ has already been embedded in the plane, i.e., it has already been established whether $u_s$ immediately precedes or follows $v_s$ in clockwise direction along $\mathcal C_s$; and
		\item for $i=0,\dots,k_s$, a, possibly partial, circular list $L_{\mathcal E}(v^s_i)$ of the edges incident to $v^s_i$ has already been fixed; this contains the two edges of $\mathcal C_s$ incident to $v^s_i$ and the edges of $G$ that are incident to $v^s_i$ and that are not in $G_s$, if any. The complete circular list $L_{\mathcal E}(v^s_i)$, representing the clockwise order of the edges incident to $v^s_i$ in $\mathcal E$, will be an extension of this partial one.
	\end{enumerate} 
	
	\begin{figure}[htb]\tabcolsep=4pt
		\centering
		\begin{tabular}{c c}
			\includegraphics[scale=0.7]{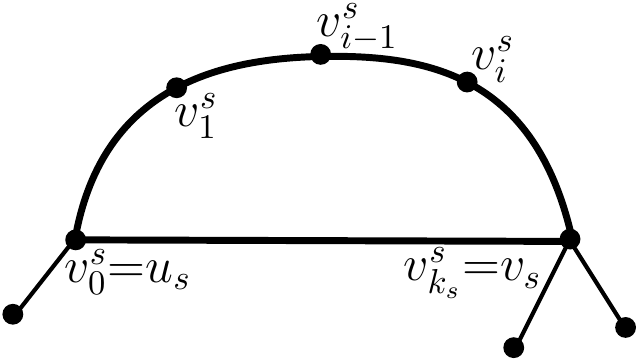} \hspace{3mm} &
			\includegraphics[scale=0.7]{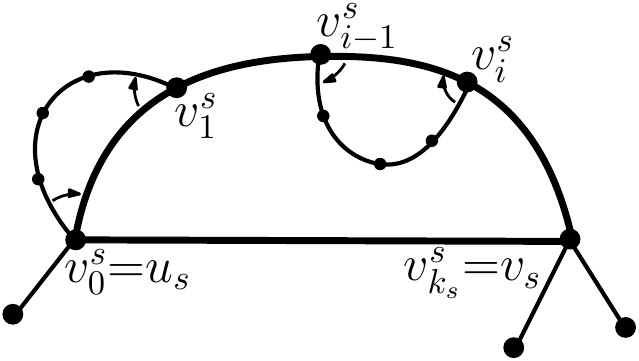} \\
			(a) \hspace{3mm} & (b)
		\end{tabular}
		\caption{(a) Before processing a node $s$ of $\mathcal T$, the cycle $\mathcal C_s$ (drawn with thick lines) is already embedded. The (partial) clockwise order of the edges incident to each vertex of $\mathcal C_s$ has been fixed. (b) Embedding the cycles $\mathcal C_{s_i}$ inside and outside $\mathcal C_s$, according to the in-out assignment $\mathcal A_s$.}
		\label{fig:variable-augment}
	\end{figure}

	Before starting the top-down visit, we perform the following initialization. We associate $(\mu_r=\mu,\nu_r=\nu)$ to $r$; this is the pair that is given in input to the overall algorithm. Further, $u$ is arbitrarily chosen to immediately precede or follow $v$ in clockwise direction along $\mathcal C_{r}$ (the two choices lead to two plane embeddings which are one the reflection of the other). Finally, for $i=0,\dots,k_r$, the list $L_{\mathcal E}(v^r_i)$ is initialized with the two edges of $\mathcal C_r$ incident to $v^r_i$; since $L_{\mathcal E}(v^r_i)$ is a circular list, the order of such edges in $L_{\mathcal E}(v^r_i)$ is unique.
	
	When the top-down visit considers an internal node $s$ of $\mathcal T$, we do the following (see Figure~\ref{fig:variable-augment}(b)). Assume that $u_s=v^s_0$ immediately follows $v_s=v^s_{k_s}$ in clockwise direction along $\mathcal C_s$, the other case is symmetric. 
	
	\begin{enumerate}
		\item For each $i=1,\dots,k_s$ such that $G_{s_i}$ is non-trivial, we associate to $s_i$ the pair $(\mu^s_i,\nu^s_i)$, where the values $\mu^s_i$ and $\nu^s_i$ are the ones in the sequence $\rho^s_0,\rho^s_1,\mu^s_1,\nu^s_1,\rho^s_2,\mu^s_2,\nu^s_2,\dots,\rho^s_{k_s},\mu^s_{k_s},\nu^s_{k_s}$ that has been stored for the $(\mu_s,\nu_s)$-representation of $G_s$; note that $(\mu^s_i,\nu^s_i)\in \mathcal N_{s_i\rightarrow s}$. 
		\item For each $i=1,\dots,k_s$ such that $G_{s_i}$ is non-trivial, let $\omega_{s_i\rightarrow s}$ be the value that represents whether $G_{s_i}$ is assigned to the outside or to the inside of $\mathcal C_s$ by the in-out assignment $\mathcal A_s$ that is stored for the $(\mu_s,\nu_s)$-representation of $G_s$. If $\omega_{s_i\rightarrow s}=$ {\sc true}, then we establish that $v^s_{i-1}=u_{s_i}$ immediately follows $v^s_i=v_{s_i}$ in clockwise direction along $\mathcal C_{s_i}$, otherwise we establish that $v^s_{i-1}=u_{s_i}$ immediately precedes $v^s_i=v_{s_i}$ in clockwise direction along $\mathcal C_{s_i}$.
		\item For each $i=1,\dots,k_s$ such that $G_{s_i}$ is non-trivial, let $\omega_{s_i\rightarrow s}$ be as in the previous item. If $\omega_{s_i\rightarrow s}=$ {\sc true}, we insert the edge of $\mathcal C_{s_i}$ incident to $v^s_{i-1}$ and different from $v^s_{i-1}v^s_{i}$ into $L_{\mathcal E}(v^s_{i-1})$, so that it immediately precedes the edge $v^s_{i-1}v^s_{i}$, and we insert the edge of $\mathcal C_{s_i}$ incident to $v^s_i$ and different from $v^s_{i-1}v^s_{i}$ into $L_{\mathcal E}(v^s_i)$, so that it immediately follows the edge $v^s_{i-1}v^s_{i}$. If $\omega_{s_i\rightarrow s}=$ {\sc false}, we insert the edge of $\mathcal C_{s_i}$ incident to $v^s_{i-1}$ and different from $v^s_{i-1}v^s_{i}$ into $L_{\mathcal E}(v^s_{i-1})$, so that it immediately follows the edge $v^s_{i-1}v^s_{i}$, and we insert the edge of $\mathcal C_{s_i}$ incident to $v^s_i$ and different from $v^s_{i-1}v^s_{i}$ into $L_{\mathcal E}(v^s_i)$, so that it immediately precedes the edge $v^s_{i-1}v^s_{i}$.
	\end{enumerate}
	
	Clearly, the above operations can be performed in $O(k_s)$ time, hence the entire top-down visit takes $O(n)$ time. The completion of the first top-down visit results in the definition of a (complete) circular list $L_{\mathcal E}(w)$ for each vertex $w$ of $G$, and hence in the definition of a plane embedding $\mathcal E$ of $G$.
	
	The second top-down visit defines a function $\phi$ such that $(\mathcal E,\phi)$ is a $(\chi,\mu,\nu)$-representation of $G$. This is done as in Theorem~\ref{th:fixed}. Let $f_s$ be the face of $\mathcal E$ incident to $u_sv_s$ and lying in the interior of the cycle $\mathcal C_s$. When we visit a node $s$ of $\mathcal T$, for $i=0,\dots,k_s$, we set $\phi(v^s_i,f_s)=\rho^s_i$, where $\rho^s_i$ is the value in the sequence $\rho^s_0,\rho^s_1,\mu^s_1,\nu^s_1,\rho^s_2,\mu^s_2,\nu^s_2,\dots,\rho^s_{k_s},\mu^s_{k_s},\nu^s_{k_s}$ that has been stored for the $(\mu_s,\nu_s)$-representation of $G_s$. After the completion of the second top-down visit of $\mathcal T$, for each vertex $w$, the value $\phi(w,f)$ has been determined for all the faces $f$ of $\cal E$ incident to $w$, except for one face $f_w$; this is the unique face of $\cal E$ incident to $w$ whose corresponding face in ${\mathcal E_s}$ is $f^*_{\mathcal E_s}$, where $s$ is the first node that is encountered in the top-down visit of $\mathcal T$ such that $w$ is a vertex of $\mathcal C_s$. We complete the rectilinear representation of $G$ by setting, for each vertex $w$ of $G$, $\phi(w,f_w)=360\degree-\sum_f \phi(w,f)$, where the sum is over all the faces $f\neq f_w$ of $\mathcal E$ incident to $w$. 
\end{proof}

By independently considering all the pairs $(\mu,\nu)$ with $\mu,\nu\in \{90\degree,180\degree,270\degree\}$, Theorem~\ref{th:2-con-variable-edge} also allows us to test whether a $\chi$-constrained representation of $G$ exists such that $uv$ is incident to the outer face. 

\subsection{$2$-Connected Outerplanar Graphs} \label{sse:var-2con}

We now get rid of the assumption that there is a prescribed edge $uv$ incident to the outer face of the rectilinear representation we seek, while maintaining the assumption that the input $n$-vertex outerplanar graph $G$ is $2$-connected. We are again required to look for $\chi$-constrained representations; that is, given a set $\chi$ of degree-$2$ vertices of $G$, we want to test whether there exists a rectilinear representation $(\mathcal E,\phi)$ of $G$ such that, for any $w\in \chi$ and for any face $f$ incident to $w$, we have either $\phi(w,f)=90\degree$ or $\phi(w,f)=270\degree$. 

Our $O(n)$-time algorithm to solve this problem will actually perform a more general task. Namely, our algorithm will label every vertex $u$ of $G$ whose degree is not larger than $3$ with a set $\gamma(u)$ which contains all the values $\mu\in \{90\degree,180\degree,270\degree\}$ such that $G$ admits a $\chi$-constrained representation $(\mathcal E,\phi)$ in which $u$ is incident to the outer face of $\mathcal E$ and the sum of the internal angles at $u$ is equal to $\mu$, i.e., $\phi^{\mathrm{int}}(u)=\mu$. 

While the condition on the degree of the labeled vertices will naturally find an explanation in Section~\ref{sse:var}, where we will deal with not necessarily $2$-connected outerplanar graphs, we here show that it does not introduce any loss of generality, even if we just look at $2$-connected outerplanar graphs.

\begin{lemma} \label{le:low-degree-external}
	Let $(\mathcal E,\phi)$ be a rectilinear representation of a $2$-connected outerplanar graph. Then there is a vertex $u$ such that: (i) the degree of $u$ is $2$, and (ii) $u$ is incident to the outer face $f^*_{\mathcal E}$ of $\mathcal E$.
\end{lemma}

\begin{proof}
	Let $\mathcal C=(u_1,\dots,u_k)$ be the cycle delimiting $f^*_{\mathcal E}$. Let $P$ be the polygon representing $\mathcal C$ in any planar rectilinear drawing corresponding to $(\mathcal E,\phi)$. Suppose, for a contradiction, that the degree of $u_i$ is larger than or equal to $3$, for $i=1,\dots,k$. Then $\phi^{\mathrm{int}}(u_i)\geq 180\degree$, for $i=1,\dots,k$; hence, the sum of the angles in the interior of $P$ is larger than or equal to $180\degree \cdot k$, which is larger than $180\degree \cdot (k-2)$, a contradiction.
\end{proof}

We first show that the problem of computing the labels $\gamma(u)$ for each vertex $u$ of $G$ whose degree is not larger than $3$ can be reduced to the problem of computing labels for the edges of $G$. More precisely, for every edge $uv$ of $G$ incident to $f^*_{\mathcal O}$, we will compute a set $\mathcal M_{uv}$ of pairs: A pair $(\mu,\nu)$ with $\mu,\nu\in\{90\degree,180\degree,270\degree\}$ belongs to $\mathcal M_{uv}$ if and only if there is a $(\chi,\mu,\nu)$-representation of $G$, i.e., a $\chi$-constrained representation $(\mathcal E,\phi)$ such that $uv$ is incident to $f^*_{\mathcal E}$, $\phi^{\mathrm{int}}(u)=\mu$, and $\phi^{\mathrm{int}}(v)=\nu$. We show that the vertex labels can be computed efficiently from the edge labels.

\begin{lemma}\label{le:edge-vertex-labels}
	Suppose that, for every edge $uv$ of $G$ incident to $f^*_{\mathcal O}$, the set $\mathcal M_{uv}$ is known. Then it is possible to compute in $O(1)$ time the set $\gamma(u)$ for each vertex $u$ of $G$ whose degree is not larger than $3$. 
\end{lemma}	 

\begin{proof}
	Consider any vertex $u$ of $G$ whose degree is not larger than $3$. We show how to construct the set $\gamma(u)$. We initialize $\gamma(u)=\emptyset$. Then we consider every edge $uv$ incident to $u$ and to $f^*_{\mathcal O}$; for every pair $(\mu,\nu)$ in $\mathcal M_{uv}$, we update $\gamma(u)$ to $\gamma(u) \cup \{\mu\}$. Since $u$ has $O(1)$ incident edges and the set $\mathcal M_{uv}$ contains $O(1)$ pairs, this computation takes $O(1)$ time.  
	
	We now show that, for each vertex $u$ of $G$ whose degree is not larger than $3$, the set $\gamma(u)$ contains a value $\mu$ if and only if $G$ admits a $\chi$-constrained representation $(\mathcal E,\phi)$ in which $u$ is incident to $f^*_{\mathcal E}$ and $\phi^{\mathrm{int}}(u)=\mu$. 
	
	In order to prove the necessity, suppose that $\gamma(u)$ contains a value $\mu$. By construction, there is an edge $uv$ of $G$ incident to $u$ and to $f^*_{\mathcal O}$ and there is a pair $(\mu,\nu)$ such that $(\mu,\nu)\in \mathcal M_{uv}$. Hence, by definition of the set $\mathcal M_{uv}$, there is a $(\chi,\mu,\nu)$-representation of $G$; this is a $\chi$-constrained representation $(\mathcal E,\phi)$ in which $u$ is incident to $f^*_{\mathcal E}$ and $\phi^{\mathrm{int}}(u)=\mu$. 
	
	In order to prove the sufficiency, consider any $\chi$-constrained representation $(\mathcal E,\phi)$ of $G$ in which $u$ is incident to $f^*_{\mathcal E}$ and $\phi^{\mathrm{int}}(u)=\mu$. Note that two edges incident to $u$ are also incident to $f^*_{\mathcal E}$. Further, two edges incident to $u$ are also incident to $f^*_{\mathcal O}$. Since $u$ has degree smaller than or equal to $3$, it follows that there is an edge incident $uv$ that is incident both to $f^*_{\mathcal E}$ and to $f^*_{\mathcal O}$. Hence, by definition of the set $\mathcal M_{uv}$, there is a pair $(\mu,\nu)$ in $\mathcal M_{uv}$, where $\nu=\phi^{\mathrm{int}}(v)$, and thus, by construction, $\mu \in \gamma(u)$. 
\end{proof}


Let $u_1v_1,\dots,u_nv_n$ be the edges of $G$ incident to $f^*_{\mathcal O}$, in any order (each vertex of $G$ has two different labels in the previous sequence of edges). The sets $\mathcal M_{u_hv_h}$ will be computed by means of $n$ postorder traversals of the extended dual tree $\mathcal T$ of $\mathcal O$; during the $h$-th traversal, $\mathcal T$ is rooted at the leaf $r^*_h$ such that the edge $r^*_hr_h$ incident to $r^*_h$ is dual to $u_hv_h$. We introduce some definitions needed to describe these traversals.

\begin{figure}[htb]\tabcolsep=4pt
	\centering
	\begin{tabular}{c c c}
		\includegraphics[scale=0.7]{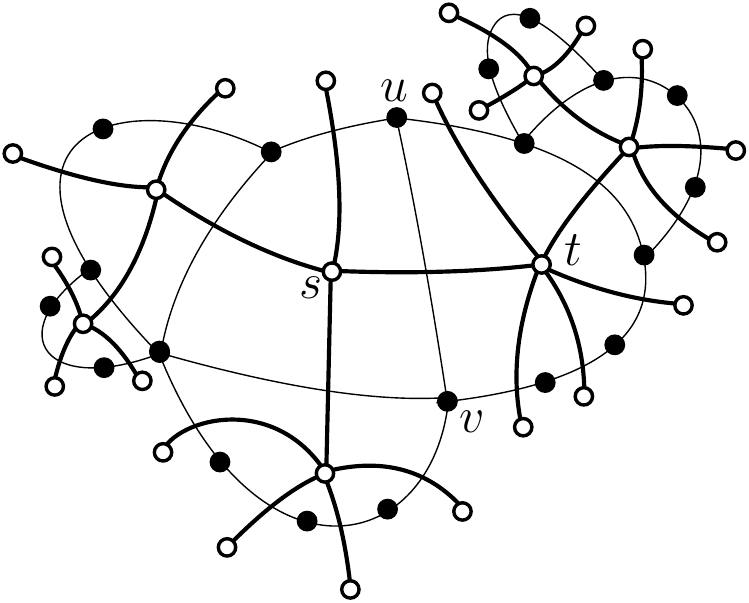} \hspace{3mm} &
		\includegraphics[scale=0.7]{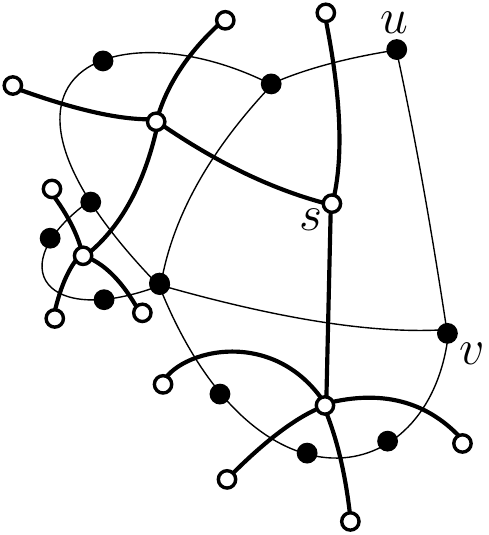} \hspace{3mm} &
		\includegraphics[scale=0.7]{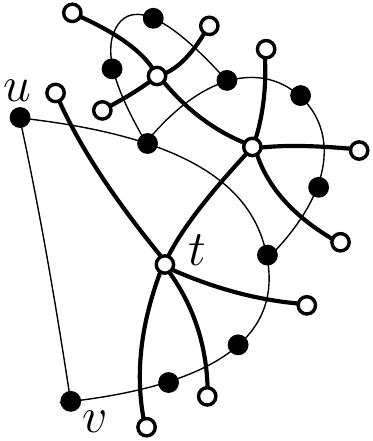} \\
		(a) \hspace{3mm} & (b) \hspace{3mm} & (c)
	\end{tabular}
	\caption{(a) The graph $G$ (represented with thin lines and black disks) and the extended dual tree $\mathcal T$ (represented with thick lines and white disks) of the outerplane embedding $\mathcal O$ of $G$. (b) The graph $G_{s\rightarrow t}$ and the tree $\mathcal T_{s\rightarrow t}$. (c) The graph $G_{t\rightarrow s}$ and the tree $\mathcal T_{t\rightarrow s}$.}
	\label{fig:outerplanar-split}
\end{figure}

Refer to Figure~\ref{fig:outerplanar-split}. Recall that each internal node $s$ of $\mathcal T$ is dual to a face of $\mathcal O$ delimited by a cycle $\mathcal C_s$. Consider any edge $st$ of $\mathcal T$. The removal of the edge $st$ splits $\mathcal T$ into two trees. Let $\mathcal T_{s\rightarrow t}$ be the one containing $s$ and let $\mathcal T_{t\rightarrow s}$ be the one containing $t$. If $s$ (resp.\ $t$) is an internal node of $\mathcal T$, then we let $G_{s\rightarrow t}$ (resp.\ $G_{t\rightarrow s}$) be the subgraph of $G$ defined as $\bigcup_{x\in \mathcal T_{s\rightarrow t}} {\mathcal C}_x$ (resp.\ $\bigcup_{x\in \mathcal T_{t-s}} {\mathcal C}_x$). If $s$ (resp.\ $t$) is a leaf of $\mathcal T$, then we let $G_{s\rightarrow t}$ (resp.\ $G_{t\rightarrow s}$) be the edge of $G$ dual to $st$. Regardless of whether $s$ and $t$ are internal nodes or not, observe that $G_{s\rightarrow t}$ and $G_{t\rightarrow s}$ share the edge $uv$ of $G$ dual to $st$, and no other edge; both $G_{s\rightarrow t}$ and $G_{t\rightarrow s}$ are then rooted at $uv$. Let $\chi_{s\rightarrow t}$ and $\chi_{t\rightarrow s}$ be the restrictions of $\chi$ to the vertices of $G_{s\rightarrow t}$ and $G_{t\rightarrow s}$, respectively.

\newcommand{\start}[1]{{\sc start}$(#1)$}
\newcommand{\terminate}[1]{{\sc end}$(#1)$}
\newcommand{\nullo}{{\sc null}}
\newcommand{\processed}[1]{{\sc processed}$(#1)$}

The traversals of $\mathcal T$ are going to  equip each edge $st$ of $\mathcal T$ with two sets $\mathcal N_{s\rightarrow t}$ and $\mathcal N_{t\rightarrow s}$. These sets are defined as follows. Let $uv$ be the edge of $G$ dual to $st$. 

\begin{itemize}
	\item If $s$ is an internal node of $\mathcal T$, then $\mathcal N_{s\rightarrow t}$ contains all the pairs $(\mu,\nu)$ with $\mu,\nu\in\{90\degree,180\degree,270\degree\}$ such that $G_{s\rightarrow t}$ admits a $(\chi_{s\rightarrow t},\mu,\nu)$-representation. Similarly, if $t$ is an internal node of $\mathcal T$, then the set $\mathcal N_{t\rightarrow s}$ stores all the pairs $(\mu,\nu)$ with $\mu,\nu\in\{90\degree,180\degree,270\degree\}$ such that $G_{t\rightarrow s}$ admits a $(\chi_{t\rightarrow s},\mu,\nu)$-representation. 
	\item If $s$ is a leaf, $uv$ is incident to $f^*_{\mathcal O}$; then $\mathcal N_{s\rightarrow t}=\{(0\degree,0\degree)\}$ and $\mathcal N_{t\rightarrow s}=\mathcal M_{uv}$, given that $G_{t\rightarrow s}=G$.  Similarly, if $t$ is a leaf, $uv$ is incident to $f^*_{\mathcal O}$; then $\mathcal N_{t\rightarrow s}=\{0\degree,0\degree\}$ and $\mathcal N_{s\rightarrow t}=\mathcal M_{uv}$. 
\end{itemize}

If $uv$ is an edge of $G$ incident to $f^*_{\mathcal O}$, the set $\mathcal M_{uv}$ can be trivially computed in $O(1)$ time from the set $\mathcal N_{s\rightarrow t}$, where $st$ is the edge of $\mathcal T$ dual to $uv$ and $s$ is a leaf of $\mathcal T$. Hence, all we have to do is show how to compute the sets $\mathcal N_{s\rightarrow t}$ efficiently.

When processing a node $s$ with parent $t$ during one of the traversals, we compute the label $\mathcal N_{s\rightarrow t}$.  The computation of $\mathcal N_{s\rightarrow t}$ exploits the values of the already computed labels $\mathcal N_{s_1\rightarrow s},\dots,\mathcal N_{s_k\rightarrow s}$, where $s_1,\dots,s_k$ are the neighbors of $s$ in $\mathcal T$ different from $t$. This is the problem we solved in Section~\ref{sse:var-2con-edge}! Namely, we want to compute the pairs $(\mu,\nu)$ with $\mu,\nu\in\{90\degree,180\degree,270\degree\}$ such that $G_{s\rightarrow t}$ admits a $(\chi_{s\rightarrow t},\mu,\nu)$-representation (these define $\mathcal N_{s\rightarrow t}$), starting from the pairs $(\mu_i,\nu_i)$ with $\mu_i,\nu_i\in\{0\degree,90\degree,180\degree\}$ such that $G_{s_i\rightarrow s}$ admits a $(\chi_{s_i\rightarrow s},\mu_i,\nu_i)$-representation (these define $\mathcal N_{s_1\rightarrow s},\dots,\mathcal N_{s_k\rightarrow s}$). When $\mathcal N_{s_i\rightarrow s}=\emptyset$, in particular, we also have $\mathcal N_{s\rightarrow t}=\emptyset$. That is, the non-existence of a $(\chi_{s_i\rightarrow s},\mu_i,\nu_i)$-representation of $G_{s_i}$ propagates towards the root of $\mathcal T$ in the current traversal; see Lemma~\ref{le:child-to-anchestor} below.

Clearly, we cannot afford to perform each traversal independently of the other ones, as this would result in a quadratic running time. Then, as in~\cite{dlop-ood-20}, we want to re-use the already computed labels $\mathcal N_{s\rightarrow t}$; this implies that a postorder traversal is not invoked on a tree $\mathcal T_{s\rightarrow t}$ if the label $\mathcal N_{s\rightarrow t}$ has been computed by a previous traversal. As a result, during the traversals of $\mathcal T$, each edge  is traversed at most once in each direction and each node with degree $k$ is processed $O(k)$ times. Differently from~\cite{dlop-ood-20}, we need to handle the possibility that, when a node $s$ of $\mathcal T$ is visited in a traversal after the first one, we might not have the sets  $\mathcal N_{s_i\rightarrow s}$ ready, even for most children of $s$. This is a consequence of the propagation of the empty sets $\mathcal N_{s\rightarrow t}$ described above. Indeed, we cannot even afford to look at all the children $s_i$ of $s$ and see which sets $\mathcal N_{s_i\rightarrow s}$ have already been computed and which have not; if the degree of $s$ is $k$, this would take $\Omega(k)$ time whenever we visit $s$ (potentially $k$ times), which would again result in a \mbox{quadratic running time.}

We cope with this problem by using, for each node of $\mathcal T$, some auxiliary labels that are dynamically computed during the traversals. For example, a label $\eta(s)$ points to a neighbor $s_i$ of $s$ for which $\mathcal N_{s_i\rightarrow s}=\emptyset$, two labels \start{s}~and \terminate{s}~delimit the interval of neighbors of $s$ for which an optimal pair has already been computed, and a label $a(s)$ stores the number of computed optimal pairs $(\mu_i,\nu_i)$ such that $\mu_i+\nu_i=360\degree$. The labels allow us to quickly determine which sets $\mathcal N_{s_i\rightarrow s}$ have already been computed and which have not, and to invoke a traversal recursively on the subtrees $\mathcal T_{s_i\rightarrow s}$ for which the sets $\mathcal N_{s_i\rightarrow s}$ have not been computed yet. Some labels (for example $a(s)$) store aggregate information on the values of the optimal pairs for the graphs $G_{s_i\rightarrow s}$. Thus, when the sets $\mathcal N_{s_i\rightarrow s}$ have been computed for all the children $s_i$ of $s$, and we are hence in a position to apply the algorithm described in Section~\ref{sse:var-2con-edge}, we do not have to spend $O(k)$ time to compute the values $a$, $b$, $c$, $d$, and $t$, but we can extract them from the labels associated to $s$ in $O(1)$ time, and then decide in $O(1)$ time whether a solution to the equation $4a'+3b'+2c'+d'=t$ subject to $0\leq a'\leq a$, $0\leq b'\leq b$, $0\leq c'\leq c$, and $0\leq d'\leq d$ exists; this ultimately determines whether a pair $(\mu,\nu)$ belongs to $\mathcal N_{s\rightarrow t}$. More precisely, we compute the following labels.

Let $s$  be any internal node of $\mathcal T$ and let $u_0,u_1,\dots,u_k$ be the clockwise order of the vertices along $\mathcal C_s$ in $\mathcal O$. Further, let $s_0,s_1,\dots,s_{k}$ be the neighbors of $s$, where $G_{s_i\rightarrow s}$ is rooted at $u_{i-1}u_i$, for $i=1,\dots,k$, and $G_{s_0\rightarrow s}$ is rooted at $u_ku_0$. 


\begin{enumerate}[(1)]
	\item The traversals compute a label $\eta(s)$ pointing to a neighbor $s_i$ of $s$ for which $\mathcal N_{s_i\rightarrow s}=\emptyset$. More precisely, if for every neighbor $s_i$ of $s$ we have $\mathcal N_{s_i\rightarrow s}\neq \emptyset$, then $\eta(s)=${\sc null}; otherwise, $\eta(s)=s_i$, where $\mathcal N_{s_i\rightarrow s}=\emptyset$.  
	\item For $i=0,\dots,k$, if $\mathcal N_{s_i\rightarrow s}\neq \emptyset$, the traversals compute two values $\mu^*_i$ and $\nu^*_i$ such that $(\mu^*_i,\nu^*_i)\in \mathcal N_{s_i\rightarrow s}$; if $G_{s_i\rightarrow s}$ is non-trivial, then $(\mu^*_i,\nu^*_i)$ is the optimal pair for $G_{s_i\rightarrow s}$, as defined after Lemma~\ref{le:fix-values}, otherwise $\mu^*_i=\nu^*_i=0\degree$.
	\item The traversals keep track of two values \start{s} and \terminate{s}; these delimit the interval $\mathcal I(s)$ of the indices of the neighbors of $s$ for which the optimal pair has already been established. More precisely, if \start{s}=\terminate{s}=\nullo, then no optimal pair has been established for any graph $G_{s_i\rightarrow s}$. Otherwise, let \start{s}=$i$ and \terminate{s}=$j$; then the optimal pairs have been established for the graphs $G_{s_i\rightarrow s}, G_{s_{i+1}\rightarrow s}, \dots, G_{s_j\rightarrow s}$. The sequence $s_0,s_1,\dots,s_{k}$ has to be intended here as a circular sequence, hence it is possible that $j<i$ (in such a case, the interval $\mathcal I(s)$ is equal to $[i,\dots,k,0,1,\dots,j]$). Observe that, if $\eta(s)=s_i$, then $i \notin \mathcal I(s)$.
	\item The traversals keep track of four values $a(s)$, $b(s)$, $c(s)$, and $d(s)$, where
	\begin{itemize}
		\item $a(s)$ is the number of $4$-components $G_{s_i\rightarrow s}$ of $G_s$ where $i\in \mathcal I(s)$, that is, $a(s)$ is the number of neighbors $s_i$ of $s$ such that $i\in  \mathcal I(s)$ and such that $\mu^*_i+\nu^*_i=360\degree$; 
		\item $b(s)$ is the number of $3$-components $G_{s_i\rightarrow s}$ of $G_s$ where $i\in \mathcal I(s)$; 
		\item $c(s)=c_1(s)+c_2(s)$, where $c_1(s)$ is the number of $2$-components $G_{s_i\rightarrow s}$ of $G_s$ where $i\in \mathcal I(s)$, while $c_2(s)$ is the number of vertices $u_i\in \chi$;
		\item $d(s)=\sum (2-\frac{\nu^*_j+\mu^*_{j+1}}{90\degree})$, where the sum ranges over all the indices $j\in  \mathcal I(s)\setminus \{\textrm{{\sc end}}(s)\}$.
	\end{itemize}  
\end{enumerate}

Before proceeding with the description of the algorithm, we introduce the following useful lemmata. 

\begin{lemma} \label{le:child-to-anchestor}
	Suppose that there exists a node $s$ in $\mathcal T$ that has a neighbor $s_i$ for which $\mathcal N_{s_i\rightarrow s}=\emptyset$. Then, for every neighbor $p\neq s_i$ of $s$, we have $\mathcal N_{s\rightarrow p}=\emptyset$.
\end{lemma}

\begin{proof}
	First, $s_i$ is an internal node of $\mathcal T$, as if it were a leaf, we would have $\mathcal N_{s_i\rightarrow s}=\{(0\degree,0\degree)\}$, while $\mathcal N_{s_i\rightarrow s}=\emptyset$. Let $uv$ be the edge of $G$ dual to $sp$. Since $s_i$ is an internal node of $\mathcal T$, it follows that $G_{s_i\rightarrow s}$ is a non-trivial $uv$-subgraph of $G_{s\rightarrow p}$. By Property~$\mathcal V_2$ of Lemma~\ref{le:structural-variable-embedding}, we have that $G_{s\rightarrow p}$ admits a $(\chi_{s\rightarrow p},\mu,\nu)$-representation, for any $\mu,\nu\in \{90\degree,180\degree,270\degree\}$, only if $G_{s_i\rightarrow s}$ admits a $(\chi_{s_i\rightarrow s},\mu_i,\nu_i)$-representation, for some $\mu_i,\nu_i\in \{90\degree,180\degree\}$. However, since $\mathcal N_{s_i\rightarrow s}=\emptyset$, we have that $G_{s_i\rightarrow s}$ admits no $(\chi_{s_i\rightarrow s},\mu_i,\nu_i)$-representation, for any $\mu_i,\nu_i\in \{90\degree,180\degree\}$, and hence $G_{s\rightarrow p}$ admits no $(\chi_{s\rightarrow p},\mu,\nu)$-representation, for any $\mu,\nu\in \{90\degree,180\degree,270\degree\}$; thus, $\mathcal N_{s\rightarrow p}=\emptyset$.
\end{proof}

\begin{lemma} \label{le:two-no-children}
	Suppose that there exist two distinct neighbors $s_i$ and $s_j$ of a node $s$ in $\mathcal T$ such that $\mathcal N_{s_i\rightarrow s}=\emptyset$ and $\mathcal N_{s_j\rightarrow s}=\emptyset$. Then, for $h=1,\dots,n$, we have $\mathcal M_{u_hv_h}=\emptyset$. 
\end{lemma}

\begin{proof}
	Consider any index $h\in \{1,\dots,n\}$. Let $r^*_hr_h$ be the edge of $\mathcal T$ dual to $u_hv_h$, where $r^*_h$ is a leaf of $\mathcal T$ (recall that $u_hv_h$ is incident to $f^*_{\mathcal O}$). Let $(s=t_1,t_2,\dots,t_{x-1}=r_h,t_x=r^*_h)$ be the path in $\mathcal T$ between $s$ and $r^*_h$. Since $s_i$ and $s_j$ are distinct, one of them is not $t_2$. By repeated applications of Lemma~\ref{le:child-to-anchestor}, we get that $\mathcal N_{t_1\rightarrow t_2}=\emptyset$, $\mathcal N_{t_2\rightarrow t_3}=\emptyset$, $\dots$, $\mathcal N_{t_{x-1}\rightarrow t_x}=\emptyset$. Hence, $\mathcal N_{r_h\rightarrow r^*_h}=\mathcal M_{u_hv_h}=\emptyset$.
\end{proof}

We are now ready to describe the postorder traversals of $\mathcal T$. We will root $G$ first at $u_1v_1$, then at $u_2v_2$, then at $u_3v_3$, and so on. When $G$ is rooted at $u_hv_h$, we also root $\mathcal T$ at the leaf $r^*_h$ such that the edge $r_h r^*_h$ of $\mathcal T$ incident to $r^*_h$ is dual to $u_hv_h$. In order to stress the rooting, we denote by $\mathcal T_h$ the tree $\mathcal T$ when rooted at $r^*_h$. Also, we say that a node $s$ is a \emph{$\mathcal T_h$-child} of a node $t$ if $s$ is a child of $t$ in $\mathcal T_h$; then $t$ is the \emph{$\mathcal T_h$-parent} of $s$. For $h=1,\dots,n$, the postorder traversal of $\mathcal T_h$ ends after visiting the child $r_h$ of the root $r^*_h$. 

Before the postorder traversal of $\mathcal T_1$ starts, we perform the following initialization. For every internal node $s$ of $\mathcal T$, we initialize the labels $\eta(s)$, \start{s}, and \terminate{s} to \nullo, and the values $a(s)$, $b(s)$, and $d(s)$ to $0$. Further, we initialize $c(s)$ to the number $c_2(s)$ of vertices in $\mathcal C(s)$ that belong to $\chi$. Finally, for each leaf $s$ of $\mathcal T$ with a unique neighbor $t$, we set $\mathcal N_{s\rightarrow t}=\{(0\degree,0\degree)\}$, and for each internal node $s$ of $T$ and each neighbor $t$ of $s$, we set $\mathcal N_{s\rightarrow t}=$\nullo.   

We now explain how the postorder traversal of $\mathcal T_h$ works, for each $h=1,\dots,n$. 

When the postorder traversal is invoked on a subtree $\mathcal T_{s\rightarrow p}$ of $\mathcal T_h$ (at first, this is the entire tree $\mathcal T_h$), it enters a $\mathcal T_h$-child $s_i$ of $s$ only if $\mathcal N_{s_i\rightarrow s}=$\nullo. Hence, if the set $\mathcal N_{s_i\rightarrow s}$ has been computed during the traversal of $\mathcal T_j$ with $j<h$, the subtree $\mathcal T_{s_i\rightarrow s}$ is not visited during the traversal of $\mathcal T_h$. This also implies that the postorder traversal is never invoked on a subtree $\mathcal T_{s_i\rightarrow s}$ such that $s_i$ is a leaf, given that $\mathcal N_{s_i\rightarrow s}=\{(0\degree,0\degree)\}$, as established in the initialization. 



Now consider an internal node $s$ of $\mathcal T$; let $u_0,u_1,\dots,u_k$ be the clockwise order of the vertices along $\mathcal C_s$. Further, let $s_0=p,s_1,\dots,s_{k}$ be the neighbors of $s$, where $s_1,\dots,s_k$ are the $\mathcal T_h$-children of $s$ and $p$ is the $\mathcal T_h$-parent of $s$; for $i=1,\dots,k$, let $u_{i-1}u_i$ be the root of $G_{s_i\rightarrow s}$ and let $u_0u_k$ be the edge of $G$ dual to $sp$. When the postorder traversal of $\mathcal T_h$ visits $s$, after all the subtrees $\mathcal T_{s_1\rightarrow s},\dots,\mathcal T_{s_k\rightarrow s}$ have been processed, it computes the set $\mathcal N_{s\rightarrow p}$ and, possibly, computes the values of some other labels. Note that, since $\mathcal T_{s_1\rightarrow s},\dots,\mathcal T_{s_k\rightarrow s}$ have been processed, the sets $\mathcal N_{s_1\rightarrow s},\dots,N_{s_k\rightarrow s}$ have been computed already. The traversal performs the following actions.

Suppose first that $\eta(s)=s_i$, for some $i\in \{1,\dots,k\}$, hence $\mathcal N_{s_i\rightarrow s}=\emptyset$. By Lemma~\ref{le:child-to-anchestor}, we can set $\mathcal N_{s\rightarrow p}=\emptyset$. If $\eta(p)=${\sc null}, we set $\eta(p)=s$; otherwise, that is, if $\eta(p)=t$ for some neighbor $t\neq s$ of $p$, by Lemma~\ref{le:two-no-children} we can terminate all the traversals reporting in output the labeling $\mathcal M_{u_1v_1}=\mathcal M_{u_2v_2}=\dots=\mathcal M_{u_nv_n}=\emptyset$.


Suppose next that $\eta(s)=p$ or that $\eta(s)=${\sc null}; hence, for every $\mathcal T_h$-child $s_i$ of $s$, we have $\mathcal N_{s_i\rightarrow s}\neq \emptyset$. We now determine the set $\mathcal N_{s\rightarrow p}$. In order to do so, we proceed as follows.

\begin{itemize}
	\item We first determine the optimal pair $(\mu^*_i,\nu^*_i)$ for every graph $G_{s_i\rightarrow s}$ with $i=2,\dots,k-1$, or conclude that $\mathcal N_{s\rightarrow p}=\emptyset$. The optimal pair $(\mu^*_i,\nu^*_i)$ for a graph $G_{s_i\rightarrow s}$ might have already been determined, if $s$ has been processed in a previous traversal of $\mathcal T$; if that is the case, then $i\in \mathcal I(s)$. We process the $\mathcal T_h$-children of $s$ one by one, and possibly update \start{s}, \terminate{s}, $a(s)$, $b(s)$, $c(s)$, and $d(s)$ whenever the optimal pair for a $\mathcal T_h$-child of $s$ is computed; the order in which the $\mathcal T_h$-children of $s$ are processed guarantees that $\mathcal I(s)$ is an interval of the circular sequence $[0,1,\dots,k]$, after every processed $\mathcal T_h$-child of $s$. 
	
	We first describe how to process a $\mathcal T_h$-child $s_i$ of $s$ (except for the possible update to \start{s} and/or to \terminate{s}) and we later discuss the order in which the  $\mathcal T_h$-children of $s$ are processed (while doing so, we also show how to update \start{s} and/or \terminate{s} after the computation of the optimal pair for a graph $G_{s_i\rightarrow s}$).
	
	When we process a $\mathcal T_h$-child $s_i$ of $s$, we do the following. For $j=i-1,i,i+1$, we known whether $G_{s_{j}\rightarrow s}$ is trivial or not (indeed, $G_{s_{j}\rightarrow s}$ is trivial if and only if $s_j$ is a leaf of $\mathcal T$). Further, if $G_{s_{j}\rightarrow s}$ is not trivial, we know whether it admits a $(\chi_{s_{j}\rightarrow s},\mu_{j},\nu_{j})$-representation or not, for every pair $(\mu_{j},\nu_{j})$ with $\mu_{j},\nu_{j} \in \{90\degree,180\degree\}$; indeed, this information is in the set $\mathcal N_{s_j\rightarrow s}$. Hence, by Lemma~\ref{le:fix-values}, we either correctly conclude that $G_{s\rightarrow p}$ has no $(\chi_{s\rightarrow p},\mu,\nu)$-representation, for any values $\mu,\nu \in \{90\degree,180\degree,270\degree\}$, or we find the optimal pair $(\mu^*_i,\nu^*_i)$ for $G_{s_{i}\rightarrow s}$. 
	\begin{itemize}
		\item If we concluded that $G_{s\rightarrow p}$ has no $(\chi_{s\rightarrow p},\mu,\nu)$-representation, for any values $\mu,\nu \in \{90\degree,180\degree,270\degree\}$, then we set $\mathcal N_{s\rightarrow p}=\emptyset$. If $\eta(p)=${\sc null}, then we set $\eta(p)=s$ and stop processing $s$, otherwise, that is, if $\eta(p)=t$ for some neighbor $t\neq s$ of $p$, by Lemma~\ref{le:two-no-children} we can terminate all the traversals reporting in output the labeling $\mathcal M_{u_1v_1}=\mathcal M_{u_2v_2}=\dots=\mathcal M_{u_nv_n}=\emptyset$. 
		\item If we found the optimal pair $(\mu^*_i,\nu^*_i)$ for $G_{s_{i}\rightarrow s}$, then we update $a(s)$, $b(s)$, $c(s)$, or $d(s)$. Namely, if $\mu^*_i+\nu^*_i=360\degree$, or $\mu^*_i+\nu^*_i=270\degree$, or $\mu^*_i+\nu^*_i=180\degree$, then we increase $a(s)$, or $b(s)$, or $c(s)$ by $1$, respectively. If $i-1 \in \mathcal I(s)$, then we add $2-\frac{\nu^*_{i-1}+\mu^*_i}{90\degree}$ to $d(s)$, and if $i+1 \in \mathcal I(s)$, then we add $2-\frac{\nu^*_{i}+\mu^*_{i+1}}{90\degree}$ to $d(s)$. 
	\end{itemize}
	
	We now discuss the order in which the $\mathcal T_h$-children of $s$ are processed. In order to determine the next $\mathcal T_h$-child of $s$ to process, we repeatedly apply the following rules.
	
	\begin{itemize}
		\item If \start{s}$=$\terminate{s}$=$\nullo, then we process $s_2$, and update \start{s}$=$\terminate{s}$=s_2$ if we found the optimal pair $(\mu^*_2,\nu^*_2)$ for $G_{s_2\rightarrow s}$, or stop processing $s$ if we concluded that $\mathcal N_{s\rightarrow p}=\emptyset$.
		\item Suppose next that \start{s}$\neq$\nullo~and \terminate{s}$\neq$\nullo. Let $i=$\start{s} and $j=$\terminate{s}. 
		\begin{itemize}
			\item If $3\leq i\leq k$ and $j \neq i-1$, then we process $s_{i-1}$, and we either update \start{s}$=i-1$ if we found the optimal pair $(\mu^*_{i-1},\nu^*_{i-1})$ for $G_{s_{i-1}\rightarrow s}$, or stop processing $s$ if we concluded that $\mathcal N_{s\rightarrow p}=\emptyset$.
			\item Otherwise, if $1\leq j\leq k-2$ and $i\neq j+1$, then we process $s_{j+1}$, and we either update \terminate{s}$=j+1$ if we found the optimal pair $(\mu^*_{j+1},\nu^*_{j+1})$ for $G_{s_{j+1}\rightarrow s}$, or stop processing $s$ if we concluded that $\mathcal N_{s\rightarrow p}=\emptyset$.
		\end{itemize}			
	\end{itemize}
	The described process either computes the optimal sequence $\mu^*_2,\nu^*_2,\dots,\mu^*_{k-1},\nu^*_{k-1}$ for $G_{s\rightarrow p}$, where $(\mu^*_i,\nu^*_i)$ is the optimal pair for $G_{s_i\rightarrow s}$ for $i=2,\dots,k-1$, or concludes that $\mathcal N_{s\rightarrow p}=\emptyset$ (possibly after computing the optimal pair $(\mu^*_i,\nu^*_i)$ for some graph $G_{s_i\rightarrow s}$). In the former case, we have $\{2,\dots,k-1\}\subseteq\mathcal I(s)$; the indices $0$, $1$, and $k$ might belong to $\mathcal I(s)$ or not. In the latter case, the processing of $s$ is finished.
	
	\item If we computed the optimal sequence $\mu^*_2,\nu^*_2,\dots,\mu^*_{k-1},\nu^*_{k-1}$ for $G_{s\rightarrow p}$, then, for any $\mu,\nu \in \{90\degree,180\degree,270\degree\}$, we determine whether $G_{s\rightarrow p}$ admits a $(\chi_{s\rightarrow p},\mu,\nu)$-representation (and hence whether $(\mu,\nu)\in \mathcal N_{s\rightarrow p}$ or not). By Lemmata~\ref{le:structural-variable-embedding} and~\ref{le:fix-values}, this is equivalent to determining whether there exist an in-out assignment $\mathcal A$ (for the $u_0u_k$-subgraphs of $G_{s\rightarrow p}$) and values $\rho_0,\rho_1,\dots,\rho_k,\mu_1,\nu_1,\mu_k,\nu_k$ such that Properties~$\mathcal V1$--$\mathcal V5$ of Lemma~\ref{le:structural-variable-embedding} are satisfied.
	
	By Lemma~\ref{le:promising-sequences}, we can construct the promising sequences for $(G_{s\rightarrow p},\mu,\nu)$ in $O(1)$ time; recall that each promising sequence is obtained by adding six values $\mu_1,\nu_1,\mu_{k},\nu_{k},\rho_0,\rho_k$ to the optimal sequence $\mu^*_2,\nu^*_2,\dots,\mu^*_{k-1},\nu^*_{k-1}$ for $G_{s\rightarrow p}$. For each promising sequence $\mu_1,\nu_1,\mu^*_2,\nu^*_2,\dots,\mu^*_{k-1},\nu^*_{k-1},\mu_{k},\nu_{k},\rho_0,\rho_k$ for $(G_{s\rightarrow p},\mu,\nu)$, we compute the values $a$, $b$, $c$, $d$, and $t$, where $a$ and $b$ denote the number of $4$- and $3$-components of $G_{s\rightarrow p}$, respectively, $c$ denotes the number of $2$-components of $G_{s\rightarrow p}$ plus the number of vertices $u_i \in \chi$ with $i\in \{1,\dots,k-1\}$, $d$ denotes the sum $(2-\frac{\nu_1+\mu^*_2}{90\degree})+\sum_{i=2}^{k-2} (2-\frac{\nu^*_i+\mu^*_{i+1}}{90\degree})+(2-\frac{\nu^*_{k-1}+\mu_k}{90\degree})$, and $t=(k-1)-\frac{\rho_0+\rho_k}{90\degree}$. The values $a$, $b$, $c$, $d$, and $t$ can be obtained by means of $O(1)$ sums, and hence in $O(1)$ time, as follows. 
	\begin{itemize}
		\item First, $a$ is equal to $a(s)$ minus the number of indices $j\in \{k,0,1\} \cap \mathcal I(s)$ such that $\mu^*_j+\nu^*_j=360\degree$ plus the number of indices $j\in \{k,1\}$ such that $\mu_j+\nu_j=360\degree$. 
		
		Indeed, the value $a(s)$ counts all the optimal pairs $(\mu^*_j,\nu^*_j)$ such that $j\in \mathcal I(s)$ and $\mu^*_j+\nu^*_j=360\degree$. However, the optimal pairs $(\mu^*_k,\nu^*_k)$, $(\mu^*_0,\nu^*_0)$, and $(\mu^*_1,\nu^*_1)$ do not contribute to the promising sequence $\mu_1,\nu_1,\mu^*_2,\nu^*_2,\dots,\mu^*_{k-1},\nu^*_{k-1},\mu_{k},\nu_{k},\rho_0,\rho_k$ for $(G_{s\rightarrow p},\mu,\nu)$; indeed, the promising sequence contains the pairs $(\mu_1,\nu_1)$ and $(\mu_k,\nu_k)$ in place of $(\mu^*_1,\nu^*_1)$ and $(\mu^*_k,\nu^*_k)$, respectively, and the promising sequence does not contain any value for $G_{s_0\rightarrow s}$.
		\item Analogously, $b$ is equal to $b(s)$ minus the number of indices $j\in \{k,0,1\}\cap \mathcal I(s)$ such that $\mu^*_j+\nu^*_j=270\degree$ plus the number of indices $j\in \{k,1\}$ such that $\mu_j+\nu_j=270\degree$.
		\item Further, $c$ is equal to $c(s)$ minus the number of indices $j\in \{k,0,1\}\cap \mathcal I(s)$ such that $\mu^*_j+\nu^*_j=180\degree$, plus the number of indices $j\in \{k,1\}$ such that $\mu_j+\nu_j=180\degree$, and minus the number of indices  $j\in \{k,0\} \cap \chi$. The last term accounts for the fact that $u_0$ and $u_k$ do not contribute to $c$, even if they belong to $\chi$.
		\item The value $d$ is equal to $d(s)$ plus $(2-\frac{\nu_1+\mu^*_2}{90\degree})$, plus $(2-\frac{\nu^*_{k-1}+\mu_k}{90\degree})$, minus $(2-\frac{\nu^*_1+\mu^*_2}{90\degree})$ if $\{1,2\}\subseteq \mathcal I(s)$, minus $(2-\frac{\nu^*_{k-1}+\mu^*_k}{90\degree})$ if $\{k-1,k\}\subseteq \mathcal I(s)$, minus $(2-\frac{\nu^*_{k}+\mu^*_0}{90\degree})$ if $\{k,0\}\subseteq \mathcal I(s)$, and minus $(2-\frac{\nu^*_{0}+\mu^*_1}{90\degree})$ if $\{0,1\}\subseteq \mathcal I(s)$.
		\item Finally, the value $t=(k-1)-\frac{\rho_0+\rho_k}{90\degree}$ can be trivially computed in $O(1)$ time. 
	\end{itemize}  
	By Lemma~\ref{le:find-values-variable}, it is possible to determine in $O(1)$ time whether the promising sequence $\mu_1,\nu_1,\mu^*_2,\nu^*_2,\dots$, $\mu^*_{k-1},\nu^*_{k-1},\mu_{k},\nu_{k},\rho_0,\rho_k$ for $(G_{s\rightarrow p},\mu,\nu)$ is extensible, i.e., whether there exist an in-out assignment $\mathcal A$ (for the $u_0u_k$-subgraphs of $G_{s\rightarrow p}$) and values $\rho_1,\rho_2,\dots,\rho_{k-1}$ that, together with the promising sequence, satisfy Properties~$\mathcal V1$--$\mathcal V5$ of Lemma~\ref{le:structural-variable-embedding}. Then we add the pair $(\mu,\nu)$ to $\mathcal N_{s\rightarrow p}$ if and only if one of the promising sequences for $(G_{s\rightarrow p},\mu,\nu)$ is extensible.		
\end{itemize} 

We are now ready to state the following.	

\begin{theorem} \label{th:2-con-variable}
	Let $G$ be an $n$-vertex $2$-connected outerplanar graph and $\chi$ be a subset of the degree-$2$ vertices of $G$. There is an $O(n)$-time algorithm that decides whether $G$ admits a $\chi$-constrained representation or not. 
	
	If $G$ admits a $\chi$-constrained representation, the algorithm labels every vertex $v$ of $G$ whose degree is not larger than $3$ with a set $\gamma(v)$; this set contains all the values $\mu\in \{90\degree,180\degree,270\degree\}$ such that $G$ admits a $\chi$-constrained representation $(\mathcal E,\phi)$ in which $v$ is incident to the outer face of $\mathcal E$ and the sum of the internal angles at $v$ is equal to $\mu$. The algorithm computes the sets $\gamma(v)$ in total $O(n)$ time. 
	
	Finally, if $\mu\in \gamma(v)$, for some vertex $v$ of $G$ whose degree is not larger than $3$ and some value $\mu\in \{90\degree,180\degree,270\degree\}$, then the algorithm can construct in $O(n)$ time a $\chi$-constrained representation $(\mathcal E,\phi)$ in which $v$ is incident to the outer face of $\mathcal E$ and the sum of the internal angles at $v$ is equal to $\mu$. 
\end{theorem}

\begin{proof}
	We show how to compute the sets $\gamma(v)$ in total $O(n)$ time. Then, by Lemma~\ref{le:low-degree-external}, in order to test whether $G$ admits a $\chi$-constrained representation, it suffices to check whether there exists a vertex $v$ whose degree is not larger than $3$	and such that $\gamma(v)\neq \emptyset$. Clearly, this can be done in $O(n)$ time.
	
	Let $\mathcal O$ be the outerplane embedding of $G$. By Lemma~\ref{le:edge-vertex-labels}, the sets $\gamma(u)$ can be computed in total $O(n)$ time from the sets $\mathcal M_{uv}$ for the edges $uv$ of $G$ incident to $f^*_{\mathcal O}$. Further, if an edge $u_hv_h$ of $G$ incident to $f^*_{\mathcal O}$ is dual to an edge $r^*_hr_h$ of the extended dual tree $\mathcal T$, where $r^*_h$ is a leaf of $\mathcal T$, then by definition $\mathcal M_{u_hv_h}=\mathcal N_{r_h\rightarrow r^*_h}$. Hence, in order to compute in total $O(n)$ time the sets $\gamma(v)$, it suffices to show how to compute in total $O(n)$ time the sets $\mathcal N_{t\rightarrow s}$ and $\mathcal N_{s\rightarrow t}$ for all the edges $st$ of $\mathcal T$. 
	
	We now prove that the algorithm described before the theorem correctly compute such sets (unless a node $s$ is found with two neighbors $s_i$ and $s_j$ such that $\mathcal N_{s_i\rightarrow s}$ and $\mathcal N_{s_j\rightarrow s}$) and does so in total linear time. 
	
	The correctness is proved inductively. Indeed, for each leaf $s$ of $\mathcal T$ with unique neighbor $t$, the set $\mathcal N_{s\rightarrow t}$ is correctly defined as $\{0\degree,0\degree\}$ during the initialization that precedes the traversals. Now consider any internal node $s$ with parent $p$ and with children $s_1,\dots,s_k$. The set $\mathcal N_{s\rightarrow p}$ is constructed only after the sets $\mathcal N_{s_i\rightarrow s}$ have been constructed, for $i=1,\dots,k$. Inductively assume that, for $i=1,\dots,k$, the set $\mathcal N_{s_i\rightarrow s}$ correctly contains all the pairs $(\mu_i,\nu_i)$ with $\mu_i,\nu_i\in \{90\degree,180\degree,270\degree\}$ such that $G_{s_i\rightarrow s}$ admits a $(\chi_{s_i\rightarrow s},\mu_i,\nu_i)$-representation. We prove that $\mathcal N_{s\rightarrow p}$ contains a pair $(\mu,\nu)$ with $\mu,\nu\in \{90\degree,180\degree,270\degree\}$ if and only if $G_{s\rightarrow p}$ admits a $(\chi_{s\rightarrow p},\mu,\nu)$-representation.
	
	Suppose first that $(\mu,\nu)\notin \mathcal N_{s\rightarrow p}$. There are three possible reasons for that. 
	
	\begin{enumerate}
		\item If $\mathcal N_{s_i\rightarrow s}=\emptyset$, for some child $s_i$ of $s$, then the algorithm sets $\mathcal N_{s\rightarrow p}=\emptyset$, which implies that $(\mu,\nu)\notin \mathcal N_{s\rightarrow p}$. Then Lemma~\ref{le:child-to-anchestor} ensures that $G_{s\rightarrow p}$ admits no $(\chi_{s\rightarrow p},\mu,\nu)$-representation. 
		\item When determining the optimal pair $(\mu^*_i,\nu^*_i)$ for $G_{s_i\rightarrow s}$, for some $i\in \{2,\dots,k-1\}$, the algorithm might set $\mathcal N_{s\rightarrow p}=\emptyset$, which implies that $(\mu,\nu)\notin \mathcal N_{s\rightarrow p}$. Then Lemma~\ref{le:fix-values} ensures that $G_{s\rightarrow p}$ admits no $(\chi_{s\rightarrow p},\mu,\nu)$-representation.
		\item After successfully computing the optimal sequence $\mu^*_2,\nu^*_2$, $\dots$, $\mu^*_{k-1},\nu^*_{k-1}$ for $G_{s\rightarrow p}$, the algorithm might not insert $(\mu,\nu)$ into $\mathcal N_{s\rightarrow p}$ because it fails to determine an in-out assignment $\mathcal A$ and a sequence of values $\rho_0,\rho_1,\dots,\rho_k,\mu_1,\nu_1,\mu_k,\nu_k$ that, together with the optimal sequence for $G_{s\rightarrow p}$, satisfy Properties~$\mathcal V_1$--$\mathcal V_5$ of Lemma~\ref{le:structural-variable-embedding} for $(G,\mu,\nu)$. Then Lemmata~\ref{le:structural-variable-embedding},~\ref{le:fix-values},~\ref{le:extensible-solution}, and~\ref{le:find-values-variable} ensure that $G_{s\rightarrow p}$ admits no $(\chi_{s\rightarrow p},\mu,\nu)$-representation.   
	\end{enumerate}     
	
	Suppose next that $(\mu,\nu)\in \mathcal N_{s\rightarrow p}$. The algorithm inserts $(\mu,\nu)$ into $\mathcal N_{s\rightarrow p}$ because there exist an in-out assignment $\mathcal A$ and a sequence of values $\rho_0,\rho_1,\dots,\rho_k,\mu_1,\nu_1,\mu_k,\nu_k$ that, together with the optimal sequence $\mu^*_2,\nu^*_2$, $\dots$, $\mu^*_{k-1},\nu^*_{k-1}$ for $G_{s\rightarrow p}$, satisfy Properties~$\mathcal V_1$--$\mathcal V_5$ of Lemma~\ref{le:structural-variable-embedding} for $(G,\mu,\nu)$. Then such a lemma ensures that $G_{s\rightarrow p}$ admits a $(\chi_{s\rightarrow p},\mu,\nu)$-representation.   
	
	There is one situation in which the algorithm actually does not construct all the sets $\mathcal N_{s\rightarrow t}$. This happens if two neighbors $s_i$ and $s_j$ of the same node $s$ are found such that $\mathcal N_{s_i\rightarrow s}=\emptyset$ and $\mathcal N_{s_j\rightarrow s}=\emptyset$. Then, by Lemma~\ref{le:two-no-children}, the algorithm correctly and directly concludes that $\mathcal M_{u_hv_h}=\emptyset$, for every edge $u_hv_h$ of $G$ incident to $f^*_{\mathcal O}$, and hence that $\gamma(v)=\emptyset$, for every vertex $v$ of $G$ whose degree is not larger than $3$.
	
	We now show that the algorithm runs in $O(n)$ time. 
	
	
	\begin{itemize}
		\item First, the initialization takes $O(n)$ time and is performed once. Namely, the initialization of the labels $\eta(s)$, \start{s}, and \terminate{s} to \nullo, and of the values $a(s)$, $b(s)$, and $d(s)$ to $0$ is performed in $O(1)$ time, for every internal node $s$ of $\mathcal T$, and hence in total $O(n)$ time. After an $O(n)$-time preprocessing that marks every vertex of $G$ that belongs to $\chi$, the initialization of $c(s)$ to the number $c_2(s)$ of vertices in $\mathcal C(s)$ that belongs to $\chi$ can be performed in $O(k)$ time, for every internal node $s$ of $\mathcal T$ such that $\mathcal C(s)$ has $k+1$ vertices, and hence in total $O(n)$ time. Finally, the initialization of the sets $\mathcal N_{s\rightarrow t}$ can be trivially performed in $O(1)$ time per edge of $\mathcal T$, and hence in total $O(n)$ time.
		\item If the algorithm finds two neighbors $s_i$ and $s_j$ of the same node $s$ such that $\mathcal N_{s_i\rightarrow s}=\emptyset$ and $\mathcal N_{s_j\rightarrow s}=\emptyset$, then in total $O(n)$ time sets $\mathcal M_{u_hv_h}=\emptyset$, for every edge $u_hv_h$ of $G$ incident to $f^*_{\mathcal O}$. Note that this operation is performed at most once by the algorithm.
		\item Now observe that, for each internal node $s$ of $\mathcal T$ and for each neighbor $p$ of $s$, there is (at most) one rooting of $\mathcal T$ for which a postorder traversal is invoked on $\mathcal T_{s\rightarrow p}$. Indeed, once a postorder traversal is invoked on $\mathcal T_{s\rightarrow p}$, the set $\mathcal N_{s\rightarrow p}$ is determined; if some subsequent traversal of $\mathcal T$ processes $p$ with $s$ as a child, then no postorder traversal is invoked on $\mathcal T_{s\rightarrow p}$. This implies that, if $\mathcal C_s$ has $k+1$ neighbors, then $s$ is processed $O(k)$ times. 
		
		\item We now show that, for each internal node $s$ of $\mathcal T$ with $k+1$ neighbors $s_0,\dots,s_k$ in $\mathcal T$, the computation of the optimal pairs for $\mathcal G_{s_0\rightarrow s},\dots,\mathcal G_{s_k\rightarrow s}$ takes $O(k)$ time. 
		
		First, for each neighbor $s_i$ of $s$, we have that Lemma~\ref{le:fix-values} is invoked at most once trying to determine the optimal pair for $\mathcal G_{s_i\rightarrow s}$.
		
		Second, the application of Lemma~\ref{le:fix-values} takes $O(1)$ time in order to either find the optimal pair for $\mathcal G_{s_i\rightarrow s}$ or stop the processing of $s$ in the current traversal of $\mathcal T$. Hence, the ``successful'' (at determining optimal pairs) invocations of Lemma~\ref{le:fix-values} take $O(k)$ time in total, as $s$ has $k+1$ neighbors, and the unsuccessful ones take $O(k)$ time in total, as $s$ is processed $O(k)$ times. 
		
		Third, after each invocation of Lemma~\ref{le:fix-values}, some updates are possibly performed on some labels. Namely, suppose that Lemma~\ref{le:fix-values} is invoked on a graph $\mathcal G_{s_i\rightarrow s}$, where $s_i$ is a child of $s$ in the current traversal of $\mathcal T$; let $p$ be the parent of $s$ in such a traversal. 
		
		\begin{itemize}
			\item If Lemma~\ref{le:fix-values} determines that $G_{s\rightarrow p}$ has no $(\chi_{s\rightarrow p},\mu,\nu)$-representation, for any values $\mu,\nu\in\{90\degree,180\degree\}$, then the algorithm sets $\mathcal N_{s\rightarrow p}=\emptyset$ and $\eta(p)=s_i$. Hence, each unsuccessful invocation of Lemma~\ref{le:fix-values} leads to $O(1)$ time to update labels. Since $s$ is processed $O(k)$ times, this amounts $O(k)$ time in total. 
			\item  If Lemma~\ref{le:fix-values} determines an optimal pair, then the algorithm either increases the value of one of $a(s)$, $b(s)$, and $c(s)$ by $1$, or adds either $0$, or $1$, or $2$ to $d(s)$. Finally, the algorithm either sets \start{s}$=$\terminate{s}$=2$, or it increases the value of  \terminate{s}~by $1$, or it decreases the value of \start{s}~by $1$. Clearly, each of these updates can be performed in $O(1)$ time. Hence, the successful invocations of Lemma~\ref{le:fix-values} lead to spend $O(k)$ time to update labels, as $s$ has $k+1$ neighbors.
		\end{itemize} 
		
		Finally, a key component in the analysis of the running time is that the ``next'' graph $\mathcal G_{s_i\rightarrow s}$ on which Lemma~\ref{le:fix-values} is invoked can be determined in $O(1)$ time, due to the values $\eta(s)$, \start{s}, and \terminate{s}. Indeed, when processing $s$ during the traversal of $\mathcal T_h$, the algorithm first checks the value of $\eta(s)$ and, possibly, concludes the processing of $s$ in the traversal of $\mathcal T_h$ without ever invoking Lemma~\ref{le:fix-values}. Otherwise, the algorithm invokes Lemma~\ref{le:fix-values} on either $\mathcal G_{s_2\rightarrow s}$, or $\mathcal G_{s_{\textrm{\sc start} -1}\rightarrow s}$, or $\mathcal G_{s_{\textrm{\sc end} -1}\rightarrow s}$, depending on the values of \start{s} and \terminate{s}. 
		
		\item As noted during the algorithm's description, once the optimal sequence for a graph $\mathcal G_{s\rightarrow p}$ has been computed, the set $\mathcal N_{s\rightarrow p}$ can be determined in $O(1)$ time due to Lemmata~\ref{le:structural-variable-embedding},~\ref{le:fix-values},~\ref{le:promising-sequences}, and~\ref{le:find-values-variable}.  
	\end{itemize}
	
	This concludes the proof that the running time of the algorithm is $O(n)$.
	
	Finally, we show that, given a vertex $u$ of $G$ whose degree is not larger than $3$ and given a value $\mu\in \{90\degree,180\degree,270\degree\}$ such that $\mu\in \gamma(u)$, we can construct in $O(n)$ time a $\chi$-constrained representation $(\mathcal E,\phi)$ in which $u$ is incident to the outer face of $\mathcal E$ and the sum of the internal angles at $u$ is equal to $\mu$. By Lemma~\ref{le:edge-vertex-labels}, there exists a set $\mathcal M_{uv}$, where $uv$ is an edge incident to $u$ and to $f^*_{\mathcal O}$, such that $(\mu,\nu) \in \mathcal M_{uv}$, for some value $\nu\in \{90\degree,180\degree,270\degree\}$. Clearly, the set $\mathcal M_{uv}$ and the pair $(\mu,\nu)$ can be found in $O(1)$ time, given that $u$ has at most three incident edges and each set $\mathcal M_{uw}$ has constant size. We can now apply Theorem~\ref{th:2-con-variable-edge} in order to construct a $(\chi,\mu,\nu)$-representation of $G$ in $O(n)$ time; this is the desired representation $(\mathcal E,\phi)$. 
\end{proof}

\subsection{General Outerplanar Graphs} \label{sse:var}

In this section we remove the assumption that the input graph is $2$-connected and show how to test whether an outerplanar graph admits a planar rectilinear drawing in linear time. Recall that we can assume that the input outerplanar graph is connected, as if it is not, then it admits a planar rectilinear drawing if and only if every of its connected components does.   

Let $G$ be an $n$-vertex connected outerplanar graph and let $\mathcal O$ be its outerplane embedding. Consider the block-cut-vertex tree $T$ of $G$~\cite{h-gt-69,ht-aeagm-73}. We denote by $G_b$ the block corresponding to a B-node $b$, by $n_b$ the number of nodes of $G_b$, and by $\mathcal O_b$ the restriction of $\cal O$ to the vertices and edges of $G_b$. 

We now define a set $\chi_b$ for every non-trivial block $G_b$ of $G$. We initialize each set $\chi_b$ to an empty set. Then, for every cut-vertex $c$ that is shared by two non-trivial blocks $G_{b_1}$ and $G_{b_2}$ of $G$, we add $c$ to both $\chi_{b_1}$ and $\chi_{b_2}$. This concludes the construction of the sets $\chi_b$. We have the following.

\begin{lemma} \label{le:chi-necessary}
	Consider any rectilinear representation $(\mathcal E,\phi)$ of $G$ and let $G_b$ be any non-trivial block of $G$. The restriction of $(\mathcal E,\phi)$ to $G_b$ is a $\chi_b$-constrained representation of $G_b$. 
\end{lemma} 	 

\begin{proof}
	Let $c$ be a cut-vertex of $G$ that is shared by two non-trivial blocks $G_{b_1}$ and $G_{b_2}$ of $G$; then $c$ has degree $2$ in each of $G_{b_1}$ and $G_{b_2}$. Further, let $(\mathcal E_1,\phi_1)$ and $(\mathcal E_2,\phi_2)$ be the restrictions of $(\mathcal E,\phi)$ to $G_{b_1}$ and $G_{b_2}$, respectively. Finally, let $f_{1,a}$ and $f_{1,b}$ be the faces of $\mathcal E_{b_1}$ incident to $c$ and let $f_{2,a}$ and $f_{2,b}$ be the faces of $\mathcal E_{b_2}$ incident to $c$. Then we only need to prove that $\phi_1(c,f_{1,a}),\phi_1(c,f_{1,b}),\phi_2(c,f_{2,a}),\phi_2(c,f_{2,b})\in \{90\degree,270\degree\}$. However, this follows from the fact that, by the planarity of $\mathcal E$, the edges of $G_{b_1}$ incident to $c$ are consecutive in the clockwise order of the edges incident to $c$ in $\mathcal E$, and so are the edges of $G_{b_2}$ incident to $c$.
\end{proof}


Refer to Figure~\ref{fig:bc-tree}. Consider any edge $bc$ of $T$, where $b$ is a B-node and $c$ is a C-node. The removal of $bc$ splits $T$ into two trees. Let $T_{b\rightarrow c}$ be the one containing $b$. Let $G_{b\rightarrow c}$ be the subgraph of $G$ composed of the blocks corresponding to B-nodes in $T_{b\rightarrow c}$. Let $\chi_{b\rightarrow c}$ be the restriction of $\chi$ to the vertices of $G_{b\rightarrow c}$.

\begin{figure}[htb]\tabcolsep=4pt
	\centering
	\begin{tabular}{c c}
		\includegraphics[scale=0.6]{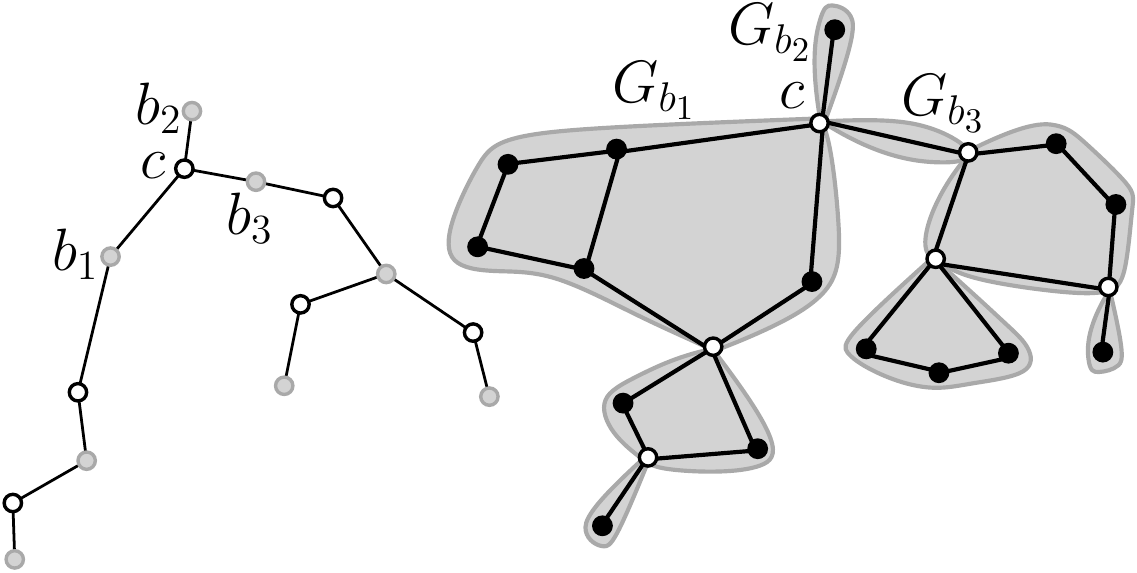} \hspace{3mm} &
		\includegraphics[scale=0.6]{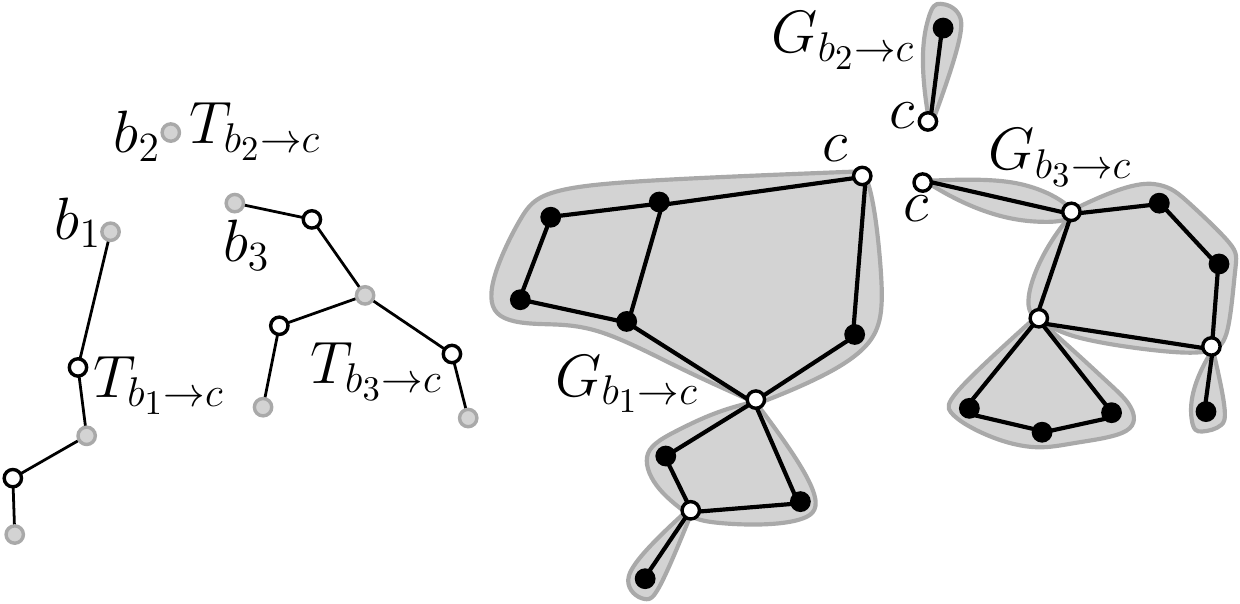} \\
		(a) \hspace{3mm} & (b)\\
	\end{tabular}
	\caption{(a) A graph $G$ and its block-cut-vertex tree $T$. The cut-vertices of $G$ are represented by empty disks; the blocks of $G$ are surrounded by gray regions. (b) The graphs $G_{b_1\rightarrow c}$, $G_{b_2\rightarrow c}$, and $G_{b_3\rightarrow c}$ and the trees $T_{b_1\rightarrow c}$, $T_{b_2\rightarrow c}$, and $T_{b_3\rightarrow c}$, where $c$ is a cut-vertex of $G$ and $b_1$, $b_2$, and $b_3$ are its adjacent nodes in $T$.}
	\label{fig:bc-tree}
\end{figure}

We are going to present an algorithm that computes, for every edge $bc$ of $T$ where $b$ is a B-node and $c$ is a C-node, a set $\mathcal N_{b\rightarrow c}$ that contains all the values $\mu\in \{0\degree,90\degree,180\degree,270\degree\}$ such that $G_{b\rightarrow c}$ admits a rectilinear representation $(\mathcal E_{b\rightarrow c}, \phi_{b\rightarrow c})$ in which $c$ is incident to the outer face of $\mathcal E_{b\rightarrow c}$ and the sum of the internal angles at $c$ is equal to $\mu$, i.e., $\phi^{\mathrm{int}}_{b\rightarrow c}(c)=\mu$.

The following definition lies at the core of our algorithm.

\begin{definition}
	Let $b$ be a B-node of $T$, let $c_i$ be a C-node adjacent to $b$, and let $b_{i,1},\dots,b_{i,m(i)}$ be the B-nodes adjacent to $c_i$ and different from $b$. 
	
	We say that $c_i$ is a \emph{friendly neighbor} of $b$ if, for every $j=1,\dots,m(i)$, we have that $\mathcal N_{b_{i,j}\rightarrow c_i}\cap \{0\degree,90\degree,180\degree\}\neq \emptyset$ and $G_b$ is trivial, or we have that $\mathcal N_{b_{i,j}\rightarrow c_i}\cap \{0\degree,90\degree\}\neq \emptyset$ and $G_b$ is non-trivial.
	
	We say that $c_i$ is an \emph{unfriendly neighbor} of $b$ if, for some $j\in\{1,\dots,m(i)\}$, we have that $\mathcal N_{b_{i,j}\rightarrow c_i}\cap \{0\degree,90\degree,180\degree\}=\emptyset$ and $G_b$ is trivial, or that $\mathcal N_{b_{i,j}\rightarrow c_i}\cap \{0\degree,90\degree\}=\emptyset$ and $G_b$ is non-trivial. 
\end{definition}


Determining the set $\mathcal N_{b\rightarrow c}$ for every edge $bc$ of $T$ is sufficient for determining whether $G$ admits a rectilinear representation, due to the following main structural lemma.

\begin{lemma} \label{le:variable-simply-central-block}
	We have that $G$ admits a rectilinear representation if and only if there exists a B-node $b^*$ in $T$ such~that:
	\begin{itemize}
		\item if $G_{b^*}$ is non-trivial, then it admits a $\chi_{b^*}$-constrained representation; and
		\item every neighbor of $b^*$ in $T$ is friendly.
	\end{itemize}
	Further, assume that we are given a $\chi_{b^*}$-constrained representation of $G_{b^*}$ and, for each $i=1,\dots,h$ and $j=1,\dots,m(i)$, a rectilinear representation $(\mathcal E_{i,j},\phi_{i,j})$ of $G_{b_{i,j}\rightarrow c_i}$ such that $c_i$ is incident to $f^*_{\mathcal E_{i,j}}$ and such that $\phi^{\mathrm{int}}_{i,j}(c_i)\in \{0\degree,90\degree,180\degree\}$ if $G_{b^*}$ is trivial, while $\phi^{\mathrm{int}}_{i,j}(c_i)\in \{0\degree,90\degree\}$ if $G_{b^*}$ is non-trivial. Then it is possible to construct a rectilinear representation of $G$ in $O(n)$ time.
\end{lemma}

\begin{proof}
	$(\Longrightarrow)$ We first prove the necessity of the characterization. Consider any rectilinear representation $(\mathcal E,\phi)$ of $G$. Let $e^*$ be any edge incident to $f^*_{\mathcal E}$ and let $b^*$ be the B-node of $T$ such that $G_{b^*}$ contains $e^*$. Let $c_1,\dots,c_h$ be the C-nodes adjacent to $b^*$ in $T$ and, for each $i=1,\dots,h$, let $b_{i,1},\dots,b_{i,m(i)}$ be the B-nodes different from $b^*$ and adjacent to $c_i$ in $T$. We distinguish two cases.
	
	\begin{itemize}
		\item Suppose first that $G_{b^*}$ is trivial; hence, $G_{b^*}=e^*$. Consider any $i\in \{1,\dots,h\}$ and any $j\in \{1,\dots,m(i)\}$. The restriction of $(\mathcal E,\phi)$ to $G_{b_{i,j}\rightarrow c_i}$ is a rectilinear representation $(\mathcal E_{i,j},\phi_{i,j})$ of $G_{b_{i,j}\rightarrow c_i}$. Since $e^*$ is incident to $f^*_{\mathcal E}$ and $c_i$ is one of the end-vertices of $e^*$ (given that $G_{b^*}=e^*$), it follows that $c_i$ is incident to $f^*_{\mathcal E_{i,j}}$. Since $e^*$ lies in $f^*_{\mathcal E_{i,j}}$, we have that $\phi_{i,j}(c_i,f^*_{\mathcal E_{i,j}}) \geq 180\degree$, and hence $\phi_{i,j}^{\mathrm{int}}(c_i)\in \{0\degree,90\degree,180\degree\}$. This implies that $\mathcal N_{b_{i,j}\rightarrow c_i} \cap \{0\degree,90\degree,180\degree\}\neq \emptyset$ and hence that every neighbor of $b^*$ in $T$ is friendly.
		\item Suppose next that $G_{b^*}$ is non-trivial. By Lemma~\ref{le:chi-necessary}, the restriction of $(\mathcal E,\phi)$ to $G_{b^*}$ is a $\chi_{b^*}$-constrained representation. Consider any $i\in \{1,\dots,h\}$ and any $j\in \{1,\dots,m(i)\}$. The restriction of $(\mathcal E,\phi)$ to $G_{b_{i,j}\rightarrow c_i}$ is a rectilinear representation $(\mathcal E_{i,j},\phi_{i,j})$ of $G_{b_{i,j}\rightarrow c_i}$. Since $e^*$ is incident to $f^*_{\mathcal E}$, the planarity of $\mathcal E$ implies that $c_i$ is incident to $f^*_{\mathcal E_{i,j}}$; indeed, if $c_i$ were not incident to $f^*_{\mathcal E_{i,j}}$, then any path in $G_{b^*}$ from $c_i$ to an end-vertex of $e^*$ would cross in $\mathcal E$ the walk delimiting $f^*_{\mathcal E_{i,j}}$. Since all the (at least two) edges of $G_{b^*}$ incident $c_i$ lie in $f^*_{\mathcal E_{i,j}}$, we have that $\phi_{i,j}(c_i,f^*_{\mathcal E_{i,j}}) \geq 270\degree$, and hence $\phi_{i,j}^{\mathrm{int}}(c_i)\in \{0\degree,90\degree\}$. This implies that $\mathcal N_{b_{i,j}\rightarrow c_i} \cap \{0\degree,90\degree\}\neq \emptyset$ and hence that every neighbor of $b^*$ in $T$ is friendly.
	\end{itemize}
	
	$(\Longleftarrow)$ We now prove the sufficiency of the characterization and show that it gives rise to an $O(n)$-time testing and embedding algorithm. For $i=1,\dots,h$ and $j=1,\dots,m(i)$, let $(\mathcal E_{i,j},\phi_{i,j})$ be a rectilinear representation of $G_{b_{i,j}\rightarrow c_i}$ such that $c_i$ is incident to $f^*_{\mathcal E_{i,j}}$ and such that $\phi^{\mathrm{int}}_{i,j}(c_i)\in \{0\degree,90\degree,180\degree\}$ if $G_{b^*}$ is trivial, while $\phi^{\mathrm{int}}_{i,j}(c_i)\in \{0\degree,90\degree\}$ if $G_{b^*}$ is non-trivial; this exists since every neighbor of $b^*$ is friendly. Further, if $G_{b^*}$ is non-trivial, let $(\mathcal E_{b^*},\phi_{b^*})$ be a $\chi_{b^*}$-constrained representation of $G_{b^*}$.
	
	We start by defining a sequence of graphs $G^+_0,G^+_1,\dots,G^+_h$. First, we have $G^+_0:=G_{b^*}$. Further, for each $i=1,\dots,h$, we have $G^+_i:=G^+_{i-1} \cup \left(\bigcup_{j=1}^{m(i)} G_{b_{i,j}\rightarrow c_i}\right)$. That is, $G^+_i$ consists of $G_{b^*}$ and of all the graphs $G_{b_{x,j}\rightarrow c_x}$ with $1\leq x \leq i$ and $1\leq j\leq m(x)$. Note that the graphs $G^+_{i-1},G_{b_{i,1}\rightarrow c_i},\dots,G_{b_{i,m(i)}\rightarrow c_i}$ share $c_i$ and no other vertex or edge, and that $c_i$ is not a cut-vertex in each of $G^+_{i-1},G_{b_{i,1}\rightarrow c_i},\dots,G_{b_{i,m(i)}\rightarrow c_i}$. Observe that $G^+_h=G$.
	
	In order to define a rectilinear representation $(\mathcal E,\phi)$ of $G$, we actually define a rectilinear representation $(\mathcal E^+_i,\phi^+_i)$ of each graph $G^+_{i}$, where $(\mathcal E^+_i,\phi^+_i)$ is a join of $(\mathcal E^+_{i-1},\phi^+_{i-1})$ and of $(\mathcal E_{i,1},\phi_{i,1}),\dots,(\mathcal E_{i,m(i)},\phi_{i,m(i)})$. For $j=0,\dots,m(i)$, our construction  guarantees that every edge of $G_{b^*}$ incident to $f^*_{\mathcal E_{b^*}}$ is also incident to $f^*_{\mathcal E^+_i}$; that is, $\mathcal E_{b^*}$ lies ``outside'' each of $\mathcal E_{i,1},\dots,\mathcal E_{i,h}$ in $\mathcal E^+_i$. 
	
	
	First, we define $\mathcal E^+_0=\mathcal E_{b^*}$. Further, if $G_{b^*}$ is non-trivial, we set $\phi^+_0=\phi_{b^*}$; that is, $\phi^+_0(w,f)=\phi_{b^*}(w,f)$, for each vertex $w$ of $G_{b^*}$ incident to a face $f$ of $\mathcal E^+_0=\mathcal E_{b^*}$. If $G_{b^*}$ is trivial, then we initialize $\phi^+_0(w_1,f^*_{\mathcal E^+_0})=360\degree$ and $\phi^+_0(w_2,f^*_{\mathcal E^+_0})=360\degree$, where $w_1$ and $w_2$ are the end-vertices of the edge $G_{b^*}$. 
	
	We now describe how to construct $(\mathcal E^+_i,\phi^+_i)$, for each $i=1,\dots,h$, as a join of the rectilinear representations $(\mathcal E^+_{i-1},\phi^+_{i-1}),(\mathcal E_{i,1},\phi_{i,1}),\dots,(\mathcal E_{i,m(i)},\phi_{i,m(i)})$. By Lemma~\ref{le:preliminaries-subgraphs-composition}, this ensures that $(\mathcal E^+_i,\phi^+_i)$ is a rectilinear representation of $G^+_i$. For ease of notation, let $G_{b_{i,0}\rightarrow c_i}:=G^+_{i-1}$, $\mathcal E_{i,0}:=\mathcal E^+_{i-1}$, and $\phi_{i,0}:=\phi^+_{i-1}$. For a face $f$ of $\mathcal E^+_i$, denote by $f_{i,j}$ the corresponding face of $\mathcal E_{i,j}$, for each $j=0,1,\dots,m(i)$. 
	
	In order to satisfy Property~(d) of a join, we insist that the restriction of $(\mathcal E^+_i,\phi^+_i)$ to  $G_{b_{i,j}\rightarrow c_i}$ coincides with $(\mathcal E_{i,j},\phi_{i,j})$, for each $j=0,\dots,m(i)$. This constraint is satisfied easily as far as each vertex $w\neq c_i$ of $G^+_i$ is concerned; indeed, $w$ belongs to a single graph $G_{b_{i,j}\rightarrow c_i}$, hence we simply preserve in $(\mathcal E^+_i,\phi^+_i)$ its clockwise order of the incident edges and its incident angles. More precisely, for every vertex $w\neq c_i$ that belongs to $G_{b_{i,j}\rightarrow c_i}$, we set the clockwise order of the edges incident to $w$ in $\mathcal E^+_i$ to be the same as in $\mathcal E_{i,j}$ and, for every face $f$ of $\mathcal E^+_i$ incident to $w$ and every occurrence $w^x$ of $w$ along the boundary of $f$, we set $\phi^+_i(w^x,f)=\phi_{i,j}(w^x,f_{i,j})$. It follows that, in order to  fully specify $(\mathcal E^+_i,\phi^+_i)$, we only need to:
	
	\begin{itemize} 
		\item fix the clockwise order of the edges incident to $c_i$ and the walk delimiting the outer face of $\mathcal E^+_i$, so that Property~(a) of a join is satisfied, i.e., $\mathcal E^+_i$ is a plane embedding; and 
		\item fix an angle $\phi^+_i(c_i^x,f)$, for each face $f$ of $\mathcal E^+_i$ incident to $c_i$ and each occurrence $c_i^x$ of $c_i$ along the boundary of $f$, so that: 
		\begin{itemize}
			\item Property~(b) of a join is satisfied, i.e., $\phi^+_i(c_i^x,f)\in \{90\degree,180\degree,270\degree\}$,
			\item Property~(c) of a join is satisfied, i.e., $\sum_{f,x} \phi^+_i(c_i^x,f)=360\degree$, where the sum is over all the faces $f$ of $\mathcal E^+_i$ incident to $c_i$ and over all the occurrences $c_i^x$ of $c_i$ along the boundary of $f$, and
			\item Property~(d) of a join is satisfied, i.e., for $j=0,\dots,m(i)$ and for each face $f_{i,j}$ of $\mathcal E_{i,j}$ incident to $c_i$, we have $\phi_{i,j}(c_i,f_{i,j})=\sum_{f,x} \phi^+_i(c_i^x,f)$, where the sum is over all the faces $f$ of $\mathcal E^+_i$ incident to $c_i$ whose corresponding face in $\mathcal E_{i,j}$ is $f_{i,j}$ and over all the occurrences $c_i^x$ of $c_i$ along the boundary of~$f$. 
		\end{itemize} 
	\end{itemize} 
	
	We distinguish four cases according to the arrangement of blocks around $c_i$. 
	
	\begin{enumerate}[(1)]
		\item Suppose first that all the blocks of $G^+_i$ containing $c_i$ (including $G_{b^*}$) are non-trivial; refer to Figure~\ref{fig:join-final-2}. Then we have $m(i)\geq 1$, since $c_i$ is a cut-vertex of $G$, and $m(i)\leq 1$, since $c_i$ has maximum degree $4$ in $G$; it follows that $m(i)=1$. For $j=0,1$, let $e_{j,a}$ and $e_{j,b}$ be the edges of $G_{b_{i,j}\rightarrow c_i}$ incident to $c_i$. 
		
		\begin{figure}[tb]\tabcolsep=4pt
			\centering
			\begin{tabular}{c c c}
				\includegraphics[scale=0.7]{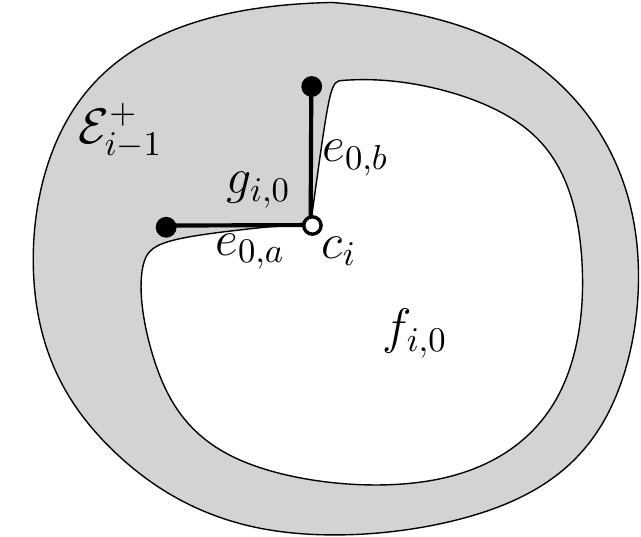} \hspace{3mm} &
				\includegraphics[scale=0.7]{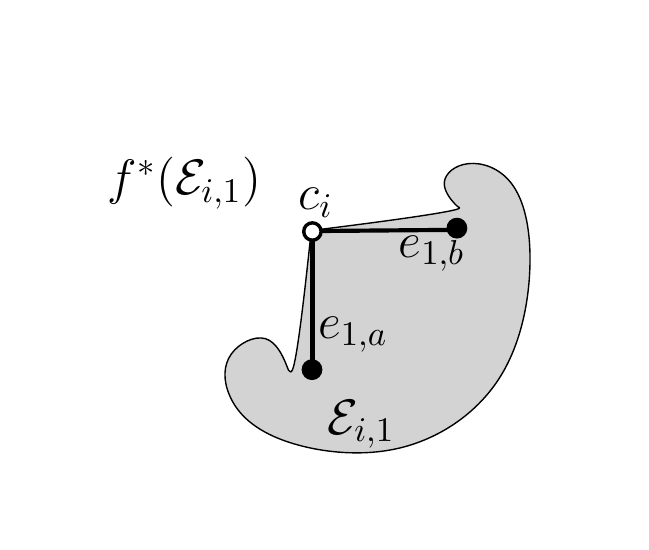} \hspace{3mm} &
				\includegraphics[scale=0.7]{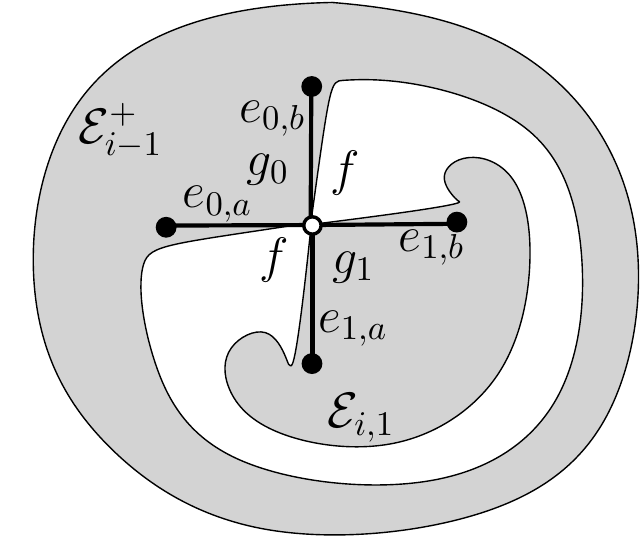} \\
				(a) \hspace{3mm} & (b) \hspace{3mm} & (c)\\
			\end{tabular}
			\caption{A join of the rectilinear representations $(\mathcal E^+_{i-1},\phi^+_{i-1})$ and $(\mathcal E_{i,1},\phi_{i,1})$ in Case~(1). Only the vertices adjacent to $c_i$ and the edges incident to $c_i$ are shown; the rest of $\mathcal E^+_{i-1}$ and $\mathcal E_{i,1}$ is hidden inside gray regions. The vertex $c_i$ is an empty disk. (a) $(\mathcal E^+_{i-1},\phi^+_{i-1})$; in this illustration $f_{i,0}$ is an internal face of $\mathcal E_{i,0}$. (b) $(\mathcal E_{i,1},\phi_{i,1})$. (c) The join $(\mathcal E^+_i,\phi^+_i)$ of $(\mathcal E^+_{i-1},\phi^+_{i-1})$ and $(\mathcal E_{i,1},\phi_{i,1})$.}
			\label{fig:join-final-2}
		\end{figure}
		
		Recall that $c_i$ is incident to $f^*_{\mathcal E_{i,1}}$ and that $\phi^{\mathrm{int}}_{i,1}(c_i)=90\degree$; assume, w.l.o.g.\ up to a relabeling, that $e_{1,a}$ immediately precedes $e_{1,b}$ in the clockwise order of the edges along the boundary of $f^*_{\mathcal E_{i,1}}$. 
		
		On the other hand, $c_i$ might be incident to $f^*_{\mathcal E_{i,0}}$ or not; regardless, there is a face $f_{i,0}$ of $\mathcal E_{i,0}$ incident to $c_i$ such that $\phi_{i,0}(c_i,f_{i,0})=270\degree$. This comes from the following three observations. First, $c_i\in \chi$, since $c_i$ is incident to two non-trivial blocks of $G$; hence, if $f_{b^*}$ and $g_{b^*}$ denote the faces of $\mathcal E_{b^*}$ incident to $c_i$, we have $\phi_{b^*}(c_i,f_{b^*})=90\degree$ and $\phi_{b^*}(c_i,g_{b^*})=270\degree$, or vice versa, given that $(\mathcal E_{b^*}, \phi_{b^*})$ is a $\chi_{b^*}$-constrained representation of $G_{b^*}$. Second, the restriction of $(\mathcal E_{i,0},\phi_{i,0})$ to $G_{b^*}$ coincides with $(\mathcal E_{b^*}, \phi_{b^*})$; this comes from repeated applications of Property~(d) of a join, applied to the sequence of rectilinear representations $(\mathcal E_{b^*}, \phi_{b^*})=(\mathcal E^+_0,\phi^+_0),(\mathcal E^+_1,\phi^+_1),\dots,(\mathcal E^+_{i-1},\phi^+_{i-1})=(\mathcal E_{i,0},\phi_{i,0})$. Third, 
		both $G_{b^*}$ and $G_{b_{i,0}\rightarrow c_i}$ have two edges incident to $c_i$, namely $e_{0,a}$ and $e_{0,b}$, hence there is a unique face of $\mathcal E_{i,0}$ whose corresponding face in $\mathcal E_{b^*}$ is $f_{b^*}$ and there is a unique face of $\mathcal E_{i,0}$ whose corresponding face in $\mathcal E_{b^*}$ is $g_{b^*}$. 
		
		Assume, w.l.o.g.\ up to a relabeling, that traversing $e_{0,a}$ towards $c_i$ and then $e_{0,b}$ away from $c_i$ leaves the face $f_{i,0}$ of $\mathcal E_{i,0}$ such that $\phi_{i,0}(c_i,f_{i,0})=270\degree$ to the right. 
		
		We fix the clockwise order of the edges incident to $c_i$ in $\mathcal E^+_i$ to be $e_{0,a},e_{0,b},e_{1,b},e_{1,a},e_{0,a}$ (thus $G_{b_{i,1}\rightarrow c_i}$ lies ``inside'' $f_{i,0}$). If $f_{i,0}$ is an internal face of $\mathcal E_{i,0}$, then the walk delimiting $f^*_{\mathcal E^+_i}$ coincides with the walk delimiting $f^*_{\mathcal E_{i,0}}$, otherwise it is composed of the walks delimiting $f^*_{\mathcal E_{i,0}}$ and $f^*_{\mathcal E_{i,1}}$, merged so that the edge $e_{0,b}$ enters $c_i$ right before the edge $e_{1,b}$ leaves $c_i$, and the edge $e_{1,a}$ enters $c_i$ right before the edge $e_{0,a}$ leaves $c_i$. This completes the definition of $\mathcal E^+_i$, which is a plane embedding, given that $\mathcal E_{i,0}$ and $\mathcal E_{i,1}$ are plane embeddings and given that $c_i$ is incident to $f^*_{\mathcal E_{i,1}}$.
		
		Let $f$ be the face of $\mathcal E^+_i$ that is incident to the edges $e_{0,a}$, $e_{0,b}$, $e_{1,b}$, and $e_{1,a}$, let $g_0$ be the face of $\mathcal E^+_i$ that is incident to $e_{0,a}$ and $e_{0,b}$ and not to $e_{1,a}$ and $e_{1,b}$, and finally let $g_1$ be the face of $\mathcal E^+_i$ that is incident to $e_{1,a}$ and $e_{1,b}$ and not to $e_{0,a}$ and $e_{0,b}$. Let $c^0_i$, $c^1_i$, $c^2_i$, and $c^3_i$ be the occurrences of $c_i$ along the boundaries of faces of $\mathcal E^+_i$ incident to $c_i$, where $c^0_i$ and $c^1_i$ are incident to $f$, $c^2_i$ is incident to $g_0$, and $c^3_i$ is incident to $g_1$. Then we set $\phi^+_i(c^0_i,f)=\phi^+_i(c^1_i,f)=\phi^+_i(c^2_i,g_0)=\phi^+_i(c^3_i,g_1)=90\degree$. This choice directly ensures that Properties~(b) and~(c) of a join are satisfied. Concerning Property~(d), recall that $\phi_{i,0}(c_i,f_{i,0})=270\degree$, hence $\phi_{i,0}(c_i,g_{i,0})=90\degree$, where $g_{i,0}$ is the face of $\mathcal E_{i,0}$ incident to $c_i$ and different from $f_{i,0}$. Since $g_0$ is the only face of $\mathcal E^+_i$ whose corresponding face of $\mathcal E_{i,0}$ is $g_{i,0}$, while $f_{i,0}$ is the corresponding face of $f$ and $g_1$, we have $\phi_{i,0}(c_i,g_{i,0})=\phi^+_i(c^2_i,g_0)=90\degree$ and $\phi_{i,0}(c_i,f_{i,0})=\phi^+_i(c^0_i,f)+\phi^+_i(c^1_i,f)+\phi^+_i(c^3_i,g_1)=270\degree$. An analogous proof can be stated for the faces of $\mathcal E_{i,1}$, given that $\phi_{i,1}(c_i,f^*_{\mathcal E_{i,1}})=270\degree$, as a consequence of $\phi^{\mathrm{int}}_{i,1}(c_i)=90\degree$. Property~(d) of a join follows.
		
		\item Suppose next that all the blocks of $G^+_i$ containing $c_i$ (including $G_{b^*}$) are trivial; see Figure~\ref{fig:join-final}(a). For $j=0,\dots,m(i)$, let $e_j$ be the edge of $G_{b_{i,j}\rightarrow c_i}$ incident to $c_i$. We fix the clockwise order of the edges incident to $c_i$ in $\mathcal E^+_i$ to be $e_0,e_1,\dots,e_{m(i)},e_{m(i)+1}:=e_0$. The walk delimiting the outer face of $\mathcal E^+_i$ is composed of the walks delimiting the outer faces of $\mathcal E_{i,0},\dots,\mathcal E_{i,m(i)}$, merged so that the edge $e_j$ enters $c_i$ right before the edge $e_{j+1}$ leaves $c_i$, for $j=0,\dots,m(i)$ where the indices are modulo $m(i)+1$. This completes the definition of $\mathcal E^+_i$, which is a plane embedding, given that each of $\mathcal E_{i,0},\dots,\mathcal E_{i,m(i)}$ is a plane embedding and given that $c_i$ is incident to the outer face of each of $\mathcal E_{i,1},\dots,\mathcal E_{i,m(i)}$, by assumption.
		
		\begin{figure}[tb]\tabcolsep=4pt
			\centering
			\begin{tabular}{c c c}
				\includegraphics[scale=0.7]{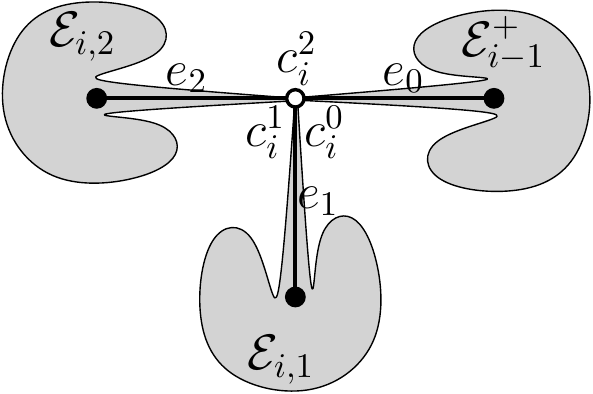} \hspace{3mm} &
				\includegraphics[scale=0.7]{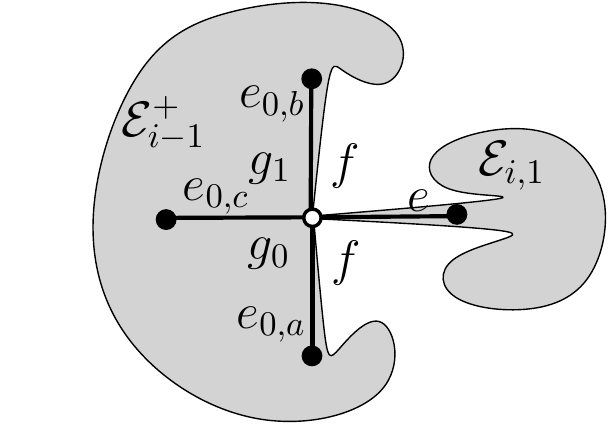} \hspace{3mm} &
				\includegraphics[scale=0.7]{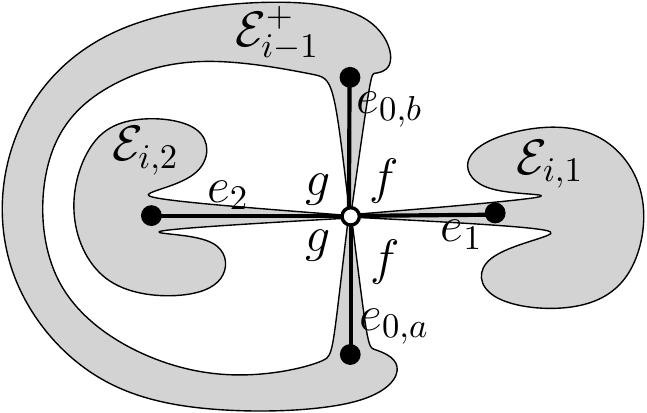} \\
				(a) \hspace{3mm} & (b) \hspace{3mm} & (c)\\
			\end{tabular}
			\caption{The join of the rectilinear representations $(\mathcal E^+_{i-1},\phi^+_{i-1})$ and $(\mathcal E_{i,1},\phi_{i,1}),\dots,(\mathcal E_{i,m(i)},\phi_{i,m(i)})$. Only the vertices adjacent to $c_i$ and the edges incident to $c_i$ are shown; the rest of $\mathcal E^+_{i-1},\mathcal E_{i,1},\dots,\mathcal E_{i,m(i)}$ is hidden inside gray regions. The vertex $c_i$ is an empty disk. Figures (a), (b), and (c) illustrate Cases~(2),~(3.1), and~(3.3.2), respectively.}
			\label{fig:join-final}
		\end{figure}
		
		Note that $c_i$ is incident to a single face $f$ of $\mathcal E^+_i$, namely to its outer face $f:=f^*_{\mathcal E^+_i}$; this follows from the fact that every block of $G^+_i$ incident to $c_i$ is trivial. For $x=0,\dots,m(i)$, let $c_i^x$ be the occurrence of $c_i$ along the boundary of $f$ incident to the edges $e_x$ and $e_{x+1}$. Then we set $\phi^+_i(c^x_i,f)=90\degree$, for $x=0,\dots,m(i)-1$, and $\phi^+_i(c^{m(i)}_i,f)=360\degree-\sum_{x=0}^{m(i)-1}\phi^+_i(c^x_i,f)$. 
		
		By construction, $\phi^+_i(c^x_i,f)=90\degree$, for $x=0,\dots,m(i)-1$. Further, we have $m(i)\geq 1$, since $c_i$ is a cut-vertex of $G$, and $m(i)\leq 3$, since $c_i$ has maximum degree $4$ in $G$. Hence, $\phi^+_i(c_i^{m(i)},f)\in \{90\degree,180\degree,270\degree\}$ and Property~(b) of a join is satisfied. Properties~(c) and~(d) of a join are trivially satisfied by the chosen values for the angles incident to $c_i$. In particular, for $j=0,\dots,m(i)$, there is a unique face $f_{i,j}$ of $\mathcal E_{i,j}$ incident to $c_i$ and $\phi_{i,j}(c_i,f_{i,j})=360\degree$, given that $c_i$ has degree $1$ in $G_{b_{i,j}\rightarrow c_i}$; further, $f_{i,j}$ corresponds to $f$ and $\sum_{x}\phi^+_i(c^x_i,f)=360\degree$, by construction.
		
		\item Suppose next that $G_{b^*}$ is non-trivial, while every other block of $G^+_i$ containing $c_i$ is trivial. Then we have $m(i)\geq 1$, since $c_i$ is a cut-vertex of $G$, and $m(i)\leq 2$, since $c_i$ has maximum degree $4$ in $G$. Note that $G_{b_{i,0}\rightarrow c_i}$ contains at least $2$ and at most $4-m(i)$ edges incident to $c_i$.
		
		\begin{enumerate}[({3.}1)]
			\item If $m(i)=1$ and $G_{b_{i,0}\rightarrow c_i}$ contains $3$ edges incident to $c_i$, as in Figure~\ref{fig:join-final}(b), then there is a face $f_{i,0}$ of $\mathcal E_{i,0}$ such that $\phi_{i,0}(c_i,f_{i,0})= 180\degree$; this trivially follows from the fact that $G_{b_{i,0}\rightarrow c_i}$ contains $3$ edges incident to $c_i$. Let $e_{0,a}$, $e_{0,b}$, $e_{0,c}$, $e_{0,a}$ be the clockwise order of the edges incident to $c_i$ in $\mathcal E_{i,0}$. Let $e_{0,a}$ and $e_{0,b}$ be the edges incident to $f_{i,0}$; assume, w.l.o.g.\ up to a relabeling, that traversing $e_{0,a}$ towards $c_i$ and then $e_{0,b}$ away from $c_i$ leaves $f_{i,0}$ to the right. Finally, let $e$ be the edge of $G_{b_{i,1}\rightarrow c_i}$ incident to $c_i$.
			
			We fix the clockwise order of the edges incident to $c_i$ in $\mathcal E^+_i$ to be $e_{0,a},e_{0,c},e_{0,b},e,e_{0,a}$ (thus $G_{b_{i,1}\rightarrow c_i}$ lies ``inside'' $f_{i,0}$). If $f_{i,0}$ is an internal face of $\mathcal E_{i,0}$, then the walk delimiting $f^*_{\mathcal E^+_i}$ coincides with the walk delimiting $f^*_{\mathcal E_{i,0}}$, otherwise it is composed of the walks delimiting $f^*_{\mathcal E_{i,0}}$ and $f^*_{\mathcal E_{i,1}}$, merged so that the edge $e_{0,b}$ enters $c_i$ right before the edge $e$ leaves $c_i$, and the edge $e$ enters $c_i$ right before the edge $e_{0,a}$ leaves $c_i$. This completes the definition of $\mathcal E^+_i$, which is a plane embedding, given that $\mathcal E_{i,0}$ and $\mathcal E_{i,1}$ are plane embeddings and given that $c_i$ is incident to the outer face of $\mathcal E_{i,1}$, by assumption.
			
			Let $f$ be the face of $\mathcal E^+_i$ that is incident to $e$, let $g_0$ be the face of $\mathcal E^+_i$ that is incident to $e_{0,a}$ and $e_{0,c}$, and let $g_1$ be the face of $\mathcal E^+_i$ that is incident to $e_{0,b}$ and $e_{0,c}$. Let $c^0_i$, $c^1_i$, $c^2_i$, and $c^3_i$ be the occurrences of $c_i$ along the boundaries of faces of $\mathcal E^+_i$ incident to $c_i$, where $c^0_i$ and $c^1_i$ are incident to $f$, $c^2_i$ is incident to $g_0$, and $c^3_i$ is incident to $g_1$. Then we set $\phi^+_i(c^0_i,f)=\phi^+_i(c^1_i,f)=\phi^+_i(c^2_i,g_0)=\phi^+_i(c^3_i,g_1)=90\degree$. This choice directly ensures that Properties~(b) and~(c) of a join are satisfied. Concerning Property~(d), recall that $\phi_{i,0}(c_i,f_{i,0})=180\degree$, hence $\phi_{i,0}(c_i,g_{i,0})=90\degree$ and $\phi_{i,0}(c_i,g_{i,1})=90\degree$, where $g_{i,0}$ and $g_{i,1}$ are the faces of $\mathcal E_{i,0}$ corresponding to $g_0$ and $g_1$, respectively. Then we have $\phi_{i,0}(c_i,f_{i,0})=\phi^+_i(c^0_i,f)+\phi^+_i(c^1_i,f)=180\degree$, $\phi_{i,0}(c_i,g_{i,0})=\phi^+_i(c^2_i,g_0)=90\degree$, and $\phi_{i,0}(c_i,g_{i,1})=\phi^+_i(c^3_i,g_1)=90\degree$. Further, $\mathcal E_{i,1}$ has a unique face incident to $c_i$, namely $f^*_{{\mathcal E_{i,1}}}$, which corresponds to $f$, $g_0$, and $g_1$. Thus, we have $\phi_{i,1}(c_i,f^*_{{\mathcal E_{i,1}}})=\phi^+_i(c^0_i,f)+\phi^+_i(c^1_i,f)+\phi^+_i(c^2_i,g_0)+\phi^+_i(c^3_i,g_1)=360\degree$. Property~(d) of a join follows.
			
			\item If $m(i)=1$ and $G_{b_{i,0}\rightarrow c_i}$ contains $2$ edges incident to $c_i$, then we can act similarly to Case (3.1), with some differences. First, $\mathcal E_{i,0}$ contains a face $f_{i,0}$ such that $\phi_{i,0}(c_i,f_{i,0})= 180\degree$ or $\phi_{i,0}(c_i,f_{i,0})= 270\degree$, given that $G_{b_{i,0}\rightarrow c_i}$ contains $2$ edges incident to $c_i$. The edges $e_{0,a}$, $e_{0,b}$, and $e$ are defined as in Case (3.1), where $f_{i,0}$ is to the right when traversing $e_{0,a}$ towards $c_i$ and then $e_{0,b}$ away from $c_i$. Let $g_{i,0}$ be the face of $\mathcal E_{i,0}$ incident to $c_i$ and different from $f_{i,0}$. 
			
			We fix the clockwise order of the edges incident to $c_i$ in $\mathcal E^+_i$ to be $e_{0,a},e_{0,b},e,e_{0,a}$ (thus $G_{b_{i,1}\rightarrow c_i}$ lies ``inside'' $f_{i,0}$) and we define the walk delimiting $f^*_{\mathcal E^+_i}$ as in Case (3.1), thus obtaining a plane embedding $\mathcal E^+_i$.
			
			The face $f$ of $\mathcal E^+_i$ is defined as in Case (3.1), however $\mathcal E^+_i$ now only contains one more face incident to $c_i$, which we denote by $g$; such a face is incident to $e_{0,a}$ and $e_{0,b}$, however not to $e$. Let $c^0_i$, $c^1_i$, and $c^2_i$ be the occurrences of $c_i$ along the boundaries of faces of $\mathcal E^+_i$ incident to $c_i$, where $c^0_i$ and $c^1_i$ are incident to $f$ and $c^2_i$ is incident to $g$. Then we set $\phi^+_i(c^0_i,f)=90\degree$, $\phi^+_i(c^1_i,f)=\phi_{i,0}(c_i,f_{i,0})-90\degree$, and $\phi^+_i(c^2_i,g)=\phi_{i,0}(c_i,g_{i,0})$. Since $\phi_{i,0}(c_i,f_{i,0})\in \{180\degree,270\degree\}$, we have $\phi^+_i(c^1_i,f)\in \{90\degree,180\degree\}$. Since $(\mathcal E_{i,0},\phi_{i,0})$ is a rectilinear representation, we have $\phi^+_i(c^2_i,g)=\phi_{i,0}(c_i,g_{i,0})=360\degree-\phi_{i,0}(c_i,f_{i,0})\in \{90\degree,180\degree\}$; Property~(b) of a join follows. Again since $(\mathcal E_{i,0},\phi_{i,0})$ is a rectilinear representation, we have $\phi^+_i(c^0_i,f)+\phi^+_i(c^1_i,f)+\phi^+_i(c^2_i,g)=\phi_{i,0}(c_i,f_{i,0})+\phi_{i,0}(c_i,g_{i,0})=360\degree$ and Property~(c) of a join follows. By construction and since $f_{i,0}$ and $g_{i,0}$ correspond to $f$ and $g$, respectively,  we have $\phi_{i,0}(c_i,f_{i,0})=\phi^+_i(c^0_i,f)+\phi^+_i(c^1_i,f)$ and $\phi_{i,0}(c_i,g_{i,0})=\phi^+_i(c^2_i,g)$. Further, $\mathcal E_{i,1}$ has a unique face incident to $c_i$, namely $f^*_{{\mathcal E_{i,1}}}$, which corresponds to $f$ and $g$. Thus, we have $\phi_{i,1}(c_i,f^*_{{\mathcal E_{i,1}}})=\phi^+_i(c^0_i,f)+\phi^+_i(c^1_i,f)+\phi^+_i(c^2_i,g)=360\degree$. Property~(d) of a join follows.
			
			\item If $m(i)=2$, then $G_{b_{i,0}\rightarrow c_i}$ contains $2$ edges incident to $c_i$, as the degree of $c_i$ is at most $4$. Let $f_{i,0}$ and $g_{i,0}$ be the faces of $\mathcal E_{i,0}$ incident to $c_i$. Assume, w.l.o.g.\ up to a relabeling, that $\phi_{i,0}(c_i,f_{i,0})\geq \phi_{i,0}(c_i,g_{i,0})$. Since $(\mathcal E_{i,0},\phi_{i,0})$ is a rectilinear representation, we have either $\phi_{i,0}(c_i,f_{i,0})=270\degree$ and $\phi_{i,0}(c_i,g_{i,0})= 90\degree$, which we call Case (3.3.1), or $\phi_{i,0}(c_i,f_{i,0})=\phi_{i,0}(c_i,g_{i,0})= 180\degree$, which we call Case (3.3.2); the latter case is illustrated in Figure~\ref{fig:join-final}(c). Let $e_{0,a}$ and $e_{0,b}$ be the edges of $G_{b_{i,0}\rightarrow c_0}$ incident to $c_i$; assume, w.l.o.g.\ up to a relabeling, that traversing $e_{0,a}$ towards $c_i$ and then $e_{0,b}$ away from $c_i$ leaves $f_{i,0}$ to the right. Finally, let $e_1$ and $e_2$ be the edges of $G_{b_{i,1}\rightarrow c_i}$ and $G_{b_{i,2}\rightarrow c_i}$ incident to $c_i$, respectively.
			
			\begin{enumerate}[({3.3.}1)]
				\item We fix the clockwise order of the edges incident to $c_i$ in $\mathcal E^+_i$ to be $e_{0,a},e_{0,b},e_1,e_2,e_{0,a}$ (thus $G_{b_{i,1}\rightarrow c_i}$ and $G_{b_{i,2}\rightarrow c_i}$ both lie ``inside'' $f_{i,0}$). If $f_{i,0}$ is an internal face of $\mathcal E_{i,0}$, then the walk delimiting $f^*_{\mathcal E^+_i}$ coincides with the walk delimiting $f^*_{\mathcal E_{i,0}}$, otherwise it is composed of the walks delimiting $f^*_{\mathcal E_{i,0}}$, $f^*_{\mathcal E_{i,1}}$, and $f^*_{\mathcal E_{i,2}}$ merged so that the edge $e_{0,b}$ enters $c_i$ right before the edge $e_1$ leaves $c_i$, the edge $e_1$ enters $c_i$ right before the edge $e_2$ leaves $c_i$, and the edge $e_2$ enters $c_i$ right before the edge $e_{0,a}$ leaves $c_i$. This completes the definition of $\mathcal E^+_i$, which is a plane embedding, given that $\mathcal E_{i,0}$, $\mathcal E_{i,1}$, and $\mathcal E_{i,2}$ are plane embeddings and given that $c_i$ is incident to the outer faces of $\mathcal E_{i,1}$ and $\mathcal E_{i,2}$, by assumption.
				
				Let $f$ be the face of $\mathcal E^+_i$ that is incident to the edges $e_{0,a}$, $e_{0,b}$, $e_1$, and $e_2$, and let $g$ be the face of $\mathcal E^+_i$ that is incident to $e_{0,a}$ and $e_{0,b}$ and not to $e_1$ and $e_2$. Let $c^0_i$, $c^1_i$, $c^2_i$, and $c^3_i$ be the occurrences of $c_i$ along the boundaries of faces of $\mathcal E^+_i$ incident to $c_i$, where $c^0_i$, $c^1_i$, and $c^1_2$ are incident to $f$, and $c^3_i$ is incident to $g$. Then we set $\phi^+_i(c^0_i,f)=\phi^+_i(c^1_i,f)=\phi^+_i(c^2_i,f)=\phi^+_i(c^3_i,g)=90\degree$. This choice directly ensures that Properties~(b) and~(c) of a join are satisfied. Concerning Property~(d), recall that $\phi_{i,0}(c_i,f_{i,0})=270\degree$ and $\phi_{i,0}(c_i,g_{i,0})=90\degree$. Since $f_{i,0}$ is the face of $\mathcal E_{i,0}$ corresponding to $f$ and $g_{i,0}$ is the face of $\mathcal E_{i,0}$ corresponding to $g$, we have $\phi_{i,0}(c_i,g_{i,0})=\phi^+_i(c^3_i,g)=90\degree$ and $\phi_{i,0}(c_i,f_{i,0})=\phi^+_i(c^0_i,f)+\phi^+_i(c^1_i,f)+\phi^+_i(c^2_i,f)=270\degree$. Further, $\mathcal E_{i,1}$ has a unique face incident to $c_i$, namely $f^*_{{\mathcal E_{i,1}}}$, which corresponds to $f$ and $g$. Thus, we have $\phi_{i,1}(c_i,f^*_{{\mathcal E_{i,1}}})=\phi^+_i(c^0_i,f)+\phi^+_i(c^1_i,f)+\phi^+_i(c^2_i,f)+\phi^+_i(c^3_i,g)=360\degree$. The argument for $\mathcal E_{i,2}$ is analogous and Property~(d) of a join follows.
				
				\item We fix the clockwise order of the edges incident to $c_i$ in $\mathcal E^+_i$ to be $e_{0,a},e_2,e_{0,b},e_1,e_{0,a}$ (thus $G_{b_{i,1}\rightarrow c_i}$ and $G_{b_{i,2}\rightarrow c_i}$ lie ``inside'' $f_{i,0}$ and $g_{i,0}$, respectively). If $f_{i,0}$ and $g_{i,0}$ are internal faces of $\mathcal E_{i,0}$, then the walk delimiting $f^*_{\mathcal E^+_i}$ coincides with the walk delimiting $f^*_{\mathcal E_{i,0}}$. If $f_{i,0}$ is the outer face of $\mathcal E_{i,0}$, then the walk delimiting $f^*_{\mathcal E^+_i}$ is composed of the walks delimiting $f^*_{\mathcal E_{i,0}}$ and $f^*_{\mathcal E_{i,1}}$, merged so that the edge $e_{0,b}$ enters $c_i$ right before the edge $e_1$ leaves $c_i$, and the edge $e_1$ enters $c_i$ right before the edge $e_{0,a}$ leaves $c_i$. Analogously, if $g_{i,0}$ is the outer face of $\mathcal E_{i,0}$, then the walk delimiting $f^*_{\mathcal E^+_i}$ is composed of the walks delimiting $f^*_{\mathcal E_{i,0}}$ and $f^*_{\mathcal E_{i,2}}$, merged so that the edge $e_{0,a}$ enters $c_i$ right before the edge $e_2$ leaves $c_i$, and the edge $e_2$ enters $c_i$ right before the edge $e_{0,b}$ leaves $c_i$.
				
				Let $f$ be the face of $\mathcal E^+_i$ that is incident to $e_{0,a}$, $e_{0,b}$, and $e_1$, let $g$ be the face of $\mathcal E^+_i$ that is incident to $e_{0,a}$, $e_{0,b}$, and $e_2$. Let $c^0_i$, $c^1_i$, $c^2_i$, and $c^3_i$ be the occurrences of $c_i$ along the boundaries of faces of $\mathcal E^+_i$ incident to $c_i$, where $c^0_i$ and $c^1_i$ are incident to $f$, and $c^2_i$ and $c^3_i$ are incident to $g$. Then we set $\phi^+_i(c^0_i,f)=\phi^+_i(c^1_i,f)=\phi^+_i(c^2_i,g)=\phi^+_i(c^3_i,g)=90\degree$. This choice directly ensures that Properties~(b) and~(c) of a join are satisfied. Concerning Property~(d), recall that $\phi_{i,0}(c_i,f_{i,0})=180\degree$ and $\phi_{i,0}(c_i,g_{i,0})=180\degree$. Since $f_{i,0}$ is the face of $\mathcal E_{i,0}$ corresponding to $f$ and $g_{i,0}$ is the face of $\mathcal E_{i,0}$ corresponding to $g$, we have $\phi_{i,0}(c_i,f_{i,0})=\phi^+_i(c^0_i,f)+\phi^+_i(c^1_i,f)=180\degree$ and $\phi_{i,0}(c_i,g_{i,0})=\phi^+_i(c^2_i,g)+\phi^+_i(c^3_i,g)=180\degree$. Further, $\mathcal E_{i,1}$ has a unique face incident to $c_i$, namely $f^*_{{\mathcal E_{i,1}}}$, which corresponds to $f$ and $g$. Thus, we have $\phi_{i,1}(c_i,f^*_{{\mathcal E_{i,1}}})=\phi^+_i(c^0_i,f)+\phi^+_i(c^1_i,f)+\phi^+_i(c^2_i,g)+\phi^+_i(c^3_i,g)=360\degree$. The argument for $\mathcal E_{i,2}$ is analogous and Property~(d) of a join follows.
			\end{enumerate}
		\end{enumerate}

		\item The case in which $G_{b^*}$ is trivial and a different block of $G^+_i$ containing $c_i$ is non-trivial can be discussed very similarly to Case~(3). Namely, assume, w.l.o.g.\ up to a relabeling, that a non-trivial block of $G_{b_{i,1}\rightarrow c_i}$ contains $c_i$. Then it suffices to let $G_{b_{i,1}\rightarrow c_i}$ play the role of $G_{b_{i,0}\rightarrow c_i}$ and vice versa, and the algorithm described in Case~(3) can be followed almost verbatim, with one little difference. In order to maintain the property that every edge of $G_{b^*}$ that is incident to $f^*_{\mathcal E_{b^*}}$ is also incident to $f^*_{\mathcal E^+_{i,j}}$, some extra care is needed when choosing the clockwise order of the edges incident to $c_i$ in $\mathcal E^+_i$. Recall that $G_{b^*}$ is a block of $G^+_{i-1}$. By assumption, $c_i$ is incident to $f^*_{\mathcal E_{i,1}}$. In the equivalent of Cases~(3.2) and (3.3.2), there might be two different faces of $\mathcal E_{i,1}$ that might accommodate $G^+_{i-1}$. Since $\phi^{\mathrm{int}}_{i,1}(c_i)\in \{90\degree,180\degree\}$, we have $\phi_{i,1}(c_i,f^*_{\mathcal E_{i,1}})\in \{180\degree,270\degree\}$, hence one of these two faces is always $f^*_{\mathcal E_{i,1}}$. Thus, we can choose the clockwise order of the edges incident to $c_i$ in $\mathcal E^+_i$ so that $G^+_{i-1}$ ends up inside $f^*_{\mathcal E_{i,1}}$. 
	\end{enumerate}
	
	The rectilinear representation $(\mathcal E,\phi):=(\mathcal E^+_h,\phi^+_h)$ is the desired rectilinear representation of $G$. This completes the proof of sufficiency.
	
	We now discuss the running time of the described algorithm to construct $(\mathcal E,\phi)$. For $i=1,\dots,h$, the construction of $(\mathcal E^+_i,\phi^+_i)$ as a join of $(\mathcal E^+_{i-1},\phi^+_{i-1}),(\mathcal E_{i,1},\phi_{i,1}),\dots,(\mathcal E_{i,m(i)},\phi_{i,m(i)})$ can be performed in $O(m(i))\subseteq O(1)$ time, as it only requires to establish the clockwise order of the $O(m(i))$ edges incident to $c_i$, to determine the $O(m(i))$ angles incident to $c_i$, and to join $O(m(i))$ walks to form the walk delimiting the outer face of $\mathcal E^+_i$. Hence, the construction of $(\mathcal E,\phi)=(\mathcal E^+_h,\phi^+_h)$ takes $O(h)\subseteq O(n)$ time.
\end{proof}


The first step of our algorithm for computing the sets $\mathcal N_{b\rightarrow c}$ labels some vertices of $G$. Namely, for each non-trivial block $G_b$ of $G$, we label each vertex $v$ of $G_b$ whose degree in $G_b$ is smaller than or equal to $3$ with a set $\gamma_b(v)$ which contains all the values $\mu\in \{90\degree,180\degree,270\degree\}$ such that $G_b$ admits a $\chi_b$-constrained representation $(\mathcal E_b,\phi_b)$ in which $v$ is incident to $f^*_{\mathcal E_b}$ and $\phi^{\mathrm{int}}_{b}(v)=\mu$. By Theorem~\ref{th:2-con-variable}, this can be done in $O(n_b)$ time for each non-trivial block $G_b$ of $G$ with $n_b$ vertices, and hence in $O(n)$ time for all the non-trivial blocks of $G$. Further, for each trivial block $G_b$ of $G$, we label each end-vertex $c$ of $G_b$ with a set $\gamma_b(c)=\{0\degree\}$. 

For each leaf $b$ of $T$, the set $\mathcal N_{b\rightarrow c}$ can be constructed directly from the set $\gamma_b(c)$, as described later. If $b$ is an internal node of $T$, in order to compute $\mathcal N_{b\rightarrow c}$, our algorithm exploits the following tool.

\begin{lemma} \label{le:variable-simply-propagation}
	Let $b$ be an internal B-node of $T$ and let $c$ be a C-node of $T$ adjacent to $b$. Further, let $c_1,\dots,c_h$ be the C-nodes adjacent to $b$ and different from $c$, where $h\geq 1$; for $i=1,\dots,h$, let $b_{i,1},\dots,b_{i,m(i)}$ be the B-nodes adjacent to $c_i$ and different from $b$. Finally, let $\mu\in \{0\degree,90\degree,180\degree,270\degree\}$. 
	
	We have that $\mu\in \mathcal N_{b\rightarrow c}$ if and only if $\mu \in \gamma_b(c)$ and $c_i$ is a friendly neighbor of $b$, for every $i=1,\dots,h$.
	
	Further, assume that we are given:
	\begin{itemize}
		\item if $G_b$ is trivial, for each $i=1,\dots,h$ and $j=1,\dots,m(i)$, a rectilinear representation $(\mathcal E_{i,j},\phi_{i,j})$ of $G_{b_{i,j}\rightarrow c_i}$ such that $c_i$ is incident to $f^*_{\mathcal E_{i,j}}$ and such that $\phi^{\mathrm{int}}_{i,j}(c_i)\in \{0\degree,90\degree,180\degree\}$; and
		\item if $G_b$ is non-trivial, a $\chi_{b}$-constrained representation $(\mathcal E_b,\phi_b)$ of $G_{b}$ such that $c$ is incident to $f^*_{\mathcal E_b}$ and such that $\phi^{\mathrm{int}}_b(c)=\mu$ and, for each $i=1,\dots,h$ and $j=1,\dots,m(i)$, a rectilinear representation $(\mathcal E_{i,j},\phi_{i,j})$ of $G_{b_{i,j}\rightarrow c_i}$ such that $c_i$ is incident to $f^*_{\mathcal E_{i,j}}$ and such that $\phi^{\mathrm{int}}_{i,j}(c_i)\in \{0\degree,90\degree\}$. 
	\end{itemize}  
	Then it is possible to construct a rectilinear representation $(\mathcal E_{b\rightarrow c},\phi_{b\rightarrow c})$ of $G_{b\rightarrow c}$ such that $c$ is incident to $f^*_{\mathcal E_{b\rightarrow c}}$ and such that $\phi^{\mathrm{int}}_{b\rightarrow c}(c)=\mu$ in $O(h)$ time.
\end{lemma} 

\begin{proof}
	The proof of the characterization is very similar to the proof of Lemma~\ref{le:variable-simply-central-block}, hence we only emphasize here the differences with respect to that proof. The catch is that $G_b$ plays here the role that in the proof of Lemma~\ref{le:variable-simply-central-block} is played by of $G_{b^*}$.
	
	$(\Longrightarrow)$ In the proof of necessity of Lemma~\ref{le:variable-simply-central-block}, given a rectilinear representation $(\mathcal E,\phi)$ of $G$, we select a B-node $b^*$ in $T$ such that $G_{b^*}$ contains an edge $e^*$ incident to $f^*_{\mathcal E}$. Here we are given a rectilinear representation $(\mathcal E_{b\rightarrow c},\phi_{b\rightarrow c})$ of $G_{b\rightarrow c}$ in which $c$ is incident to $f^*_{\mathcal E_{b\rightarrow c}}$ and $\phi^{\mathrm{int}}_{b\rightarrow c}(c)=\mu$, and we do not need to select a B-node in $T$, as we focus on $b$ and the corresponding block $G_b$ of $G$. It can be proved exactly as in the proof of Lemma~\ref{le:variable-simply-central-block} that the restriction of $\mathcal E_{b\rightarrow c}$ to each graph $G_{b_{i,j}\rightarrow c_i}$ is a rectilinear representation $(\mathcal E_{i,j},\phi_{i,j})$ of $G_{b_{i,j}\rightarrow c_i}$ such that $c_i$ is incident to $f^*_{\mathcal E_{i,j}}$ and such that $\phi_{i,j}^{\mathrm{int}}(c_i)\in \{0\degree,90\degree,180\degree\}$ if $G_b$ is trivial, while $\phi_{i,j}^{\mathrm{int}}(c_i)\in \{0\degree,90\degree\}$ if $G_b$ is non-trivial; hence, $c_i$ is a friendly neighbor of $b$, for every $i=1,\dots,h$. By Lemma~\ref{le:chi-necessary}, the restriction of $\mathcal E_{b\rightarrow c}$ to $G_b$ is a $\chi_b$-constrained representation $(\mathcal E_b,\phi_b)$ of $G_b$. Further, since $c$ is incident to $f^*_{\mathcal E_{b\rightarrow c}}$ and $\phi^{\mathrm{int}}_{b\rightarrow c}(c)=\mu$, we have that $c$ is incident to $f^*_{\mathcal E_b}$ and $\phi^{\mathrm{int}}_{b}(c)=\mu$, given that every edge of $G_{b\rightarrow c}$ that is incident to $c$ belongs to $G_b$. It follows that $\mu \in \gamma_b(c)$. 
	
	$(\Longleftarrow)$ The proof of sufficiency is almost identical to that in the proof of Lemma~\ref{le:variable-simply-central-block}. For each $i=1,\dots,h$ and $j=1,\dots,m(i)$, we are given a rectilinear representation $(\mathcal E_{i,j},\phi_{i,j})$ of $G_{b_{i,j}\rightarrow c_i}$ such that $c_i$ is incident to $f^*_{\mathcal E_{i,j}}$ and such that $\phi^{\mathrm{int}}_{i,j}(c_i)\in \{0\degree,90\degree,180\degree\}$ if $G_{b}$ is trivial, while $\phi^{\mathrm{int}}_{i,j}(c_i)\in \{0\degree,90\degree\}$ if $G_{b}$ is non-trivial. Further, if $G_{b}$ is non-trivial, we are given a $\chi_{b}$-constrained representation $(\mathcal E_b,\phi_b)$ of $G_{b}$ such that $c$ is incident to $\mathcal E_b$ and such that $\phi_b^{\mathrm{int}}(c)=\mu$ (if $G_{b}$ is trivial, then $\mathcal E_b$ is the unique embedding of the edge $G_{b}$ and $\phi_b$ assigns $360\degree$ to the angle incident to each end-vertex of $G_{b}$). Then we define a sequence of graphs $G^+_0,G^+_1,\dots,G^+_h$, where $G^+_0:=G_{b}$ and $G^+_i:=G^+_{i-1} \cup \left(\bigcup_{j=1}^{m(i)} G_{b_{i,j}\rightarrow c_i}\right)$, for each $i=1,\dots,h$. Further, we define a rectilinear representation $(\mathcal E^+_i,\phi^+_i)$ of $G^+_{i}$, for each $i=0,\dots,h$. The desired rectilinear representation of $G_{b\rightarrow c}$ is then $(\mathcal E_{b\rightarrow c},\phi_{b\rightarrow c}):=(\mathcal E^+_h,\phi^+_h)$. We let $(\mathcal E^+_0,\phi^+_0)$ coincide with $(\mathcal E_b,\phi_b)$; then, as in the proof of Lemma~\ref{le:variable-simply-central-block}, we construct $(\mathcal E^+_i,\phi^+_i)$ as a join of $(\mathcal E^+_{i-1},\phi^+_{i-1})$ and of $(\mathcal E_{i,1},\phi_{i,1}),\dots,(\mathcal E_{i,m(i)},\phi_{i,m(i)})$, for $i=1,\dots,h$. We use here the property that, for $i=0,\dots,h$, every edge of $G_{b}$ incident to $f^*_{\mathcal E_{b}}$ is also incident to $f^*_{\mathcal E^+_i}$; that is, $\mathcal E_{b}$ is in the outer face of each embedding $\mathcal E_{i,j}$. This implies that $c_i$ is incident to $f^*_{\mathcal E^+_h}$ and hence that $\phi_{b\rightarrow c}^{\mathrm{int}}(c)=\mu$. The join of $(\mathcal E^+_{i-1},\phi^+_{i-1}),(\mathcal E_{i,1},\phi_{i,1}),\dots,(\mathcal E_{i,m(i)},\phi_{i,m(i)})$ is defined exactly as in the proof of Lemma~\ref{le:variable-simply-central-block}. As in that proof, the time needed to constructed $(\mathcal E^+_i,\phi^+_i)$ from $(\mathcal E^+_{i-1},\phi^+_{i-1}),(\mathcal E_{i,1},\phi_{i,1}),\dots,(\mathcal E_{i,m(i)},\phi_{i,m(i)})$ is constant, hence the construction of $(\mathcal E_{b\rightarrow c},\phi_{b\rightarrow c})$ takes $O(h)$ time.
\end{proof}


The sets $\mathcal N_{b\rightarrow c}$ will be constructed by performing several traversals of $T$. The following lemma, akin to Lemma~\ref{le:two-no-children}, is used to abruptly terminate these traversals, once we can conclude that $G$ admits no rectilinear representation.

\begin{lemma} \label{le:two-no-neighbors}
	Suppose that a B-node $b$ of $T$ has two unfriendly neighbors. Then $G$ admits no rectilinear representation. 
\end{lemma}

\begin{proof}
	Let $b$ be a B-node of $T$ with two unfriendly neighbors $c_i$ and $c_{i'}$. By Lemma~\ref{le:variable-simply-central-block}, we have that $G$ admits a rectilinear representation only if there exists a B-node $b^*$ in $T$ such that every neighbor of $b^*$ is friendly. Suppose, for a contradiction, that such a B-node $b^*$ exists. Let $(b^*_0:=b, c^*_0,b^*_1,c^*_1,\dots,b^*_{x-1},c^*_{x-1},b^*_x:=b^*)$ be the path between $b$ and $b^*$ in $T$. Since $c_i$ and $c_{i'}$ are distinct, one of them is not $c^*_0$. By repeated applications of Lemma~\ref{le:variable-simply-propagation}, we get that $\mathcal N_{b^*_0\rightarrow c^*_0}=\emptyset$, $\mathcal N_{b^*_1\rightarrow c^*_1}=\emptyset$, $\dots$, $\mathcal N_{b^*_{x-1}\rightarrow c^*_{x-1}}=\emptyset$. The last equality implies that $c^*_{x-1}$ is an unfriendly neighbor of $b^*$. This contradiction proves the lemma.
\end{proof}

We now describe our algorithm for computing the sets $\mathcal N_{b\rightarrow c}$. Similarly to Section~\ref{sse:var-2con}, the algorithm consists of several postorder traversals and, in order to achieve linear running time, computes and exploits the following additional information, for each B-node $b$ of $T$. Let  $c_0,c_1,\dots,c_h$ be the neighbors of $b$ in $T$. 


\begin{itemize}
	\item The algorithm computes a label $\eta(b)$ pointing to an unfriendly neighbor $c_i$ of $b$. While an unfriendly neighbor $c_i$ of $b$ has not been found, $\eta(b)$=\nullo.
	\item The algorithm maintains a label $\xi(b)$ which can have three values. \begin{itemize}
		\item The label $\xi(b)$ has value {\sc done} if the algorithm has established that every neighbor $c_i$ of $b$ is friendly.
		\item The label $\xi(b)$ has value $c_i$ if the algorithm has established that every neighbor of $b$ different from $c_i$ is friendly, while the algorithm has either established that $c_i$ is unfriendly or has not yet established whether $c_i$ is friendly or unfriendly.
		\item The label $\xi(b)$ has value \nullo~if no neighbor of $b$ has yet established to be friendly.
	\end{itemize}
\end{itemize}


We start by performing the following initialization. For every B-node $b$ of $T$, we initialize the labels $\eta(b)$ and $\xi(b)$ to \nullo~and, for every neighbor $c$ of $b$, we initialize the set $\mathcal N_{b\rightarrow c}$ to \nullo.    

Let $\beta_1,\dots,\beta_p$ be the leaves of $T$. For $i=1,\dots,h$, let $c_i$ be the unique neighbor of $\beta_i$ in $T$; note that the degree of $c_i$ in $G_{\beta_i}$ is smaller than or equal to $3$. As anticipated earlier, we have $\mathcal N_{\beta_i\rightarrow c_i}=\gamma_{\beta_i}(c_i)$. This allows us to compute the sets $\mathcal N_{\beta_1\rightarrow c_1},\dots,\mathcal N_{\beta_p\rightarrow c_p}$ in total $O(n)$ time. We also propagate the computed information to the B-nodes that are adjacent to the C-nodes $c_1,\dots,c_p$. Namely, for each $i=1,\dots,p$, we consider the B-nodes $b_{i,1},\dots,b_{i,m(i)}$ adjacent to $c_i$ in $T$ and different from $\beta_i$. When we consider $b_{i,j}$ we do the following.

\begin{itemize}
	\item Suppose first that $\eta(b_{i,j})$=\nullo. If $G_{b_{i,j}}$ is trivial and $\mathcal N_{{\beta_i}\rightarrow c_i}\cap \{0\degree,90\degree,180\degree\}= \emptyset$, or if $G_{b_{i,j}}$ is non-trivial and $\mathcal N_{\beta_i\rightarrow c_i}\cap \{0\degree,90\degree\}= \emptyset$, then we set $\eta(b_{i,j})=c_i$ (by Lemma~\ref{le:variable-simply-propagation}), otherwise we do nothing.
	\item Suppose next that $\eta(b_{i,j})=c_x$, where $x\neq i$. If $G_{b_{i,j}}$ is trivial and $\mathcal N_{{\beta_i}\rightarrow c_i}\cap \{0\degree,90\degree,180\degree\}= \emptyset$, or if $G_{b_{i,j}}$ is non-trivial and $\mathcal N_{\beta_i\rightarrow c_i}\cap \{0\degree,90\degree\}= \emptyset$, then we conclude that $G$ admits no rectilinear representation (by Lemma~\ref{le:two-no-neighbors}), otherwise we do nothing.
\end{itemize}

Since the degree of each vertex $c_i$ is at most $4$, there are $O(1)$ B-nodes $b_{i,j}$ to consider for each C-node $c_i$. Hence, the described computation can be performed in total $O(n)$ time.

We now describe the postorder traversals of $T$, which aim at computing the sets $\mathcal N_{b\rightarrow c}$ for every internal B-node $b$ of $T$. We will root $T$ first at $\beta_1$, then at $\beta_2$, and so on. In order to stress the rooting, we denote by $T_h$ the tree $T$ when rooted at $\beta_h$. We say that a (B- or C-) node $s$ of $T$ is a \emph{$T_h$-child} of a (C- or B-) node $t$ of $T$ if $s$ is a child of $t$ in $T_h$; then $t$ is the \emph{$T_h$-parent} of $s$. A \emph{$T_h$-subtree} of a B-node $b$ of $T$ is a tree $T_{b_{i,j}\rightarrow c_i}$ such that $c_i$ is a $T_h$-child of $b$ and $b_{i,j}$ is a $T_h$-child of $c_i$.

We now explain how the postorder traversal of $T_h$ works, for each $h=1,\dots,p$. 

When the postorder traversal is invoked on a subtree $T_{b\rightarrow c}$ of $T_h$ (at first, this is the entire tree $T_h$), it enters a $T_h$-subtree $T_{b_{i,j}\rightarrow c_i}$ of $b$ only if $\mathcal N_{b_{i,j}\rightarrow c_i}=$\nullo. That is, if $b_{i,j}$ is a leaf of $T$ or if $\mathcal N_{b_{i,j}\rightarrow c_i}$ has been constructed during the traversal of $T_g$ with $g<h$, then $T_{b_{i,j}\rightarrow c_i}$ is not visited during the traversal of $T_h$. 

It is crucial here, as far as the running time is concerned, that the algorithm does not check, for each individual $T_h$-child $c_i$ of $b$ and for each $T_h$-child $b_{i,j}$ of $c_i$, whether $\mathcal N_{b_{i,j}\rightarrow c_i}=$\nullo, and thus whether it should enter $T_{b_{i,j}\rightarrow c_i}$; in fact, performing all these checks would imply spending time proportional to the degree of $b$ in $T$ each time a traversal processes $b$, which would result in a quadratic running time. Indeed, the information we need can be recovered for all the $T_h$-children $c_i$ of $b$ and for all the $T_h$-children $b_{i,j}$ of $c_i$ from the labels $\xi(b)$ and $\eta(b)$, in total $O(1)$ time. In particular, if $\xi(b)$={\sc done} or if $\xi(b)=c$, where $c$ is the $T_h$-parent of $b$, then all the $T_h$-children of $b$ are friendly, hence there is no need to check any set $\mathcal N_{b_{i,j}\rightarrow c_i}$. If $\eta(b)=c_x$, where $c_x$ is a $T_h$-child of $b$, then $b$ has an unfriendly $T_h$-child, and the algorithm does not check any set $\mathcal N_{b_{i,j}\rightarrow c_i}$ either. If $\eta(b)\neq c_x$ and $\xi(b)=c_x$, where $c_x$ is a $T_h$-child of $b$, then all the $T_h$-children of $b$ different from $c_x$ are friendly, and the traversal only needs to check the sets $\mathcal N_{b_{x,j}\rightarrow c_x}$; note that there is a constant number of these sets, given that the degree of $c_x$ is at most $4$. Finally, if $\xi(b)$={\sc nullo} and if $\eta(b)$={\sc null} or $\eta(b)=c$, where $c$ is the $T_h$-parent of $b$, then the algorithm does indeed check every set $\mathcal N_{b_{i,j}\rightarrow c_i}$, for each $T_h$-child $c_i$ of $b$ and for each $T_h$-child $b_{i,j}$ of $c_i$ different from $b$, and possibly enters the subtree $T_{b_{i,j}\rightarrow c_i}$; however, this case happens at most twice, namely once with $\eta(b)$={\sc null} and once with $\eta(b)=c$, as will be argued in the proof of the upcoming Theorem~\ref{th:main}.



Now suppose that the traversal of $T_h$ visits an internal B-node $b$ of $T$; let $c_0,c_1,\dots,c_k$ be the C-nodes adjacent to $b$, where $c:=c_0$ is the $T_h$-parent of $b$ and $c_1,\dots,c_k$ are the $T_h$-children of $b$. For $i=1,\dots,k$, let $b_{i,1},\dots,b_{i,m(i)}$ the B-nodes of $T$ adjacent to $c_i$ and different from $b$. When the postorder traversal of $T_h$ visits $b$, after each subtree $T_{b_{i,j}\rightarrow c_i}$ with $i\in \{1,\dots,k\}$ and $j\in \{1,\dots,m(i)\}$ has been processed, it constructs the set $\mathcal N_{b\rightarrow c}$, it possibly updates $\xi(b)$, and it possibly updates $\eta(b_x(c))$, for each B-node $b_x(c)$ adjacent to $c$ and different from $b$. Note that, since every tree $T_{b_{i,j}\rightarrow c_i}$ has been processed, the set $\mathcal N_{b_{i,j}\rightarrow c_i}$ has been computed already, for every $i=1,\dots,k$ and $j=1,\dots,m(i)$. The algorithm performs the following actions.

Suppose first that $\eta(b)=c_i$, for some $i\in \{1,\dots,h\}$, hence $c_i$ is an unfriendly neighbor of $b$. By Lemma~\ref{le:variable-simply-propagation}, we can set $\mathcal N_{b\rightarrow c}=\emptyset$. We also propagate the computed information to each B-node $b_x(c)$ that is adjacent to $c$ and that is different from $b$. Namely, if $\eta(b_x(c))$ points to some C-node different from $c$, by Lemma~\ref{le:two-no-neighbors} we can terminate all the traversals reporting in output that $G$ admits no rectilinear representation. Otherwise, we set $\eta(b_x(c))=c$. Since the degree of $c$ in $G$ (and in $T$) is at most $4$, this can be done in total $O(1)$ time.

Suppose next that $\eta(b)=${\sc null}~or that $\eta(b)=c$. Since every tree $T_{b_{i,j}\rightarrow c_i}$ has been processed and since $\eta(b)\neq c_i$ for every $T_h$-child $c_i$ of $b$, it follows that every $T_h$-child of $b$ is friendly. By Lemma~\ref{le:variable-simply-propagation}, we can set $\mathcal N_{b\rightarrow c}=\gamma_b(c)$. Further, if $\xi(b)=${\sc done} or if $\xi(b)=c_i$, for some $T_h$-child $c_i$ of $b$, we set $\xi(b)=${\sc done}; otherwise, we set $\xi(b)=c$. Finally, as in the previous case, we consider each B-node $b_x(c)$ that is adjacent to $c$ and that is different from $b$; differently from the previous case, $c$ might be a friendly neighbor of $b_x(c)$. We check whether $G_{b_x(c)}$ is trivial and $\mathcal N_{b\rightarrow c}\cap \{0\degree,90\degree,180\degree\}= \emptyset$, or whether $G_{b_x(c)}$ is non-trivial and $\mathcal N_{b\rightarrow c}\cap \{0\degree,90\degree\}= \emptyset$. In the positive case, we have that $c$ is an unfriendly neighbor of $b_x(c)$. If $\eta(b_x(c))$ points to some C-node different from $c$, by Lemma~\ref{le:two-no-neighbors} we can terminate all the traversals reporting in output that $G$ admits no rectilinear representation. Otherwise, we set $\eta(b_x(c))=c$. This concludes the processing of $b$ in $T_h$. 

We are now ready to state and prove the following main theorem.

\begin{theorem} \label{th:main}
	There is an $O(n)$-time algorithm that tests whether an $n$-vertex outerplanar graph admits a planar rectilinear drawing; in the positive case, the algorithm constructs such a drawing in $O(n)$ time.
\end{theorem}

\begin{proof}
	We compute in $O(n)$ time the block-cut-vertex tree $T$ of $G$~\cite{h-gt-69,ht-aeagm-73} and label each edge with the block it belongs to. We visit $G$ and label each vertex that is in two non-trivial blocks $G_b$ and $G_{b'}$ as belonging to $\chi_b$ and $\chi_{b'}$. Hence, this visit computes in total $O(n)$ time a set $\chi_b$ for every non-trivial block $G_b$ of $G$. 
	
	By Theorem~\ref{th:2-con-variable}, we can label each vertex $v$ of $G$ that belongs to a non-trivial block $G_b$ and whose degree in $G_b$ is not larger than $3$ with a set $\gamma_b(v)$; this set contains all the values $\mu\in \{90\degree,180\degree,270\degree\}$ such that $G_b$ admits a $\chi_b$-constrained representation $(\mathcal E_b,\phi_b)$ in which $v$ is incident to $f^*_{\mathcal E_b}$ and $\phi_b^{\mathrm{int}}(v)=\mu$. For each non-trivial block $G_b$ of $G$ with $n_b$ vertices, this can be done in total $O(n_b)$ time for all the vertices $v$ of $G_b$ whose degree is not larger than $3$. Hence, this computation takes $O(n)$ time over all the non-trivial blocks of $G$. Further, for each trivial block $G_b$ of $G$, we label in $O(1)$ time each end-vertex $v$ of $G_b$ with a set $\gamma_b(v)=\{0\degree\}$. Hence, this computation takes $O(n)$ time over all the trivial blocks of $G$.
	
	We now equip each edge $bc$ of $T$, where $b$ is a B-node and $c$ is a C-node, with a set $\mathcal N_{b\rightarrow c}$ that contains all the values $\mu\in \{0\degree,90\degree,180\degree,270\degree\}$ such that $G_{b\rightarrow c}$ admits a rectilinear representation $(\mathcal E_{b\rightarrow c}, \phi_{b\rightarrow c})$ in which $c$ is incident to $f^*_{\mathcal E_{b\rightarrow c}}$ and $\phi^{\mathrm{int}}_{b\rightarrow c}(c)=\mu$. In order to do that, we apply the algorithm described before the statement of the theorem, for which we now prove the correctness and analyze the running time. 
	
	The correctness is proved inductively. More precisely, by induction on the number of sets $\mathcal N_{b\rightarrow c}\neq$\nullo, we prove that, before the processing of each B-node, the following statements are satisfied:
	\begin{itemize}
		\item For every set $\mathcal N_{b\rightarrow c}\neq$\nullo~and every value $\mu\in \{0\degree,90\degree,180\degree,270\degree\}$, we have $\mu\in \mathcal N_{b\rightarrow c}$ if and only if $G_{b\rightarrow c}$ admits a rectilinear representation $(\mathcal E_{b\rightarrow c}, \phi_{b\rightarrow c})$ in which $c$ is incident to $f^*_{\mathcal E_{b\rightarrow c}}$ and $\phi^{\mathrm{int}}_{b\rightarrow c}(c)=\mu$.
		\item For each B-node $b$ of $T$, the label $\eta(b)$ is equal to \nullo~if and only if the algorithm has not yet found any unfriendly neighbor $c_i$ of $b$. Further, the label $\eta(b)$ is equal to $c_i$~if and only if the algorithm has established that a neighbor $c_i$ of $b$ is unfriendly.
		\item For each B-node $b$ of $T$, the label $\xi(b)$ is equal to {\sc done} if and only if the algorithm has established that every neighbor $c_i$ of $b$ is friendly. Further, the label $\xi(b)$ has value $c_i$ if and only if the algorithm has established that every neighbor of $b$ different from $c_i$ is friendly, while the algorithm has either established that $c_i$ is unfriendly or has not yet established whether $c_i$ is friendly or unfriendly. Finally, the label $\xi(b)$ has value \nullo~if and only if no neighbor of $b$ has yet established to be friendly.
	\end{itemize}
	
	The inductive hypothesis is trivially true initially, before processing the first B-node.
	
	First, the algorithm determines the set $\mathcal N_{\beta_i\rightarrow c_i}$ for each leaf $\beta_i$ of $T$ with a unique neighbor $c_i$. When processing $\beta_i$, the set $\mathcal N_{\beta_i\rightarrow c_i}$ is correctly defined as $\gamma_{\beta_i}(c_i)$, given that $G_{\beta_i\rightarrow c_i}$ coincides with the block $G_{\beta_i}$. Further, the algorithm considers each B-node $b_{i,j}\neq b$ adjacent to $c_i$; if $c_i$ is determined to be an unfriendly neighbor of $b_{i,j}$, then the algorithm either, by Lemma~\ref{le:variable-simply-propagation}, correctly sets $\eta(b_{i,j})=c_i$ (if previously $\eta(b_{i,j})$=\nullo) or, by Lemma~\ref{le:two-no-neighbors}, correctly concludes that $G$ admits no rectilinear representation (if previously $\eta(b_{i,j})=c_x$, for some neighbor $c_x$ of $b_{i,j}$ different from $c_i$).
	
	The algorithm processes each internal B-node $b$ of $T$ several times. Consider the processing of $b$ during the traversal of $T_h$. Let $c$ be the $T_h$-parent of $b$, let $c_1,\dots,c_k$ be the $T_h$-children of $b$, and, for $i=1,\dots,k$, let $b_{i,1},\dots,b_{i,m(i)}$ be the $T_h$-children of $b_i$. The algorithm processes $b$ in $T_h$ only after each set $\mathcal N_{b_{i,j}\rightarrow c_i}$ has been constructed. We first prove that, after processing $b$ in $T_h$, the set $\mathcal N_{b\rightarrow c}$ contains a value $\mu\in \{0\degree,90\degree,180\degree,270\degree\}$ if and only if $G_{b\rightarrow c}$ admits a rectilinear representation $(\mathcal E_{b\rightarrow c}, \phi_{b\rightarrow c})$ in which $c$ is incident to $f^*_{\mathcal E_{b\rightarrow c}}$ and $\phi^{\mathrm{int}}_{b\rightarrow c}(c)=\mu$. 
	
	Suppose first that $\mu \notin \mathcal N_{b\rightarrow c}$. There are two possible reasons for that. 
	
	\begin{enumerate}
		\item First, when the algorithm processes $b$, it sets $\mathcal N_{b\rightarrow c}=\emptyset$ (which implies that $\mu \notin \mathcal N_{b\rightarrow c}$), since $\mathcal N_{b_{i,j}\rightarrow c_i}\cap \{0\degree,90\degree,180\degree\}=\emptyset$ and $G_b$ is trivial, or since $\mathcal N_{b_{i,j}\rightarrow c_i}\cap \{0\degree,90\degree\}=\emptyset$ and $G_b$ is non-trivial, for some $i\in \{1,\dots,h\}$ and $j\in \{1,\dots,m(i)\}$. Then Lemma~\ref{le:variable-simply-propagation} ensures that $G_{b\rightarrow c}$ admits no rectilinear representation $(\mathcal E_{b\rightarrow c}, \phi_{b\rightarrow c})$ in which $c$ is incident to $f^*_{\mathcal E_{b\rightarrow c}}$ and $\phi^{\mathrm{int}}_{b\rightarrow c}(c)=\mu$. 
		
		\item Second, when the algorithm processes $b$, it sets $\mathcal N_{b\rightarrow c}=\gamma_b(c)$ if, for $i=1,\dots,h$ and $j=1,\dots,m(i)$, it holds true that $\mathcal N_{b_{i,j}\rightarrow c_i}\cap \{0\degree,90\degree,180\degree\}\neq \emptyset$ and $G_b$ is trivial, or that $\mathcal N_{b_{i,j}\rightarrow c_i}\cap \{0\degree,90\degree\}\neq \emptyset$ and $G_b$ is non-trivial. Thus, if $\mu \notin \gamma_b(c)$, then $\mu \notin \mathcal N_{b\rightarrow c}$. Then Lemma~\ref{le:variable-simply-propagation} ensures that $G_{b\rightarrow c}$ admits no rectilinear representation $(\mathcal E_{b\rightarrow c}, \phi_{b\rightarrow c})$ in which $c$ is incident to $f^*_{\mathcal E_{b\rightarrow c}}$ and $\phi^{\mathrm{int}}_{b\rightarrow c}(c)=\mu$. 
	\end{enumerate}     
	
	Suppose next that $\mu\in \mathcal N_{b\rightarrow c}$. The algorithm inserts $\mu$ into $\mathcal N_{b\rightarrow c}$ because $\mu \in \gamma_b(c)$ and because, for $i=1,\dots,h$ and $j=1,\dots,m(i)$, it holds true that $\mathcal N_{b_{i,j}\rightarrow c_i}\cap \{0\degree,90\degree,180\degree\}\neq \emptyset$ and $G_b$ is trivial, or that $\mathcal N_{b_{i,j}\rightarrow c_i}\cap \{0\degree,90\degree\}\neq \emptyset$ and $G_b$ is non-trivial.  Then Lemma~\ref{le:variable-simply-propagation} ensures that $G_{b\rightarrow c}$ admits a rectilinear representation $(\mathcal E_{b\rightarrow c}, \phi_{b\rightarrow c})$ in which $c$ is incident to $f^*_{\mathcal E_{b\rightarrow c}}$ and $\phi^{\mathrm{int}}_{b\rightarrow c}(c)=\mu$.
	
	The construction of the set $\mathcal N_{b\rightarrow c}$ might lead the algorithm to conclude that $c$ is an unfriendly neighbor of some B-node of $T$. Indeed, the algorithm checks each B-node $b_x(c)\neq b$ adjacent to $c$; if $c$ is determined to be an unfriendly neighbor of $b_x(c)$, then the algorithm either correctly sets $\eta(b_x(c))=c$ (if previously $\eta(b_x(c))$=\nullo), or correctly concludes, by Lemma~\ref{le:two-no-neighbors}, that $G$ admits no rectilinear representation (if previously $\eta(b_x(c))=c_x$, for some neighbor $c_x$ of $b_x(c)$ different from $c$).
	
	Finally, when processing $b$ the algorithm might conclude that every neighbor of $b$ is friendly, or that every neighbor of $b$ different from $c$ is friendly. If before processing $b$ we have $\xi(b)=c_i$, for some $T_h$-child $c_i$ of $b$, then every neighbor of $b$ different from $c_i$ is friendly; when $b$ is processed, if $\eta(b)=${\sc null}~or $\eta(b)=c$, then every $T_h$-child of $b$, including $c_i$, is friendly, hence the algorithm can correctly set $\xi(b)=${\sc done}. If before processing $b$ we have $\xi(b)=${\sc null}, then either $c$ has been determined to be an unfriendly neighbor of $b$ or it is unknown whether $c$ is a friendly or unfriendly neighbor of $b$; when $b$ is processed, if $\eta(b)=${\sc null}~or $\eta(b)=c$, then every $T_h$-child of $b$ is friendly, hence the algorithm can correctly set $\xi(b)=c$.
	
	We now show that the algorithm to construct the sets $\mathcal N_{b\rightarrow c}$ runs in $O(n)$ time. 
	
	\begin{itemize}
		\item First, the initialization to \nullo~of each label $\eta(b)$ and $\chi(b)$ and of each set $\mathcal N_{b\rightarrow c}$ can clearly be performed in $O(n)$ time.
		\item Second, all the sets $\mathcal N_{\beta_i\rightarrow c_i}$ such that $\beta_i$ is a leaf of $T$ are constructed in $O(n)$ time, as argued during the algorithm's description. Further, after computing each set $\mathcal N_{\beta_i\rightarrow c_i}$, the labels $\eta(b_{i,j})$ for the B-nodes adjacent to $c_i$ are possibly updated in $O(1)$ time, and hence in total $O(n)$ time over the entire tree, as argued during the algorithm's description.
		\item Now observe that, for each B-node $b$ of $T$, for each neighbor $c_i$ of $b$, and for each neighbor $b_{i,j}$ of $c_i$ different from $b$, there is (at most) one rooting of $T$ for which a postorder traversal is invoked on $T_{b_{i,j}\rightarrow c_i}$. Indeed, once a postorder traversal is invoked on $T_{b_{i,j}\rightarrow c_i}$, the set $\mathcal N_{b_{i,j}\rightarrow c_i}$ is determined; if some subsequent traversal of $T$ processes $b$ with $c_i$ as a child, then no postorder traversal is invoked on $T_{b_{i,j}\rightarrow c_i}$. This implies that, if a B-node has $k+1$ neighbors in $T$, then it is processed $O(k)$ times. 
		\item We now show that the total running time for processing a B-node $b$ of $T$ with $k+1$ neighbors, over all the traversals of $T$, is in $O(k)$. This implies that the overall running time of the algorithm is in $O(n)$, given that the sum of the degrees of the vertices of $T$ is in $O(n)$. 
		
		Let $c_0,c_1,\dots,c_k$ be the C-nodes neighbors of $b$ and, for $i=0,\dots,k$, let $b_{i,0},\dots,b_{i,m(i)}$ be the B-nodes neighbors of $c_i$ and different from $b$. Suppose that a traversal of $T_h$ visits $b$ with $c_i$ as a $T_h$-parent, for some $i\in \{0,\dots,k\}$. We distinguish some cases, based on the values of the labels $\xi(b)$ and $\eta(b)$. Common to all the cases is that the processing of $b$ determines the set $\mathcal N_{b\rightarrow c_i}$; then the computed information $\mathcal N_{b\rightarrow c_i}$ is propagated to each neighbor $b_x(c_i)$ of $c_i$ different from $b$, possibly updating the label $\eta(b_x(c_i))$ in $O(1)$ time. Since $c_i$ has $O(1)$ neighbors in $T$, this is done in $O(1)$ time, and hence in $O(k)$ time over all the traversals of $T$ that visit $b$.
		\begin{itemize}
			\item If $\xi(b)=${\sc done} or if $\xi(b)=c_i$, then no $T_h$-child of $b$ is visited, as all the $T_h$-children of $b$ are already known to be friendly; then we set $\mathcal N_{b\rightarrow c_i}=\gamma_b(c_i)$. Hence, the entire processing of $b$ in $T_h$ takes $O(1)$ time. Over all the traversals of $T$ this amounts to $O(k)$ time.
			\item If $\eta(b)=c_j$, for some $j\in \{0,\dots,k\}$ such that $c_j$ is a $T_h$-child of $b$, then no $T_h$-child of $b$ is visited, as we already know that not all the $T_h$-children of $b$ are friendly; then we set $\mathcal N_{b\rightarrow c_i}=\emptyset$. Hence, the entire processing of $b$ in $T_h$ takes $O(1)$ time. Over all the traversals of $T$ this amounts to $O(k)$ time.
			\item If $\xi(b)=c_j$ and $\eta(b)\neq c_j$, for some $j\in \{0,\dots,k\}$ with $j\neq i$, then traversals of $T_{b_{j,0}\rightarrow c_j},\dots,T_{b_{j,m(j)}\rightarrow c_j}$ are invoked in $O(m(j))\subseteq O(1)$ time and the values $\mathcal N_{b_{j,0}\rightarrow c_j},\dots,\mathcal N_{b_{j,m(j)}\rightarrow c_j}$ are used in order to determine in $O(1)$ time whether $c_j$ is a friendly or unfriendly neighbor of $b$. Then $\mathcal N_{b\rightarrow c_i}$ is set equal to $\gamma_b(c_i)$ if $c_j$ is a friendly neighbor of $b$, and to $\emptyset$, otherwise; in both cases, the computation takes $O(1)$ time. Finally, if $c_j$ is a friendly neighbor of $b$, then we set $\xi(b)=${\sc done}; otherwise, the processing of $c_j$ either sets $\eta(b)=c_j$ (if $\eta(b)$ was previously \nullo), or stops all the traversals of $T$ by concluding that $G$ admits no rectilinear representation. Hence, the entire processing of $b$ in $T_h$ takes $O(1)$ time. Over all the traversals of $T$ this amounts to $O(k)$ time.
			\item If $\xi(b)=${\sc null}, and if $\eta(b)=c_i$ or $\eta(b)=${\sc null}, then, for each $j=0,\dots,k$ with $j\neq i$, traversals of $T_{b_{j,0}\rightarrow c_j},\dots,T_{b_{j,m(j)}\rightarrow c_j}$ are invoked in $O(m(j))\subseteq O(1)$ time and the values $\mathcal N_{b_{j,0}\rightarrow c_j},\dots,\mathcal N_{b_{j,m(j)}\rightarrow c_j}$ are used in order to determine in $O(1)$ time whether $c_j$ is a friendly or unfriendly neighbor of $b$. Hence, this computation takes $O(k)$ time over all the $T_h$-children of $b$. 
			
			If all the $T_h$-children of $b$ are friendly, then $\mathcal N_{b\rightarrow c_i}$ is set equal to $\gamma_b(c_i)$ and $\xi(b)$ is set equal to $c_i$. Otherwise, the processing of an unfriendly child $c_j$ of $b$ either sets $\eta(b)=c_j$ (if $\eta(b)$ was previously \nullo), or stops all the traversals of $T$ by concluding that $G$ admits no rectilinear representation.
			
			The entire processing of $b$ in $T_h$ thus takes $O(k)$ time. However, there are at most two traversals of $T$ for which the conditions of this case are met. Namely, if $\xi(b)=${\sc null} and $\eta(b)=${\sc null}, then after processing $b$ we either have $\xi(b)=c_i$, or we have $\eta(b)=c_j$, for some $T_h$-child $c_j$ of $b$, or the traversals of $T$ have stopped because it was concluded that $G$ admits no rectilinear representation. Hence, $\xi(b)=${\sc null} and $\eta(b)=${\sc null}~occurs in at most one traversal of $T$. Further, if $\xi(b)=${\sc null} and $\eta(b)=c_i$,  then after processing $b$ we either have $\xi(b)=c_i$, or the traversals of $T$ have stopped because it was concluded that $G$ admits no rectilinear representation. Hence, $\xi(b)=${\sc null} and $\eta(b)=c_i$, where $c_i$ is the parent of $b$, happens in at most one traversal of $T$.
		\end{itemize}
	\end{itemize}
	
	This concludes the proof that the running time of the algorithm to construct the sets $\mathcal N_{b\rightarrow c}$ is $O(n)$.
	
	By Lemma~\ref{le:variable-simply-central-block}, we have that $G$ admits a rectilinear representation if and only a B-node $b^*$ exists in $T$ such~that: (i) if $G_{b^*}$ is non-trivial, then it admits a $\chi_{b^*}$-constrained representation; and (ii) every neighbor of $b^*$ in $T$ is friendly. An $O(n)$-time visit of $T$ allows us to find every B-node $b$ satisfying (ii); indeed, whether a neighbor $c_i$ of $b$ is friendly can be tested in $O(1)$ time, given that the sets $\mathcal N_{b_{i,j}\rightarrow c_i}$ for the neighbors $b_{i,j}$ of $c_i$ different from $b$ are known. Further, in order to check whether a block $G_b$ of $G$ admits a $\chi_b$-constrained representation, it suffices to check whether $G_b$ contains a vertex $v$ of maximum degree $3$ such that $\gamma_b(v)\neq \emptyset$. Since every vertex $v$ of $G_b$ has already been labeled with the set $\gamma_b(v)$, this only requires a visit of $G_b$, which takes $O(n_b)$ time if $G_b$ has $n_b$ vertices, and thus $O(n)$ time for all the blocks of $G$.
	
	If we found a B-node $b^*$ such that every neighbor of $b^*$ in $T$ is friendly and such that, if $G_{b^*}$ is non-trivial, then it admits a $\chi_{b^*}$-constrained representation, we can construct a rectilinear representation of $G$ in $O(n)$ time as follows. We root $T$ at $b^*$ and we perform a bottom-up visit of $T$. We now describe the processing of a B-node $b$ of $T$. If $b\neq b^*$, then let $c$ be the parent of $b$ and let $\hat b$ the parent of $c$. Denote by $n_b$ the number of vertices of $G_b$. Then the processing of $b$ constructs a rectilinear representation $(\mathcal E_{b\rightarrow c},\phi_{b\rightarrow c})$ of $G_{b\rightarrow c}$ such that $c$ is incident to $f^*_{\mathcal E_{b\rightarrow c}}$ and such that $\phi^{\mathrm{int}}_{b\rightarrow c}(c)\in \{0\degree,90\degree,180\degree\}$ if $G_{\hat b}$ is trivial, while $\phi^{\mathrm{int}}_{b\rightarrow c}(c)\in \{0\degree,90\degree\}$ if $G_{\hat b}$ is non-trivial. Such a representation exists given that $c$ is a friendly neighbor of $\hat b$; indeed, if $c$ was an unfriendly neighbor of $\hat b$, then the child of $b^*$ that is an ancestor of $c$ would not be a friendly neighbor of $b^*$, by repeated applications of Lemma~\ref{le:variable-simply-propagation}, contradicting the assumptions.
	
	Suppose first that $b$ is a leaf of $T$. If $G_b$ is trivial, then the desired representation is easily constructed. If $G_b$ is non-trivial then, by Theorem~\ref{th:2-con-variable}, we can construct in $O(n_b)$ time a rectilinear representation $(\mathcal E_{b\rightarrow c},\phi_{b\rightarrow c})$ of $G_{b\rightarrow c}=G_b$ such that $c$ is incident to $f^*_{\mathcal E_{b\rightarrow c}}$, and such that $\phi^{\mathrm{int}}_{b\rightarrow c}(c)\in \{0\degree,90\degree,180\degree\}$ if $G_{\hat b}$ is trivial, while $\phi^{\mathrm{int}}_{b\rightarrow c}(c)\in\{0\degree,90\degree\}$ if $G_{\hat b}$ is non-trivial. 
	
	Suppose next that $b\neq b^*$ is an internal node of $T$. Let $c_1,\dots,c_h$ be the children of $b$ and, for $i=1,\dots,h_b$, let $b_{i,1},\dots,b_{i,m(i)}$ be the children of $c_i$ in $T$. Assume that, for every $i=1,\dots,h_b$ and for every $j=1,\dots,m(i)$, we have a rectilinear representation $(\mathcal E_{b_{i,j}\rightarrow c_{i}},\phi_{b_{i,j}\rightarrow c_{i}})$ of $G_{b_{i,j}\rightarrow c_{i}}$ such that $c_i$ is incident to $f^*_{\mathcal E_{b_{i,j}\rightarrow c_i}}$, and such that $\phi^{\mathrm{int}}_{b_{i,j}\rightarrow c_i}(c_i) \in \{0\degree,90\degree,180\degree\}$ if $G_{b}$ is trivial, while $\phi^{\mathrm{int}}_{b_{i,j}\rightarrow c_i}(c_i) \in \{0\degree,90\degree\}$ if $G_{b}$ is non-trivial. Further, if $G_b$ is non-trivial then, by Theorem~\ref{th:2-con-variable}, we can construct in $O(n_b)$ time a $\chi_{b}$-constrained representation $(\mathcal E_{b},\phi_{b})$ of $G_{b}$ such that $c$ is incident to $f^*_{\mathcal E_{b}}$ and such that $\phi^{\mathrm{int}}_{b}(c)\in \{0\degree,90\degree,180\degree\}$ if $G_{\hat b}$ is trivial, while $\phi^{\mathrm{int}}_{b}(c)\in\{0\degree,90\degree\}$ if $G_{\hat b}$ is non-trivial. By Lemma~\ref{le:variable-simply-propagation}, it is possible to construct in $O(h_b)$ time a rectilinear representation $(\mathcal E_{b\rightarrow c},\phi_{b\rightarrow c})$ of $G_{b\rightarrow c}$ such that $c$ is incident to $f^*_{\mathcal E_{b\rightarrow c}}$ and such that $\phi^{\mathrm{int}}_{b\rightarrow c}(c)\in \{0\degree,90\degree,180\degree\}$ if $G_{\hat b}$ is trivial, while $\phi^{\mathrm{int}}_{b\rightarrow c}(c)\in\{0\degree,90\degree\}$ if $G_{\hat b}$ is non-trivial.
	
	The root $b^*$ of $T$ is treated very similarly to any other B-node, with two exceptions. Let $n_{b^*}$ be the number of vertices in $G_{b^*}$. First, if $G_{b^*}$ is non-trivial then, by Theorem~\ref{th:2-con-variable}, we construct in $O(n_{b^*})$ time a $\chi_{b^*}$-constrained representation $(\mathcal E_{b^*},\phi_{b^*})$ of $G_{b^*}$, with no requirement on the angle incident to a specific vertex. Second, the rectilinear representations $(\mathcal E_{b_{i,j}\rightarrow c_{i}},\phi_{b_{i,j}\rightarrow c_{i}})$ of the graphs $G_{b_{i,j}\rightarrow c_i}$ are merged in order to form a rectilinear representation $(\mathcal E,\phi)$ of $G$ by means of Lemma~\ref{le:variable-simply-central-block}, rather than Lemma~\ref{le:variable-simply-propagation}.
	
	The time needed for processing a B-node $b$ is in $O(n_b+h_b)$, where $h_b+1$ is the degree of $b$ in $T$. Since the sum of the degrees of the nodes of a tree with $O(n)$ nodes is in $O(n)$ and since $\sum_b n_b \in O(n)$, the algorithm constructs a rectilinear representation $(\mathcal E,\phi)$ of $G$ in $O(n)$ time.
	
	This concludes the construction of a rectilinear representation $(\mathcal E,\phi)$ of $G$. Finally, a planar rectilinear drawing can be constructed from $(\mathcal E,\phi)$ in $O(n)$ time~\cite{t-eggmnb-87}.
\end{proof}
\subsubsection*{Comparison with~\cite{dlop-ood-20}.} In the proofs of Theorems~\ref{th:2-con-variable} and~\ref{th:main}, we are given a tree describing the structure of an $n$-vertex outerplanar graph $G$, namely the extended dual tree of $G$ in Theorem~\ref{th:2-con-variable} and the block-cut-vertex tree of $G$ in Theorem~\ref{th:main}, and we want to visit the tree in such a way that each edge is traversed in both directions. Didimo et al.~\cite{dlop-ood-20} showed that this can be done in $O(n)$ time. More precisely, they established a general ``Reusability Principle''~\cite[Lemma 6.6]{dlop-ood-20} stating the following. If an algorithm $\mathcal A$ exists that traverses bottom-up a rooted tree\footnote{Didimo et al.~\cite{dlop-ood-20} stated this principle for the SPQR-tree of a $2$-connected graph, however their idea immediately extends to any tree.} and, for each node $s$ with parent $p$, computes in $O(f(n))$ time a value $\mathcal V_{s\rightarrow p}$ from the values $\mathcal V_{s_i\rightarrow s}$ computed for the children $s_i$ of $s$, then it is possible to execute $\mathcal A$ for all possible re-rootings of the tree in total $O(n\cdot f(n))$ time. In practice, one needs to perform a first $O(n)$-time bottom-up visit of the tree (which can be regarded as a pre-processing step), which ensures that, for each node $s$ of the tree, the value $\mathcal V_{s_i\rightarrow s}$ is computed for all the neighbors $s_i$ of $s$, except for one. Then one needs to show how to update in $O(f(n))$ time (indeed in constant time, if one is aiming for total linear time) the value $\mathcal V_{s\rightarrow p}$, once a different rooting of the tree is chosen and hence the parent of $s$ is not $p$ anymore. Key to this approach is the fact that a postorder traversal does not enter a subtree if this has been visited already, hence each value $\mathcal V_{s\rightarrow p}$ is computed just once if the tree is visited by means of postorder traversals.

In the proofs of Theorems~\ref{th:2-con-variable} and~\ref{th:main}, we follow the approach presented by Didimo et al.~\cite{dlop-ood-20}. In particular, no significant variation is made to the order in which the nodes of the tree are visited;  this order is still based on the fact that a postorder traversal does not enter a subtree if this has been visited already, and hence each subtree is visited just once. However, in both our proofs we have to face a possibility that can not occur in the algorithm by Didimo et al., namely that the value we need to compute does not exist. Roughly speaking, while Didimo et al.\ need to compute the minimum number of bends in certain orthogonal representations, and these numbers always exist, we need to compute the angles for which a certain rectilinear representation exists, and have to deal with the possibility that there might not be any such an angle. If that happens, the non-existence of the angles propagates towards the root of the current traversal, resulting in possibly many nodes of the tree for which no substantial information is computed during the traversal. This breaks the approach of having a pre-processing step that ensures that, for each node $s$ of the tree, the value we seek is computed for all the neighbors of $s$, except for one. Rather, we need to handle the possibility that, even when a node $s$ of the tree is visited in a traversal after the first one, we might not have the numerical values we seek available for all the children of $s$. One thing we can not afford, though, is to recurse on each subtree of $s$ every time we visit $s$, or even to just look at all the children of $s$ to see if the required angles have been computed successfully for them. Hence, we need to efficiently keep track of the neighbors of $s$ that have been visited, and of the visited neighbors of $s$ for which the computation of the required angles was successful. Luckily, in both Theorems~\ref{th:2-con-variable} and~\ref{th:main}, we can not have too many neighbors of a node for which the computation of the required angles was unsuccessful, as otherwise we can directly conclude that the sought representation does not exist. Then it suffices to keep track of a single neighbor $\eta(s)$ of each node $s$ for which we have already concluded that the required angles do not exist; hence, the label $\eta(s)$ allows us to avoid looking at the children of $s$ during a traversal of the tree if one of them is $\eta(s)$. Handling the neighbors of $s$ for which the required angles have been computed successfully only amounts to maintain a further label $\xi(s)$ in the proof of Theorem~\ref{th:main}, while it is more involved in the proof of Theorem~\ref{th:2-con-variable}. Namely, in the latter case, each successful computation leaves three neighbors of $s$ for which an optimal pair of angles might not have been established (even if it exists). Then we keep track of an interval $\mathcal I(s)$ of neighbors of $s$ for which such optimal pairs have been computed and update $\mathcal I(s)$  at each traversal that successfully computes a new optimal pair. We also keep track of some aggregate values (these are the labels $a(s)$, $b(s)$, $c(s)$, and $d(s)$) on the values of the angles of the computed optimal pairs; these are used to quickly check for the existence of the required angles after a re-rooting of the tree.

\section{Conclusions} \label{se:conclusions}

In this paper, we proved that the existence of a planar rectilinear drawing of an outerplanar graph can be tested in linear time, both if the plane embedding of the outerplanar graph is prescribed and if it is not; we also described how our algorithm can be adapted to test for the existence of an outerplanar rectilinear drawing. We believe that the tools we established to deal with the fixed and the variable embedding scenarios can be suitably merged in order to also deal with outerplanar graphs with a prescribed \emph{combinatorial embedding}, i.e., with a prescribed clockwise order of the edges incident to each vertex and with no prescribed outer face. Given the length of the paper, we decided to leave this for future research.

We conclude with two questions which are natural generalizations of the ones we answered in this paper.

First, is it possible to determine in $O(n)$ time the minimum number of bends needed for a planar orthogonal drawing of an $n$-vertex outerplanar graph $G$? This question seems challenging both in the fixed and in the variable embedding setting. A generalization of our approach would be feasible if a lemma akin to Lemma~\ref{le:structural-fixed-embedding} could be established. In particular, it is not clear to us whether the minimum number of bends in an orthogonal representation of $G$ in which a certain edge $uv$ is incident to the outer face can be determined  by only looking at a constant number of orthogonal representations for each $uv$-subgraph of $G$.

Second, is it possible to test in $O(n)$ time whether an $n$-vertex series-parallel graph admits a planar rectilinear drawing? The fastest known testing algorithms run in $O(n\log^3 n)$ and $O(n^3 \log n)$ time for the fixed~\cite{bkmnw-msms-17,t-eggmnb-87} and the variable~\cite{glm-sr-19} embedding setting, respectively; an $O(n)$-time algorithm for $n$-vertex $2$-connected series-parallel graphs with fixed embedding was recently presented by Didimo et al.~\cite{dklo-rpt-20}. We believe that the algorithm for outerplanar graphs from this paper could be extended to handle series-parallel graphs in which every parallel composition involves the edge between the poles; however, we are less confident that an extension of our approach to general series-parallel graphs is doable.

\subsubsection*{Acknowledgments} Thanks to Maurizio ``Titto'' Patrignani for many explanations about results and techniques from the state of the art.
\bibliographystyle{splncs03} 
\bibliography{bibliography}

\begin{thebibliography}{10}
\providecommand{\url}[1]{\texttt{#1}}
\providecommand{\urlprefix}{URL }

\bibitem{blr-odie-16}
Bl{\"{a}}sius, T., Lehmann, S., Rutter, I.: Orthogonal graph drawing with
  inflexible edges. Comp.\ Geom.  55,  26--40 (2016)

\bibitem{bkmnw-msms-17}
Borradaile, G., Klein, P.N., Mozes, S., Nussbaum, Y., Wulff{-}Nilsen, C.:
  Multiple-source multiple-sink maximum flow in directed planar graphs in
  near-linear time. {SIAM} J. Comput.  46(4),  1280--1303 (2017)

\bibitem{cy-bmod-17}
Chang, Y., Yen, H.: On bend-minimized orthogonal drawings of planar 3-graphs.
  In: Aronov, B., Katz, M.J. (eds.) {SoCG} '17. LIPIcs, vol.~77, pp.
  29:1--29:15. Schloss Dagstuhl - Leibniz-Zentrum f{\"{u}}r Informatik (2017)

\bibitem{cnao-lta-85}
Chiba, N., Nishizeki, T., Abe, S., Ozawa, T.: A linear algorithm for embedding
  planar graphs using {PQ}-trees. J. Comput. Syst. Sci.  30(1),  54--76 (1985)

\bibitem{ck-abm-12}
Cornelsen, S., Karrenbauer, A.: Accelerated bend minimization. J. Graph
  Algorithms Appl.  16(3),  635--650 (2012)

\bibitem{d-iroga-07}
Deng, T.: On the implementation and refinement of outerplanar graph algorithms.
  Master's thesis, University of Windsor, Ontario, Canada (2007)

\bibitem{dlv-sood-98}
{Di Battista}, G., Liotta, G., Vargiu, F.: Spirality and optimal orthogonal
  drawings. {SIAM} J. Comput.  27(6),  1764--1811 (1998)

\bibitem{glm-sr-19}
{Di Giacomo}, E., Liotta, G., Montecchiani, F.: Sketched representations and
  orthogonal planarity of bounded treewidth graphs. In: Archambault, D.,
  T{\'{o}}th, C.D. (eds.) {GD} '19. LNCS, vol. 11904, pp. 379--392. Springer
  (2019)

\bibitem{dklo-rpt-20}
Didimo, W., Kaufmann, M., Liotta, G., Ortali, G.: Rectilinear planarity testing
  of plane series-parallel graphs in linear time. CoRR  abs/2008.03784 (2020)

\bibitem{dlop-ood-20}
Didimo, W., Liotta, G., Ortali, G., Patrignani, M.: Optimal orthogonal drawings
  of planar 3-graphs in linear time. In: Chawla, S. (ed.) {SODA} '20. pp.
  806--825. {SIAM} (2020)

\bibitem{dlp-bmodqt-18}
Didimo, W., Liotta, G., Patrignani, M.: Bend-minimum orthogonal drawings in
  quadratic time. In: Biedl, T.C., Kerren, A. (eds.) {GD} '18. LNCS, vol.
  11282, pp. 481--494. Springer (2018)

\bibitem{DBLP:conf/gd/GargT96a}
Garg, A., Tamassia, R.: A new minimum cost flow algorithm with applications to
  graph drawing. In: North, S.C. (ed.) {GD} '96. LNCS, vol. 1190, pp. 201--216.
  Springer (1997)

\bibitem{gt-ccurpt-01}
Garg, A., Tamassia, R.: On the computational complexity of upward and
  rectilinear planarity testing. {SIAM} J. Comput.  31(2),  601--625 (2001)

\bibitem{h-gt-69}
Harary, F.: Graph Theory. Addison-Wesley Pub.\ Co.\, Reading, Massachusetts
  (1969)

\bibitem{hr-nbocd-19}
Hasan, M.M., Rahman, M.S.: No-bend orthogonal drawings and no-bend orthogonally
  convex drawings of planar graphs (extended abstract). In: Du, D., Duan, Z.,
  Tian, C. (eds.) {COCOON} '19. LNCS, vol. 11653, pp. 254--265. Springer (2019)

\bibitem{ht-aeagm-73}
Hopcroft, J.E., Tarjan, R.E.: Algorithm 447: efficient algorithms for graph
  manipulation. Communications of the {ACM}  16(6),  372--378 (1973)

\bibitem{ht-ept-74}
Hopcroft, J.E., Tarjan, R.E.: Efficient planarity testing. J. {ACM}  21(4),
  549--568 (1974)

\bibitem{m-laarogmog-79}
Mitchell, S.L.: Linear algorithms to recognize outerplanar and maximal
  outerplanar graphs. Inf. Process. Lett.  9(5),  229--232 (1979)

\bibitem{ntu-odog-05}
Nomura, K., Tayu, S., Ueno, S.: On the orthogonal drawing of outerplanar
  graphs. {IEICE} Transactions  88-A(6),  1583--1588 (2005)

\bibitem{ren-nb-06}
Rahman, M.S., Egi, N., Nishizeki, T.: No-bend orthogonal drawings of
  series-parallel graphs. In: Healy, P., Nikolov, N.S. (eds.) {GD} '05. LNCS,
  vol. 3843, pp. 409--420. Springer (2006)

\bibitem{rn-bmod-02}
Rahman, M.S., Nishizeki, T.: Bend-minimum orthogonal drawings of plane
  3-graphs. In: Kucera, L. (ed.) {WG} '02. LNCS, vol. 2573, pp. 367--378.
  Springer (2002)

\bibitem{rnn-odpwb-03}
Rahman, M.S., Nishizeki, T., Naznin, M.: Orthogonal drawings of plane graphs
  without bends. J. Graph Algorithms Appl.  7(4),  335--362 (2003)

\bibitem{st-ncmegpg-80}
Storer, J.A.: The node cost measure for embedding graphs on the planar grid
  (extended abstract). In: Miller, R.E., Ginsburg, S., Burkhard, W.A., Lipton,
  R.J. (eds.) {STOC} '80. pp. 201--210. {ACM} (1980)

\bibitem{t-eggmnb-87}
Tamassia, R.: On embedding a graph in the grid with the minimum number of
  bends. {SIAM} J. Comput.  16(3),  421--444 (1987)

\bibitem{w-rolt-87}
Wiegers, M.: Recognizing outerplanar graphs in linear time. In: Tinhofer, G.,
  Schmidt, G. (eds.) {WG} '86. LNCS, vol. 246, pp. 165--176. Springer (1987)

\bibitem{nz-odspg-08}
Zhou, X., Nishizeki, T.: Orthogonal drawings of series-parallel graphs with
  minimum bends. {SIAM} J. Discrete Math.  22(4),  1570--1604 (2008)

\end{thebibliography}

\end{document}